\newif\ifabstract
\abstracttrue
\newif\iffull
\ifabstract \fullfalse \else \fulltrue \fi

\documentclass[11pt,letterpaper]{article}
\usepackage{amsfonts}
\usepackage{amssymb}
\usepackage{amstext}
\usepackage{amsmath}
\usepackage{xspace}
\usepackage{ntheorem}
\usepackage{graphicx}

\usepackage{url}
\usepackage{graphics}
\usepackage{colordvi}
\usepackage[colorlinks,linkcolor=black,citecolor=black,filecolor=blue,urlcolor=blue]{hyperref}
\usepackage{hyperref}
\usepackage{subfigure}
\usepackage{xcolor}
\usepackage{cleveref}

\usepackage{thmtools}
\usepackage{thm-restate}
\usepackage{enumitem}

\advance \oddsidemargin by -0.8in
\newcommand{\myparskip}{3pt}
\parskip \myparskip

\newcommand{\val}{\operatorname{val}}

\newcommand\tO{\ensuremath{\tilde O}}

\newcommand{\td}{\tilde d}
\newcommand{\hd}{\operatorname{hd}}
\newcommand{\ld}{\operatorname{ld}}
\newcommand{\Vhd}{V^{\operatorname{hd}}}
\newcommand{\Vld}{V^{\operatorname{ld}}}

\newcommand{\Jld}{J^{\operatorname{ld}}}
\newcommand{\Nld}{N^{\operatorname{ld}}}
\newcommand{\Nhd}{N^{\operatorname{hd}}}

\newcommand{\algterm}{\ensuremath{\operatorname{AlgSelectTerminals}}\xspace}
\newcommand{\algpairs}{\ensuremath{\operatorname{AlgSelectPairs}}\xspace}

\newcommand{\alglocal}{\ensuremath{\operatorname{AlgLocalFlow}}\xspace}

\newcommand{\algtrip}{\ensuremath{\operatorname{AlgSelectTriple}}\xspace}
\newcommand{\algspec}{\ensuremath{\operatorname{AlgSpecial}}\xspace}

\newcommand{\algselectterm}{\ensuremath{\operatorname{AlgSelectTerminals}}\xspace}
\newcommand{\alg}{\ensuremath{\operatorname{Alg}}\xspace}

\newcommand{\ceil}[1]{\ensuremath{\left\lceil#1\right\rceil}}
\newcommand{\floor}[1]{\ensuremath{\left\lfloor#1\right\rfloor}}

\newcommand{\band}{\wedge}

\newcommand{\event}{{\cal{E}}}

\newcommand{\polylog}{\mathrm{polylog}}

\newcommand{\opt}{\mathsf{OPT}}

\newcommand{\set}[1]{\left\{ #1 \right\}}

\newcommand{\iset}{{\mathcal{I}}}
\newcommand{\pset}{{\mathcal{P}}}

\newcommand{\bset}{{\mathcal{B}}}
\newcommand{\aset}{{\mathcal{A}}}

\newcommand{\dset}{{\mathcal{D}}}

\newcommand{\be}{\begin{enumerate}}
	\newcommand{\ee}{\end{enumerate}}
\newcommand{\bd}{\begin{description}}
	\newcommand{\ed}{\end{description}}
\newcommand{\bi}{\begin{itemize}}
	\newcommand{\ei}{\end{itemize}}

\newtheorem{theorem}{Theorem}[section]

\newtheorem{lemma}[theorem]{Lemma}
\newtheorem{observation}[theorem]{Observation}
\newtheorem{corollary}[theorem]{Corollary}
\newtheorem{claim}[theorem]{Claim}

\newtheorem{definition}[theorem]{Definition}

\newenvironment{proof}{\par \smallskip{\bf Proof:}}{\hfill\stopproof}
\def\stopproof{\square}
\def\square{\vbox{\hrule height.2pt\hbox{\vrule width.2pt height5pt \kern5pt
			\vrule width.2pt} \hrule height.2pt}}

\newenvironment{proofof}[1]{\noindent{\bf Proof of #1.}}%
{\hfill\stopproof}

\newenvironment{prog}[1]{
	\begin{minipage}{5.8 in}
		\begin{center}
			{\sc #1}
		\end{center}
	}
	{
	\end{minipage}
}

\renewcommand{\phi}{\varphi}
\newcommand{\eps}{\epsilon}

\newcommand{\half}{\ensuremath{\frac{1}{2}}}

\newcommand{\poly}{\operatorname{poly}}
\newcommand{\dist}{\mbox{\sf dist}}

\newcommand{\reals}{{\mathbb R}}

\newcommand{\expect}[2][]{\text{\bf E}_{#1}\left [#2\right]}
\newcommand{\prob}[2][]{\text{\bf Pr}_{#1}\left [#2\right]}

\newenvironment{properties}[2][0]
{
\begin{enumerate} \setcounter{enumi}{#1}}{\end{enumerate}}

\newcommand{\edel}{E^{\operatorname{del}}}

\newcommand{\vol}{\operatorname{Vol}}

\newcommand{\AFMC}{{\sf AFMC}\xspace}

\newcommand{\wmax}{W_{\mbox{\textup{\footnotesize{max}}}}}

\newcommand{\OUT}{\mathsf{OUT}}
\newcommand{\IN}{\mathsf{IN}}
\newcommand{\reg}{\operatorname{reg}}
\newcommand{\spec}{\operatorname{spec}}
\newcommand{\inn}{\operatorname{in}}
\newcommand{\out}{\operatorname{out}}

\newcommand{\oracle}{\operatorname{oracle}}

\makeatletter
\newcommand{\customlabel}[2]{%
	\protected@write \@auxout {}{\string \newlabel {#1}{{#2}{\thepage}{#2}{#1}{}} }%
	\hypertarget{#1}{}
}
\makeatother

\makeatletter
\def\ifempty#1{%
	\def\@tmp@a{#1}%
	\ifx\@tmp@a\@empty%
}

\newtheoremstyle{ams-theorem}%
{\item[\hskip\labelsep \theorem@headerfont ##1\ ##2\theorem@separator]}%
{\item[\hskip\labelsep {\theorem@headerfont ##1\ ##2}{\normalfont\ (##3)}{\theorem@headerfont
		\theorem@separator}]}
\newtheoremstyle{ams-restatedtheorem}
{\item[\hskip\labelsep \theorem@headerfont ##1\ ##2\theorem@separator]}%
{\item[\hskip\labelsep {\theorem@headerfont ##1\ ##2}{\normalfont\ (##3)}{\theorem@headerfont
		\theorem@separator}]}
\newtheoremstyle{nonumberams-restatedtheorem}%
{\item[\theorem@headerfont \hskip\labelsep ##1\theorem@separator]}%
{\item[\hskip\labelsep \theorem@headerfont ##3\theorem@separator]}%

\begin{document}
	
\begin{titlepage}
	
	\title{Faster Algorithms for Global Minimum Vertex-Cut in Directed Graphs \footnote{A preliminary version is to appear in SODA 2026.}}

	\author{Julia Chuzhoy\thanks{Toyota Technological Institute at Chicago. Email: {\tt cjulia@ttic.edu}. Supported in part by NSF grant CCF-2402283 and NSF HDR TRIPODS award 2216899.}\and Ron Mosenzon\thanks{Toyota Technological Institute at Chicago. Email: {\tt ron.mosenzon@ttic.edu}. Supported in part by NSF grant CCF-2402283.} \and Ohad Trabelsi\thanks{University of Haifa, Israel. Email: {\tt otrabelsi@cs.haifa.ac.il}. Work mostly done while at Toyota Technological Institute at
Chicago.}}
	\date{}
	\maketitle
	\pagenumbering{gobble}
	
		\begin{abstract}
	
We study the directed global minimum vertex-cut problem: 
given a directed vertex-weighted graph $G$, compute a vertex-cut $(L,S,R)$ in $G$ of minimum value, which is defined to be the total weight of all vertices in $S$. The problem, together with its edge-based variant, is one of the most basic in graph theory and algorithms, and has been studied extensively.
The fastest currently known algorithm for directed global minimum vertex-cut (Henzinger, Rao and Gabow, FOCS 1996 and
J. Algorithms 2000) has running time $\tO(mn)$, where $m$ and $n$ denote the number of edges and vertices in the input graph, respectively. A long line of work over the past decades led to faster algorithms for other main versions of the problem, including the  undirected edge-based setting (Karger, STOC 1996 and J. ACM 2000),  directed edge-based setting (Cen et al., FOCS 2021), and  undirected vertex-based setting (Chuzhoy and Trabelsi, STOC 2025). However, for the vertex-based version in directed graphs, the 29 year-old $\tO(mn)$-time algorithm of  Henzinger, Rao and Gabow remains the state of the art to this day, in all edge-density regimes.

In this paper we break the $\Theta(mn)$ running time barrier for the first time, by 
providing a randomized algorithm  for directed global minimum vertex-cut, with running time $O\left(mn^{0.976}\cdot \poly\log W\right )$ where $W$ is the ratio of largest to smallest vertex weight. 
Our algorithm can also be viewed as improving and significantly simplifying the recent randomized $O(mn^{0.99+o(1)}\cdot\poly\log W)$-time algorithm for the undirected version of the problem (Chuzhoy and Trabelsi, STOC 2025).

Additionally, we provide a randomized  $O\left ( \min\set{m^{1+o(1)}\cdot k,n^{2+o(1)}}\right )$-time algorithm for the unweighted version of directed global minimum vertex-cut,
 where $k$ is the value of the optimal solution.
 The best previous algorithm for the problem achieved running time
 $\Tilde{O}\left (\min\set{k^2 \cdot m, mn^{11/12+o(1)}, n^{2+o(1)}}\right )$  (Forster et al., SODA 2020, Li et al., STOC 2021). Our result  almost matches the best current
 running time of $\tilde O\left(\min\set{k\cdot m,n^2}\right )$ for the directed unweighted edge-based version of the problem, due to (Gabow, STOC 1991 and J. Comput. Syst. Scie., 95) and (Chekuri, Quanrud, ICALP 2021).
\end{abstract}

	\newpage
\end{titlepage}

\tableofcontents

\newpage
\pagenumbering{arabic}

\section{Introduction}

Given a directed graph $G=(V,E)$ with integral weights $0\leq w(v)\leq W$ on its vertices $v\in V$, a \emph{vertex-cut} is a partition $(L,S,R)$ of $V$ into three disjoint subsets with $L,R\neq \emptyset$, such that no edge connects a vertex of $L$ to a vertex of $R$ (but edges connecting vertices of $R$ to vertices of $L$ are allowed). The \emph{value} of the cut is $w(S)$ -- the total weight of all vertices in $S$. In the global minimum vertex-cut problem, the goal is to compute a vertex-cut  of minimum value in the input directed vertex-weighted graph $G$. We also consider the \emph{unweighted} version of the problem, where the weights of all vertices are unit, and the value of the vertex-cut $(L,S,R)$ is $|S|$. In this paper we explore the design of fast algorithms for directed global minimum vertex-cut in both the unweighted and the general weighted settings.
A closely related and extensively studied variant of the problem is based on edge-cuts. In this variant, the weights are on the graph's edges instead of vertices, and an edge-cut is a bipartition $(L,R)$ of $V(G)$ into non-empty subsets. The \emph{value} of the cut is the total weight of all edges connecting vertices of $L$ to vertices of $R$ in the weighted version of the problem, or just the number $|E_G(L,R)|$ of such edges in the unweighted version. In the global minimum edge-cut problem, the goal is to compute an edge-cut of minimum value in the given input graph. Throughout, we use $m$ and $n$ to denote the number of edges and vertices, respectively, in the input graph $G$, and $W$ to denote the maximum vertex weight.

Edge- and vertex-cuts in graphs, together with the closely related notion of flows, are among the most fundamental graph-theoretic objects, that arise routinely in algorithms and structural results involving graphs. The unweighted version of global min-cut is closely related to another central graph-theoretic notion: vertex-connectivity (or edge-connectivity in the case of edge-cuts).  
It is then not surprising that the design of fast algorithms for global minimum cut has been studied extensively over many decades. 
 In fact one of the most well-known results in the area of graph optimization is the edge contraction based algorithm  of Karger and Stein~\cite{KS93}
for the edge version of the problem in undirected graphs.

Clearly, the value of the global minimum cut in a graph is equal to the smallest value, over all pairs $(s,t)$ of its vertices, of the minimum $s$-$t$ cut, in both the edge- and the vertex-versions of the problem. A central paradigm in designing efficient algorithms for global minimum cut is to identify a small number of vertex pairs $(s,t)$, so that computing a minimum $s$-$t$ cut for each such pair is sufficient in order to find the global min-cut. In order to optimize the running time, algorithms for computing a minimum cut separating a subset $S$ of vertices of $G$ from some vertex $t$, that we refer to as the minimum $S$-$t$ cut, are also sometimes utilized. The classical algorithm of Hao and Orlin~\cite{HO94} for global minimum edge-cut in directed graphs provides an ingenious implementation of this paradigm. Their algorithm uses a variant of the Push-Relabel scheme of~\cite{GT88} in order to compute a minimum $S$-$t$ cut for some initial set $S$ of vertices and a vertex $t\not\in S$. The inner state of the Push-Relabel algorithm is then exploited in order to identify a new set $S'$ of vertices and a new vertex $t'\not\in S'$, for which the minimum $S'$-$t'$ cut can be computed efficiently by building on the current inner state of the Push-Relabel algorithm. After $O(n)$ such iterations, the algorithm is guaranteed to terminate with a global minimum edge-cut, in time $\tO(mn)$.
Henzinger et al.~\cite{HRG00} extended this method to the vertex-weighted version of the problem in directed graphs, obtaining a randomized $\tilde O(mn)$-time algorithm for directed global minimum vertex-cut. This remains the fastest current algorithm for the problem in every edge-density regime, although its running time was very recently almost matched, to within an $n^{o(1)}$ factor, by a deterministic algorithm~\cite{JNS25}.

For the edge version of the problem, significantly faster algorithms have been known for some time. In undirected graphs,  Karger~\cite{Kar00}, improving upon the algorithm of~\cite{KS93} mentioned above, provided a randomized algorithm with close to the best possible running time of $\tO(m)$; a recent sequence of works~\cite{HRW17,KawT19,Li21,HLRW24} has led to a deterministic algorithm with a similar running time. In directed graphs, the best current algorithm for global minimum edge-cut, due to  Cen et al.~\cite{CLN21} solves the problem in the time of $\tO\left (\min\set{\frac{n}{m^{1/3}},\sqrt{n}}\right )$ applications of Maximum $s$-$t$ Flow, that, combined with the recent almost-linear time algorithm for the latter problem by \cite{CKLP22,BCP23}, yields a randomized algorithm with $O\left (\min\set{\frac{n}{m^{1/3}},\sqrt{n}}\cdot m^{1+o(1)}\right )$ running time for directed edge-weighted global minimum cut.
The edge version of the problem was also studied extensively in directed unweighted graphs \cite{ET75,Sch79,MS89,Gabow95,CQ21},
with the best current algorithms achieving running time $\tilde O(mk)$ where $k$ is the value of the optimal solution~\cite{Gabow95} and  $\tilde O(n^2)$ \cite{CQ21}.

Despite receiving a significant amount of attention starting from as early as the late 1960's
\cite{Kle69,Pod73,ET75,LLW88}, the progress on the vertex version of the problem has been slower, and the state of the art results are significantly weaker. One exception is the {\bf unweighted undirected}  version, where a recent series of results~\cite{NSY19,FNY20} culminated with an $O\left(m^{1+o(1)}\right)$-time randomized algorithm~\cite{LNP21}; this bound relies on the recent almost-linear time algorithm for directed Maximum $s$-$t$ Flow~\cite{CKLP22,BCP23}. A very recent work of \cite{JNS25} provides a deterministic algorithm for this version of the problem, with running time $O(m^{1+o(1)}\cdot k)$, where $k$ is the size of the global min-cut.
For {\bf weighted undirected} graphs, the $\Theta(mn)$ barrier of~\cite{HRG00} on the running time was broken only recently by \cite{CT24}, who obtained a randomized algorithm with running time $O(\min\{mn^{0.99+o(1)},m^{3/2+o(1)}\}\cdot\poly\log W)$.
For {\bf unweighted directed} graphs, the best current algorithm, due to~\cite{FNY20,LNP21}, runs in time $\Tilde{O}\left (\min\set{k^2 \cdot m, mn^{11/12+o(1)},n^{2+o(1)}}\right )$, when combined with the almost-linear time algorithm for Maximum $s$-$t$ Flow of~\cite{CKLP22,BCP23}. For {\bf weighted directed} graphs, the 29 year-old randomized algorithm of~\cite{HRG00} with $\tilde O(mn)$ running time remains the fastest known to this day in every edge-density regime; a very recent result of ~\cite{JNS25} provided a deterministic algorithm with an almost matching running time $O(mn^{1+o(1)})$. 

The directed vertex-cut version of the problem appears to be especially difficult, because some of the powerful techniques that were  pivotal to obtaining fast algorithms in other settings
 no longer seem to apply in this setting. This includes, for example, the \emph{tree packing} method of Tutte and Nash-Williams~\cite{Tut61, Nash61}, that is used by most algorithms for the edge-cut version, and the \emph{Isolating Cuts} lemma of \cite{LP20,AKT21_stoc}, that is used in algorithms for undirected vertex-weighted graphs. Additionally, a standard reduction from directed vertex-capacitated to directed edge-capacitated graphs via \emph{split graphs} (see the definition in Section \ref{subsec: split graph}), that works in the context of Maximum $s$-$t$ Flow, does not appear to work in the context of global minimum vertex-cut (see the discussion in the introduction of \cite{CT24}), and so it is unclear whether the recent faster algorithms for directed global minimum edge-cut of \cite{CLN21} can be leveraged for the vertex-cut version.

In this paper we break the $\Theta(mn)$ running time barrier for directed global minimum vertex-cut, by designing a randomized algorithm with running time $O\left(m n^{0.976}\cdot \poly\log W\right)$. Our algorithm can also be viewed as improving and simplifying the recent $O\left (\min\{mn^{0.99+o(1)},m^{3/2+o(1)}\}\cdot\poly\log W\right )$-time algorithm of \cite{CT24} for undirected global minimum vertex-cut.
Additionally, we provide a new randomized $O\left ( \min\set{m^{1+o(1)}\cdot k,n^{2+o(1)}}\right )$-time algorithm for the unweighted global minimum vertex-cut in directed graphs, where $k$ is the value of the optimal solution. This running time almost matches the $\tilde O\left(\min\set{k\cdot m,n^2}\right )$-time algorithm for the unweighted edge version of the problem in directed graphs of \cite{Gabow95,CQ21}. The best previous algorithm for the unweighted directed global minimum vertex-cut achieved running time $\Tilde{O}(\min\{k^2 \cdot m, mn^{11/12+o(1)}, n^{2+o(1)}\})$~\cite{FNY20,LNP21}.

In order to put our results for directed vertex-weighted graphs into context, we provide a very high-level overview of the recent $O(mn^{0.99+o(1)}\cdot\poly\log W)$-time algorithm of \cite{CT24} for the undirected version of the problem. At a very high level, \cite{CT24} provide two simple algorithms for two special cases of the problem respectively. The first special case is where there exists an optimal vertex-cut $(L,S,R)$ that is roughly balanced (that is, $|L|,|R|\geq n^{\eps}$ for some small constant $\eps$).
The algorithm for this special case is very simple and works as is in directed graphs; we employ it as well in this special case.
The second special case is where every global minimum vertex-cut is unbalanced but the graph is not too dense (e.g. $m\leq n^{1.98}$). Their algorithm for this special case is also simple, but it crucially relies on the powerful Isolating Cuts lemma of  \cite{LP20,AKT21_stoc}, that does not have known analogues in directed graphs. One of the main difficulties in designing a fast algorithm for directed global minimum vertex-cut is overcoming this obstacle, and we do so in this paper.
Finally, \cite{CT24} provide a very technically involved algorithm for the remaining special case, where the input graph is very dense, and every optimal vertex-cut is unbalanced. We provide an algorithm for this special case in directed graphs, that we believe is significantly simpler, and that achieves a faster running time. 
Overall, our technical contribution to the weighted version of the problem is twofold. First, we provide a randomized algorithm with running time $O\left(m n^{11/12+o(1)}\cdot d^{1/12}\cdot \poly\log W\right )$, where $d$ is the average vertex degree. This algorithm allows us to handle graphs that are not too dense, a task that was relatively easy for the undirected setting, but is much more challenging in directed graphs. This algorithm is completely new. Second, we provide an $O\left(n^{2.677}\cdot \poly\log W\right )$-time algorithm for directed global minimum vertex-cut, that can be used in the case where the input graph is dense. This algorithm builds on and generalizes the algorithm of \cite{CT24}  to the directed setting, while also improving its running time and significantly simplifying it.

\paragraph{Independent work.}
Independently from our work, \cite{FJMY} showed a simple reduction from directed to undirected vertex connectivity in dense graphs. In particular, their results lead to a much simpler $n^{2+o(1)}$-time algorithm for the directed unweighted Global Minimum Vertex-Cut, than the algorithm of \cite{LNP21}. They also obtain a subcubic-time algorithm, with running time $O(n^{2.99+o(1)})$ for weighted directed Global Minimum Vertex-Cut by combining this reduction with the algorithm of \cite{CT24} for undirected Global Minimum Vertex-Cut.

\paragraph{Other related work.}
A thorough and extensive overview of past results for global min-cut and the related techniques can be found in \cite{CT24}. We only mention here several additional recent developments on directed vertex-weighted global minimum cut, that include a $(1+\varepsilon)$-approximation algorithm with running time $\tO(n^2/\varepsilon^2)$ by~\cite{CLN21}, and an exact algorithm  with pseudo-polynomial running time $\tO(\opt \cdot n \cdot \sum_{u\in V} w(u))$ by ~\cite{CQ21}.

\subsection{Our Results and Techniques}

Note that it is possible that a directed graph $G$ contains no vertex-cuts, for example, if the out-degree of every vertex in $G$ is $n-1$. Given a graph $G$, we can check this in time $O(|E(G)|)$. Therefore, in the statements of our results, we assume that the input graph $G$ has a valid vertex-cut.

Our main result is a new algorithm for the global minimum vertex-cut problem in directed graphs, that is summarized in the following theorem.

\begin{theorem}\label{thm: main: weighted}
	There is a randomized algorithm, that, given a simple directed $n$-vertex and $m$-edge graph $G$ with integral weights $0\leq w(v)\leq W$ on its vertices $v\in V(G)$ that contains some vertex-cut, computes  a vertex-cut $(L',S',R')$ in $G$, such that, with probability at least $\half$,  $(L',S',R')$ is a global minimum vertex-cut. The running time of the algorithm is: 
	$$O\left(\min\set{m n^{11/12+o(1)}\cdot d^{1/12}, n^{2.677}}\cdot (\log W)^{O(1)}\right ) \leq O\left(m n^{0.976}\cdot (\log W)^{O(1)}\right),$$
	
	where $d$ is the average vertex degree in $G$.
\end{theorem}

The best previous randomized algorithm for the problem achieved running time $\tilde O(mn)$~\cite{HRG00}, and a deterministic  algorithm with an almost matching $O(mn^{1+o(1)})$ running time was recently provided by~\cite{JNS25}. We also note that the best current algorithm for the undirected version of the problem has running time  $O(\min\{mn^{0.99+o(1)},m^{3/2+o(1)}\}\cdot\poly\log W)$ \cite{CT24}; our result can be viewed as improving and simplifying this algorithm.

Our second result is a new algorithm for the unweighted directed global minimum vertex-cut problem, that is summarized in the following theorem.

\begin{theorem}\label{thm: main: unweighted}
	There is a randomized algorithm, that, given a simple directed unweighted $n$-vertex and $m$-edge graph $G$ that contains some vertex-cut, computes a vertex-cut $(L',S',R')$ in $G$, such that, with probability at least $\half$,  $(L',S',R')$ is a global minimum vertex-cut.
	 The expected running time of the algorithm is: $$O\left ( \min\set{m^{1+o(1)}\cdot k,n^{2+o(1)}}\right ),$$ where $k$ is the value of the global minimum vertex-cut in $G$.
\end{theorem}

The best previous algorithms for the unweighted version of the problem achieve running times $\tilde O(k^2m)$ \cite{FNY20}  and 
$O(\min\{mn^{11/12+o(1)}, n^{2+o(1)}\})$ \cite{LNP21} (the latter runtime bound relies on the recent almost linear time algorithm for maximum $s$-$t$ flow by \cite{CKLP22,BCP23}). Both these algorithms are randomized. 
We note that the running time of our algorithm almost matches the $\tilde O\left(\min\set{k\cdot m,n^2}\right )$-time bound of the algorithm for the unweighted edge version of the problem in directed graphs of \cite{Gabow95,CQ21}.

We now provide a high-level overview of our techniques, starting with the algorithm for the unweighted version of the problem, which is simpler.

\subsection*{Unweighted Directed Graphs: Proof of \Cref{thm: main: unweighted}}
We assume that we are given as input a directed unweighted $m$-edge and $n$-vertex graph $G$, that contains a vertex-cut, and we denote by $k$ the value of the global minimum vertex-cut in $G$. Our main technical subroutine is an algorithm that, given a ``guess'' $1\leq \hat k\leq n/10$ on the value of the global minimum vertex-cut in $G$, produces some vertex-cut $(\hat L,\hat S,\hat R)$ in $G$, in time $O\left (m^{1+o(1)}\cdot \hat k\right )$. The subroutine further guarantees that, if $\frac{\hat k}{2}\leq k\leq \hat k$, then, with a sufficiently high probability, $(\hat L,\hat S,\hat R)$ is a global minimum vertex-cut in $G$. We execute this subroutine in parallel for $O(\log n)$ ``guesses'' $\hat k$ that are integral powers of $2$ between $1$ and $n/10$, and combine it with the $O\left(n^{2+o(1)}\right )$-time randomized algorithm of \cite{LNP21} in a natural way, in order to complete the proof of \Cref{thm: main: unweighted}. From now on we focus on the description of this subroutine.

For the sake of the analysis, we fix an arbitrary global minimum vertex-cut $(L,S,R)$ in $G$ that we refer to as the \emph{distinguished cut}. We can assume without loss of generality that $|L|\leq |R|$ holds, since otherwise we can instead focus on solving the problem on the reversed graph $\overline{G}$, that is obtained from $G$ by reversing the direction of all its edges. We also assume without loss of generality that we are given a guess $\lambda$ on $\vol^+(L)$, that is an itegeral power of $2$, such that $\frac{\lambda}{2}\leq \vol^+(L)\leq \lambda$ holds; here, $\vol^+(L)$ is  the total out-degree of all vertices in $L$.  In fact, our algorithm will try all such values of $\lambda$ and then return the smallest-value vertex-cut computed using any such value $\lambda$.

A natural approach for solving the problem, that was also employed by~\cite{LNP21} in the context of undirected graphs, is as follows. We compute a set $T\subseteq V(G)$ of vertices called \emph{terminals}, by adding every vertex $v\in V(G)$ to $T$ independently with probability roughly $\frac{\deg^+(v)}{\lambda}$, where $\deg^+(v)$ is the out-degree of $v$. We also select an additional vertex $t\in V(G)$ uniformly at random. We say that a good event $\event$ happens if all of the following hold: (i) $t\in R$; (ii) $T\cap L\neq \emptyset$; and (iii) $|T|\leq \tilde O(m/\lambda)$. It is not hard to show that $\prob{\event}\geq \Omega\left(\frac{1}{\log n}\right )$; we assume for simplicity in the remainder of the exposition that Event $\event$ has happened. For every vertex $v\in T$, we can then compute the value $c_v$ of the minimum $v$-$t$ vertex-cut in $G$. Let $v^*\in T$ be the vertex for which $c_{v^*}$ is minimized. Then we  compute the minimum $v^*$-$t$ vertex-cut in $G$ using the almost linear-time algorithm for maximum $s$-$t$ flow of  \cite{CKLP22,BCP23}, and output this vertex-cut as the algorithm's outcome. Since we have assumed that $t\in R$ and that $T\cap L\neq \emptyset$, it is easy to see that the resulting cut must be a global minimum vertex-cut. The main challenge in employing this approach is that $|T|$ may be quite large, especially if $\vol^+(L)$ is small. Therefore, in order to achieve a running time that is below $\Theta(mn)$, we cannot afford to spend much time in order to compute the value of the minimum $v$-$t$ vertex-cut in $G$ for every vertex $v\in T$; we may need to perform each such computation in time that is significantly lower than $m$ -- the size of the input graph $G$. In order to overcome this difficulty, \cite{LNP21} design an ingenious and very efficient algorithm that, for every vertex $v\in T$, computes a \emph{sparsified} subgraph $G_v\subseteq G$, such that, if $v\in L$ and $t\in R$, then the value of the minimum $v$-$t$ cut in $G_v$ is preserved, and otherwise it does not decrease. Computing the minimum $v$-$t$ cut in each resulting graph $G_v$ is then sufficient in order to complete the algorithm. Unfortunately, their sparsification procedure appears tailored to the undirected setting, and it is not clear how to generalize it to directed graphs. To summarize, the main difficulty with employing the above approach in directed graphs is that, since $|T|$ may be very large, it is not clear how to compute the value of the minimum $v$-$t$ vertex-cut for all vertices $v\in T$ efficiently.

Our main idea for overcoming this difficulty can be thought of as \emph{committing to the terminals gradually}. Specifically, we use a parameter $z=O(\log n)$, and we compute a hierarchical partition $(\bset_0,\bset_1,\ldots,\bset_z)$ of the terminals. We ensure that $\bset_0=\set{T}$, and, for all $1\leq i\le z$, $\bset_i$ is       a partition of the set $T$ of terminals into roughly $2^i$ subsets of roughly equal cardinality; we refer to the subsets  $B\in \bset_i$ as \emph{level-$i$ batches}. We also ensure that every level-$z$ batch $B\in \bset_z$ contains exactly one terminal, and that, for all $1\leq i\leq z$, partition $\bset_i$ is a refinement of the partition $\bset_{i-1}$ of $T$, so that every level-$i$ batch is contained in some level-$(i-1)$ batch. At a high level, our algorithm consists of $z$ phases. For all $1\leq i\leq z$, in the $i$th phase, we compute, for every level-$i$ batch $
B\in \bset_i$, a ''sparsified'' graph $G_B$, that has the following properties. First, $V(G_B)\subseteq V(G)$, $t\in V(G_B)$,  and $|E(G_B)|\leq \tilde O\left(\frac{m\cdot \hat k}{2^i}\right )$ must hold. Additionally, we ensure that, for every vertex-cut $(L',S',R')$ in $G_B$ with $t\in R'$, $(L',S',V(G)\setminus (L'\cup S'))$ is a valid vertex-cut in $G$. Finally, we ensure that, if $B\cap L\neq \emptyset$, then $L\cup S\subseteq V(G_B)$, and, moreover, $(L,S,V(G_B)\setminus(L\cup S))$ is a valid vertex-cut in $G_B$. 

Recall that every level-$z$ batch $B\in \bset_z$ contains a single terminal. Since we have assumed that $T\cap L\neq \emptyset$, there must be a level-$z$ batch $B^*\in \bset_z$, whose unique terminal $v^*\in B^*$ lies in $L$. We are then guaranteed that there is a vertex-cut $(L',S',R')$ in $G_{B^*}$ with $v^*\in L'$ and $t\in R'$, whose value is $k$; this is precisely the cut $(L,S,V(G_B)\setminus(L\cup S))$. For every level-$z$ batch $B\in \bset_z$, we compute a minimum vertex-cut in $G_B$ separating the unique terminal of $B$ from $t$, and we let $(L'',S'',R'')$ be the smallest-value resulting cut, whose value must be $k$ from the above discussion. We then output the vertex-cut $(L'',S'',V(G)\setminus(L''\cup S''))$ in $G$ as the algorithm's outcome. Note that our algorithm ensures that, for every level $0\leq i\leq z$, the total size of all graphs in $\set{G_B\mid B\in \bset_i}$ is bounded by $\tilde O(m\cdot \hat k)$. In particular, we can perform a single $s$-$t$ vertex-cut computation in every graph $G_B$ for $B\in \bset_z$, in total time 
$ O\left (m^{1+o(1)}\cdot \hat k\right )$.

We now provide additional technical details of our algorithm. Recall that there is a single level-$0$ batch of terminals, $B=T$. We let the corresponding graph $G_B$ be $G$. It is immediate to verify that this graph has all required properties. Consider now some level $1\le i<z$, and some level-$i$ batch $B\in \bset_i$. Assume that we have computed the corresponding graph $G_B$ with the required properties. For every level-$(i+1)$ batch $B'\subseteq B$, we construct the corresponding graph $G_{B'}$ as follows. First, we carefully compute a vertex-cut $(L',S',R')$ in $G_B$ with $t\in R'$, such that $|S'|\leq O\left (\frac{\hat k\cdot |T|}{2^i}\right )$, and, moreover, if $L\cap B'\neq \emptyset$, then $L\subseteq L'$ holds. This cut is computed by performing a single $s$-$t$ vertex-cut computation in a graph $\hat G_B$ that is obtained by a slight modification of the graph $G_B$. We then let the new graph $G_{B'}$ contain the vertex set $L'\cup N^+_G(L')\cup\set{t}$. For every vertex $v\in L'$, we include all edges of $\delta^+_G(v)$ in $G_{B'}$, and for every vertex $v'\in N^+_G(L')$, we include the edge $(v',t)$. Recall that, if $B'\cap L\ne \emptyset$, then our algorithm guarantees that $L\subseteq L'$ must hold. In this case, it is immediate to verify that $L\cup S\subseteq V(G_B)$, and, moreover, that $(L,S,V(G_B)\setminus(L\cup S))$ is a valid vertex-cut in $G_{B'}$. It is also not hard to see that, for every vertex-cut $(L'',S'',R'')$ in $G_{B'}$ with $t\in R''$,  $(L'',S'',V(G)\setminus(L''\cup S''))$ is a valid vertex-cut in $G$. The main challenge however is to ensure that $|E(G_{B'})|$ is sufficiently small; since $|E(G_{B'})|=\Theta(\vol^+_G(L'))$, equivalently, it is enough to ensure that $\vol^+_G(L')$ is sufficiently small.

In order to overcome this difficulty, at the beginning of the algorithm, in addition to selecting the set $T$ of terminals, we also select a set $T'\subseteq V(G)$ of vertices that we refer to as \emph{anti-terminals}, using a similar procedure. We say that a good event $\event'$ happens if Event $\event$ happened, and, additionally, $T'\cap L=\emptyset$. It is not hard to show that $\prob{\event'}\geq \Omega(1/\poly\log n)$. Whenever we process a level-$i$ batch $B$ and its level-$(i+1)$ child-batch $B'\subseteq B$ as described above, we will ensure that the vertex-cut $(L',S',R')$ in $G_B$ is selected so that $L'\cap T'=\emptyset$. 
Recall that the cut $(L',S',R')$ defines a valid vertex-cut $(L',S',R'')$ in graph $G$, where $R''=V(G)\setminus (L'\cup S')$.
The key to the analysis of our algorithm is to show that, for every vertex-cut $(L',S',R'')$ in $G$ that may arise via this process, $\vol^+_G(L')$ must be sufficiently small. This property follows from the selection of the set $T'$ of the anti-terminals, and since, for each vertex-cut $(L',S',R'')$ in $G$ that may arise over the course of the algorithm, $L'\cap T'=\emptyset$ must hold. 

\subsection*{Weighted Graphs: Proof of \Cref{thm: main: weighted}}

We now provide a high-level description of our algorithm for weighted graphs, that is used in the proof of \Cref{thm: main: weighted}.
Recall that we are given as input a directed $n$-vertex and $m$-edge graph $G=(V,E)$ with integral weights $0\leq w(v)\leq W$ on its vertices $v\in V$; in fact, using a simple standard transformation described in \Cref{subsec: positive weights}, we can assume that $w(v)\geq 1$ for all $v\in V$, which we do in the remainder of this exposition. For a subset $X\subseteq V$ of vertices, we denote by $w(X)=\sum_{v\in X}w(v)$ its \emph{weight}, and by $\vol(X)=\sum_{v\in X}\deg_G(v)$ its \emph{volume}. We also denote by $\vol^+(X)$ the total out-degree of all vertices in $X$. As before, we can assume without loss of generality that there is a global minimum vertex-cut $(L,S,R)$ in $G$ with $w(L)\leq w(R)$; if this is not the case, we can equivalently solve the problem on the reversed graph $\overline{G}$, obtained from $G$ by reversing the directions of all is edges.
As before, for the sake of the analysis, we fix  a single global minimum vertex-cut $(L,S,R)$, that we refer to as the \emph{distinguished cut} in the remainder of this exposition, to be any minimum vertex-cut with $w(L)\leq w(R)$. 
Throughout, we denote by $d$ the average vertex degree in $G$ and by $\opt$ the value of the global minimum vertex-cut in $G$.
 We design two separate algorithms: one, with running time $O\left(m n^{11/12+o(1)}\cdot d^{1/12}\cdot \poly\log W\right )$, that will be employed when the input graph $G$ is not too dense, and one with running time $O\left( n^{2.677}\cdot \poly\log W\right )$ that will be used in dense graphs. The proof of \Cref{thm: main: weighted} follows immediately by combining these two algorithms. We now describe each of these two algorithms in turn.

\subsection*{Weighted Non-Dense Graphs}

We now describe our  $O\left(m n^{11/12+o(1)}\cdot d^{1/12}\cdot \poly\log W\right )$-time algorithm, which is most useful when the input graph $G$ is not too dense.
We note that \cite{CT24} provide a rather simple algorithm for the undirected setting that beats the $\Theta(mn)$ running time barrier if the input graph is not too dense. Their algorithm relies on the Isolating Cuts lemma of \cite{LP20,AKT21_stoc}. Unfortunately, this technique is not known to extend to directed graph, and so designing an algorithm with running time below $\Theta(mn)$ for directed non-dense graphs appears to be a challenge in its own right.

We use a parameter $0\leq \eps\leq 1$, that will be set to $\frac{1}{12}-\frac{\log d}{12\log n}$. As in much of previous work, we design a separate algorithm to deal with several simple special cases, that we discuss next.

\paragraph{Simple special cases.}
Consider the distinguished min-cut $(L,S,R)$. We say that it is \emph{balanced} if  either $|L|\geq n^{\eps}$ or $\vol(L)\geq d\cdot n^{\eps}$ hold, and we say that it is \emph{unbalanced} otherwise. For the case where cut $(L,S,R)$ is balanced, it is easy to obtain a randomized algorithm with  $O(mn^{1-\eps+o(1)}\cdot \log W)$ running time, using techniques similar to those employed in previous work, such as, e.g. \cite{CT24}. For example, assume that  $|L|\geq n^{\eps}$, and consider selecting a vertex $s$ from $V(G)$ uniformly at random, and then selecting another vertex $t$ from $V(G)\setminus \left (\set{s}\cup N^+(s)\right )$ with probability proportional to its weight. Then it is not hard to show that, with probability at least $\Omega\left(\frac{1}{n^{1-\eps}}\right )$, $s\in L$ and $t\in R$ must hold. We can then compute a minimum $s$-$t$ vertex-cut in $G$ using the almost linear-time algorithm of \cite{CKLP22,BCP23}. By repeating this process $\Theta(n^{1-\eps}\cdot \log n)$ times and returning the smallest-value cut among the resulting vertex-cuts, we are guaranteed to obtain a global minimum vertex-cut with probability at least $0.99$; the running time of this algorithm is $O(mn^{1-\eps+o(1)}\cdot \log W)$. In the case where 
$\vol(L)\geq d\cdot n^{\eps}$, we can employ a similar algorithm, except that we select the vertex $s$ slightly differently: we first select an edge $e$ from $E(G)$ uniformly at random, and then let $s$ be its random endpoint. The remainder of the algorithm remains unchanged, and is guaranteed to return a global minimum vertex-cut with probability at least $0.99$, in time $O(mn^{1-\eps+o(1)}\cdot \log W)$.  To summarize, the special case where the distinguished min-cut is balanced is easy to solve using known techniques, so from now on we focus on the case where $(L,S,R)$ is unbalanced, so $|L|\leq n^{\eps}$ and  $\vol(L)\leq d\cdot n^{\eps}$ hold. Since every vertex of $S$ must have a neighbor in $L$ (as the vertex-cut $(L,S,R)$ could be improved otherwise), we get that $|S|\leq d\cdot n^{\eps}$ must hold as well in this case. It will be convenient for us to assume that the algorithm is given an approximate estimate $\lambda$ (to within factor $2$) on $|L|$. In our actual algorithm we will try all $O(\log n)$ such possible estimates that are integral powers of $2$. Clearly, it is enough to ensure that the vertex-cut that the algorithm returns is optimal with a sufficiently high probability, if the estimate on $|L|$ that it receives is approximately correct. Therefore, we assume from now on that the algorithm is given a value $\lambda$ with $\lambda/2\leq |L|\leq \lambda$, and that the distinguished cut is unbalanced, so $|L|\leq n^{\eps}$, $\vol(L)\leq n^{\eps}\cdot d$ and $|S|\leq n^{\eps}\cdot d$ hold. 
We say that a vertex $v\in V$ is a \emph{low-degree} vertex if $\deg(v)\leq d\cdot n^{\eps}$, and we say that it is a \emph{high-degree} vertex otherwise. Note that, from the above discussion, $L$ may  not contain high-degree vertices. 
 We now describe the main tools that our algorithm uses, that include a critical threshold, $\beta$-sets, $\sigma$-sets, and suspicious vertices.

\paragraph{Critical threshold.}
One of the central notions that our algorithm uses is that of a \emph{critical threshold}. Consider some value $\tau\geq 1$ that is an integral power of $2$. We say that $\tau$ is the \emph{critical threshold} for the distinguished cut $(L,S,R)$, if every {\bf low-degree} vertex $v\in S$ has $w(v)\leq \tau$, and, moreover, $\tau$ is the smallest integral power of $2$ with this property. It will be convenient for us to assume that our algorithm is given as input the value $\tau$ of the critical threshold for $(L,S,R)$. In fact our algorithm will try all $O(\log(W))$ such values, and it is sufficient for us to ensure that the vertex-cut that the algorithm returns is optimal (with a sufficiently high probability), if the guessed value $\tau$ is indeed equal to the critical threshold. Therefore, we assume from now on that the critical threshold $\tau$ is known to the algorithm.

\paragraph{$\beta$-sets.}
The second main tool that we use is vertex sets $\beta(v)$ for all $v\in V(G)$, that are constructed as follows. Consider a vertex $v\in V(G)$, and consider performing a DFS search in $G$ starting from $v$, while only exploring low-degree vertices whose weight is greater than $\tau$. Once the DFS search discovers $|L|+1$ vertices, we terminate it (unless it terminated earlier), and we let $\beta(v)$ be the set of all vertices that the DFS has discovered, including $v$. A central observation is that, if $v\in L$, then $\beta(v)\subseteq L$ must hold. Indeed, otherwise, $\beta(v)$ must contain a vertex of $S$. But since all vertices in $\beta(v)\setminus \set{v}$ are low-degree vertices of weight greater than $\tau$, from the definition of the critical threshold this is impossible. Therefore, if $v\in L$, then $\beta(v)\subseteq L$ and $N^+(\beta(v))\subseteq L\cup S$ must hold. We will use this fact later.

\paragraph{$\sigma$-sets and suspicious vertices.} For every vertex $u\in V(G)$, we then define the vertex set $\sigma(u)=\set{v\in V(G)\mid u\in \beta(v)}$. Note that, from the definition, $\sum_{v\in V(G)}|\beta(v)|\leq 2n|L|$, so on average, $|\sigma(u)|\leq 2|L|$ must hold. We say that a vertex $u$ is \emph{suspicious}, if $|\sigma(u)|>2|L|\cdot n^{\eps}$, and we denote by $U$ the set of all suspicious vertices in $G$; it is easy to verify that $|U|\leq n^{1-\eps}$. Intuitively, the algorithm that we present below for weighted non-dense graphs only works efficiently if no vertex of $L$ is suspicious. However, the special case where $U\cap L\neq \emptyset$ is easy to deal with, via an algorithm that is very similar to the one that we used for the case where the cut $(L,S,R)$ is balanced: we consider all suspicious vertices one by one. For each such vertex 
  $s\in U$, we select a vertex $t$ from $V(G)$ using a simple random process that ensures that, if $s\in L$, then, with a constant probability, $t\in R$ must hold. We then compute a minimum $s$-$t$ vertex-cut in $G$ using the almost linear-time algorithm of \cite{CKLP22,BCP23}. Finally, we return the smallest-value vertex-cut among all resulting cuts. Like in the algorithm for the case where the cut $(L,S,R)$ is balanced, we are guaranteed to obtain a global minimum vertex-cut with a sufficiently  high probability, in time $O(mn^{1-\eps+o(1)}\cdot \log W)$. Therefore, we assume from now on that no vertex of $L$ is suspicious. In particular, if $v\in L$, then no vertex of $\beta(v)$ may be suspicious, since, as observed already, $\beta(v)\subseteq L$ must hold.

We are now ready to describe our algorithm for the weighted non-dense setting. 
We select a set $\Gamma$ of $\Theta(n/|L|)$ pairs of vertices of $G$ using a simple randomized procedure that ensures that, with a sufficiently high probability, there is some pair $(x^*,y^*)\in \Gamma$, that we refer to as the \emph{distinguished pair}, with $x^*\in L$ and $y^*\in R$. We can also ensure that, for every pair $(x,y)\in \Gamma$, $x$ is a low-degree vertex and $(x,y)\not\in E(G)$.
In the remainder of the algorithm, our goal is to compute, for every pair $(x,y)\in \Gamma$, a value $c_{x,y}$ that is at least as high as the value of the minimum $x$-$y$ vertex-cut in $G$; additionally, if $(x,y)\in \Gamma$ is the distinguished pair, then we need to ensure that $c_{x,y}$ is equal to the value  of the minimum $x$-$y$ vertex-cut in $G$. Once we compute the values $c_{x,y}$ for all pairs $(x,y)\in \Gamma$, we will select a pair $(x',y')\in \Gamma$ for which the value $c_{x',y'}$ is the smallest, and compute the minimum $x'$-$y'$ vertex-cut in $G$, that is guaranteed to be a global minimum vertex-cut with a sufficiently high probability.

 Consider now any such pair $(x,y)\in \Gamma$. In order to compute the value $c_{x,y}$ with the desired properties, our main strategy is to compute a relatively small subset $A_{x,y}\subseteq V(G)$ of vertices, such that all vertices in $A_{x,y}$ are low-degree, $x\in A_{x,y}$ and $y\not\in A_{x,y}\cup N^+(A_{x,y})$. Additionally, we will ensure that, if $(x,y)$ is the distinguished pair, then $L\subseteq A_{x,y}$ holds. Assume first that we can indeed compute the vertex set $A_{x,y}$ with these properties, and let $A'_{x,y}\subseteq V(G)$ contain all vertices $v\in V(G)\setminus A_{x,y}$, such that there is an edge connecting some vertex of $A_{x,y}$ to $v$ in $G$. Note that, if $(x,y)$ is the distinguished pair, then $L\cup S\subseteq A_{x,y}\cup A'_{x,y}$. 
 Consider now the graph $\hat G_{x,y}$, whose vertex set is $A_{x,y}\cup A'_{x,y}\cup \set{\tilde t}$, and the weights of the vertices in $A_{x,y}\cup A'_{x,y}$ remain the same as in $G$. For every vertex $v\in A_{x,y}$, we include in $\hat G_{x,y}$ all edges leaving $v$ in $G$. For every vertex $u\in A'_{x,y}$, we include the edge $(u,\tilde t)$ in $\hat G_{x,y}$.
Since $|A_{x,y}|$ is relatively small, and all vertices in $A_{x,y}$ are low-degree vertices, $|E(\hat G)|$ is relatively small as well. We can then afford to compute a minimum $x$-$\tilde t$ vertex-cut $(\hat L,\hat S,\hat R)$ in $\hat G_{x,y}$. It is not hard to verify that the tripartition $(L',S',R')$ of vertices of $G$, where $L'=\hat L$, $S'=\hat S$ and $R'=V(G)\setminus (L'\cup S')$, defines a valid $x$-$y$ vertex-cut in $\hat G_{x,y}$. Moreover,  if $(x,y)$ is the distinguished pair and $L\subseteq A_{x,y}$ holds, then $(L',S',R')$ must be a global minimum vertex-cut. We then return the value $c_{x,y}=w(\hat S)$.
The main challenge with this approach is to compute the desired set $A_{x,y}$ of vertices efficiently. We now describe our algorithm for doing so, for a fixed pair $(x,y)\in \Gamma$ of vertices. Throughout, we use a parameter $\tau'=\frac{\tau}{64|L|^2}$, and we denote by $\hat S\subseteq S$ the set containing all low-degree vertices $v\in S$ with $w(v)\geq \tau'$ (that is not known to the algorithm).
We design two different algorithms for computing the set $A_{x,y}$ of vertices, that separately deal with the cases where $|\hat S|$ is large or small.

\paragraph{Case 1: $|\hat S|\geq 2^{12}|L|^3$.}
For every vertex $v\in V(G)$, we define a set $J(v)\subseteq V(G)$, that contains all low-degree vertices of $N^+_G(\beta(v))$ whose weight is at least $\tau'$. Recall that we have established that, if $v\in L$, then $N^+_G(\beta(v))\subseteq L\cup S$. It is then not hard to see that $J(v)\subseteq L\cup \hat S$, and, moreover, $J(v)$ must contain the vast majority of the vertices of $\hat S$. We then let $A_{x,y}$ contain all low-degree vertices $v\in V(G)$, for which $|J(v)\triangle J(x)|\leq 256|L|^3$ holds. We show that, if $x\in L$, then $L\subseteq A_{x,y}$ must hold. We provide an efficient randomized algorithm for constructing the set $A_{x,y}$ of vertices, that essentially samples $O(\log n)$ vertices from $J(x)$ uniformly at random, and, for each such sampled vertex $v$, adds to $A_{x,y}$ all low-degree vertices in the sets $\set{\sigma(u)\mid u\in N^-_G(v)\setminus U}$. In order to ensure that the algorithm is efficient and that the resulting set $A_{x,y}$ of vertices is small, we crucially rely on the assumption that no vertex of $L$ is suspicious. 

This provides an efficient algorithm for constructing the set $A_{x,y}$ of vertices with the desired properties in the case where $|\hat S|\geq 2^{12}\cdot |L|^3$. The main challenge in designing an algorithm for the case where $|\hat S|$ is small is that, if we use the same strategy as above, we can no longer guarantee that $|A_{x,y}|$ is small. Specifically, if $|\hat S|$ is sufficiently large (like in Case 1), then the vast majority of the vertices of $J(x)$ must lie in $\hat S$, and, from the definition of the set $A_{x,y}$ of vertices, for every vertex $v\in A_{x,y}$, some vertex in $J(v)$ must lie in $\hat S$ (and in fact a large fraction of the vertices of $\hat S$ must lie in $J(v)$). This allows us to bound the cardinality of $A_{x,y}$ by the total in-degree of all vertices in $\hat S$ times $O(n^{\eps}\cdot |L|)$ (the bound on  $|\sigma(u)|$ for non-suspicious vertices $u$). Assume now that $|\hat S|$ is small, for example, that $|\hat S|\ll |L|$. In this case, the vast majority of the vertices of $J(x)$ may lie in $L$. In order to ensure that $L\subseteq A_{x,y}$, we need to allow $A_{x,y}$ to contain all vertices $v\in V(G)$ with $|J(v)\triangle J(x)|\leq 2|L|$, and in particular $A_{x,y}$ may contain vertices $v$ with $J(v)\cap J(x)=\emptyset$. But then we can no longer guarantee that, for every vertex $v\in A_{x,y}$, $J(v)\cap \hat S\neq\emptyset$, and so we no longer have a way to bound $|A_{x,y}|$ appropriately. In view of these difficulties, we employ a completely different strategy that exploits the fact that $|\hat S|$ is small directly.

\paragraph{Case 2: $|\hat S|< 2^{12}|L|^3$.}
Our algorithm for this case is more involved. We use a parameter $\hat \tau=\frac{\tau}{8|L|}$, and slightly modify the definition of the sets $\beta(v)$ of vertices: for every vertex $v\in V(G)$, in order to construct the set $\beta(v)$, we perform the DFS search in graph $G$ starting from $v$ as before, except that now we only explore low-degree vertices of weight greater than $\hat \tau$ (instead of $\tau$). 
For every vertex $u\in V(G)$, we define the vertex set $\sigma(u)=\set{v\in V(G)\mid u\in \beta(v)}$ exactly as before (except that we use the new definition of the vertex sets $\beta(v)$), and we define the set $U=\set{u\in V(G)\mid |\sigma(u)|\geq 2n^{\eps}\cdot |L|}$ of suspicious vertices exactly as before. The special case where some vertex in $L$ is suspicious is also dealt with exactly as before, so we assume from now on that $U\cap L\neq \emptyset$.

We introduce a new notion of a \emph{promising vertex set}, that is defined as follows. A set $Z\subseteq V(G)$ is a promising vertex set, if every vertex in $Z$ is a low-degree vertex, and, additionally, every low-degree vertex $v\in S$ of weight $w(v)\geq \hat \tau$ lies in $Z$. 

Consider now some vertex $v\in L$. With the new modified definition of the vertex set $\beta(v)$, we are no longer guaranteed that $\beta(v)\subseteq L$. However, we show that $\beta(v)$ has another crucial property that our algorithm can exploit: if $Z$ is any  promising vertex set, then $\beta(v)$ must contain a low-degree non-suspicious vertex $u'$ that is connected with an edge to some vertex of $Z$.

The key subroutine of our algorithm constructs a relatively small subset $Z$ of vertices of $G$, such that, if $(x,y)$ is the distinguished pair, then, with a sufficiently high probability, $Z$ is a promising vertex set. The desired set $A_{x,y}$ of vertices is then obtained by including, for every vertex $u\in Z$, for every low-degree non-suspicious in-neighbor $u'\in N^-(u)$, all vertices of $\sigma(u')$ in $A_{x,y}$. From the above discussion, if $Z$ is a promising vertex set, then $L\subseteq A_{x,y}$ must hold.

Lastly, we need to provide an algorithm for our key subroutine, that computes a relatively small set $Z\subseteq V(G)$ of vertices, such that, if $(x,y)$ is the distinguished pair, then, with a sufficiently high probability, $Z$ is a promising vertex set. In order to construct such a vertex set $Z$, we employ the local flow augmentation technique, that was first introduced in~\cite{CHILP17} and then further refined in~\cite{FNY20,CQ21,CT24}.

In order to utilize the local flow augmentation technique, we use the ``split graph'' $G'$ of the input graph $G$, that can be thought of as the edge-capacitated version of $G$: the set $V(G')$ of its vertices contains two copies of every vertex $v\in V(G)$, an \emph{out-copy}, denoted by $v^{\out}$ and an \emph{in-copy}, denoted by $v^{\inn}$. For every vertex $v\in V(G)$, graph $G'$ contains an edge $e_v=(v^{\inn},v^{\out})$ of capacity $c(e_v)=w(v)$, that we refer to as the \emph{special edge representing $v$}. Additionally, for every edge $e=(u,v)\in E(G)$, graph $G'$ contains an edge $(u^{\out},v^{\inn})$ of capacity $\wmax=\Theta(n \cdot W)$, that we refer to as the \emph{regular edge representing $e$}. Note that the value of the minimum $x$-$y$ vertex-cut in $G$ is equal to the value of the minimum $x^{\out}$-$y^{\inn}$ edge-cut in $G'$, which, in turn, from the Max-Flow / Min-Cut theorem, is equal to the value of the maximum $x^{\out}$-$y^{\inn}$ flow in $G'$.

We use a graph $G'_{x,y}$, that is initially identical to $G'$, and we denote $x^{\out}$ by $s$ and $y^{\out}$ by $t$. Recall that, if $v\in V(G)$ is a high-degree vertex, then $v$ may not lie in $L$. Therefore, for every high-degree vertex $v\in V(G)$, we can add a ''shortcut'' edge $(v^{\out},t)$ of capacity $\infty$ to $G'_{x,y}$, without increasing the value of the maximum $s$-$t$ flow in $G'_{x,y}$. We then let $G^h\subseteq G'_{x,y}$ be the graph that is obtained from $G'_{x,y}$ by deleting from it the copies of all vertices $v\in V(G)$ with $w(v)<\tau'$ (recall that $\tau'=\frac{\tau}{64|L|^2}$). Let $\opt_{x,y}$ denote the value of the maximum $s$-$t$ flow in graph $G^h$. 

Our algorithm starts by computing an initial $s$-$t$ flow $f$ in $G^h$, whose value is close to $\opt_{x,y}$: $\val(f)\geq \opt_{x,y}-2|\hat S|\cdot \tau$, such that,  for every vertex $v\in \hat S$, $f(v^{\inn},v^{\out})=0$ holds, via a simple procedure. Then we gradually augment the flow $f$ via the local flow augmentation paradigm, over the course of $O(|L|^2\cdot |\hat S|)$ iterations, where in every iteration, we send $\tau'$ flow units over a single augmenting flow-path in the residual flow network $H$ of $G^h$. The main difference from the standard local flow augmentation algorithm  is that we only employ augmenting paths of length at most $16|L|$, and we terminate the algorithm once the total number of vertices in $\set{v^{\out}\mid v\in V(G)}$ reachable from $s$ via paths of length at most $16|L|$ in the residual flow network $H$ becomes sufficiently small. We then add to $Z$ every low-degree vertex $v\in V(G)$, such that either (i) $v^{\out}$ is reachable from $s$ via a path of length at most $16|L|$ in the residual flow network $H$ obtained at the end of the local flow augmentation algorithm; or (ii) edge $(v^{\inn},v^{\out})$ lied on at least one of the augmenting paths employed by the local flow augmentation algorithm. 

We prove that, if $(x,y)$ is the distinguished pair, then, with  a sufficiently high probability, the resulting set $Z$ of vertices must be promising. The intuitive reason for this is as follows. Assume that $(x,y)$ is the distinguished pair, and let $v\in S$ be a low-degree vertex with $w(v)\geq \hat \tau$. Then $v\in \hat S$ must hold, so, in the initial flow $f$, $f(v^{\inn},v^{\out})=0$ must hold. If the local flow augmentation algorithm ever sent flow on the edge $(v^{\inn},v^{\out})$, then $v$ is added to $Z$, since then the edge $(v^{\inn},v^{\out})$ must have lied on one of the augmenting paths. Otherwise, at the end of the local flow augmentation algorithm, $f(v^{\inn},v^{\out})=0$ holds, and, from the properties of Maximum Flow, since $v\in S$, we must be able to send at least $w(v)\geq \hat \tau$ flow units from $s$ to $v^{\inn}$, in the residual flow network $\tilde H$ of $G'_{x,y}$ with respect to the current flow $f$, such that this new flow only utilizes vertices of $G'_{x,y}$ that are out-copies of the vertices of $L$ or in-copies of the vertices of $L \cup S$. It is then not hard to show that the residual flow network $H$ of $G^h$ with respect to the current flow $f$ must contain an $s$-$v^{\inn}$ path of length at most $16|L|-1$, and so $v\in Z$ must hold.
   To summarize, in case where $|\hat S|$ is small, we use the local flow augmentation technique to compute a relatively small set $Z$ of vertices of $G$, such that, if $(x,y)$ is the distinguished pair, then, with a sufficiently high probability, $Z$ is a promising vertex set. We then compute the vertex set $A_{x,y}\subseteq V(G)$ with the desired properties from $Z$ as described above.

\subsection*{Weighted Dense Graphs}
We now provide a high-level overview of our $O\left(n^{2.677} \cdot \poly\log W\right )$-time algorithm for the weighted version of the problem, that is most useful in relatively dense graphs. We fix a distinguished cut $(L,S,R)$ in the input vertex-weighted graph $G$ exactly as before, and, as before, we assume that we are given  an approximate estimate $\lambda$ (to within factor 2) of $|L|$. As before, we can assume that $|L|$ is not too large, so, for example, $|L|\leq \sqrt{n}$, since otherwise we can use the algorithm for easy special cases, where the distinguished cut is balanced, that was described above.

As in our algorithm for the non-dense setting, we employ a simple randomized procedure to compute a collection $\Gamma$ of roughly $\Theta\left (\frac{n}{|L|}\right )$ pairs of vertices of $G$, such that, with a sufficiently high probability, there exists some pair $(x^*,y^*) \in \Gamma$ (that we refer to as the \emph{distinguished pair}) with $x^* \in L$ and $y^* \in R$. We can also ensure that, for every pair $(x,y)\in \Gamma$, $(x,y)\not\in E(G)$. As before, the key part of our algorithm is a subroutine, that, given a pair $(x,y)\in \Gamma$, computes an $x$-$y$ vertex-cut $(L_{x,y},S_{x,y},R_{x,y})$ in $G$, such that, if $(x,y)$ is the distinguished pair, then, with a sufficiently high probability, $(L_{x,y},S_{x,y},R_{x,y})$ is a minimum $x$-$y$ vertex-cut; we refer to this subroutine as  the \emph{main subroutine}. Our algorithm applies the main subroutine to every pair $(x,y)\in \Gamma$, and then outputs the smallest-value vertex-cut $(L_{x,y},S_{x,y},R_{x,y})$ from among the resulting cuts. In the remainder of this overview we focus on the description of the main subroutine.

Recall that $|\Gamma| \approx \Theta\left ( \frac{n}{|L|}\right )$, and that our goal is to design an algorithm for global minimum vertex-cut with running time $O\left(n^{2.677} \cdot \poly\log W\right )$. 
Since the main subroutine  is applied to every vertex pair in $\Gamma$ separately,
 we need to ensure that its running time is bounded by $O\left(n^{1.677} \cdot |L| \cdot \poly\log W\right)$ --- the running time that may be sublinear in $m$ if $|L|$ is small. In particular, the main subroutine may not even be able to read the entire graph, and this is one of the main challenges in designing it.

At a high level, the main subroutine can be thought of as an adaptation of a similar subroutine of \cite{CT24} to directed graphs (see Theorem 3.16 in \cite{CT24}).
While it is generally not difficult to extend their subroutine to the directed setting, our main contribution in the dense regime lies in significantly simplifying their algorithm and its analysis, while also obtaining a faster running time. Next, we provide a high-level overview of the procedure of \cite{CT24} that is analogous to our main subroutine, and then describe our implementation of the main subroutine.

\paragraph{The Main Subroutine of \cite{CT24}.}
We start by defining a new notion, that was implicitly used in \cite{CT24} and in several previous works, that we call \emph{shortcut flow}.
Consider the input graph $G$ and the given pair $(x,y)\in \Gamma$ of its vertices. We say that a collection $E'$ of pairs of vertices of $G$ is a collection of \emph{shortcut edges} if, for every pair $(u,v)\in E'$, $v=y$ holds. We say that it is a \emph{valid collection of shortcut edges}, if, for every pair $(u,y)\in E'$, $u\not\in L$ (recall that $(L,S,R)$ is the distinguished cut we fixed in $G$). Observe that, if $(x,y)\in \Gamma$ is the distinguished pair, and $E'$ is a valid collection of shortcut edges, then $(L,S,R)$ remains a valid $x$-$y$ vertex-cut in the graph $G \cup E'$ that is obtained from $G$ by inserting the edges of $E'$. An \emph{$x$-$y$ shortcut flow} is a pair $(E',f)$, where $E'$ is a set of {shortcut edges}, and $f$ is an $x$-$y$ vertex-capacitated flow in $G \cup E'$.
We say that a shortcut flow $(E',f)$ is \emph{valid} if the set $E'$ of shortcut edges is valid.


%
%
%

At a high level, the main subroutine of \cite{CT24}, given the input graph $G$ and the pair $(x,y)\in \Gamma$ of its vertices,
gradually constructs a shortcut $x$-$y$ flow $(E',f)$ in $G$, over the course of  $O(\log(Wn))$ phases. Initially, $E'=\emptyset$ and $f=0$. Then, in every phase, new edges may be added to $E'$, and the flow $f$ is augmented, so that its value becomes closer and closer to the value of the maximum $x$-$y$ flow in $G$, which, in turn, is equal to the value of the minimum $x$-$y$ vertex-cut.
The subroutine guarantees that, if the vertex pair $(x,y)\in \Gamma$ is the distinguished pair, then, with high probability, $E'$ remains a valid collection of shortcut edges throughout the algorithm, and, additionally, the gap $w(S)-\val(f)$ decreases by a constant factor from phase to phase. Therefore, after $O(\log(Wn))$ phases, $w(S)-\val(f)<\half$ holds, and so the value $w(S)$ can be computed  by rounding the value of the flow $f$ up to the nearest integer.

While the above description conceptually matches the high-level structure of the subroutine of \cite{CT24}, it is imprecise and  omits important detail. 
Specifically, the algorithm of  \cite{CT24}, in order to gradually augment the flow $f$, uses a standard paradigm of computing augmenting paths in the residual flow network, and then augmenting the flow $f$ via these paths. This technique inherently only works with edge-capacitated, rather than vertex-capacitated graphs. 
In order to circumvent this difficulty, \cite{CT24} do not directly use the shortcut flow as defined above, and instead use an analogous object defined with respect to the split graph $G'$ of $G$.
Each phase of the subroutine of \cite{CT24} is then executed on the residual network $G'_f$ of the split graph $G'$ with respect to the current flow $f$, which leads to a rather complicated algorithm, due to the complex structure of the graph $G'_f$.

\paragraph{Our approach: well-behaving cuts.}
From the above description, the algorithm of \cite{CT24} implementing the main subroutine can be thought of as gradually augmenting a shortcut $x$-$y$ flow $f$ from phase to phase, until its value becomes very close to the optimal one. 
In contrast, our algorithm uses a different object in order to make progress from phase to phase, that we refer to as a  \emph{well-behaving cut}. Consider an $x$-$y$ vertex-cut $(L',S',R')$ in $G$, and a parameter $M$ that is an integral power of $2$ (where possibly $M<1$).
 We say that the cut $(L',S',R')$ is \emph{well-behaving for scale $M$}, if all of the following hold: (i) $\vol^+(L')\leq \frac{n^2}{\gamma}$ (here, $\gamma=n^c$ for a constant $0<c<1$ is a parameter that we choose in order to optimize the algorithm's running time); (ii) $w(S')\leq w(S)+10nM$; and (iii) the set $L' \cup S'$ contains all but at most $|L|\cdot \gamma$ vertices of $L \cup S$ whose weight is at least $M$. One can think of the scale parameter $M$ as measuring the quality of the well-behaving cut. 
  Our algorithm for the main subroutine starts with the $x$-$y$ cut $(L_0,S_0,R_0)$ in $G$ where $L_0=\set{x}$, $R_0=\set{y}$, and $S_0=V(G)\setminus\set{x,y}$; it is not hard to see that this cut must be well-behaving for a scale $M_0=10W$. It then performs $O(\log(Mn))$ phases, where, in the $i$th phase, it must produce an $x$-$y$ vertex-cut $(L_i,S_i,R_i)$ that is well-behaving for scale $M_i=\frac{M_0}{2^i}$. Let $z=O(\log(nW))$ be the smallest integer for which $M_z\leq \frac{1}{20n}$. We terminate the algorithm after $z$ phases are completed, and output the vertex-cut $(L_z,S_z,R_z)$ that was computed in Phase $z$. Since  $w(S_z)-w(S)<10nM_z<1$ holds, and, since all vertex weights are integral, the cut $(L_z,S_z,R_z)$ must be optimal.
  
We note that the algorithm of \cite{CT24} produces an intermediate object in every phase, that they refer to as the ``final cut''. This object is an edge-cut in the residual flow network $G'_f$ of the split-graph $G'$ with respect to the current flow $f$, 
and it has some useful properties, whose description is somewhat complex and is omitted here. This cut
is used by \cite{CT24} in order to sparsify the graph $G'_f$, so that the flow augmentation can be computed efficiently. However, it is not hard to show that one can also derive a well-behaving $x$-$y$ vertex-cut in $G$ from their ``final cut''. Therefore, while the algorithm of   \cite{CT24} is not directly presented in this way, one can view it as following the high-level approach of computing better and better well-behaving cuts. The main drawback of their algorithm is that it explicitly focuses on maintaining an 
analogue of the shortcut $x$-$y$ flow in the split-graph $G'$ of $G$, that is gradually augmented over the course of the entire algorithm. As such, each phase of their algorithm needs to work with the corresponding residual flow network of $G'$ with respect to this flow, whose structure is quite complex, adding an extra level of complexity to the algorithm and its analysis. In contrast, our algorithm focuses directly on obtaining better and better  well-behaving cuts in the input graph $G$. While some steps involved in computing improved well-behaving cuts do require computing an $x$-$y$ flow in the split graph $G'$, this use of the split graph is limited to rather simple subroutines, and the resulting flow is discarded once the improved well-behaving cut is computed. Therefore, the majority of our algorithm works directly with the vertex-capacitated input graph $G$, which greatly simplifies the algorithm and its analysis. The main ingredient of our subroutine is then an algorithm for implementing a single phase, that we describe next.

\paragraph{Implementing a single phase.} We now describe our algorithm for the $i$th phase, for $i\geq 1$.
We assume that, at the beginning of the phase, we are given an $x$-$y$ vertex-cut $(L_{i-1},S_{i-1},R_{i-1})$ in $G$.
 For brevity, we say that Condition (C) holds, if the vertex pair $(x,y)\in \Gamma$ is distinguished, and the cut  $(L_{i-1},S_{i-1},R_{i-1})$  is well-behaving for scale $M_{i-1}$.
 Our goal is to produce an $x$-$y$ vertex-cut $(L_{i},S_{i},R_{i})$ in $G$ with the following property: if Condition (C) holds, then, with a sufficiently high probability, cut $(L_{i},S_{i},R_{i})$ must be well-behaving for scale $M_i=\frac{M_{i-1}}2$.
The algorithm for the $i$th phase consists of three steps. In Step 1, we exploit the vertex-cut $(L_{i-1},S_{i-1},R_{i-1})$  in order to compute a valid shortcut  $x$-$y$  flow $(E',f)$ in $G$, whose value is close to the value of the maximum $x$-$y$ flow.
In Step 2, the shortcut flow $(E',f)$ is used in order to compute a set $A\subseteq V(G)$ of vertices with $\vol^+_G(A)\leq \frac{n^2}{\gamma}$. 
We say that the vertex set $A$ is \emph{promising}, if
there exists an $x$-$y$ vertex-cut  $(\hat{L},\hat{S},\hat{R})$  in $G$ with $\hat L\subseteq A$, that is well-behaving for scale $M_i$; in fact, we require that this cut satisfies even slightly stronger properties. Our algorithm for Step 2 guarantees that, if Condition (C) holds, then the vertex set $A$ that it computes is promising.  We note however that the algorithm for Step 2 does not compute the cut $(\hat{L},\hat{S},\hat{R})$  explicitly; instead it only explicitly computes the vertex set $A$ and only guarantees the existence of the cut $(\hat{L},\hat{S},\hat{R})$  with the above properties, if Condition (C) holds.
In the third and the last step, we exploit the vertex set $A$ in order to compute the vertex-cut $(L_i,S_i,R_i)$, such that, if Condition (C) holds, then cut $(L_i,S_i,R_i)$ is well-behaving for scale $M_i$.
Intuitively, the third step can be thought of as "approximately recovering" the vertex-cut $(\hat{L},\hat{S},\hat{R})$ that is guaranteed to exist by the definition of the promising set -- since $(\hat{L},\hat{S},\hat{R})$ satisfies somewhat stronger properties than those required by a well-behaving cut, an "approximation" of $(\hat{L},\hat{S},\hat{R})$ must also be a well-behaving cut.
We note that, at a very high level, our algorithm for recovering a well-behaving cut from a shortcut flow can be viewed as somewhat similar to the approach of \cite{CT24}, who obtain their "final cut" (an analogue of a well-behaving cut) via an intermediate object, that can be shown to implicitly encode a promising vertex set; this object is a subgraph of $G'_f$ that they denote by $J$.
However, both our algorithms for computing the promising vertex set from a shortcut flow, and for computing a well-behaving cut from a promising set, are faster and significantly simpler than the methods used by \cite{CT24} for the analogous tasks.

\paragraph{Organization.}

We start with preliminaries in \Cref{sec: prelims}. We describe algorithms for several easy special cases in \Cref{sec: special}, and  introduce some key tools that we use in our algorithms for both the unweighted and the weighted settings in \Cref{sec: tools}. In  \Cref{sec: alg: unweigthed}, we present our algorithm for unweighted graphs, proving \Cref{thm: main: unweighted}.
We present our $O\left(m n^{11/12+o(1)}\cdot d^{1/12} \cdot \poly\log W\right )$-time algorithm for the weighted setting in \Cref{sec: weighted non-dense}, and our $O\left(n^{2.677} \cdot \poly\log W\right )$-time algorithm in \Cref{sec: weighted dense}, where we also complete the proof of \Cref{thm: main: weighted} by combining these two algorithms.

\section{Preliminaries}
\label{sec: prelims}

All logarithms in this paper are to the base of $2$. We use the $\tilde O(\cdot)$ notation to hide $(\log n)^{O(1)}$ factors, where $n$ is the number of vertices in the input graph. For an integer $k>0$, we use $[k]$ to denote the set $\set{1,\ldots,k}$ of integers.

\subsection{Graph-Theoretic Notation}

Suppose we are given a directed graph $G=(V,E)$ with weights $w(v)\geq 0$ on its vertices $v\in V$. For a vertex $v\in V$, the set of its \emph{out-neighbors} is $N^+_G(v)=\set{u\in V(G)\mid (v,u)\in E(G)}$, and the set of its \emph{in-neighbors} is 
$N^-_G(v)=\set{u\in V(G)\mid (u,v)\in E(G)}$. We denote by $N_G(v)=N_G^+(v)\cup N_G^-(v)$ the set of all neighbors of $v$. We also denote by $\deg^+_G(v)=|N^+_G(v)|$, $\deg^-_G(v)=|N^-_G(v)|$ and $\deg_G(v)=\deg^+_G(v)+\deg^-_G(v)$ the out-degree, the in-degree, and the total degree of $v$, respectively. 
Additionally, we denote by $\delta^-_G(v)=\set{(u,v)\in E(G)\mid u\in V(G)}$ and $\delta^+_G(v)=\set{(v,u)\in E(G)\mid u\in V(G)}$ the sets of all incoming and outgoing edges of $v$, respectively. 
Given a subset $X\subseteq V$ of vertices, its \emph{volume} is $\vol(X)=\sum_{v\in X}\deg_G(v)$, and its \emph{weight} is $w(X)=\sum_{v\in X}w(v)$. Given a value $\tau\in \reals$, we denote by $X^{\geq \tau}=\set{v\in X\mid w(v)\geq \tau}$ the set of all vertices of $X$ whose weight is at least $\tau$.
We denote  the \emph{out-volume of $X$} by $\vol^+_G(X)=\sum_{v \in X}\deg^+_G(v)$.
We also denote by $N_G^+(X)=\left(\bigcup_{x\in X}N_G^+(x)\right )\setminus X$, and by $N_G^-(X)=\left(\bigcup_{x\in X}N_G^-(x)\right )\setminus X$ the sets of all out-neighbors and all in-neighbors of the set $X$, respectively, excluding the vertices that lie in $X$.
Given two disjoint vertex sets $A,B\subseteq V(G)$, we denote by $E_G(A,B)=\set{(u,v)\in E(G)\mid u\in A,v\in B}$.
Given a graph $G=(V,E)$ and a set $E'\subseteq V\times V$ of pairs of its vertices, we use $G \cup E'$ to denote the graph obtained by inserting the edges of $E' \setminus E(G)$ into $G$.
We may omit the subscript $G$ in the above notations when the graph $G$ is clear from context.

For a path $P \subseteq G$, its \emph{length} is $|E(P)|$ -- the number of edges on $P$.
For a pair $s,t \in V(G)$ of vertices, the \emph{distance} from $s$ to $t$, denoted $\dist_G(s,t)$, is the length of the shortest path connecting $s$ to $t$; if no such path exists, then $\dist_G(s,t)=\infty$. 
Given a vertex $s \in V(G)$ and a subset $T \subseteq V(G)$ of vertices, the distance from $s$ to $T$ is $\dist_G(s,T) = \min_{t \in T} \{\dist_G(s,t)\}$.
We may omit the subscript $G$ in the above notations when it is clear from the context.

\paragraph{Adjacency-list representation of graphs.}
Some of our algorithms utilize adjacency-list representations of graphs, that we recap here. Suppose we are given a directed graph $G$ with capacities $c(e)$ on its edges $e\in E(G)$ and weights $w(v)$ on its vertices $v\in V(G)$. The adjacency-list representation of $G$ stores, for every vertex $v\in V(G)$, its weight $w(v)$, its in-degree $\deg^-_G(v)$, and its out-degree $\deg^+_G(v)$.
Additionally, it stores a doubly-linked list $\IN(v)$, listing all edges $e\in \delta^-_G(v)$, together with the capacity $c(e)$ of each such edge. It also stores the edges of $\delta^-_G(v)$ in a separate efficient search structure, such as a binary search tree, that allows to look up an edge $(v,u)$ by its endpoint $u$. Similarly, it stores a doubly-linked list $\OUT(v)$ of all edges in $\delta^+_G(v)$, together with the capacity $c(e)$ of each such edge. The edges of $\delta^+_G(v)$ are also stored in an efficient search structure, where an edge $(u,v)$ can be looked up by its endpoint $u$.
If the edge capacities in $G$ are undefined then we assume that they are unit, and similarly, if vertex weights in $G$ are undefined then we assume that they are unit.

\paragraph{Vertex-Cuts in Directed Graphs and Global Minimum Vertex-Cut.}
Given a directed graph $G$ with weights $w(v)\geq 0$ on its vertices $v\in V(G)$, a \emph{vertex-cut} in $G$ is a partition $(L,S,R)$ of $V(G)$ with $L,R\neq \emptyset$, so that no edge of $G$ connects a vertex of $L$ to a vertex of $R$ (but edges connecting vertices of $R$ to vertices of $L$ are allowed). The \emph{value} of the cut is the total weight $w(S)=\sum_{v\in S}w(v)$ of the vertices of $S$.   A \emph{global minimum vertex-cut} is a vertex-cut $(L,S,R)$ of minimum value. Given an input graph $G$ with weights $w(v)\geq 0$ on its vertices, we denote by $\opt$ the value of the global minimum vertex-cut in $G$.

Notice that we can assume without loss of generality that there is a global minimum vertex-cut $(L,S,R)$ with $w(L)\leq w(R)$; if this is not the case, then we can consider the graph $\overline{G}$ that is obtained from $G$ by reversing the direction of all its edges. Clearly, if $(A,B,C)$ is a vertex-cut in $G$, then $(C,B,A)$ is a vertex-cut in $\overline{G}$ of the same value, and vice versa. We can now focus on computing a minimum vertex-cut in $\overline{G}$.
Therefore, from now on we assume that there is a global minimum vertex-cut $(L,S,R)$ in $G$ with $w(L)\leq w(R)$, and we will use the term ``global minimum vertex-cut'' to only include such cuts (in other words, a cut $(L',S',R')$ with $w(L')>w(R')$ will not be considered a global minimum vertex-cut regardless of its value $w(S')$).

Consider now a directed vertex-weighted graph $G$ and a pair $s,t$ of its vertices. We say that a vertex-cut $(L,S,R)$ \emph{separates} $s$ from $t$, or that $(L,S,R)$ is an \emph{$s$-$t$ vertex-cut}, if $s\in L$ and $t\in R$ holds. We say that it is a \emph{minimum $s$-$t$ vertex-cut}, if the value of the cut is minimized among all such cuts.
Notice that, if there is an edge connecting $s$ and $t$ in $G$, then $G$ does not contain a vertex-cut separating $s$ from $t$; in such a case, we say that the value of the minimum $s$-$t$ vertex-cut is infinite, and that the cut is undefined.

\paragraph{Induced Tripartitions and Vertex-Cuts.}
Let $G$ be a directed graph, and let $A\subseteq V(G)$ be a non-empty subset of its vertices. Consider the tri-partition $(L,S,R)$ of the vertices of $G$, where $L=A$, $S=N^+_G(A)$, and $R=V(G)\setminus(L\cup S)$. We say that  $(L,S,R)$ is the \emph{tripartition of $V(G)$ induced by $A$}, or that $A$ \emph{induces} the tripartition $(L,S,R)$. Note that, if $A\cup N^+_G(A)\neq V(G)$, then $(L,S,R)$ is also valid vertex-cut in $G$. In such a case, we may say that $(L,S,R)$ is the \emph{vertex-cut induced by $A$}, and that $A$ \emph{induces} the vertex-cut $(L,S,R)$. When the set $A$ consists of a single vertex $x$, we may refer to $\left(x,N^+_G(x),V(G)\setminus\left(\set{x}\cup N^+_G(x)\right )\right )$ as the \emph{tripartition induced by $x$}, and, if it is a valid vertex-cut, then we may refer to it as the \emph{vertex-cut induced by $x$}. 
Note that if, for every vertex $v\in V(G)$, $N^+_G(v)=V(G)\setminus\set{v}$, then $G$ does not have a valid vertex-cut. Otherwise, every vertex $v\in V(G)$ with $\deg^+_G(v)<|V(G)|-1$ induces a valid vertex-cut in $G$.

\paragraph{Flows in vertex-weighted graphs.}
Suppose we are given a directed graph $G$ with weights $w(v)\geq 0$ on its vertices $v\in V(G)$, and two special vertices $s$ and $t$. An \emph{$s$-$t$ flow} in $G$ is an assignment of a flow value $f(e)\geq 0$ to every edge $e\in E(G)\setminus\left(\delta^-(s)\cup \delta^+(t)\right )$, such that, for every vertex $v\in V(G)\setminus\set{s,t}$, $\sum_{e\in \delta^+(v)}f(e)=\sum_{e\in \delta^-(v)}f(e)$ and $\sum_{e\in \delta^+(v)}f(e)\leq w(v)$ hold. The \emph{value} of the flow is $\val(f)=\sum_{e\in \delta^+(s)}f(e)$.
In our algorithms, we always assume that a flow $f$ in a graph $G$ is given by listing a collection  $E^f\subseteq E(G)$ that contains all edges $e\in E(G)$ with $f(e)>0$, together with the flow value $f(e)$ for each edge $e\in E^f$.

A \emph{flow-path decomposition} of the $s$-$t$ flow $f$ consists of a collection $\pset$ of simple $s$-$t$ paths, and a value $f(P)>0$ for every path $P\in \pset$, such that $\sum_{P\in \pset}f(P)=\val(P)$ and, for every edge $e\in E(G)$, $\sum_{\stackrel{P\in \pset:}{e\in E(P)}}f(P)\leq f(e)$. We sometimes refer to the paths in $\pset$ as \emph{flow-paths}.
Recall that, from the Max-flow / Min-cut theorem, the value of the maximum $s$-$t$ flow is equal to the value of the minimum $s$-$t$ vertex-cut in any vertex-weighted graph.

\paragraph{Edge-cuts in directed graphs.} 
Let $G=(V,E)$ be a directed graph with capacities $c(e)\geq 0$ on its edges $e\in E$. Given a subset $E'\subseteq E$ of edges, we denote by $c(E')=\sum_{e\in E'}c(e)$ the \emph{total capacity of the edges in $E'$}. An \emph{edge-cut} in $G$ is a partition $(X,Y)$ of the vertices of $G$ into two disjoint non-empty subsets, and the \emph{value} of the cut is $c(E_G(X,Y))$, that we also denote by $c_G(X,Y)$. We say that a cut $(X,Y)$ is a \emph{global minimum edge-cut for $G$}, if it is an edge-cut in $G$ of minimum value. Given a pair $s,t$ of vertices of $G$, we say that an edge-cut $(X,Y)$ is an \emph{$s$-$t$ edge-cut}, or that it is an \emph{edge-cut that separates $s$ from $t$}, if $s\in X$ and $t\in Y$ holds.
We say that it is a \emph{minimum $s$-$t$ edge-cut} if it is an $s$-$t$ edge-cut whose value is minimal among all $s$-$t$ edge-cuts.
We define edge-cuts separating a vertex $s$ from a subset $T\subseteq V\setminus \set{s}$ of vertices, and edge-cuts separating disjoint subsets $S,T\subseteq V(G)$ of vertices similarly.

\paragraph{Flows in edge-capacitated directed graphs.}
Let $G=(V,E)$ be a directed graph with capacities $c(e)\geq 0$ on its edges $e\in E$, and let $s,t\in V$ be a pair of vertices of $G$. An \emph{$s$-$t$ flow} in $G$ is an assignment of flow value $0\leq f(e)\leq c(e)$ to every edge $e\in E$, such that, for every vertex $v\in V\setminus\set{s,t}$, the following \emph{flow conservation constraint} holds: $\sum_{e\in \delta^+(v)}f(e)=\sum_{e\in \delta^-(v)}f(e)$. We also require that, if $e\in \delta^-(s)\cup \delta^+(t)$, then $f(e)=0$. 
For a value $M>0$, we say that the flow $f$ is \emph{$M$-integral}, if, for every edge $e\in E$, $f(e)$ is an integral multiple of $M$ (note that it is possible that $M<1$ under this definition).
The \emph{value} of the flow is $\val(f)=\sum_{e\in \delta^+(s)}f(e)$. As in vertex-capacitated flows, we always assume that a flow $f$ in a graph $G$ is given by listing a collection  $E^f\subseteq E(G)$ that contains all edges $e\in E(G)$ with $f(e)>0$, together with the flow value $f(e)$ for each edge $e\in E^f$.

%

\paragraph{Residual Flow Network.}
Let $G=(V,E)$ be a directed graph with capacities $c(e)\geq 0$ on its edges $e\in E$, let $s,t\in V$ be a pair of vertices of $G$, and let $f$ be an $s$-$t$ flow in $G$.
The residual flow network of $G$ with respect to $f$, denoted $G_f$, is an edge-capacitated graph with $V(G_f)=V(G)$. The set $E(G_f)$ of edges is partitioned into two subsets: a set of \emph{forward} edges and a set of \emph{backward edges}. The set of forward edges of $G_f$ contains, for every edge $(u,v) \in E(G)$ with $f(u,v) < c(u,v)$, the edge $(u,v)$ of capacity $c_{G_f}(u,v)=c(u,v)-f(u,v)$. The set of backward edges of $G_f$ contains, for every edge $(u,v) \in E(G)$ with $f(u,v) > 0$, the edge $(v,u)$ of capacity $c_{G_f}(v,u)=f(u,v)$.
For an edge $(u,v) \in G_f$, its capacity in $G_f$ is called the \emph{residual capacity of $(u,v)$}.
The following is a basic observation about the residual network $G_f$; similar observations were used in numerous previous works. For example, \cite{BBR25} proved a similar statement (see Theorem 10 of \cite{BBR25}), though their formulation is less general as it only covers the case where $f$ is a maximum $s$-$t$ flow.
\begin{observation}\label{fact : residual-capacity-of-edge-cut}
	Let $G=(V,E)$ be a directed edge-capacitated graph, let $s,t \in V$ be a pair of vertices of $G$, and let $f$ be an $s$-$t$ flow in $G$. Consider the residual network $H=G_f$. For every edge $e\in E$, we denote its capacity in $G$ by $c_G(e)$, and for every edge $e\in E(H)$, we denote its capacity in $H$ by $c_H(e)$.
	Then, for every $s$-$t$ edge-cut $(X,Y)$ in $G$,  $c_{H}(X,Y) = c_G(X,Y) - \val(f)$.
\end{observation}
\begin{proof}
	By the definition of the residual flow network, the total residual capacity of forward edges in $E_H(X,Y)$ is $\sum_{e \in E_G(X,Y)} (c_G(e) - f(e))$, and the total residual capacity of backwards edges in $E_H(X,Y)$ is $\sum_{e \in E_G(Y,X)} f(e)$.
	Therefore,
	\begin{align*}
		c_{H}(X,Y)
		&= \sum_{e \in E_G(X,Y)} (c_G(e) - f(e)) + \sum_{e \in E_G(Y,X)} f(e)\\
		&= c_G(X,Y) - \left(\sum_{e \in E_G(X,Y)} f(e) - \sum_{e \in E_G(Y,X)} f(e)\right).
	\end{align*}
	The observation follows using the well-known fact that $\val(f)=\left(\sum_{e \in E_G(X,Y)} f(e) - \sum_{e \in E_G(Y,X)} f(e)\right)$.
\end{proof}

\subsection{Graph Derived via a Vertex Set} 
\label{subsec: derived subgraphs}

One of the tools that we use in order to sparsify graphs is a graph that is derived from the input graph $G$ via some vertex set $A\subseteq V(G)$, that we define next. We note that such graphs were used in past work, though to the best of our knowledge this notion was never formalized. 

\begin{definition}[Graph derived via vertex set]\label{def: defined graph} Let $G$ be a directed graph with weights $w(v)\geq 0$ on its vertices $v\in V(G)$, and let $A\subseteq V(G)$ be a non-empty set of vertices. 
	The \emph{graph derived from $G$ via $A$}, denoted by $G^{|A}$ is defined as follows. The vertex set of $G^{|A}$ is $A\cup N^+_G(A)\cup \set{t}$, where $t$ is a new vertex that does not lie in $G$. The weight of the vertex $t$ is set to $\max_{v\in V(G)}\set{w(v)}$, while the weights of all other vertices remain the same as in $G$. For simplicity, for every vertex $v\in V(G^{|A})$, its weight in $G^{|A}$ is still denoted by $w(v)$. The set of edges of $G^{|A}$ is defined as follows. First, we include, for every vertex $v\in A$, all edges of $\delta^+_G(A)$ in $G^{|A}$. Additionally, for every vertex $u\in N^+_G(A)$, we include the edge $(u,t)$.
\end{definition}

In the following claim, we summarize the central properties of the graph $G^{|A}$.

\begin{claim}\label{claim: properties of defined graph}
	Let  $G$ be a directed graph with weights $w(v)\geq 0$ on its vertices $v\in V(G)$, and let $A\subseteq V(G)$ be a non-empty set of vertices. Consider the graph $G^{|A}$ derived from $G$ via $A$. Then $|E(G^{|A})|\leq 2\vol^+_G(A)$ must hold, and graph $G^{|A}$ can be constructed in time $O(\vol^+_G(A))$ given access to the adjacency-list representation of $G$ and the vertex set $A$. Moreover, if $(L,S,R)$ is a vertex-cut in $G$ that is induced by the set $L$ of vertices with $L\subseteq A$, then $(L,S,V(G^{|A})\setminus (L\cup S))$ is a valid vertex-cut in $G^{|A}$.
\end{claim}
\begin{proof}
	For convenience, in this proof we denote $G^{|A}$ by $\hat G$. Since $|N^+_G(A)|\leq \vol^+_G(A)$, and since every edge $e\in \hat G$ either lies in $\bigcup_{a\in A}\delta^+_G(A)$, or connects a vertex of $N^+_G(A)$ to $t$, it is immediate to verfiy that $|E(\hat G)|\le 2\vol^+_G(A)$. It is also immediate to verify that graph  $G^{|A}$ can be constructed in time $O(\vol^+_G(A))$ given access to the adjacency-list representation of $G$ and the vertex set $A$.
	
	Consider now any vertex-cut   $(L,S,R)$ in $G$ that is induced by some vertex set $L$, such that $L\subseteq A$ holds. Since $S=N_G^+(L)\subseteq A\cup N_G^+(A)$ by the definition of induced cuts, we get that $S\subseteq V(\hat G)$. Consider now the tri-partition $(L,S,R')$ of $V(\hat G)$, where $R'=V(\hat G)\setminus (L\cup S)$. Notice that $t\in R'$, so $R'\neq \emptyset$. Moreover, no edge of $\hat G$ may connect a vertex of $L$ to a vertex of $R'$. Indeed, assume otherwise, and let $(u,v)\in E(\hat G)$ be an edge with $u\in L$ and $v\in R'$. Assume first that $v\in R$. Then edge $(u,v)$ lies in $G$, which is impossible since $(L,S,R)$ is a valid vertex-cut in $G$. Therefore, $v=t$ must hold. But since $L\subseteq A$, $u\in A$ must hold. However,  graph $\hat G$ contains no edges connecting vertices of $A$ to $t$, a contradiction. We conclude that $(L,S,R')$ is a valid vertex-cut in $\hat G$.
\end{proof}

The following simple claim allows us to transform vertex-cuts in $G^{|A}$ to vertex-cuts in $G$.

\begin{claim}\label{claim: transforming cuts from defined graph}
	Let  $G$ be a directed graph with weights $w(v)\geq 0$ on its vertices $v\in V(G)$, and let $A\subseteq V(G)$ be a non-empty set of vertices with $A\cup N^+_G(A)\neq V(G)$. Consider the graph $G^{|A}$ derived from $G$ via $A$, and let $(L,S,R)$ be a vertex-cut in $G^{|A}$ with $t\in R$. Then $L \subseteq A$ must hold, and, moreover, $(L,S,V(G)\setminus(L\cup S))$ is a valid vertex-cut in $G$.
\end{claim}
\begin{proof}
	We first prove that $L \subseteq A$.
	Indeed, since $(L,S,R)$ is a vertex-cut in $G^{|A}$ with $t\in R$, $L$ may only vertices $v\in V(G^{|A})\setminus\set{t}$ for which $(v,t)\not\in E(G^{|A})$. From the construction of the graph $G^{|A}$, it is then easy to verify that each such vertex must lie in $A$, so $L\subseteq A$ must hold.
	
	Next, we prove that $(L,S,V(G)\setminus(L\cup S))$ is a valid vertex-cut in $G$.
	Let $X=V(G)\setminus (A\cup N_G^+(A))$ and note that $X\neq \emptyset$ must hold. Then $V(G)\setminus (L\cup S)=(R\setminus\set{t})\cup X$; we denote this set by $R'$. It is now enough to show that $(L,S,R')$ is a valid vertex-cut in $G$. 
	Since $(L,S,R)$ is a valid vertex-cut in $G^{|A}$, and since $X\neq \emptyset$, it is immediate to vertify that $L,R'\neq \emptyset$. 
	
	It now remains to show that no edge connects a vertex of $L$ to a vertex of $R'$ in $G$. Indeed, assume otherwise, and let $(u,v)\in E(G)$ be an edge with $u\in L$ and $v\in R'$. 
	Assume first that $v\in X$. Since $(L,S,R)$ is valid vertex-cut in $G^{|A}$, edge $(u,t)$ may not lie in $G^{|A}$, and  so, by the definition of $G^{|A}$, $u\in A$ must hold. But then by construction of $G^{|A}$, $v\in N^+_G(A)$ must hold, so $v\not\in X$, a contradiction. We conclude that $v\in R\setminus\set{t}$ must hold. Since edge $(u,v)$ may not lie in $G^{|A}$ but it lies in $G$, we get that $u\not\in A$. But then edge $(u,t)$ must lie in $G^{|A}$, contradicting the fact that $(L,S,R)$ is a valid vertex-cut in $G^{|A}$.	
\end{proof}

We also use the following immediate corollary of \Cref{claim: transforming cuts from defined graph}.

\begin{corollary}\label{cl: finding global mincut using defined graph}\label{cor: minimum cut in defined graph}
	Let $G$ be a directed graph with weights $w(v)\geq 0$ on its vertices $v\in V(G)$, and let $A\subseteq V(G)$ be a non-empty set of vertices with $A\cup N^+_G(A)\neq V(G)$. Let $x,y \in V(G)$ be a pair of vertices with $x \in A$ and $y \notin A\cup N^+_G(A)$.
	Consider the graph $G^{|A}$ derived from $G$ via $A$, and let $(L',S',R')$ be a minimum $x$-$t$ vertex-cut in $G^{|A}$.
	Then $(L',S',V(G)\setminus(L'\cup S'))$ is a valid $x$-$y$ vertex-cut in $G$.
	In particular, if we  denote by $\opt_{x,y}$ the value of the minimum $x$-$y$ vertex-cut in $G$, and by $\opt'_{x,t}$ the value of the minimum $x$-$t$ vertex-cut in $G^{|A}$, then $\opt'_{x,t}\geq \opt_{x,y}$ must hold.
	Moreover, if there exists a global minimum vertex-cut $(L,S,R)$ in $G$ with $x \in L$, $y \in R$, and $L \subseteq A$, then $(L',S',V(G)\setminus(L'\cup S'))$ is a global minimum vertex-cut in $G$.
\end{corollary}
\begin{proof}
	Let $(L',S',R')$ be a minimum $x$-$t$ vertex-cut in $G^{|A}$, so $w(S')=\opt'_{x,y}$. Denote by $R''=V(G)\setminus (L'\cup S')$, and note that, since $y \notin A\cup N^+_G(A)$, $y\in R''$ must hold. From \Cref{claim: transforming cuts from defined graph}, $(L',S',R'')$ is a valid vertex-cut in $G$. Clearly, since $x\in L'$ and $y\in R''$, it is also a valid $x$-$y$ vertex-cut in $G$. In particular, $\opt_{x,y}\leq w(S')=\opt'_{x,t}$ must hold.
	
	Assume now that there exists a global minimum vertex-cut $(L,S,R)$ in $G$ with $x \in L$, $y \in R$, and $L \subseteq A$.
	Consider the vertex-cut $(\hat{L},\hat{S},\hat{R})$ in $G$ induced by the set $L$ of vertices, so $\hat{L}=L$, $\hat{S} = N^+_G(L)$ and $\hat{R} = V(G) \setminus (\hat{L} \cup \hat{S})$.
	Clearly, $\hat S\subseteq S$ must hold, so $(\hat{L},\hat{S},\hat{R})$ is also a global minimum vertex-cut in $G$, with $x \in \hat{L}$, $y \in \hat{R}$, and $\hat{L} \subseteq A$.
	Then, by \Cref{claim: properties of defined graph}, $(\hat{L},\hat{S},V(G^{|A})\setminus (\hat{L}\cup \hat{S}))$ is a valid vertex-cut in $G^{|A}$. Since $x \in \hat{L}$ and $t \in V(G^{|A})\setminus (\hat{L}\cup \hat{S})$, it is a valid $x$-$t$ vertex-cut in $G^{|A}$.
	Since $(L',S',R')$ is a minimum $x$-$t$ vertex-cut in $G^{|A}$, it follows that $w(S') \leq w(\hat{S})$.
However, since, as we have shown above, $(L',S',R'')$ is a valid vertex-cut in $G$, and $w(S')\leq w(\hat S)\leq w(S)$, we get that it must be a global minimum vertex-cut in $G$.
\end{proof}

\subsection{An Assumption on Strictly Positive Vertex Weights}
\label{subsec: positive weights}

Suppose we are given a simple directed $n$-vertex and $m$-edge graph $G$ with integral weights $0\leq w(v)\leq W$ on its vertices.
It will be convenient for us to assume that all vertex weights are strictly positive. In order to achieve this, we can modify the vertex weights using a standard transformation as follows: for every vertex $v\in V(G)$, we let  $w'(v)=n^2\cdot w(v)+1$. Let $\tilde G$ be the graph that is identical to $G$, except that the weight of every vertex $v$ is set to $w'(v)$, and let $W'=n^2\cdot W+1$, so for every vertex $v\in V(G)$, $1\leq w'(v)\leq W'$ holds. Note that, if $(L,S,R)$ is a global minimum vertex-cut in $\tilde G$, then it must also be a global minimum vertex-cut in $G$. Therefore, in order to obtain an algorithm for the global minimum vertex-cut problem, it is enough to design an algorithm for the special case where all vertex weights are strictly positive.

Note that, if all vertex weights in $G$ are strictly positive, and  $(L,S,R)$ is a global minimum vertex-cut in $ G$, then every vertex $v\in S$ must have an in-neighbor $u\in N_G^-(v)$ that lies in $L$, since otherwise, by moving $v$ from $S$ to $R$, we obtain a vertex-cut of a smaller value. 
Therefore, $S=N_G^+(L)$ must hold, and in particular, $(L,S,R)$ is a vertex-cut induced by $L$. Using a similar reasoning, every vertex $v\in S$ must have an out-neighbor $u'\in N_G^+(v)$ that lies in $R$, and so $S=N_G^-(R)$.

\subsection{Parameters $\wmax$ and $\wmax'$ and The Split Graph}
\label{subsec: split graph}

In this subsection we assume that we  are given a directed graph $G=(V,E)$ with integral weights $w(v)\geq 0$ on its vertices $v\in V(G)$.

\paragraph{Parameters $\wmax$ and $\wmax'$.}
These parameters are defined exactly like in \cite{CT24}.
We denote by $W'_{\max}(G)$ the smallest integral power of $2$, such that $\wmax'(G)\geq 1+ \max_{v\in V}\set{w(v)}$ holds. We let $W_{\max}(G)$ be the smallest integral power of $2$ with $\wmax(G)\geq8n\cdot W'_{\max}(G)$. When the graph $G$ is clear from context, 
we denote $\wmax'(G)$ and $\wmax(G)$ by $\wmax'$ and $\wmax$, respectively. Notice that:

\begin{equation}\label{eq: bound on wmax}
	\wmax(G)\leq 16n\cdot \wmax'(G)\leq 32n\cdot(1+ \max_{v\in V}\set{w(v)}).
\end{equation}

\paragraph{The Split Graph.}
The split graph is defined exactly like in \cite{CT24}, except that we extend the definition to directed graphs.
Given a directed vertex-weighted graph $G=(V,E)$, its \emph{split graph} is a directed graph $G'$, that is defined as follows. The set of vertices of $G'$ contains two copies of every vertex $v\in V$, an \emph{in-copy} $v^{\inn}$, and an \emph{out-copy} $v^{\out}$, so $V(G')=\set{v^{\inn},v^{\out}\mid v\in V}$. The set of edges of $G'$ is partitioned into two subsets: the set $E^{\spec}$ of \emph{special} edges, and the set $E^{\reg}$ of \emph{regular edges}. For every vertex $v\in V(G)$, the set $E^{\spec}$ of special edges contains the edge $e_v=(v^{\inn},v^{\out})$, whose capacity $c(e_v)$ is set to be $w(v)$; we refer to $e_v$ as the \emph{special edge representing $v$}. Additionally, for every edge $e=(x,y)\in E(G)$, we add a regular edge $(x^{\out},y^{\inn})$ of capacity $c(x^{\out},y^{\inn})=\wmax(G)$ to $G'$, representing $e$.
This completes the definition of the split graph $G'$. Observe that $|V(G')|=2|V(G)|$ and $|E(G')|\leq O(|E(G)|)$; moreover, there is an algorithm that, given $G$, computes $G'$ in time $O(|E(G')|)$.

Given a subset $Z\subseteq V(G)$ of vertices of $G$, we denote by $Z^{\inn}=\set{v^{\inn}\mid v\in Z}$, $Z^{\out}=\set{v^{\out}\mid v\in Z}$, and $Z^*=Z^{\inn}\cup Z^{\out}$.
We also denote by  $V^{\inn}=\set{v^{\inn}\mid v\in V(G)}$ and $V^{\out}=\set{v^{\out}\mid v\in V(G)}$.

Note that, for every pair $s,t\in V(G)$ of vertices, the value of the minimum $s$-$t$ vertex-cut in $G$ is equal to the value of the minimum $s^{\out}$--$t^{\inn}$ edge-cut in $G'$, which, in turn, is equal to the value of the maximum  $s^{\out}$--$t^{\inn}$ flow in $G'$ from the Max-Flow/Min-Cut theorem. In particular, the value of the global minimum vertex-cut in $G$ is equal to maximum value of the $s^{\out}$--$t^{\inn}$ flow in $G'$, over all vertex pairs  $s,t\in G$.

The following two simple observations, that were proved implicitly or explicitly  numerous times in prior work, show that an $s$-$t$ vertex-cut in $G$ can be transformed into an $s^{\out}$-$t^{\inn}$ edge-cut  of the same value in the split graph, and vise versa.  For completeness, we include their proofs in sections \ref{subsec : proof of obs existential from regualr to split} and \ref{subsec: proof of obs existential from split to regular} of appendix, respectively.

\begin{observation}\label{obs: existential from regular to split}
	Let $G$ be an directed graph with integral weights $0\leq w(v)\leq W$ on its vertices $v\in V(G)$, let $s,t\in V(G)$ be a pair of distinct vertices, and let $(L,S,R)$ be an $s$-$t$ vertex cut in $G$.
	Consider the corresponding split graph $G'$, and for every edge $e\in E(G')$, denote its capacity in $G'$ by $c(e)$. Let $X = L^* \cup S^{\inn}$ and $Y = V(G') \setminus X = S^{\out} \cup R^*$.
	Then $(X,Y)$ is an $s^{\out}$-$t^{\inn}$ edge-cut in $G'$ whose value is equal to $w(S)$.
\end{observation}



\begin{observation}\label{obs: existential from split to regular}
	Let $G$ be an directed graph with integral weights $0\leq w(v)\leq W$ on its vertices $v\in V(G)$, and let $s,t\in V(G)$ be a pair of distinct vertices.
	Consider the corresponding split graph $G'$, and for every edge $e\in E(G')$, denote its capacity in $G'$ by $c(e)$. Let $(X,Y)$ be an $s^{\out}$-$t^{\inn}$ edge-cut in $G'$ whose value $c(X,Y)<\wmax(G)$. Let $L=\{v \in V(G) \setminus \{t\} \mid v^{\out} \in X\}$, $S=\set{v\in V(G)\mid v^{\inn}\in X,v^{\out}\in Y}$, and let $R=V(G)\setminus(L\cup S)$. Then $(L,S,R)$ is an $s$-$t$ vertex-cut in $G$, and $w(S) \leq c(X,Y)$. Moreover, if we let $(L,S',R')$ be the tripartition of $V(G)$ induced by $L$, then $(L,S',R')$ is an $s$-$t$ vertex-cut in $G$ and $S'\subseteq S$ must hold; in particular, $w(S') \leq c(X,Y)$. 
\end{observation}

We obtain the following simple corollary of \Cref{obs: existential from split to regular}, whose proof follows standard arguments and is included in Section \ref{subsec: proof of cor minimum cut in split graph} of Appendix for completeness.

\begin{corollary}\label{cor : minimum cut in split graph}
	Let $G$ be an directed graph with integral weights $0\leq w(v)\leq W$ on its vertices $v\in V(G)$, and let $s,t\in V(G)$ be a pair of distinct vertices such that $(s,t)\not\in E(G)$.
	Consider the corresponding split graph $G'$.
	Then, the value of the minimum $s$-$t$ vertex-cut in $G$ is equal to the value of the minimum $s^{\out}$-$t^{\inn}$ edge-cut in $G'$, and this value is less than $\frac{\wmax(G)}{2}$.
\end{corollary}

\subsection{A Split Graph with Shortcuts and the Modified Adjacency-List Representation}
\label{subsec: split w shortcuts}

\paragraph{A split graph with shortcuts.}
Assume that we are given a directed graph $G=(V,E)$ with integral weights $w(v)\geq 1$ on its vertices $v\in V(G)$, and consider the corresponding split graph $G'$. We will sometimes use a \emph{split graph with shortcuts}, that is defined as follows. Let $B\subseteq V(G')$ be any subset of vertices of $G'$. The corresponding split graph with $B$-shortcuts, denoted by $G''(B)$, is obtained from the graph $G'$ by first adding a new vertex $t$ to it, and then connecting every vertex $v\in B$ to the vertex $t$ with a directed edge whose capacity is $\wmax$. These newly added edges are considered regular edges of the resulting graph $G''(B)$.

\paragraph{Modified adjacency-list representation.}
Our algorithm may sometimes need to utilize a slightly unusual analogue of the adjacency-list representation of the split graph $G'$ and of the split graph with shortcuts $G''(B)$. In order to define this representation, we first partition the vertices of the graph $G$ into groups by their weight geometrically: for all $1\leq i\leq \log(\wmax')+1$, we define $V_i=\set{v\in V(G)\mid 2^{i-1}\leq w(v)<2^i}$.
Additionally, assume that we are given a degree threshold $\td$, and denote by $\Vhd=\set{v\mid \deg_G(v)\geq \td}$ and by $\Vld=V\setminus\Vhd=\set{v\mid \deg_G(v)<\td}$; we refer to the vertices in $\Vhd$ as \emph{high-degree vertices} and to the vertices in $\Vld$ as \emph{low-degree vertices}.
Finally, assume that we are given a subset $B\subseteq V(G')$ of vertices of the split graph $G'$, and consider the corresponding split graph with shortcuts $G''(B)$. The modified adjacency-list representation of $G''(B)$ is defined as follows.

First, it contains an array $\IN(t)$, listing all edges that enter the destination vertex $t$ in $G''(B)$. Consider now some vertex $v\in V(G')\setminus\set{t}$, and assume first that $v$ is an in-copy of some vertex $u\in V(G)$, so $v=u^{\inn}$. The modified adjacency-list representation of $G''(B)$ contains the following data structures for $v$:

\begin{itemize}
	\item a bit $b(v)$ indicating whether the edge $(v,t)$ is present in $G''(B)$;
	\item the unique special edge $\spec(v)=(u^{\inn},u^{\out})$ that is incident to $v=u^{\inn}$ together with its capacity $c(u^{\inn},u^{\out})$; and
	\item for all $1\leq i\leq \log(\wmax')+1$, a doubly-linked list $\IN^{\hd}_i(v)$, containing all regular edges $(a^{\out},v)$ that are incident to $v$, for which $a\in V_i\cap \Vhd$,
	 and a doubly-linked list $\IN^{\ld}_i(v)$, containing all regular edges $(a^{\out},v)$ that are incident to $v$, for which $a\in V_i\cap \Vld$; we  also denote by $\IN_i(v)=\IN^{\hd}_i(v)\cup \IN^{\ld}_i(v)$, and we store all edges of $\IN_i(v)$ in an efficient search data structure, that allows to look up an edge $(u,v)$ via its endpoint $u$.
\end{itemize}

Similarly, if $v$ is an out-copy of some vertex $u\in V(G)$, so $v=u^{\out}$,  the modified adjacency-list representation of $G''(B)$ contains the following data structures for $v$:

\begin{itemize}
	\item a bit $b(v)$ indicating whether the edge $(v,t)$ is present in $G''(B)$;
	\item the unique special edge $\spec(v)=(u^{\inn},u^{\out})$ that is incident to $v=u^{\out}$ together with its capacity $c(u^{\inn},u^{\out})$; and
	\item for all $1\leq i\leq \log(\wmax')+1$,  a doubly-linked list $\OUT^{\hd}_i(v)$, that contains all regular edges $(v,a^{\inn})$ that are incident to $v$, for which $a\in V_i\cap \Vhd$ and a doubly-linked list $\OUT^{\ld}_i(v)$, that contains all regular edges $(v,a^{\inn})$ that are incident to $v$, for which $a\in V_i\cap \Vld$; we will also denote by $\OUT_i(v)=\OUT^{\hd}_i(v)\cup \OUT^{\ld}_i(v)$, and store all edges of $\OUT_i(v)$ in an efficient search  data structure, that allows to look up an edge $(v,u)$ via its endpoint $u$.
\end{itemize}

This completes the definition of the modified adjacency-list representation of the graph $G''(B)$. The modified adjacency-list representation of graph $G'$ is defined analogously.

Consider now some value $1\leq \tau\leq \log(\wmax')+1$ that is an integral power of $2$. Let $G_{\geq \tau}$ be the graph obtained from $G$ by deleting from it all vertices $v$ with $w(v)<\tau$. Let $G'_{\geq\tau}$ be the split graph corresponding to $G_{\geq \tau}$. Note that graph $G'_{\geq \tau}$ can be equivalently obtained from the split graph $G'$ by deleting from it the copies of all vertices $v\in V(G)$ with $w(v)<\tau$. Assume now that we are given a modified adjacency-list representation of the graph $G'$. Then we can use it to obtain access to the adjacency-list representation of $G'_{\geq \tau}$, by simply ignoring the copies of all vertices $v\in V(G)$ with $w(v)<\tau$, and the arrays $\IN_i(u^{\inn}),\OUT_i(u^{\out})$ for all $u\in V(G)$ with $w(u)\geq \tau$ and $i<\log \tau$. 

Similarly, assume that we are given a subset $B\subseteq V(G')$ of vertices of the split-graph. 
Let $G''_{\geq \tau}(B)$ be the graph obtained from $G''(B)$ by deleting from it all vertices that are copies of the vertices of $\set{u\in V(G)\mid w(u)<\tau}$.
Then a modified adjacency-list representation of the graph $G''(B)$ can be used to obtain access to the adjacency-list representation of $G''_{\geq \tau}(B)$, exactly as before.

Finally, the modified adjacency-list representation of the graph $G''(B)$ allows us to efficiently access, for every vertex $u\in V(G'')\setminus \set{t}$, all neighbors of $u$ in $G''$ that correspond to low-degree vertices of $G$ whose weight is at least as high some given threshold $\tau$ that is an integral power of $2$.

\subsection{Fast Algorithms for Flows and Cuts}

Our algorithm relies on the recent breakthrough almost-linear time algorithm for Minimum Cost Maximum $s$-$t$ Flow \cite{CKLP22,BCP23}. The algorithm of \cite{BCP23} first computes a fractional solution to the problem, whose cost is close to the optimal one, and then rounds the resulting flow via the Link-Cut Tree data structure~\cite{ST83,KP15} in order to compute the optimal flow. In particular, the flow that they obtain is integral, provided all edge capacities are integral as well (see Lemma $4.1$ in~\cite{BCP23}). This result is summarized in the following theorem (see Theorem $1.1$ and Lemma 4.1 in \cite{BCP23}).

\begin{theorem}
	\label{thm: maxflow}
	There is a deterministic algorithm, that receives as input a directed $m$-edge graph $G$ with integral capacities $1\leq u(e)\leq U$ and  integral costs $1\leq c(e)\leq C$ on edges $e\in E(G)$, together with two vertices $s,t\in V(G)$. The algorithm computes a minimum-cost maximum $s$-$t$ flow $f$ in $G$, such that $f$ is integral. The running time of the algorithm is $O\left (m^{1+o(1)}\cdot \log U\cdot  \log C\right )$.
\end{theorem}

We also use the following immediate corollary of \Cref{thm: maxflow}, whose proof follows standard arguments and can be found, e.g. in \cite{CT24}.

\begin{corollary}[Corollary 2.9 in \cite{CT24}]\label{cor: mincut}
	There is a deterministic algorithm, that receives as input a directed $m$-edge graph $G$ with integral capacities $1\leq u(e)\leq U$ on edges $e\in E(G)$, together with two vertices $s,t\in V(G)$. The algorithm computes a minimum $s$-$t$ edge-cut $(S,T)$ in $G$. The running time of the algorithm is $O\left (m^{1+o(1)}\cdot \log U\right )$.
\end{corollary}

Finally, we use the following immediate corollary of \Cref{cor: mincut}, whose proof also uses standard techniques. The corollary was proven for undirected graphs in \cite{CT24} (see Corollary 2.10). The proof for directed graphs, that appears below, is esentially identical. 

\begin{corollary}\label{cor: min_vertex_cut}
	There is a deterministic algorithm, that receives as input a directed $m$-edge graph $G$ with integral weights $1\leq w(v)\leq W$ on vertices $v\in V(G)$, together with two disjoint subsets $S,T\subseteq V(G)$ of vertices, such that no edge of $G$ connects a vertex of $S$ to a vertex of $T$. The algorithm computes a minimum $S$-$T$ vertex-cut $(X,Y,Z)$ in $G$. The running time of the algorithm is $O\left (m^{1+o(1)}\cdot \log W\right )$.
\end{corollary}
\begin{proof}
	Let $\hat G$ be the graph obtained from $G$ by setting the weights of the vertices in $S\cup T$ to $\wmax(G)$, and then adding vertices $s$ and $t$ of unit weight, where $s$ connects with an edge to every vertex in $S$, and  every vertex in $T$ connects with an edge to $t$. It is immediate to verify that, if $(X,Y,Z)$ is a minimum $s$-$t$ vertex-cut in $\hat G$, then $(X\setminus\set{s},Y,Z\setminus\set{t})$ is a minimum $S$-$T$ vertex-cut in $G$.
	We then compute a split graph $\hat G'$ corresponding to $\hat G$, and
	use the algorithm from \Cref{cor: mincut} to compute the minimum $s^{\out}$-$t^{\inn}$ edge-cut $(L,R)$ in $\hat G'$, whose value must be equal to the value of a minimum $s$-$t$ vertex-cut in $\hat G$, from \Cref{cor : minimum cut in split graph}. 
	We convert this edge-cut into  an $s$-$t$ vertex-cut $(X,Y,Z)$ in $\hat G$, using the procedure outlined in \Cref{obs: existential from split to regular}, namely: we let $X=\{v \in V(\hat G) \setminus \{t\} \mid v^{\out} \in L\}$, $Y=\set{v\in V(\hat G)\mid v^{\inn}\in L,v^{\out}\in R}$, and $Z=V(\hat G)\setminus(X\cup Y)$. From \Cref{obs: existential from split to regular}, $(X,Y,Z)$ is an $s$-$t$ vertex-cut in $G$, with $w(Y) \leq c(L,R)$, so it must be a minimum $s$-$t$ vertex-cut in $\hat G$. We then output 
	$(X\setminus\set{s},Y,Z\setminus\set{t})$ as a minimum $S$-$T$ vertex-cut in $G$. \end{proof}

\subsection{Some Useful Claims}

We will use the following claim that can be thought of as a generalization of Proposition $3.2$ in~\cite{LNP21} from unit vertex weights to general weights.
The claim was proved in \cite{CT24} for undirected graphs, but the proof works as is in directed graphs as well. For completeness we provide it here.

\begin{claim}\label{claim: neighbors in S}
	Let $G$ be a simple directed graph with weights $w(v)> 0$ on its vertices $v\in V(G)$, and let  $(L,S,R)$ be a global minimum vertex-cut in $G$. Let $x\in L$ be any vertex, and denote by $S'=S\setminus N^+_G(x)$. Then $w(S')\leq w(N^+_G(x)\cap L)$ must hold, and in particular $w(S')\leq w(L)$. 
\end{claim}
\begin{proof}	
	Let $t\in R$ be any vertex. From the Max-Flow/Min-Cut theorem, there is a flow $f$ in $G$ from $x$ to $t$, that respects all vertex capacities (that are equal to their weights), whose value is $\opt=w(S)$. Let $\pset$ be the collection of flow-paths in the flow-path decomposition of $f$. We partition $\pset$ into two sets: set $\pset'\subseteq \pset$ contains all paths $P$, whose second vertex lies in $S$, and $\pset'' = \pset\setminus\pset'$ contains all remaining flow-paths. Notice that, for every path $P\in \pset''$, the second vertex of $P$ must lie in $N^+_G(x)\cap L$, and so $\sum_{P\in \pset''}f(P)\leq w(N^+_G(x)\cap L)$ must hold. Therefore, $\sum_{P\in \pset'}f(P)\geq \opt-w(N^+_G(x)\cap L)=w(S)-w(N^+_G(x)\cap L)$. For each flow-path $P\in \pset'$, the second vertex of $P$ lies in $S\cap N^+_G(x)$, so $w(S\cap N^+_G(x))\geq w(S)-w(N^+_G(x)\cap L)$, and $w(S')=w(S)-w(S\cap N^+_G(x))\leq w(N^+_G(x)\cap L)$.
\end{proof}

We will also use the following lemma, that was proved in \cite{CT24} for undirected graphs. The proof of \cite{CT24} works for directed graphs as well; we provide it here for completeness. The only difference is  the use of \Cref{claim: neighbors in S} instead of its counterpart for undirected graphs. The lemma uses the convention mentioned above that, in a global minimum vertex-cut $(L,S,R)$, $w(L)\leq w(R)$ must hold.

\begin{lemma}[Analogue of Lemma 3.7 in \cite{CT24}]\label{lem: cut if vertex of L}
	There is a randomized algorithm, that, given a simple directed $n$-vertex and $m$-edge graph $G$ with integral weights $1\leq w(v)\leq W'$ on its vertices $v\in V(G)$, and a vertex $x\in V(G)$, returns a vertex $y\in V(G)$, and the value $c$ of the minimum $x$-$y$ vertex-cut in $G$. Moreover, if there is a global minimum vertex-cut $(L,S,R)$ in $G$ with $x\in L$, then with probability at least $1-1/n^4$, $c=w(S)$ holds. The running time of the algorithm is $O\left (m^{1+o(1)}\cdot \log W'\right )$.
\end{lemma}
\begin{proof}
	Let $Z=V(G)\setminus\left(\set{x}\cup N_G^+(x)\right )$.  Let $\aset$ be an algorithm that selects a vertex $y\in Z$ at random, where the probability to select a vertex $v\in V(G)$ is $\frac{w(v)}{w(Z)}$. It then uses the algorithm from \Cref{cor: min_vertex_cut} to compute the value $c$ of minimum $x$-$y$ vertex cut in $G$, in time $O\left (m^{1+o(1)}\cdot \log W'\right )$, and returns vertex $y$ and the value $c$. 
	We use the following observation.

	\begin{observation}\label{obs: alg2 success}
		Suppose there is a global minimum vertex-cut $(L,S,R)$ in $G$ with $x\in L$. Let $(c,y)$ be the outcome of the algorithm $\aset$. Then with probability at least $\frac{1}{3}$, $y\in R$ and $c=\opt$ holds.
	\end{observation}
	\begin{proof}
		Consider a global minimum vertex-cut $(L,S,R)$ with $x\in L$. Let $S'=S\setminus N^+_G(x)$. Then from \Cref{claim: neighbors in S},  $w(S')\leq w(L)$ must hold.
		
		We conclude that $w(Z)\leq w(L)+w(R)+w(S')\leq 2w(L)+w(R)\leq 3w(R)$, since we always assume by default that, in a global minimum vertex-cut $(L,S,R)$, $w(L)\leq w(R)$ holds. Therefore,  the probability that Algorithm $\aset$ chooses a vertex $y\in R$ is $\frac{w(R)}{w(Z)}\geq \frac{1}{3}$. Clearly, if $y\in R$ holds, then the value $c$ of the minimum $x$-$y$ vertex cut in $G$ is equal to $\opt$.	
	\end{proof}
	
	We are now ready to complete the proof of \Cref{lem: cut if vertex of L}. We execute Algorithm $\aset$ $\Theta(\log n)$ times, and output the smallest value $c$ that the algorithm returned, together with the corresponding vertex $y$. It is immediate to verify that $c$ is indeed the value of the minimum $x$-$y$ vertex cut in $G$. Assume now that there is a global minimum vertex cut $(L,S,R)$ in $G$ with $x\in L$. We say that an application of Algorithm $\aset$ is \emph{successful}, if it returns a value $c'=\opt$. From \Cref{obs: alg2 success}, the probability that a single application of the algorithm is successful is at least $1/3$. It is then easy to verify that the probability that every application of Algorithm $\aset$ was unsuccessful is at most $1/n^4$. If at least one application of the algorithm was successful, then we are guaranteed that the value $c$ returned by our algorithm is equal to $\opt$. Since the running time of Algorithm $\aset$ is $O(m^{1+o(1)}\cdot \log W')$, the total running time of the entire algorithm is $O(m^{1+o(1)}\cdot \log W')$.
\end{proof}

\subsection{Local Flow}

Our algorithm for dense weighted graphs uses a poweful tool called the \emph{local flow technique}, that was also exploited in numerous previous works (see, e.g. \cite{OZ14,HRW17,NSWN17,SW19}).
At a very high level, in the local flow problem, the goal is to compute an $s$-$t$ flow of a prescribed value $\Delta$ in a given directed edge-capacitated graph, in time that is proportional to $\Delta$; so in particular, the running time of the algorithm may be significantly lower than the size of the input graph. More specifically, the input to the local flow problem consists of an edge-capacitated directed graph $G$ with two special vertices $s$ and $t$, and two additional parameters $\Delta\geq 1$ and $h\geq 1$; here, $\Delta$ is the desired flow value, and the parameter $h$ can be thought of as controlling the tradeoff between the approximation factor of the algorithm and its running time. The algorithm must output a valid $s$-$t$ flow $f$ in $G$, together with an $s$-$t$ edge-cut $(X,Y)$, so that, if $\val(f)<\Delta$, then the sparsity of the edge-cut $(X,Y)$ in the residual flow network $G_f$ is at most roughly $\frac{1}{h}$. We note that, for some specific and often used settings of the edge capacities, an edge-cut $(X,Y)$ that is sparse in $G_f$ can also be shown to be sparse in the original graph $G$.
Moreover, typically, as long as every vertex $v$ of the input graph is connected to the destination vertex $t$ with an edge whose capacity is at least roughly $\deg_G(v)$, one can obtain an algorithm for local flow whose running time  is close to $O(\Delta \cdot h\cdot n^{o(1)})$ (see e.g. \cite{HRW17,NSWN17,SW19}.)
Several algorithms were developed for this problem in the past, some of which are based on the push-relabel framework of Goldberg and Tarjan \cite{GT88} (see e.g. \cite{HRW17,NSWN17,SW19}), and some on the blocking flows of Dinic \cite{dinic1970algorithm} (see e.g. \cite{OZ14}). Both types of algorithms follow the same approach of starting with an initial flow $f=0$, and then gradually augmenting it, until the distance from $s$ to $t$ in the residual flow network $G_f$ becomes at least roughly $h$. Using the standard ball-growing technique one can then compute an $s$-$t$ edge-cut $(X,Y)$ that is sparse in $G_f$.
For technical reasons, we need a slight variation of this technique, namely: in the case where the value of the $s$-$t$ flow $f$ that the algorithm computes is lower than the target value $\Delta$, we need the algorithm to compute an intermediate object, instead of a sparse cut in $G_f$. This intermediate object is a low-volume vertex subset $A\subseteq V(G)$, containing all vertices that are sufficiently close to the source vertex $s$ in the residual flow network $G_f$. Specifically, we require that, for every vertex $v\in V(G)\setminus A$, every path connecting $s$ to $v$ in $G_f$ contains at least $h$ edges. Generally, given such a set $A$ of vertices, it is not difficult to construct an $s$-$t$ edge-cut whose sparsity in $G_f$ is at most roughly $\frac{1}{h}$ using the standard ball-growing technique; however, our algorithm needs to use the properties of the vertex set $A$ directly, instead of exploiting such a cut. We also slightly modify the standard setting in that we only require that every vertex $v\in V(G)$ is connected to the destination vertex $t$ with an edge whose capacity is at least roughly $\deg^+_G(v)$, instead of the more standard bound of $\deg_G(v)$. Our algorithm for computing local flow is summarized in the following theorem. The proof uses standard techniques and is defferred to Section \ref{sec: proof-of-unit-flow} of Appendix.

\begin{theorem} \label{thm : unit-flow}
	There is a deterministic algorithm, called \alglocal, whose input consists of a simple directed graph $G$ with integral capacities $c(e) \geq 0$ on its edges $e \in E(G)$ given in the adjacency-list representation, a pair $s,t\in V(G)$ of its vertices, and integral parameters $\Delta,h \geq 1$, such that $\deg^+_{G}(s) \leq \Delta$ holds.
	Let $\hat G$ be the graph obtained from $G$ by adding, for every vertex $v\in V(G)\setminus\set{t}$ with $\deg^+_G(v)>0$, an edge $(v,t)$ of capacity $\deg^+_G(v)$ to the graph.
	The algorithm returns an $s$-$t$ flow $f$ in $\hat{G}$ (by only listing the flow values $f(e)$ for edges $e\in E(\hat G)$ with $f(e)>0$), and a set $A \subseteq V(G) \setminus \set{t}$ of vertices with $\vol^+_{G}(A) \leq \Delta$ and $s \in A$.
	Let $H=\hat G_{f}$ be the residual flow network of graph $\hat G$ with respect to the flow $f$. The algorithm additionally guarantees that,  if $\val(f) < \Delta$, then every path in $H$ that connects $s$ to a vertex of $V(\hat G) \setminus A$ contains at least $h$ edges.
	The running time of the algorithm is $\Tilde{O}(h \cdot \Delta)$.
\end{theorem}

We note that \cite{CT24} use a technique they call "local flow augmentation", that was first introduced in~\cite{CHILP17} and then further refined in~\cite{FNY20,CQ21}, and which we also use in our algorithm for weighted non-dense graphs.
This technique, while having some similarities to the local flow algorithm from \Cref{thm : unit-flow}, also has some fundamental differences.
Specifically, one can think of the setting of \Cref{thm : unit-flow}, in which every vertex of the graph $\hat G$ is connected to the destination vertex $t$, as allowing each such vertex to absorb a small amount of flow. In contrast, the local flow augmentation technique can be thought of as randomly selecting a small subset of vertices and allowing each such vertex to absorb an arbitrarily large amount of flow.
While the local flow augmentation technique allows the resulting flow $f$  to satisfy somewhat stronger properties than the flow $f$ that we construct in \Cref{thm : unit-flow}, one drawback of this technique is that, 
in the event where $\val(f)$ is small, the algorithm returns a subset of vertices\footnote{As an example, the algorithm of \cite{CT24} computes a subgraph $J$ of the residual flow network of the split-graph of $G$; we refer here to the subset of vertices of the original graph $G$ that $J$ naturally defines.} with significantly weaker properties than those required from the vertex set $A$ in \Cref{thm : unit-flow}. 
Additionally, using \Cref{thm : unit-flow} instead of the local flow augmentation technique allows our algorithm for weighted global minimum vertex-cut in dense graphs to compute the flow $f$ faster than the algorithm of \cite{CT24}, which, in turn, is crucial to obtaining a faster running time in the dense regime. 

\subsection{Chernoff Bound}

We use the following standard version of the Chernoff Bound (see. e.g., \cite{dubhashi2009concentration}).

\begin{lemma}[Chernoff Bound]
	\label{lem: Chernoff}
	Let $X_1,\ldots,X_n$ be independent random variables taking values in $\set{0,1}$. Let $X=\sum_{1\le i\le n}X_i$, and let $\mu=\expect{X}$. Then for any $t>2e\mu$,
	\[\Pr\Big[X>t\Big]\le 2^{-t}.\]
	Additionally, for any $0\le \delta \le 1$,
	\[\Pr\Big[X<(1-\delta)\cdot\mu\Big]\le e^{-\frac{\delta^2\cdot\mu}{2}}.\]
\end{lemma}

\subsection{A Subgraph Oracle}

Our algorithm uses a subgraph oracle, that was defined in \cite{CT24} for undirected graphs; we extend their definition to directed graphs in a straightforward manner.

\begin{definition}[A Subgraph Oracle]\label{def: subgraph oracle}
	A $(\delta,\tau)$-subgraph oracle,  for parameters $0<\delta\leq 3/4$ and $\tau\geq 2$, receives as input a directed $n$-vertex graph $G$ in the adjacency list representation and a set $Z\subseteq V(G)$ of its vertices. The oracle must return a partition $(Y^h,Y^{\ell})$ of $V(G)$, and a collection $E'$ of at most $n^{2-\delta}\cdot \log^2\tau$ edges of $G$. We say that the oracle \emph{errs}, if, for some vertex $v\in Y^{h}$, $|N^+_G(v)\cap Z|<\frac{n^{1-\delta}}{1000\log n}$ holds, or, for some vertex $v'\in Y^{\ell}$, $|N^+_G(v)\cap Z|>n^{1-\delta}\cdot \log \tau$ holds, or $E'\neq E_G(Y^{\ell},Z)$. We require that the probability that the oracle errs is at most $\frac{1}{n^5\cdot \log^4\tau}$. 
\end{definition}


The following lemma is a straightforward extension of an analogous result of \cite{CT24} to directed graphs, that provides a relatively efficient algorithm for responding to up to $n$ subgraph oracle queries in bulk. We slightly adapt the lemma statement by explicitly stating the dependence of the running time on the number of queries, which is straightforward to do. The lemma statement uses the matrix multiplication exponent $\omega$, whose current bound is $\omega\leq 2.371552$~\cite{WXXZ24}. The proof is essentially identical to that from \cite{CT24} and is included in Section \ref{subsec: implementing subgraph oracle} of Appendix for completeness.

\begin{lemma}[Extension of Lemma 3.17 in \cite{CT24} to directed graphs]\label{lem: oracle in bulk}
	There is a randomized algorithm, that is given as input a directed $n$-vertex graph $G$, parameters $\frac{1}{\sqrt{\log n}}<\delta\leq \frac{3}{4}$ and $\tau>0$ that is greater than a large enough constant, and $q = n^{1-\epsilon}$ queries $(Z_1,\ldots,Z_q)$ to the $(\delta,\tau)$-subgraph oracle on $G$, for any $0\leq\epsilon\leq 1$. The algorithm computes responses $(Y^h_1,Y^{\ell}_1,E'_1),\ldots,(Y^h_q,Y^{\ell}_q,E'_q)$ to the queries, such that the probability that the algorithm errs in its response to any query is at most $\frac{1}{n^9\cdot\log^4\tau}$. The running time of the algorithm is: 
	
	\[O\left (\left(n^{3-(\delta+\epsilon)\cdot 2(3-\omega)/(5-\omega)+o(1)} + n^{2+\delta+o(1)}\right)\cdot \poly\log \tau\right )\leq O\left (\left(n^{3-0.478(\delta+\epsilon)+o(1)} + n^{2+\delta+o(1)}\right)\cdot \poly\log \tau\right ),\]
	where $\omega$ is the matrix multiplication exponent.
\end{lemma}

%

%
\section{Algorithms for Easy Special Cases}
\label{sec: special}

Suppose we are given a directed graph $G$ with integral weights $w(v)\geq 1$ on its vertices. Observe first that, if $G$ is not a strongly connected graph, then there is a partition $(L,R)$ of vertices of $V(G)$ into two non-empty subsets, so that no edge connects a vertex of $L$ to a vertex of $R$, and so the value of the global minimum vertex-cut in $G$ is $0$. Moreover, by performing two DFS searches from an arbitrary vertex $v$ of $G$, one in the graph $G$ itself, and one in the graph obtained from $G$ by reversing the direction of all its edges, we can check, in time $O(m)$ whether $G$ is strongly connected, and, if this is not the case, compute the vertex-cut of value $0$. Therefore, we will assume from now on that the input graph $G$ is strongly connected.

In the remainder of this section we provide a simple algorithm for some additional easy special cases of the problem: namely, where there is a global minimum cut $(L,S,R)$, such that either $|L|$ is sufficiently large, or $\vol_G(L)$ is high. 
The algorithm is summarized in the following theorem, and it closely resembles algorithms that were employed in previous work for similar settings.

\begin{theorem}\label{thm: alg for special}\label{thm : alg for when L is large}
	There is a randomized algorithm, that we denote by $\algspec$, that receives as input a directed $n$-vertex and $m$-edge graph $G=(V,E)$ with integral weights $1\leq w(v)\leq W$ on its vertices $v\in V$, such that $G$ contains some vertex-cut, together with a parameter $0<\eps<1$.
	Let $d=\frac{2m}{n}$ denote the average vertex degree in $G$. The algorithm  returns a vertex-cut $(L',S',R')$ in $G$.  Moreover, if there exists a global minimum vertex-cut $(L,S,R)$ in $G$, for  which either $|L|\geq n^{\eps}$ or $\vol_G(L)\geq n^{\eps}\cdot d$ hold, then, with probability at least $(1-1/n^2)$, the vertex-cut $(L',S',R')$ is a global minimum vertex-cut. The running time of the algorithm is $O(mn^{1-\eps+o(1)}\log W)$.
\end{theorem}

\begin{proof}
	We design two simple algorithms, denoted by $\alg_1$ and $\alg_2$, respectively. Each of the two algorithms, given the graph $G$ and the parameter $\eps$ as above, produces a vertex-cut in $G$. Additionally, if $|L|\geq n^{\eps}$, then the cut produced by $\alg_1$ is guaranted  to be a global minimum vertex-cut  with probability at least $(1-1/n^2)$, and, if $\vol_G(L)\geq n^{\eps}\cdot d$, then the cut produced by $\alg_2$ has a similar guarantee. Our algorithm then returns the smaller-value cut among the two resulting cuts. In order to complete the proof of \Cref{thm: alg for special}, it is now enough to provide Algorithms $\alg_1$ and $\alg_2$ with these properties. Both algorithms are elementary and are similar to algorithms that appeared in prior work in similar settings, e.g. see Algorithm $\alg_2$ in Section 3.2 of \cite{CT24}. We now describe Algorithm $\alg_1$.
	
\paragraph{Algorithm $\alg_1$.}
Let $q=128n^{1-\eps}\cdot \ceil{\log n}$. We perform $q$ iterations. In the $i$th iteration, we select a vertex $x_i\in V(G)$ uniformly at random, and then 
apply the algorithm from 	\Cref{lem: cut if vertex of L} to vertex $x_i$. Recall that the algorithm returns a vertex $y_i\in V(G)$, and the value $c_i$ of the minimum $x_i$-$y_i$ vertex-cut in $G$. The algorithm also guarantees that,  if there is a global minimum vertex-cut $(L,S,R)$ in $G$ with $x_i\in L$, then with probability at least $1-1/n^4$, $c_i=w(S)$ holds. Let $1\leq i^*\leq q$ be the index for which $c_{i^*}$ is minimized, breaking ties arbitrarily. We use the algorithm from \Cref{cor: min_vertex_cut} to compute a minimum $x_{i^*}$-$y_{i^*}$ cut in $G$, that we denote by $(L_1,S_1,R_1)$. Recall that the value of this cut is guaranteed to be equal to $c_{i^*}$. We then return the cut $(L_1,S_1,R_1)$ as the outcome of the algorithm. This completes the description of Algorithm $\alg_1$. We now bound its running time. The running time of a single execution of the algorithm from \Cref{lem: cut if vertex of L} is $O\left (m^{1+o(1)}\cdot \log W\right )$, and this algorithm is executed $O(n^{1-\eps}\log n)$ times. The running time of the algorithm from  \Cref{cor: min_vertex_cut} is $O\left (m^{1+o(1)}\cdot \log W\right )$, and this algorithm is executed only once. Therefore, the total running time of algorithm $\alg_1$ is bounded by $O(mn^{1-\eps+o(1)}\log W)$.

We now proceed to analyze Algorithm $\alg_1$.
Let $\event$ be the bad event that no vertex of $\set{x_1,\ldots,x_q}$ lies in $L$. Note that, if $|L|\geq n^{\eps}$, then the probability that a randomly chosen vertex $x_i$ of $V(G)$ does not lie in $L$ is bounded by $\left(1-\frac{1}{n^{1-\eps}}\right )$. Since $q=128n^{1-\eps}\cdot \ceil{\log n}$, if $|L|\geq n^{\eps}$, then:

\[\prob{\event}\leq \left(1-\frac{1}{n^{1-\eps}}\right )^q=\left(1-\frac{1}{n^{1-\eps}}\right )^{128n^{1-\eps}\cdot \ceil{\log n}}\leq \frac{1}{n^{3}}. \]

Assume now that Event $\event$ does not happen, and let $1\leq i'\leq q$ be any index with $x_{i'}\in L$. Let $\opt$ be the value of the global minimum vertex-cut in $G$, and $\event'$ be the bad event that $c_{i'}\neq \opt$. Then from \Cref{lem: cut if vertex of L}, $\prob{\event'\mid \neg\event}\leq \frac{1}{n^4}$. If neither of the events $\event,\event'$ happen, then we are guaranteed that $c_{i^*}=c_{i'}=\opt$, and so the value of the vertex-cut $(L_1,S_1,R_1)$ returned by the algorithm is $\opt$. Therefore, altogether, if $|L|\geq n^{\eps}$, then the probability that the cut $(L_1,S_1,R_1)$ is not a global minimum cut is bounded by:

\[\prob{\event}+\prob{\event'\mid\neg\event}\leq \frac{1}{n^{3}}+\frac{1}{n^{4}}\leq \frac{1}{n^2}. \]

We now proceed to describe Algorithm $\alg_2$, which is quite similar.

\paragraph{Algorithm $\alg_2$.}
We again denote $q=128n^{1-\eps}\cdot \ceil{\log n}$ and perform $q$ iterations. In the $i$th iteration, we select an edge $e_i\in E(G)$ uniformly at random, and we denote its endpoints by $x'_{2i-1}$ and $x'_{2i}$, respectively. Let $X=\set{x'_1,\ldots,x'_{2q}}$ be the resulting set of vertices. For all $1\leq j\leq 2q$, we
apply the algorithm from 	\Cref{lem: cut if vertex of L} to vertex 
$x'_j$. Recall that the algorithm returns a vertex $y'_j\in V(G)$, and the value $c'_j$ of the minimum $x'_j$-$y'_j$ vertex-cut in $G$. The algorithm also guarantees that,  if there is a global minimum vertex-cut $(L,S,R)$ in $G$ with $x'_j\in L$, then with probability at least $1-1/n^4$, $c'_j=w(S)$ holds. Let $1\leq j^*\leq 2q$ be the index for which $c'_{j^*}$ is minimized, breaking ties arbitrarily. We use the algorithm from \Cref{cor: min_vertex_cut} to compute a	 minimum $x_{j^*}$-$y_{j^*}$ cut in $G$, that we denote by $(L_2,S_2,R_2)$. Recall that the value of this cut is guaranteed to be equal to $c_{j^*}$. We then return the cut $(L_2,S_2,R_2)$ as the outcome of the algorithm. This completes the description of Algorithm $\alg_2$. Its running time is bounded similarly to that of Algorith $\alg_1$: the running time of a single execution of the algorithm from \Cref{lem: cut if vertex of L} is $O\left (m^{1+o(1)}\cdot \log W\right )$, and this algorithm is executed $O(n^{1-\eps}\log n)$ times. The running time of the algorithm from  \Cref{cor: min_vertex_cut} is $O\left (m^{1+o(1)}\cdot \log W\right )$, and this algorithm is executed only once. Therefore, the total running time of algorithm $\alg_2$ is bounded by: $O(mn^{1-\eps+o(1)}\log W)$.

Let $\hat \event$ be the bad event that no vertex of $X$ lies in $L$. 
Let $E'\subseteq E(G)$ be the set of all edges incident to the vertices of $L$.  Clearly, Event $\hat \event$ may only happen if neither of the edges $e_1,\ldots,e_{q}$ lies in $E'$.
Note that, if $\vol(L)\geq n^{\eps}\cdot d$, then $|E'|\geq \frac{n^{\eps}\cdot d}{2}$, and the probability that a randomly chosen edge $e_i$ of $E(G)$ does not lie in $E'$ is bounded by: 

\[\left(1-\frac{|E'|}{m}\right )\leq \left(1-\frac{n^{\eps}d/2}{nd/2}\right ) =\left(1-\frac{1}{n^{1-\eps}}\right ).\] 

Since $q=128n^{1-\eps}\cdot \ceil{\log n}$, we get that,  if $\vol(L)\geq n^{\eps}\cdot d$, then:

\[\prob{\hat \event}\leq \left(1-\frac{1}{n^{1-\eps}}\right )^q=\left(1-\frac{1}{n^{1-\eps}}\right )^{128n^{1-\eps}\cdot \ceil{\log n}}\leq \frac{1}{n^{3}}. \]

Assume now that Event $\hat \event$ does not happen, and let $1\leq i'\leq 2q$ be any index with $x_{i'}\in L$. Let $\opt$ be the value of the global minimum vertex-cut in $G$, and let $\hat \event'$ be the bad event that $c_{i'}\neq \opt$. Then from \Cref{lem: cut if vertex of L}, $\prob{\hat \event\mid \neg\hat \event}\leq \frac{1}{n^4}$. If neither of the events $\hat \event,\hat\event'$ happen, then we are guaranteed that $c_{i^*}=c_{i'}=\opt$, and so the value of the vertex-cut $(L_2,S_2,R_2)$ returned by the algorithm is $\opt$. Therefore, altogether, if $\vol(L)\geq n^{\eps}\cdot d$, then the probability that the cut $(L_2,S_2,R_2)$ is not a global minimum cut is bounded by:

\[\prob{\hat \event}+\prob{\hat \event'\mid\neg\hat \event}\leq \frac{1}{n^2}. \]

Algorithm $\algspec$ executes Algorithms $\alg_1$ and $\alg_2$ and then returns one of the cuts $(L_1,S_1,R_1)$ or $(L_2,S_2,R_2)$, whose value is smaller. From the above discussion, if there exist a global minimum vertex-cut $(L,S,R)$ in $G$, for  which either $|L|\geq n^{\eps}$ or $\vol_G(L)\geq n^{\eps}\cdot d$ hold, then, with probability at least $(1-1/n^2)$, the returned cut is a global minimum vertex-cut. The running time of the algorithm is $O(mn^{1-\eps+o(1)}\log W)$.
\end{proof}

\section{Some Useful Tools}
\label{sec: tools}

In this section we introduce several additional tools that will be used in our algorithms for both unweighted and weighted graphs. 

Suppose we are given a directed $n$-vertex and $m$-edge graph $G=(V,E)$ with integral weights $1\leq w(v)\leq W$ on its vertices $v\in V$ and a parameter $0<\eps<1$. Throughout, we denote by $d=2m/n$ the average vertex degree in $G$. Recall that algorithm $\algspec$ from \Cref{thm: alg for special} allows us to 
 compute a global minimum vertex-cut in $G$ in the special case where there exists a global minimum vertex-cut $(L,S,R)$ in $G$, for  which either $|L|\geq n^{\eps}$ or $\vol_G(L)\geq n^{\eps}\cdot d$ hold.
Therefore, for some of our algorithms, it is sufficient to consider the case where, for every global minimum vertex-cut $(L,S,R)$ in $G$, $|L|<n^{\eps}$ and $\vol_G(L)<n^{\eps}\cdot d$. Note that, if $(L,S,R)$ is a global minimum vertex-cut in $G$, then, for every vertex $v\in S$, there must be an edge $(u,v)$ with $u\in L$, since otherwise we could move $v$ from $S$ to $R$ and obtain a valid vertex-cut of a smaller value. Therefore, we can assume that $|S|\le \vol(L)< n^{\eps}\cdot d$ holds as well.
To summarize, some of our algorithms will assume that there is a global minimum vertex-cut $(L,S,R)$ in $G$  with the following properties:

\begin{properties}{P}
	\item $|L|< n^{\eps}$;\label{prop: small L}
	\item $|S|< n^{\eps}\cdot d$; and \label{prop: small S}
	\item $\vol_G(L)<n^{\eps}\cdot d$. \label{prop: small volume of L}
\end{properties}

This motivates our definition of high-degree vertices:

\begin{definition}[High-degree vertices]\label{def: high-deg}
	We say that a vertex $v\in V(G)$ is a \emph{high-degree vertex} if $\deg_G(v)\geq d\cdot n^{\eps}$, where $d$ is the average vertex degree in $G$; otherwise, it is a \emph{low-degree vertex}. We denote by $\Vhd=\set{v\in V(G)\mid \deg_G(v)\geq d\cdot n^{\eps}}$ the set of all high-degree vertices of $G$, and by  $\Vld=\set{v\in V(G)\mid \deg_G(v)< d\cdot n^{\eps}}$ the set of all low-degree vertices.
\end{definition}

Note that, if $(L,S,R)$ is a global minimum vertex-cut for which Property \ref{prop: small volume of L} holds, then set $L$ may not contain high-degree vertices.

\subsection{Selection of Terminals and Parameter $\lambda$}

It will be convenient for us to assume that the algorithm is given a correct (approximate) estimate $\lambda$ on the value $|L|$. Additionally, in some of our algorithms, we need to select a relatively small collection $T$ of vertices of $G$ called terminals, so that at least one vertex $v$ of $T$ lies in $L$, and the degree of $v$ in $G$ is not much higher than the average degree of all vertices in $L$. In the following claim we provide a simple algorithm that allows us to achieve both these goals.

\begin{claim}\label{claim: compute terminals}
	There is a randomized algorithm, that we call $\algterm$, that receives as input an $n$-vertex directed graph $G$ with integral weights $w(v)\geq 1$ on its vertices $v\in V(G)$, and returns 
	a value $1\le \lambda\leq n$ that is an integral power of $2$, together with a set $T$ of at most $\ceil{\frac{100n\log n}{\lambda}}$ vertices of $G$. We say that the algorithm is \emph{successful} with respect to a fixed global minimum vertex-cut $(L,S,R)$, if  (i) $\frac{\lambda}{2}\leq |L|\leq\lambda$; (ii) $T\cap L\neq \emptyset$; and (iii) there is a vertex $v\in T\cap L$ with $\deg_G(v)\leq 4\ceil{\frac{\vol_G(L)}{\lambda}}$. Let $(L,S,R)$ be a global minimum vertex-cut in $G$. The running time of the algorithm is $O(n)$, and it is successful with respect to $(L,S,R)$ with probability at least  $\frac{1}{4\log n}$.
\end{claim}
\begin{proof}
	The algorithm starts by selecting an integer $1\leq i\leq \ceil{\log n}$ uniformly at random and setting $\lambda=2^i$. Additionally, it selects a set $T$ of  vertices of $G$ by performing  $\ceil{\frac{100n\log n}{\lambda}}$ trials; in every trial, a single vertex is selected uniformly at random from $G$ and is added to $T$. Clearly, the running time of the algorithm is bounded by $O(n)$.
	
	Consider now some fixed global minimum vertex-cut $(L,S,R)$. Let $L'\subseteq L$ be the set of all vertices $v\in L$ with $\deg_G(v)\leq 4\ceil{\frac{\vol_G(L)}{\lambda}}$. 
	Let $\event_1$ be the good event that $\frac{\lambda}{2}\leq |L|\leq \lambda$. It is easy to verify that $\prob{\event_1}\geq \frac{1}{2\log n}$.
	
	Next, we let $\event_2$ be the good event that in some trial, a vertex $x\in L'$ was selected. Note that, if Event $\event_1$ happened, 
	then the number of vertices $v\in L$ with $\deg_G(v)> 4\ceil{\frac{\vol_G(L)}{\lambda}}$ must be bounded by $\frac{\lambda}{4}\leq \frac{|L|}{2}$. Therefore, $|L'|\geq \frac{|L|}{2}$ must hold. The probability that in a single trial a vertex of $L'$ is selected is at least $\frac{|L'|}{n}\geq \frac{\lambda}{4n}$. Therefore:

	\[\prob{\neg \event_2\mid \event_1}\leq \left(1-\frac{\lambda}{4n}\right )^{\ceil{(100n\log n)/\lambda}}\leq e^{-25\log n},\]
	
	and $\prob{\event_2\mid \event_1}\geq 1-\frac{1}{n}$.
	
	Overall, if both of the events $\event_1$ and $\event_2$ happen, then Algorithm \algselectterm is successful with respect to $(L,S,R)$. Therefore, 
	the  the probability that the algorithm is successful is at least:
	
	\[\prob{\event_1\band\event_2}\geq \prob{\event_1}\cdot \prob{\event_2\mid \event_1}\geq \frac{1}{4\log n}. \]
\end{proof}

\subsection{A New Tool: an Approximate Furthest Min-Cut}\label{sec : approximately-minimum-containing-cut-problem}

It is well known that, for any edge-capacitated graph $G$ and a pair $x,y\in V(G)$ of distinct vertices, there exists a minimum $x$-$y$ edge-cut $(X,Y)$ in $G$ with the following property: for every other minimum $x$-$y$ edge-cut $(X',Y')$ in $G$, $X' \subseteq X$ holds. Such a cut $(X,Y)$ is sometimes referred to as the \emph{furthest minimum $x$-$y$ cut} (in fact, an analogous notion of ``closest min-cut'' was also used in the past, perhaps even more often, see e.g. \cite{ASY25}).
The furthest minimum $x$-$y$ cut $(X,Y)$ can be computed by first computing a maximum $x$-$y$ flow $f$ in $G$, and then letting $X$ contain all vertices $v\in G$ such that there is no $v$-$y$ path in the resulting residual flow network $G_f$.
In this section, we generalize the notion of furthest cuts to capture not just exact minimum $x$-$y$ cuts, but also approximate ones, and provide an efficient algorithm for computing such cuts. 
We refer to the corresponding problem as the \emph{Approximate Furthest Min-Cut} problem, or \AFMC for short. We distinguish between the edge- and the vertex-versions of the problem, that we refer to as edge-\AFMC and vertex-\AFMC, respectively. We start by defining the edge version of the problem.

\begin{definition}[Edge-\AFMC]\label{def : almost-minimum-furthest-edge-cut}
In the edge version of the Approximate Furthest Min-Cut problem (edge-\AFMC), the input consists of a directed $n$-vertex and $m$-edge graph $G$ with integral capacities $0\leq c(e) \leq W$ on its edges $e \in E(G)$, a pair $x,y \in V(G)$ of distinct vertices, and an integral parameter $\alpha>0$.
We denote the value of the minimum $x$-$y$ edge-cut in $G$ by $\opt_{x,y}^{E}$.
The goal is to compute an $x$-$y$ edge-cut $(X,Y)$ in $G$ of value $c(X,Y)\leq \opt_{x,y}^{E}+\alpha$, such that, for every other $x$-$y$ edge-cut $(X',Y')$ the following holds:
\begin{equation}\label{eq : condition-in-almost-furthest-edge-cut-problem}
	|X' \setminus X| \leq \left(c(X',Y') - \opt_{x,y}^{E}\right )\cdot\frac{n}{\alpha}.
\end{equation}
\end{definition}

 Observe that, unlike in the case of the exact furthest minimum edge-cut, we no longer require that the cut $(X,Y)$ is the minimum $x$-$y$ edge-cut; instead, it is sufficient that its value is within an additive $\alpha$ value from the optimal one. Additionally, for all $\beta>0$, for every other $x$-$y$ edge-cut $(X',Y')$ whose value is at most $\opt^E_{x,y}+\beta$, we only require that all but at most $\frac{\beta\cdot n}{\alpha}$ vertices of $X'$ are contained in $X$, relaxing the stricter requirement that $X'\subseteq X$ in the exact furthest minimum edge-cut.
%
%
In the following theorem, whose  proof appears in \Cref{sec : proof of edge capacitated approximately minimum containing cut}, we provide an efficient algorithm for the edge-\AFMC problem.

\begin{theorem}\label{thm : edge capacitated approximately minimum containing cut}
There is a deterministic algorithm for the edge-\AFMC  problem with running time $O\left ((n+m)^{1+o(1)}\log W\cdot \log(2+\alpha)\right )$.
\end{theorem}

In our algorithm for the weighted version of global minimum vertex-cut in dense graphs, we will need to solve the vertex version of the \AFMC problem, that we refer to as \emph{vertex-\AFMC} and define next.
The vertex version is very similar to the edge version, except that it naturally uses vertex-cuts $(L,S,R)$ and $(L',S',R')$ instead of the edge-cuts $(X,Y)$ and $(X',Y')$, and it requires that the cardinality of $(L' \cup S') \setminus (L \cup S)$ is bounded, instead of the cardinality of $X' \setminus X$.

\begin{definition}[Vertex-\AFMC]\label{def : almost-minimum-furthest-vertex-cut}
 In the vertex version of the Approximate Furthest Min-Cut problem (vertex-\AFMC), the input is a directed $n$-vertex and $m$-edge graph $G$ with integral weights $0\leq w(v) \leq W$ on its vertices $v \in V(G)$, a pair of distinct vertices $x,y \in V(G)$ such that $(x,y) \notin E(G)$, and an integral parameter $0<\alpha<\frac{\wmax(G)}{2}$. 
We denote the value of the minimum $x$-$y$ vertex-cut in $G$ by $\opt_{x,y}^{V}$.
The goal is to compute an $x$-$y$ vertex-cut $(L,S,R)$ in $G$, with $w(S)\leq \opt_{x,y}^{V}+2\alpha$, such that, for every other $x$-$y$ vertex-cut $(L',S',R')$, the following holds:
\begin{equation}\label{eq : guarnatee in vertex weighted approximately minimum containing cut problem}
	|(L' \cup S') \setminus (L \cup S)| \leq \left(w(S') - \opt_{x,y}^{V}\right )\cdot \frac{n}{\alpha}.
\end{equation}
\end{definition}

We note that Condition \ref{eq : guarnatee in vertex weighted approximately minimum containing cut problem} in the above definition is an analogue of Condition \ref{eq : condition-in-almost-furthest-edge-cut-problem} for edge-\AFMC. In particular, it requires that, for every value $\beta\geq 0$, for every $x$-$y$ vertex-cut $(L',S',R')$ of value $w(S')\leq \opt_{x,y}^{V} + \beta$, all but at most $\frac{\beta\cdot n}{\alpha}$ vertices of $L' \cup S'$ are contained in $L \cup S$.

In the following theorem we provide an efficient algorithm for the vertex-\AFMC problem, by reducing it to edge-\AFMC via a split graph, and then employing the algorithm from \Cref{thm : edge capacitated approximately minimum containing cut}.

\begin{theorem}\label{thm : vertex weighted approximately minimum containing cut}
There is a deterministic algorithm for the vertex-\AFMC problem with running time $O\left ((n+m)^{1+o(1)}\log W\cdot \log(2+\alpha)\right )$.
\end{theorem}

\begin{proof}
Recall that we are given as input a directed $n$-vertex and $m$-edge graph $G$ with integral weights $0\leq w(v) \leq W$ on its vertices $v \in V(G)$, a pair of distinct vertices $x,y \in V(G)$ such that $(x,y) \notin E(G)$, and an integral parameter $0<\alpha<\wmax(G)/2$.

We start by computing, in time $O(m)$, the split graph $G'$ of $G$  (see \Cref{subsec: split graph} for the definition). We then apply the algorithm from \Cref{thm : edge capacitated approximately minimum containing cut} to the instance of edge-\AFMC defined by the graph $G'$, the pair $x^{\out},y^{\inn} \in V(G')$ of its vertices, and the parameter $\alpha$ that remains unchanged. Recall that the running time of the algorithm from \Cref{thm : edge capacitated approximately minimum containing cut} is $O\left((|V(G')|+|E(G')|)^{1+o(1)} \cdot \log \wmax(G)\cdot \log(2+\alpha)\right )\leq O\left ((n+m)^{1+o(1)}\log W\cdot \log(2+\alpha)\right )$, since $|V(G')|\leq 2n$, $|E(G')|\leq m+n$, and $\wmax(G) \leq 32n \cdot (W+1)$.
Let $(X,Y)$ denote the $x^{\out}$-$y^{\inn}$ edge-cut in $G'$ that the algorithm from \Cref{thm : edge capacitated approximately minimum containing cut} outputs.

Let $\opt_{x,y}^{V}$ denote the value of the minimum $x$-$y$ vertex-cut in $G$, and let $\opt_{x^{\out},y^{\inn}}^{E}$ denote the value of the minimum $x^{\out}$-$y^{\inn}$ edge-cut in $G'$.
Note that, from \Cref{cor : minimum cut in split graph},  $\opt_{x^{\out},y^{\inn}}^{E}=\opt_{x,y}^{V} < \frac{\wmax(G)}{2}$ must hold.
We will use the following simple observation.

\begin{observation}\label{obs: small cut value}
	 $c(X,Y) < \wmax(G)$.
\end{observation}
\begin{proof}
As observed already, $\opt_{x^{\out},y^{\inn}}^{E}=\opt_{x,y}^{V} < \frac{\wmax(G)}{2}$ must hold. From the definition of vertex-\AFMC, $\alpha<\frac{\wmax(G)}{2}$, and,  from the definition of edge-\AFMC:

\[c(X,Y) \leq \opt_{x^{\out},y^{\inn}}^{E} + \alpha<\wmax(G).\] 
\end{proof}

Let $L\subseteq V(G)$ be a vertex set that is defined as follows:

\[L=\set{v \in V(G) \setminus \{y\}\mid v^{\out} \in X}.\]

Let $S=\set{v\in V(G)\mid v^{\inn}\in X,v^{\out}\in Y}$, and let $R=V(G)\setminus(L\cup S)$.
By combining \Cref{obs: existential from split to regular} with \Cref{obs: small cut value}, we get that $(L,S,R)$ is a valid $x$-$y$ vertex-cut in $G$ with $w(S)\leq c(X,Y)\leq \opt_{x^{\out},y^{\inn}}^E+\alpha=\opt_{x,y}^V+\alpha$.
We return $(L,S,R)$ as the algorithm's outcome. Note that, given the edge-cut $(X,Y)$ in $G'$, the corresponding vertex-cut $(L,S,R)$ in $G$ can be computed in time $O(n)$. Therefore, the total running time of the algorithm is $O\left ((n+m)^{1+o(1)}\log W\cdot \log(2+\alpha)\right )$.

We claim that $(L,S,R)$  is a valid output for the vertex-\AFMC problem instance. Indeed, we already established that $(L,S,R)$ is a valid $x$-$y$ vertex-cut in $G$, and that $w(S)\leq \opt_{x,y}^V+\alpha$. Consider now any other $x$-$y$ vertex-cut $(L',S',R')$ in $G$, and let $(X',Y')$ be a partition of vertices of $G'$ with $X'=(L')^* \cup (S')^{\inn}$ and $Y' = (S')^{\out} \cup (R')^*$.
From \Cref{obs: existential from regular to split}, $(X',Y')$ is an $x^{\out}$-$y^{\inn}$ edge-cut in $G'$ whose value is $c(X',Y')=w(S')$. From the definition of the edge-\AFMC problem: 

\[|X' \setminus X| \leq \left(c(X',Y') - \opt_{x^{\out},y^{\inn}}^{E}\right )\cdot \frac{n}{\alpha} = \left(w(S') - \opt_{x,y}^{V}\right )\cdot \frac{n}{\alpha}\]

In order to complete the proof of \Cref{thm : vertex weighted approximately minimum containing cut}, it is now enough to prove that $|(L' \cup S') \setminus (L \cup S)| \leq |X' \setminus X|$ must hold. 
We do so by showing that, for every vertex $v\in (L' \cup S') \setminus (L \cup S)$, $v^{\inn}\in X'\setminus X$ holds.
Consider now any vertex $v\in (L' \cup S') \setminus (L \cup S)$. From the definition of the edge-cut $(X',Y')$, since $v\in L'\cup S'$, $v^{\inn}\in X'$ must hold. On the other hand, $v^{\inn}$ may not lie in $X$: indeed, if both $v^{\inn},v^{\out}\in X$, then $v\in L$ must hold from the definition of the set $L$; and if $v^{\inn}\in X$ and $v^{\out}\not\in X$ then $v\in S$ by the definition of the set $S$. We conclude that, for every vertex $v\in (L' \cup S') \setminus (L \cup S)$, $v^{\inn}\in X' \setminus X$ must hold, and so:

\[|(L' \cup S') \setminus (L \cup S)| \leq |X' \setminus X|\leq \left(w(S') - \opt_{x,y}^{V}\right )\cdot \frac{n}{\alpha},\]

as required, and, in particular, $(L,S,R)$ is a valid output for the vertex-\AFMC problem instance.
\end{proof}

\subsubsection{An Algorithm for the Edge-\AFMC Problem: Proof of \Cref{thm : edge capacitated approximately minimum containing cut}} \label{sec : proof of edge capacitated approximately minimum containing cut}

In this subsection we prove \Cref{thm : edge capacitated approximately minimum containing cut}, by designing an algorithm for the  edge-\AFMC Problem.
Recall we are given as input a directed $n$-vertex and $m$-edge graph $G$ with integral capacities $0\leq  c(e) \leq W$ on its edges $e \in E(G)$, a pair $x,y \in V(G)$ of distinct vertices, and an integral parameter $\alpha>0$.
We denote the value of the minimum $x$-$y$ edge-cut in $G$ by $\opt_{x,y}^{E}(G)$.
Our goal is to compute an $x$-$y$ edge-cut $(X,Y)$ in $G$ of value $c(X,Y)\leq \opt_{x,y}^{E}(G)+\alpha$, such that, for every other $x$-$y$ edge-cut $(X',Y')$,
$|X' \setminus X| \leq \left(c(X',Y') - \opt_{x,y}^{E}(G)\right )\cdot\frac{n}{\alpha}$ holds.

Our algorithm consists of two steps. In the first step, we use the  algorithm from \Cref{thm: maxflow}, in order to compute a maximum $x$-$y$ flow $f$ in $G$, in time $O(m^{1+o(1)} \cdot \log W)$. We then compute the residual flow network $H=G_f$ of the graph $G$ with respect to the flow $f$.  We denote the edge capacities in $H$ by $c_H(\cdot)$. In the remainder of the algorithm, we will solve the edge-\AFMC problem instance defined by the graph $H$, the pair $(x,y)$ of its vertices, and the parameter $\alpha$ that remains unchanged. In the following claim we show that, if $(X,Y)$ is a valid solution to the edge-\AFMC problem instance defined by the graph $H$ (that we denote by $(H,x,y,\alpha)$), then it must also be a valid solution to the edge-\AFMC problem instance defined by the graph $G$ (that we denote by $(G,x,y,\alpha)$.

\begin{claim}\label{cl : reduction-to-special-case-in-proof-of-furthest-edge-cut}
	If $(X,Y)$ is a valid solution to the instance $(H,x,y,\alpha)$ of the edge-\AFMC problem, then it is a valid solution to the instance $(G,x,y,\alpha)$ of the edge-\AFMC problem.
\end{claim}
\begin{proof}
	Note that, since $f$ is the maximum $x$-$y$ flow in $G$, from the properties of the residual flow network, the value of the maximum 
	$x$-$y$ flow, and hence of the minimum $x$-$y$ in $H$ is $0$. Consider now a solution $(X,Y)$ to the instance $(H,x,y,\alpha)$ of edge-\AFMC. Then, from the problem definition, $c_H(X,Y)\leq \alpha$. Moreover, for every other $x$-$y$ edge-cut $(X',Y')$ in $H$, $|X'\setminus X|\leq c_H(X',Y')\cdot \frac{n}{\alpha}$.
	
	From \Cref{fact : residual-capacity-of-edge-cut}:
	
	\[c(X,Y)=c_{H}(X,Y)+\val(f)\leq \alpha+\opt^E_{x,y}(G).\]
	
	Consider now any other $x$-$y$ cut $(X'',Y'')$ in $G$, and observe that $(X'',Y'')$ also defines an $x$-$y$ cut in $H$.
	From \Cref{fact : residual-capacity-of-edge-cut}, $c_H(X'',Y'')=c(X'',Y'')-\val(f)=c(X'',Y'')-\opt^E_{x,y}(G)$.  Therefore:

\[
|X''\setminus X| \leq c_H(X'',Y'')\cdot \frac{n}{\alpha}
=\left(c(X'',Y'')-\opt^E_{x,y}(G)\right )\cdot \frac{n}{\alpha}.\]
\end{proof}

From \Cref{cl : reduction-to-special-case-in-proof-of-furthest-edge-cut}, it is now sufficient to design an algorithm for the instance $(H,x,y,\alpha)$ of edge-\AFMC. As observed already, the value of the minimum $x$-$y$ cut in $H$ is $0$.
We start by constructing a graph $\hat H$, by starting with $\hat H=H$ and then adding, for every vertex $v \in V(H)$ a new edge $(x,v)$ of capacity $\eta = \frac{\alpha}{n}$; we refer to these new edges as \emph{fake edges}. We denote the edge capacities in $\hat H$ by $c_{\hat H}(\cdot)$.
We use the algorithm from \Cref{cor: mincut} in order to compute a minimum $x$-$y$ edge-cut $(X,Y)$ in $\hat H$, and return $(X,Y)$  as the outcome of our algorithm. 
Note that the algorithm from \Cref{cor: mincut} requires that all edge capacities in the graph are integral, so we may need to scale all edge capacities up by factor $n$ to ensure this. The running time of the algorithm from \Cref{cor: mincut}  is bounded by $O\left (m^{1+o(1)}\log W\cdot \log(2+\alpha)\right )$, and so the total running time of our algorithm is bounded by $O\left ((n+m)^{1+o(1)}\log W\cdot \log(2+\alpha)\right )$.

From \Cref{cl : reduction-to-special-case-in-proof-of-furthest-edge-cut}, in order to complete the proof of \Cref{thm : edge capacitated approximately minimum containing cut}, it is now enough to prove that $(X,Y)$ is a valid solution to the instance $(H,x,y,\alpha)$ of the edge-\AFMC problem instance.
We do so in the following two claims.

\begin{claim}\label{claim : value-of-final-cut-in-proof-of-furthest-edge-cut}
$c_H(X,Y)\leq \alpha$.
\end{claim}
\begin{proof}
Consider the minimum $x$-$y$ cut $(X^*,Y^*)$ in $H$, and recall that $c_H(X^*,Y^*)=0$. Recall that graph $\hat H$ is obtained from $H$ by adding, for every vertex $v \in V(H)$ a new edge $(x,v)$ of capacity $\eta = \frac{\alpha}{n}$. Clearly:

\[c_{\hat H}(X^*,Y^*)\leq c_H(X,Y)+\eta\cdot n=\alpha. \]

Since $(X,Y)$ is a minimum $x$-$y$ edge-cut in $\hat H$, we get that $c_H(X,Y)\leq c_{\hat H}(X,Y)\leq c_{\hat H}(X^*,Y^*)\leq \alpha$. 
\end{proof}

\begin{claim}\label{claim: second prop of furthest cut}
	Let $(X',Y')\neq (X,Y)$ be any $x$-$y$ edge-cut in $H$. Then 	$|X' \setminus X| \leq c_H(X',Y')\cdot\frac{n}{\alpha}$.
\end{claim}
\begin{proof}
Assume, for contradiction, that for some $x$-$y$ 	edge-cut $(X',Y')$ in $H$, $|X' \setminus X| > c_H(X',Y')\cdot\frac{n}{\alpha}$ holds.
Let $X''=X\cup X'$ and let $Y''=V(H)\setminus X''$. Clearly, $x\in X''$ and $y\not\in X''$, so $(X'',Y'')$ is a valid $x$-$y$ cut in $H$ and hence in $\hat H$. We now show that $c_{\hat H}(X'',Y'')<c_{\hat H}(X,Y)$, contradicting the fact that $(X,Y)$ is a minimum $x$-$y$ cut in $\hat H$.

Indeed, notice that cut $(X'',Y'')$ can be obtained from the cut $(X,Y)$ by moving the vertices of $X'\setminus X$ from $Y$ to $X$. Notice that the fake edges incident to the vertices of $X'\setminus X$ contribute to $c_{\hat H}(X,Y)$ but do not contribute to $c_{\hat H}(X'',Y'')$, and their total capacity is at least $|X' \setminus X|\cdot \eta>c_H(X',Y')$. On the other hand, the only edges that contribute to $c_{\hat H}(X'',Y'')$ but not to $c_{\hat H}(X,Y)$ are edges of $H$ connecting vertices of $X'\setminus X$ to vertices of $Y''$, whose total capacity is bounded by $c_H(X'\setminus X,Y'')\leq c_H(X',Y'')\leq c_H(X',Y')$, since $Y''\subseteq Y'$ as $X'\subseteq X''$. Altogether, we get that $c_{\hat H}(X'',Y'')<c_{\hat H}(X,Y)$, contradicting the fact that $(X,Y)$ is a minimum $x$-$y$ cut in $\hat H$.
\end{proof}

\subsection{Maximum $s$-$t$ Flow with Small Support}

For some of our algorithms, we need to be able to quickly compute a maximum $s$-$t$ flow $f$ in a vertex-weighted directed $n$-vertex graph, such that the flow $f$ is non-zero on at most $O(n)$ edges. In the following theorem we provide an algorithm for computing such a flow $f$ in time that is almost-linear in the number of edges of the input graph.
We note that the existence of such a flow was already shown in \cite{Tra25} (see Lemma 2.5 in \cite{Tra25}), and their proof also implicitly provides an algorithm for computing it. Unfortunately, the running time of this algorithm is at least $\Omega(nm)$ in the worst case, which is too slow for our purposes.
We build on the techniques from \cite{Tra25} to provide an algorithm with an almost linear runing time.
The algorithm is summarized in the following theorem.
\begin{theorem}\label{thm : computing-flow-with-few-edges}
	There is a randomized algorithm, whose input consists of $n$-vertex and $m$-edge graph $G$ with integral weights $0 \leq w(v)\leq W$ on its vertices $v \in V(G)$ and a pair $s,t \in V(G)$ of its vertices. The algorithm computes a maximum $s$-$t$ flow $f$ in $G$ that obeys the vertex capacities (defined by their weights), such that $f$ is integral, and, moreover, the total number of edges $e\in E(G)$ with $f(e)>0$ is at most $4n$.
	The expected running time of the algorithm is $O(m^{1+o(1)} \cdot \log W)$.
\end{theorem}

\begin{proof}
	Throughout the proof, we refer to an $s$-$t$ flow in $G$ that obeys the vertex capacities as a \emph{valid $s$-$t$ flow}. Observe that it is sufficient to design a randomized algorithm that computes a valid maximum $s$-$t$ flow $f$ in $G$ in time $O\left(m^{1+o(1)}\cdot \log W\right )$, such that $f$ is integral and, with probability at least $\half$, the number of edges $e\in E(G)$ with $f(e)>0$ is at most $4n$. This is since we can check, in time $O(m)$, whether the number of edges  $e\in E(G)$ with $f(e)>0$ is bounded by $4n$, and, if this is not the case,  repeat the algorithm; the expected running time of the resulting algorithm is $O\left(m^{1+o(1)}\cdot \log W\right )$ as required.  
	Therefore, from now on we focus on desiging an algorithm that computes a valid maximum $s$-$t$ flow $f$ in $G$ in time $O\left(m^{1+o(1)}\cdot \log W\right )$, such that $f$ is integral, and, with probability at least $\half$, the number of edges $e\in E(G)$ with $f(e)>0$ is at most $4n$.

	Our algorithm uses a parameter $U = 2^{32\ceil{\log n}^2}$.
	We start by computing  the split graph $G'$ of $G$, and then assign, to every edge $e\in E(G')$, a cost $\gamma(e)$, that is chosen independently uniformly at random from $\set{1,\ldots,U}$. We then use the algorithm from \Cref{thm: maxflow} to compute  a minimum-cost maximum $s^{\out}$-$t^{\inn}$ flow $f'$ in the resulting flow network $G'$, such that $f'$ is integral, in time $O\left(m^{1+o(1)}\cdot \log W\cdot \log U\right )\le O\left(m^{1+o(1)}\cdot \log W\right )$, since $U\leq 2^{O(\log^2n)}$. Note that the flow $f'$ in $G'$ naturally defines a valid and integral $s$-$t$ flow $f$ in $G$, as follows: for every edge $e=(u,v)\in E(G)$, we set $f(e)=f'(u^{\out},v^{\inn})$. From the construction of the split graph, it is immediate to verify that the flow $f$ is a valid $s$-$t$ flow in $G$ of value $\val(f')$, and moreover, from \Cref{cor : minimum cut in split graph}, it is the maximum $s$-$t$ flow in $G$. It is immediate to verify that flow $f$ is integral, and that it can be computed in time $O(m)$ given $f'$. Therefore, the total running time of our algorithm is $O\left(m^{1+o(1)}\cdot \log W\right )$.

	It now remains to show that, with probability at least $\half$, the total number of edges $e\in E(G)$ with $f(e)>0$ is bounded by $4n$. Notice that an edge $e=(u,v)\in E(G)$ has $f(e)>0$ if and only if the corresponding regular edge $e'=(u^{\out},v^{\inn})\in E(G')$ has $f'(e')>0$. Therefore, it is enough to prove that, with probability at least $\half$, the total number of regular edges $e'\in E(G')$ with $f'(e')>0$ is bounded by $4n$. In order to do so, we consider a new undirected graph $\hat G$, with $V(\hat G)=V(G')$; for every regular edge $e'=(u^{\out},v^{\inn})\in E(G')$ with $f'(e')>0$, we add an undirected edge $(u^{\out},v^{\inn})$ to $\hat G$. Let $\event$ be the bad event that $\hat G$ contains a cycle of length at most $4\ceil{\log (2n)}$. The following claim is key to our analysis.
	
	\begin{claim}\label{claim: small girth}
		$\prob{\event}\leq \half$.
	\end{claim}
	
	We provide the proof of \Cref{claim: small girth} below, after we complete the proof of \Cref{thm : computing-flow-with-few-edges} using it. In order to do so, it is enough to prove that, if Event $\event$  did not happen, then the total number of regular edges 
	$e'\in E(G')$ with $f'(e')>0$ (and hence the total number of edges $e\in E(G)$ with $f(e)>0$), is bounded by $4n$. Indeed, assume that Event $\event$ did not happen. Notice that the number  of regular edges 
	$e'\in E(G')$ with $f'(e')>0$ is precisely $|E(\hat G)|$. However, if Event $\event$ did not happen, then the girth of $\hat G$ is at least 
	$4\ceil{\log (2n)}$. We show, using standard techniques, that, in this case, $|E(\hat G)|\leq 4n$ must hold. Indeed, assume for contradiction that
	the girth of $\hat G$ is at least $4\ceil{\log (2n)}$ but $|E(\hat G)|> 4n$. Consider an iterative process, where in every iteration, we delete a vertex of $\hat G$ whose current degree is at most $2$. When this process terminates, the total number of edges that were deleted from $\hat G$ is at most $2|V(\hat G)|\leq 4n$, and the degree of every vertex is at least $3$; in particular, the graph remains non-empty. Consider now performing a BFS from any vertex $v$ in the resulting subgraph of $\hat G$. Since the girth of $\hat G$ is at least 
	$4\ceil{\log (2n)}=4\ceil{\log(|V(\hat G)|)}$, and since the degree of each remaining vertex is at least $3$, for all $1\leq i< 2\ceil{\log (2n)}$, the number of vertices in the $(i+1)$th layer of the BFS must be at least twice that in the $i$th layer. This is impossible, because then the total number of vertices encountered in the first $2\ceil{\log (2n)}$ layers is greater than $2n=|V(\hat G)|$. We concldue that, if Event $\event$ happened, then $|E(\hat G)|\leq 4n$, and so the total number of edges $e\in E(G)$ with $f(e)>0$ is at most $4n$.
	%
	%
	%
	In order to complete the proof of \Cref{thm : computing-flow-with-few-edges}, it is now enough to prove \Cref{claim: small girth}.
	
	\begin{proofof}{\Cref{claim: small girth}}
		For every integer $1\leq \ell\leq 4\ceil{\log (2n)}$, for every sequence $\Sigma=(v_1,\ldots,v_{\ell})$ of $\ell$ distinct vertices of $G'$, we let $\event(\Sigma)$ be the bad event that there is a cycle $C$ in $\hat G$, whose vertices are $(v_1,\ldots,v_{\ell})$, and that they appear on the cycle in this order. In the following observation, we bound the probability of Event $\event(\Sigma)$.
		
		\begin{observation}\label{obs: bad cycle}
			For every integer $1\leq \ell\leq 4\ceil{\log (2n)}$ and a sequence $\Sigma=(v_1,\ldots,v_{\ell})$ of $\ell$ vertices of $G$, $\prob{\event(\Sigma)}\leq \frac{1}{U}$.
		\end{observation}
		\begin{proof}
			Note that Event $\event(\Sigma)$ may only happen if there is a sequence $E'=(e_1,\ldots,e_{\ell})$ of directed regular edges in graph $G'$, for which the following hold. First,  for all $1\leq i<\ell$, $f'(e_i)>0$, and the endpoints of $e_i$ are $v_i$ and $v_{i+1}$. Additionally, $f'(e_{\ell})>0$, and the endpoints of $e_{\ell}$ are $v_{\ell}$ and $v_1$. For $1\leq i<\ell$, we say that edge $e_i$ is oriented to the right if $e_i=(v_i,v_{i+1})$, and we say that it is oriented to the left if $e_i=(v_{i+1},v_{i})$. Similarly, we say that edge $e_{\ell}$ is oriented to the right if $e_{\ell}=(v_{\ell},v_1)$, and that it is oriented to the left\footnote{In fact due to the structure of the split graph, where all regular edges are directed from the vertices of $V^{\out}$ to the vertices of $V^{\inn}$, $\ell$ must be an even integer, and either all even-indexed edges are oriented to the right and odd-indexed ones to the left, or the other way around. This structure may be convenient to keep in mind for intuition but it is immaterial to the proof.} if $e_{\ell}=(v_1,v_{\ell})$. Let $S_1$ denote the total cost of all edges of $E'$ that are oriented to the left, and let $S_2$ denote the total cost of all edges of $E'$ that are oriented to the left. We claim that Event $\event(\Sigma)$ may only happen if $S_1=S_2$ holds. Indeed, assume, for contradiction that Event $\event(\Sigma)$ happened, but $S_1\neq S_2$; assume without loss of generality that $S_1>S_2$. Recall that, by construction of the split graph, the capacity of every regular edge in $G'$ is $\wmax(G)$, while the total weight of all vertices in $G$ (and hence the values of the flows $f$ and $f'$) is less than $\wmax(G)$. Therefore, for every edge $e_i\in E'$, $0<f'(e_i)<c(e_i)$ holds. We can then select a value $\eps>0$, so that, for every edge $e_i\in E'$, $\eps\le f'(e_i)\leq c(e_i)-\eps$ holds. Consider the $s^{\out}$-$t^{\inn}$ flow in $G'$ obtained by starting from $f'$ and then decreasing the flow value on every edge of $E'$ that is oriented to the left by $\eps$, and increasing the flow value on every edge of $E'$ that is oriented to the right by the same amount. It is immediate to verify that we obtain a valid $s^{\out}$-$t^{\inn}$ flow in $G'$ whose value remains equal to $\val(f')$, while the cost becomes strictly smaller than that of $f'$,  a contradiction.
			
			We conclude that Event $\event(\Sigma)$ may only happen if $S_1=S_2$. Assume without loss of generality that edge $e_1$ is oriented to the right, and that the costs of all edges in $E'\setminus\set{e_1}$ were selected before the cost of $e_1$ was selected. Let $E_1$ be the set of all edges of $E'$ oriented to the left, and let $E_2=E'\setminus(E_1\cup \set{e_1})$. 
			Finally, let $S=\sum_{e\in E_1}\gamma(e)-\sum_{e\in E_2}\gamma(e)$.
			Then Event $\event(\Sigma)$ may only happen if the cost of the edge $e_1$ is chosen to be precisely $S$, and, since the cost of $e_1$ is chosen uniformly at random from $\set{1,\ldots,U}$, the  probability of this happening is $\frac{1}{U}$.
		\end{proof}

		Notice that, for all $1\leq \ell\leq 4\ceil{\log (2n)}$, the total number of sequences of $\ell$ vertices in $G'$ is bounded by $(2n)^{\ell}$. Therefore, the total number of length-$\ell$ vertex sequences, for all $1\leq \ell\leq 4\ceil{\log (2n)}$, is bounded by $n^{16\ceil{\log n}}$. Using the Union Bound over the events $\event(\Sigma)$ for all such sequences $\Sigma$, and the fact that $U=2^{32\ceil{\log n}^2}$, we get that $\prob{\event}\leq \half$.
	\end{proofof}
\end{proof}

\section{Algorithm for Unweighted Directed Global Minimum Vertex-Cut}
\label{sec: alg: unweigthed}

In this section we design an algorithm for the unweighted version of Global Minimum Vertex-Cut, proving \Cref{thm: main: unweighted}. 
We use the following theorem, that provides an $O(n^{2+o(1)})$-time algorithm for the problem, and was proved in \cite{LNP21} (their initial theorem statement was conditional on the existence of a near-linear time algorithm for maximum $s$-$t$ flow, but since such an algorithm (albeit with an almost-linear running time) was shown by \cite{CKLP22,BCP23}, it now holds unconditionally).


\begin{theorem}[follows from Theorem 1.2 in \cite{LNP21} and \cite{CKLP22,BCP23}]\label{thm: old unweighted}
	There is a randomized algorithm that, given a simple directed unweighted $n$-vertex and $m$-edge graph $G$ that contains some vertex-cut, returns a vertex-cut $(L',S',R')$ in $G$, so that, with probability at least $0.99$,  $(L',S',R')$ is a global minimum vertex-cut. The running time of the algorithm is $O\left(n^{2+o(1)}\right )$.
\end{theorem}

The main technical result of this section is the following theorem, that provides an algorithm for the unweighted directed global minimum vertex-cut problem with running time $O(m^{1+o(1)}\cdot \hat k)$, where $\hat k$ is a ``guess'' on the value of the optimal solution that the algorithm receives as input.

\begin{theorem}\label{thm: main: inner unweighted}
	There is a randomized algorithm, that, given a simple directed unweighted $n$-vertex and $m$-edge graph $G$ that contains some vertex-cut, together with an integer $1\leq \hat k\leq \frac{n}{1000}$, computes a vertex-cut $(L',S',R')$ in $G$,
  in time $O\left (m^{1+o(1)}\cdot \hat k\right )$. Moreover, if we denote by $k$ the value of the global minimum vertex-cut in $G$, and if $\frac{\hat k}{2}< k\leq \hat k$ holds, then, with probability at least $\left(1-\frac{1}{n}\right )$, $(L',S',R')$ is a global minimum vertex-cut in $G$.
\end{theorem}

We prove \Cref{thm: main: inner unweighted} in the remainder of this section, after completing the proof of \Cref{thm: main: unweighted} using it. 
We assume that $m\geq n/2$ holds, since otherwise the input graph $G$ is not strongly connected, and vertex-cut of value $0$ can be computed in time $O(m)$ by performing two DFS searches.
Let $z=\floor{\log(n/1000)}$. For all $0\leq i\leq z$, 
we execute the algorithm from \Cref{thm: main: inner unweighted} with the parameter $\hat k=2^i$; we denote this algorithm by $\aset_i$. We also denote by $\aset_{z+1}$ the algorithm from \Cref{thm: old unweighted}. We execute all algorithms $\aset_0,\ldots,\aset_z,\aset_{z+1}$ in parallel. For all $0\leq i\leq z+1$, we may execute Algorithm $\aset_i$ in its entirety (in which case we say that the algorithm has run its course), or we may choose to terminate it prematurely.

Whenever, for any $0\leq i\leq z$ algorithm $\aset_i$ terminates after running its course, if the cut $(L_i,S_i,R_i)$ that it returns has $|S_i|\leq 2^i$, then we immediately terminate algorithms $\aset_{i+1},\ldots,\aset_{z}$ prematurely, and if, additionally, $2^i\cdot m<n^2$, then we also terminate Algorithm $\aset_{z+1}$ prematurely.
We then record the vertex-cut $(L_i,S_i,R_i)$, and refer to it as a \emph{candidate solution}.
If Algorithm $\aset_{z+1}$ terminates after running its course, then we record the vertex-cut $(L_{z+1},S_{z+1},R_{z+1})$ as a candidate solution, and we terminate  prematurely all algorithms $\aset_0,\ldots,\aset_z$ that have not terminated yet.

Once every algorithm $\aset_j$ terminates (either because it has run its their course, or because we terminated it prematurely), we output the smallest-value candidate solution.
This completes the description of the algorithm. We now analyze its running time and correctness.

Observe first that the running time of the algorithm may not exceed $O\left (n^{2+o(1)}\right )$. Indeed, if Algorithm $\aset_{z+1}$ is allowed to run its course, then it terminates in time $O\left (n^{2+o(1)}\right )$, and, once it does so, all other algorithms $\aset_1,\ldots,\aset_z$ are terminated.
Otherwise, Algorithm $\aset_{z+1}$ was terminated prematurely. This may only happen if, for some $0\leq i\leq z$  with $2^i\cdot m<n^2$, Algorithm $\aset_i$ terminates with a cut $(L_i,S_i,R_i)$, and $|S_i|\leq 2^i$ holds. 
Recall that, from \Cref{thm: main: inner unweighted}, for all $0\leq i'\leq i$, the running time of algorithm $\aset_{i'}$ is bounded by  $O\left (m^{1+o(1)}\cdot 2^{i'}\right )\leq O\left (m^{1+o(1)}\cdot 2^{i}\right )\leq O\left(n^{2+o(1)}\right)$. Once Algorithm $\aset_i$ terminates, all algorithms $\aset_{i+1}\ldots,\aset_{z+1}$ are terminated as well, and so the total running time of the entire algorithm is bounded by $O\left(n^{2+o(1)}\right)$.

Let $k$ denote the value of the global minimum vertex-cut in $G$, and let $i\geq 0$ be the unique integer for which $2^{i-1}<k\leq 2^i$ holds. Assume first that $i>z$, so $k\geq \frac{n}{2000}$. In this case, $m\cdot k\geq\Omega(n^2)$, and, from the above discussion, the running time of the algorithm is $O\left(n^{2+o(1)}\right)\leq O\left ( \min\set{n^{2+o(1)},m^{1+o(1)}\cdot k}\right )$, as required.
Moreover, for all $0\leq i\leq z$, algorithm $\aset_i$ may not terminate with a vertex-cut $(L_i,S_i,R_i)$ of value $|S_i|\leq 2^i<k$, so Algorithm $\aset_{z+1}$ was allowed to run its course, and, from \Cref{thm: old unweighted}, with probability at least $0.99$, the cut $(L_{z+1},S_{z+1},R_{z+1})$ that it returns (which must be the only candidate solution for our algorithm) is optimal.

 Therefore, we assume from now on that $i\leq z$. Recall that the running time of Algorithm $\aset_i$ is bounded by $O\left (m^{1+o(1)}\cdot 2^i\right )\leq O\left (m^{1+o(1)}\cdot k\right )$. From the choice of the value of $i$, for $0\leq i'<i$, algorithm $\aset_{i'}$ may not terminate with a cut of value at most $2^{i'}<k$. Therefore, one of the algorithms $\aset_i$ or $\aset_{z+1}$ must have run its course. If it is the former, then, with probability at least $\left(1-\frac{1}{n}\right )$, the cut $(L_i,S_i,R_i)$ that $\aset_i$ outputs has $|S_i|=k$. If this is indeed the case, then the entire algorithm terminates, and its running time is bounded by $O\left (m^{1+o(1)}\cdot k\right )$. With probability at most $\frac 1 n$, the cut $(L_i,S_i,R_i)$ that Algorithm $\aset_i$ returns is not optimal, and then the running time of the entire algorithm is bounded by $O\left(n^{2+o(1)}\right)$. It is now easy to verify that the expected running time of the entire algorithm is bounded by $ O\left ( \min\set{n^{2+o(1)},m^{1+o(1)}\cdot k}\right )$, and, with probability at least $\half$, the cut that it returns is a global minimum vertex-cut.
In order to complete the proof of  \Cref{thm: main: unweighted}, it is now enough to prove \Cref{thm: main: inner unweighted}, which we do in the remainder of this section.

\subsection*{Proof of \Cref{thm: main: inner unweighted}}

We assume that the the number of vertices $n$ in the input graph is greater than a sufficiently large constant, so that, for example, $n>\log^{20}n$ holds; if this is not the case, then we can use brute-force search to find and return a global minimum vertex-cut of $G$, in time that is upper-bounded by a sufficiently large constant.

\paragraph{Distinguished cut.}
Throughout the algorithm, for the sake of the analysis, we fix a global minimum vertex-cut $(L,S,R)$ in $G$, that minimizes $|L|$, that we refer to as the \emph{distinguished cut}. As before, we assume without loss of generality that $|L|\leq |R|$, since otherwise we can equivalently solve the problem on the graph $\overline{G}$ that is obtained from $G$ by reversing the direction of its edges. Notice that, from the selection of the distinguished cut, it must be induced by the vertex set $L$.

\paragraph{Fixing an arbitrary vertex-cut.} Our algorithm may sometimes need to return an arbitrary vertex-cut in $G$. For convenience, we compute such a vertex-cut $(\tilde L,\tilde S,\tilde R)$ at the beginning of the algorithm, as follows. Recall that the input graph $G$ has some vertex-cut, so for some vertex $v\in V(G)$, $\deg^+_G(v)<n-1$ must hold. We let $v$ be any such vertex, and we let $(\tilde L,\tilde S,\tilde R)$ be the tripartition of $V(G)$ induced by $v$; clearly, $(\tilde L,\tilde S,\tilde R)$ is a valid vertex-cut in $G$. Note that the cut $(\tilde L,\tilde S,\tilde R)$ can be computed in time $O(m)$, and we do so at the beginning of the algorithm. Generally, whenever our algorithm needs to output an arbitrary vertex-cut in $G$ (which, intuitively, means that it failed), we output this fixed cut $(\tilde L,\tilde S,\tilde R)$.


\paragraph{High-level overview of the algorithm.}
Our algorithm starts with a preprocessing step, in which we compute two disjoint sets $T$, $T'$ of vertices whose cardinalities are bounded by roughly $\tilde O\left(\frac{m}{\vol^+(L)}\right )$, and another vertex $t^*\in V(G)$ using a simple randomized procedure that ensures that, with a sufficiently high probability, $T\cap L\neq \emptyset$, $T'\cap L=\emptyset$ and $t^*\in R$ holds. We say that the good event $\event^*$ happens if these conditions are indeed satisfied. Throughout the algorithm, we refer to the vertices in the sets $T$ and $T'$ as \emph{terminals} and \emph{anti-terminals}, respectively. We then compute a hierachical partition of the collection $T$ of terminals. The hierarchical partition has $z=O(\log |T|)$ levels. For all $0\leq i\leq z$, the \emph{level-$i$ partition}, denoted by $\bset_i$, contains at most $N_i\leq 2^i$ sets (that we refer to as \emph{level-$i$ batches of terminals}), each of which contains roughly $\eta_i=\frac{|T|}{2^{\Theta(i)}}$ terminals. Additionally, every set in the level-$i$ partition $\bset_i$ is contained in some level-$(i-1)$ batch. We also ensure that every level-$z$ batch $B\in \bset_z$ contains exactly one terminal. Our algorithm consists of $z$ phases, where, in the $i$th phase, we process the level-$i$ partition $\bset_i$ of the terminals.
We partition the set $\bset_i$ into a collection $\bset_i^A$ of \emph{active} level-$i$ batches and a collection $\bset_i^I$ of \emph{inactive} level-$i$ batches. Additionally, for every active batch $B\in \bset_i^A$, we will compute a set $A_B\subseteq V(G)\setminus (T'\cup \set{t^*}\cup N^-_G(t^*))$ of vertices with $A_B\cap B\neq \emptyset$. 
We will ensure that, on the one hand, for every active batch $B\in \bset_i^A$, $\vol^+_G(A_B)$ is sufficiently small, while, on the other hand, if, for some level-$i$ batch $B\in \bset_i$, $B\cap L\neq \emptyset$ holds, then $B$ must be an active batch, and $L\subseteq A_B$ must hold. 
For every level $0\leq i\leq z$, every active level-$i$ batch $B\in \bset_i^A$ is associated with the graph $G_B=G^{|A_B}$ --- the graph derived from $G$ via the vertex set $A_B$. Since $\vol_G^+(A_B)$ is relatively small, so is $|E(G_B)|$.

Consider now the last level $z$ of the hierarchical partition of the terminals. 
For every level-$z$ batch $B\in \bset_z$, $|B|=1$ holds. When processing this last level of the partition, our algorithm will compute, for each active batch $B\in \bset^A_z$, a vertex-cut $(L_B,S_B,R_B)$ in $G$, so that, if the unique terminal of $B$ lies in $L$, then $(L_B,S_B,R_B)$  is a global minimum vertex-cut. Cut $(L_B,S_B,R_B)$ is computed by first computing a minimum vertex-cut in $G_B$ separating the unique terminal in $B$ from $t$, and then converting the resulting cut into a vertex-cut in $G$. The key is that, since we ensure that the sizes of the graphs in $\set{G_B}_{B\in \bset_z}$ are sufficiently small, and each batch $B\in \bset_z$ only contains one terminal, we can compute the cut  $(L_B,S_B,R_B)$ in time $O\left (|E(G_B)|^{1+o(1)}\right )$, and we can afford to spend this running time for all $B\in \bset_z$.
We now proceed to formally describe the preprocessing step.

\subsection{The Preprocessing Step}
\label{subsec: preproc}

In the preprocessing step, we construct the collections $T$ and $T'$ of terminals and anti-terminals, respectively, and select the vertex $t^*\in V(G)$. We also construct the hierarchical partition of the terminals. 

We start by selecting an integer $1\leq j\leq \ceil{\log m}$ uniformly at random and setting $\lambda=\frac{m}{2^j}$.
We say that the good event $\event_1$ happens if $\frac{\lambda} 2\leq \vol^+_G(L)\leq \lambda$. It is immediate to verify that $\prob{\event_1}\geq \frac{1}{\ceil{\log m}}\geq\frac{1}{4\log n}$.

Next, we select a vertex $t^*\in V(G)$ uniformly at random. We let $\event_2$ be the good event that $t^*\in R$. Note that, since $\hat k\leq \frac{n}{1000}$, if $k\leq \hat k$ holds, then $k\leq \frac{n}{1000}$ must hold.  Since we also assumed that $|L|\leq |R|$, we get that, if $k\leq \hat k$, then $|R|\geq \frac{n}{4}$. Therefore, if $k\leq \hat k$, $\prob{\event_2}\geq \frac{1}{4}$.

Our next step is to select a random bit $b\in\set{0,1}$ uniformly at random. We say that the good event $\event_3$ happens if either (i) $b=0$ and $\vol^+(L)\geq \frac{m}{4}$; or (ii) $b=1$ and $\vol^+(L)< \frac{m}{4}$. Clearly, $\prob{\event_3}=\half$.

We now consider two cases. Assume first that $b=0$. In this case, we select a vertex $s^*\in V(G)$ at random, where, for every vertex $v\in V(G)$, the probability to select $s^*=v$ is $\frac{\deg^+(v)}{m}$. We say that the good event $\event'$ happens if $s^*\in L$. Note that, if Event $\event_3$ has happened, then $\vol^+_G(L)\geq\frac{m}{4}$, and so $\prob{\event'\mid \event_3}\geq \frac{1}{4}$. We then compute a minimum $s^*$-$t^*$ vertex-cut $(L',S',R')$ in $G$ using the algorithm from \Cref{cor: min_vertex_cut} in time $O\left (m^{1+o(1)})\right )$, and return this cut as the algorithm's outcome. Note that, if $b=0$, and if both events $\event_2$ and $\event'$ happened, then $s^*\in L$ and $t^*\in R$ must hold, and the cut $(L',S',R')$ must be a global minimum vertex-cut. Note that, if $\vol^+_G(L)\geq\frac{m}{4}$ and $k\leq \hat k$, then $\prob{\event_2\band\event_3\band\event'}\geq \prob{\event_2}\cdot \prob{\event_3}\cdot \prob{\event'\mid \event_3}\geq \frac{1}{32}$.
Altogether, if $\vol^+_G(L)\geq\frac{m}{4}$ and $k\leq \hat k$, then, with probability at least $\frac{1}{32}$, the algorithm described so far outputs a global minimum vertex-cut. We assume from now on that $b=1$.

Our next step is to compute a set $T_0\subseteq V(G)$ of vertices that we call \emph{initial terminals}, by performing $r=\ceil{\frac{m}{\lambda}}$ independent trials. In every trial we select a vertex from $V(G)$, where the probability for selecting a vertex $v$ is $\frac{\deg^+_G(v)}{m}$; the selected vertex is then added to $T_0$. These trials are performed independently and with repetitions. We construct a set $T'\subseteq V(G)$ of vertices, that we refer to as \emph{anti-terminals}, using the same process. Finally, we let $T=T_0\setminus \left(T'\cup \set{t^*}\cup N^-_G(t^*)\right )$ be the set of \emph{terminals}.
We say that a good event $\event$ happens if all of the following hold:

\begin{itemize}
	\item $\frac{\lambda} 2\leq \vol^+_G(L)\leq \lambda$;
	\item $t^*\in R$;
\item Event $\event_3$ happens;
	\item $T_0\cap L\neq \emptyset$; and
	\item $T'\cap L=\emptyset$.
\end{itemize}

In the following simple observation we show that, if $\vol^+_G(L)<\frac{m}{4}$ and $k\leq \hat k$, then Event $\event$ happens with a sufficiently high probability. The proof is elementary and is deferred to Section \ref{subsec: proof of good terminal selection} of Appendix.
\begin{observation}\label{obs: preprocessing successful}
If $\vol^+_G(L)<\frac{m}{4}$  and $k\leq \hat k$, then	$\prob{\event}\geq \frac{1}{2^{64}\cdot \log n}$.
\end{observation}

Throughout, we refer to the vertices in $T'$ as \emph{anti-terminals}, and we let $T=T_0\setminus \left(T'\cup \set{t^*}\cup N^-_G(t^*)\right )$ be the set of vertices that we refer to as \emph{terminals}. The vertices of $T_0$ are called \emph{initial terminals} and are only used in the analysis of the preprocessing step. Note that, if Event $\event$ happened, then $T\cap L\neq \emptyset$ must hold.

\paragraph{Hierarchical Partition of the Terminals.}
Our algorithm computes a hierarchical partition of the terminals, as follows. The hierarchy will have $z=O(\log n)$ \emph{levels}, and, for all $0\leq i\leq z$, we will construct a partition $\bset_i$ of $T$, that we refer to as the \emph{level-$i$ partition}. We refer to the sets $B\in \bset_i$ in the level-$i$ partition as \emph{level-$i$ batches of terminals}.
We will denote by $\eta_i=\max_{B\in \bset_i}\set{|B|}$ the largest cardinality of a level-$i$ batch, and we will ensure that, for all $B\in \bset_i$, $\frac{\eta_i}{4}\leq |B|\leq \eta_i$ holds. We will also ensure that, for $1\leq i\leq z$, $\frac{\eta_{i-1}}{2}\leq \eta_i\leq \frac{2\eta_{i-1}}{3}$. We now provide a formal description of the construction of the hierarchical partition of the terminals.

The level-$0$ partition, $\bset_0$, consists of a single set of terminals $T$, and we let $\eta_0=|T|$. Consider now some level $i>0$ and assume that the level-$(i-1)$ partition $\bset_{i-1}$ of the terminals was computed already. 
Let $\eta_{i-1}=\max_{B\in \bset_{i-1}}\set{|B|}$, and assume additionally that, for every set $B\in \bset_{i-1}$, $|B|\geq\frac{\eta_{i-1}}{4}$ holds.

In order to compute the level-$i$ partition, we start with $\bset_i=\emptyset$, and then consider every level-$(i-1)$ batch $B\in \bset_{i-1}$ one by one. For each such batch $B$, if $|B|=1$ or $|B|\leq \frac{\eta_{i-1}}{2}$, then we add the set $B$ to $\bset_i$. Otherwise, we compute an arbitrary partition $(B',B'')$ of $B$, where $B'$ contains exactly $\ceil{\frac{|B|}{2}}$ and $B''$ exactly $\floor{\frac{|B|}{2}}$ terminals of $B$, and add both $B'$ and $B''$ to $\bset_{i-1}$. 
This concludes the description of the algorithm for computing the level-$i$ partition $\bset_i$. We denote by $\eta_i=\max_{B'\in \bset_i}\set{|B'|}$. Since, for every level-$(i-1)$ batch $B\in \bset_{i-1}$, $\frac{\eta_{i-1}}{4}\leq |B|\leq \eta_{i-1}$ holds, it is immediate to verify that, if $\eta_{i-1}\geq 2$, then $\eta_{i}\leq \ceil{\frac{\eta_{i-1}}{2}}\leq \frac{2\eta_{i-1}}{3}$.
Moreover, for every set $B'\in \bset_{i-1}$, $|B'|\geq \frac{\eta_{i-1}}{4}\geq \frac{\eta_i}4$. It is also immediate to verify that $\eta_{i}\geq \frac{\eta_{i-1}}{2}$ must hold.

Once we reach a level $i$ for which $\eta_i=1$, we set $z=i$ and terminate the algorithm for constructing the hierarchical partition of the terminals. 
From the above discussion, $\eta_0=|T|\leq r=\ceil{\frac{m}{\lambda}}$, and for all $1\leq i\leq z$: 

\begin{equation}\label{eq: compare set sizes}
\frac{\eta_{i-1}}{2}\leq \eta_i\leq  \frac{2\eta_{i-1}}{3}.\end{equation}

 Therefore, for all $1\leq i\le z$:

\begin{equation}
\frac{|T|}{2^i}\leq \eta_i\leq |T|\cdot \left(\frac{2}{3}\right )^i.
\end{equation}

In particular, $z\leq \frac{\log r}{\log (3/2)}\leq 2\log n$.
For all $0\leq i\leq z$, we denote $N_i=|\bset_i|$.
Since $\bset_i$ is a partition of $T$, and since, for every level-$i$ batch $B\in \bset_i$, $\frac{\eta_i}{4}\leq |B|\leq \eta_i$ holds, we get that:

\begin{equation}\label{eq: bound 1 on Ni}
 \frac{|T|}{\eta_i}\leq N_i\leq \frac{4|T|}{\eta_i}.
\end{equation}


It is immediate to verify that the hierarchical partition of the terminals can be computed in time $O(nz)\leq O(n\log n)$, and overall, the total running time of the preprocessing step is bounded by $O(m+n\log n)$.

\paragraph{Algorithm Overview.}
Our algorithm consists of $z$ phases. For all $1\leq i\leq z$, in Phase $i$ we compute a partition of the set $\bset_i$ of level-$i$ batches into the collection $\bset^A_i$ of \emph{active batches}, and the collection $\bset^I_i$ of \emph{inactive batches}. For every active batch $B\in \bset^A_i$, we will compute a subset $A_B\subseteq V(G)\setminus \left(\set{t^*}\cup N^-_G(t^*)\cup T'\right )$ of vertices of $G$, with $\vol^+_G(A_B)\leq \frac{1000\hat k\cdot m\cdot \log n}{N_i}$ and $B\cap A_B\ne \emptyset$. 
We will ensure that, 
if $B\in \bset_i$ is a level-$i$ batch of terminals with  $L\cap B\neq \emptyset$, then $B\in \bset^A_i$ holds, and, moreover, $L\subseteq A_B$.


After the last phase is completed, for every active level-$z$ batch $B\in \bset_z^A$, we will compute the graph $G_B=G^{|A_B}$ that is derived from $G$ via the vertex set $A_B$, and
a minimum vertex-cut in $G_B$ between the unique terminal $s\in B$ and $t$. We will then select the smallest-value cut among the resulting cuts, transform it into a vertex-cut in $G$ of the same value, and output this vertex-cut.

\paragraph{The Main Technical Claim.}
At a very high level, in order to execute the $i$th phase, for $i>0$, the algorithm processes each level-$i$ batch $B\in \bset_i$  one by one. Consider any such batch $B$, and let $B'\in \bset_{i-1}$ be the unique level-$(i-1)$ batch containing $B$. If $B'$ is an inactive batch, then we add $B$ to the set $\bset_i^I$ of inactive level-$i$ batches. Assume now that $B'$ is an active batch. The algorithm considers the corresponding graph $G_{B'}$ derived from $G$ via $A_{B'}$, and then computes a vertex-cut $(\hat L,\hat S,\hat R)$ in $G_{B'}$ with $t\in \hat R$, such that $T'\cap \hat L=\emptyset$, and  $|\hat S|\leq 3\hat k\cdot \eta_i$.
It then sets $A_{B}=\hat L$. Let $G_{B}=G^{|A_{B}}$ be the graph derived from $G$ via the vertex set $A_B$. From the properties of graphs derived via vertex sets, $(\hat L,\hat S, V(G)\setminus (\hat L\cup \hat S))$ is a valid vertex-cut in $G$  (see \Cref{claim: transforming cuts from defined graph}).
In order to ensure that  $\vol^+_G(A_B)\leq \frac{1000\hat k\cdot m\cdot \log n}{N_i}$, it is enough to ensure that $\vol_G^+(\hat L)\leq\frac{1000\hat k\cdot m\cdot \log n}{N_i}$. Our main technical claim shows that, with a sufficiently high probability, for every vertex-cut $(\hat L,\hat S, V(G)\setminus (\hat L\cup \hat S))$  in $G$ that may arise from the above process, $\vol_G^+(\hat L)$ must indeed be sufficiently low. In fact, the main purpose of selecting the anti-terminals is to ensure precisely this property. 
We now define the notion of a \emph{bad tripartition} and show that, with a sufficiently high probability, no such bad tripartitions of vertices of $G$ exist.

\begin{definition}[Bad tripartition]\label{def: bad tripartition}
	Consider a level $0\leq i\leq z$ and a tripartition $(L',S',R')$ of $V(G)$. We say that $(L',S',R')$ is a \emph{bad tripartition for level $i$} if all of the following hold:
	
	\begin{properties}{B}
		\item $|S'|\leq 3\hat k\cdot \eta_i$; \label{prop: bad cuts: small separator}
		\item $L'\cap T'=\emptyset$; \label{prop: bad cuts: missed cut}
		\item $\vol^+_G(L')\geq \frac{1000\hat k\cdot m\cdot \log n}{N_i}$; and \label{prop: bad cuts: large volume}
		\item there is a collection $\hat B\subseteq T\setminus S'$ of at most $\eta_i$ terminals, such that $L'$ is precisely the collection of all vertices of $G$ that are reachable from the vertices of $\hat B$ in $G\setminus S'$. In other words, for every vertex $v\in L'$, there is a path from some vertex $u\in \hat B$ to $v$ in $G$ that avoids the vertices of $S'$, and all vertices of $G$ with this property lie in $L'$. \label{prop: bad cuts: terminals}
	\end{properties}
\end{definition}

We let $\hat \event$ be the bad event that, for some level $0\leq i\leq z$, a  bad tripartition of $V(G)$  for level $i$ exists. The following technical claim is key to our algorithm.

\begin{claim}\label{claim: avoiding bad cuts}
	$\prob{\hat \event}\leq\frac{1}{n^2}$.
\end{claim}
\begin{proof}
	For the sake of the proof, we view the process of constructing the set of terminals and anti-terminals as follows. We first select the set $T_0$ of initial terminals and the vertex $t^*$, and only then construct the set $T'$ of anti-terminals. Lastly, we let the set of terminals be $T=T_0\setminus \left(T'\cup N^-_G(t^*)\cup \set{t^*}\right )$. In this proof, we consider the time when the set $T_0$ of initial terminals was constructed already, but the set $T'$ of anti-terminals was not yet constructed. Let $\Pi$ be the collection of all triples $(i,\hat S,\hat B)$, where $0\leq i\leq z$ is a level, $\hat S\subseteq V(G)$ is a collection of at most $3\hat k\eta_i$ vertices of $G$, and $\hat B\subseteq T_0\setminus \hat S$ is a collection of at most $\eta_i$ initial terminals. We may sometimes refer to a triple $(i,\hat S,\hat B)$ as a \emph{level-$i$ triple}. For each such triple $(i,\hat S,\hat B)$, let $\rho(\hat S,\hat B)$ be the collection of all vertices $v\in V(G)$, such that $v$ is reachable from some vertex of $\hat B$ in $G\setminus \hat S$; in other words, there is a path in $G$ connecting a vertex of $\hat B$ to $v$ that avoids $\hat S$. Notice that $\hat B\subseteq \rho(\hat S,\hat B)$ must hold. Let $\hat \event(i,\hat S,\hat B)$ be the bad event that no vertex of $\rho(\hat S,\hat B)$ was added to $T'$. Finally, we let $\Pi'\subseteq \Pi$ be the collection of triples $(i,\hat S,\hat B)$ with $\vol^+_G(\rho(\hat S,\hat B))\geq \frac{1000\hat k\cdot m\cdot \log n}{N_i}$; we refer to the triples in $\Pi'$ as \emph{suspicious triples}. We start by showing that, if Event $\hat \event$ has happened, then Event $\hat \event(i,\hat S,\hat B)$ must have happened for some suspicious triple $(i,\hat S,\hat B)\in \Pi'$, in the following simple observation.
	
	\begin{observation}\label{obs: relating events}
		If Event $\hat \event$ has happened, then there must be a triple  $(i,\hat S,\hat B)\in \Pi'$ for which Event $\hat \event (i,\hat S,\hat B)$ has happened.
	\end{observation}

\begin{proof}
	Assume that Event $\hat \event$ has happened. Then for some level $0\le i\leq z$, there is a bad tripartition for level $i$ in $G$; denote this tripartition by $(L',S',R')$. From Property \ref{prop: bad cuts: terminals} of bad tripartitions,
	there is a collection $\hat B\subseteq T\setminus S'$ of at most $\eta_i$ terminals, such that $L'$ is precisely the set of all vertices of $G$ that are reachable from the vertices of $\hat B$ in $G\setminus S'$, so in particular $\hat B\subseteq T\subseteq T_0$ and $\hat B\cap S'=\emptyset$.
	Additionally, from Property \ref{prop: bad cuts: small separator}, $|S'|\leq 3\hat k\cdot \eta_i$ must hold. Therefore, triple $(i,S',\hat B)$ lies in $\Pi$. Notice also that, from Property \ref{prop: bad cuts: terminals}, $L'=\rho(S',\hat B)$, and, from Property \ref{prop: bad cuts: large volume}, $\vol^+_G(L')=\vol^+_G(\rho(S',\hat B))\geq \frac{1000\hat k\cdot m\cdot \log n}{N_i}$. We conclude that $(i,S',\hat B)\in \Pi'$. Lastly, from Property \ref{prop: bad cuts: missed cut} of bad tripartitions,  $\rho(\hat S,B')\cap T'=L'\cap T'=\emptyset$. Therefore, Event  $\hat \event (i,\hat S,\hat B)$ must have happened.
\end{proof}

In the following observation, we bound the probability of the event $\hat \event(i,\hat S,\hat B)$ for a single suspicious triple $(i,\hat S,\hat B)$; we will then use the Union Bound over all such triples to complete the proof of \Cref{claim: avoiding bad cuts}.

\begin{observation}\label{obs: bounding a single event}
	For every suspicious triple $(i,\hat S,\hat B)\in \Pi'$, $\prob{\hat \event(i,\hat S,\hat B)}\leq n^{-250\hat k\cdot \eta_i}$.	
\end{observation}
\begin{proof}
Consider a  suspicious triple $(i,\hat S,\hat B)\in \Pi'$, and denote $\Lambda=\rho(\hat S,\hat B)$. Recall that, from the definition of a suspicious triple, $\vol^+_G(\Lambda)\geq\frac{1000\hat k\cdot m\cdot\log n}{N_i}$.
Since $\vol^+_G(\Lambda)\leq m$ must hold, the triple $(i,\hat S,\hat B)$ may only be suspicious if $N_i\geq 1000\hat k\cdot \log n$.

 Recall that the bad event $\hat \event(i,\hat S,\hat B)$ may only happen if no vertex of $\Lambda$ was added to $T'$. Recall also that the set $T'$ of anti-terminals is constructed over the course of $r=\ceil{\frac{m}{\lambda}}$ independent trials, where in every trial we select a single vertex from $V(G)$ to be added to $T'$, and the probability of selecting a vertex $v$ is $\frac{\deg^+(v)}{m}$. Therefore, the probability that the vertex selected in a single trial lies in $\Lambda$ is at least $\frac{\vol^+_G(\Lambda)}{m}\geq \frac{1000\hat k\cdot \log n}{N_i}$, and so:

\[\prob{\hat \event(i,\hat S,\hat B)}\leq \left(1-\frac{1000\hat k\cdot \log n}{N_i}\right )^{r}\leq e^{-(1000\hat kr\log n)/N_i}\leq n^{-250\hat k\cdot \eta_i}.
\]

since $N_i\leq \frac{4|T|}{\eta_i}\leq \frac{4r}{\eta_i}$ from Inequality \ref{eq: bound 1 on Ni}.
\end{proof}

Finally, we bound the number of triples in the set $\Pi'$. Fix a level $0\leq i\leq z$. Note that there are at most $n^{3\hat k\cdot \eta_i}$ possible choices for a collection
$\hat S\subseteq V(G)$ of at most $3\hat k\eta_i$ vertices of $G$, and at most $n^{\eta_i}$ possible choices for a collection $\hat B\subseteq T_0\setminus \hat S$ of at most $\eta_i$ initial terminals. Therefore, the total number of level-$i$ triples in $\Pi$, and hence in $\Pi'$, is bounded by $n^{4\hat k\cdot \eta_i}$. From \Cref{obs: bounding a single event}, combined with the union bound, the probability that Event $\hat \event(i,\hat S,\hat B)$ happens for any level-$i$ triple $(i,\hat S,\hat B)$ is bounded by $n^{-250\hat k\cdot \eta_i}\cdot n^{4\hat k\cdot \eta_i}\leq n^{-100\hat k\cdot \eta_i}$, and, from \Cref{obs: relating events}, combined with the Union Bound, $\prob{\hat \event}\leq \frac{1}{n^2}$.
\end{proof}

We let $\event^*$ be the good event that the Event $\event$ has happened and the Event $\hat \event$ did not happen. By combining \Cref{obs: preprocessing successful} with \Cref{claim: avoiding bad cuts} and the Union Bound, we get that, if $\vol^+_G(L)<\frac{m}{4}$  and $k\leq \hat k$, then:

\[\prob{\neg\event^*}\leq \prob{\neg \event}+\prob{\hat \event}\leq 1-\frac{1}{2^{64}\cdot \log n}+\frac{1}{n^2}\leq 1-\frac{1}{2^{65}\cdot \log n}. \]

Therefore, if $\vol^+_G(L)<\frac{m}{4}$  and $k\leq \hat k$, then:
\begin{equation}\label{eq: bounding success prob}
\prob{\event^*}\geq \frac{1}{2^{65}\cdot \log n}.
\end{equation}

For convenience, in the remainder of the proof, we denote by $T''=T'\cup N^-_G(t^*)\cup \set{t^*}$. Note that $T''\cap T=\emptyset$ must hold, and, if Event $\event^*$ happened, then $T''\cap L=\emptyset$.

\subsection*{Processing a Single Batch of Terminals}

The key subroutine of our algoritm is a procedure for processing a single batch of terminals, that is summarized in the following claim.

\begin{claim}\label{claim: process one batch}
	There is a deterministic algorithm, whose input consists of an integer $1\leq i\leq z$, a level-$(i-1)$ batch $B\in \bset_{i-1}$ of terminals, a set $A_B\subseteq V(G)\setminus T''$ of vertices of $G$ with $B\cap A_B\neq \emptyset$ and $\vol^+_G(A_B) \leq \frac{1000\hat k\cdot m\cdot \log n}{N_{i-1}}$, together with a level-$i$ batch $B'\subseteq B$ of terminals. The algorithm either returns ``FAIL'', or computes a set $A_{B'}\subseteq  V(G)\setminus T''$ of vertices of $G$ with $B'\cap A_{B'}\neq\emptyset$ and $\vol^+_G(A_{B'}) \leq \frac{1000\hat k\cdot m\cdot \log n}{N_{i}}$. The algorithm guarantees that, if all of the following hold:
	
	\begin{itemize}
\item 	Event $\event^*$ has happened;
\item $k\leq \hat k$;

	\item $L\subseteq A_B$; and 
	
	\item $B'\cap L\neq\emptyset$,
	
	\end{itemize} 

then it may not return ``FAIL'', and moreover, $L\subseteq A_{B'}$ must hold.
	The running time of the algorithm is $O\left(\frac{\hat k\cdot m^{1+o(1)}}{N_{i-1}}\right )$.
\end{claim}
	
We provide the proof of \Cref{claim: process one batch} below, after we complete the proof of \Cref{thm: main: inner unweighted} using it.

\subsection*{The Algorithm Description}

At the beginning of the algorithm, we consider the unique level-$0$ batch $B\in \bset_0$ of terminals (recall that $B=T$). 
We set $\bset_0^I=\emptyset$ and $\bset_0^A=\bset_0=\set{B}$.
We also define $A_B=V(G)\setminus T''$. 
Since $N_0=1$, it is immediate to verify that $\vol^+(A_B)\leq \frac{1000\hat k\cdot m\cdot \log n}{N_0}$. Assume now that Event $\event^*$ has happened, so, in particular, $L\cap B=L\cap T\neq \emptyset$ and $t^*\in R$ hold. Then $N^-_G(t^*)\subseteq S\cup R$ must hold, and so $L\subseteq A_B$. 
Next, we perform at most $z$ phases.

\paragraph{Execution of a Phase.}
We now fix an index $1\leq i\leq z$ and describe the execution of the $i$th phase. We assume that, at the beginning of Phase $i$, we are given a partition $(\bset_{i-1}^A,\bset_{i-1}^I)$ of the set $\bset_{i-1}$ of level-$(i-1)$ batches, where the batches in $\bset_{i-1}^A$ are referred to as \emph{active}, and the batches in 
$\bset_{i-1}^I$ as \emph{inactive}. We also assume that, for every active level-$(i-1)$ batch $B\in \bset_{i-1}^A$, we are given 
a set $A_B\subseteq V(G)\setminus T''$ of vertices of $G$, such that $B\cap A_B\neq \emptyset$ and $\vol^+_G(A_B) \leq \frac{1000\hat k\cdot m\cdot \log n}{N_{i-1}}$ hold. Lastly, we assume that, if Event $\event^*$ has happened and $k\leq \hat k$ then for every level-$(i-1)$ batch $B\in \bset_{i-1}$, if  $L\cap B\neq \emptyset$, then $B\in \bset^A_{i-1}$ holds, and, moreover, $L\subseteq A_B$.

We consider every level-$i$ batch  one by one. Let $B'\in \bset_i$ be any such batch, and let $B\in \bset_{i-1}$ be the unique level-$(i-1)$ batch with $B'\subseteq B$. If $B\in \bset_{i-1}^I$, then we add $B'$ to the set $\bset_i^I$ of inactive level-$i$ batches. Otherwise, we apply the algorithm from \Cref{claim: process one batch} to the batches $B$ and $B'$ and the set $A_B$ of vertices of $G$. If the algorithm returns ``FAIL'', then we add $B'$  to the set $\bset_i^I$ of inactive level-$i$ batches. Otherwise, the algorithm must have computed a set $A_{B'}\subseteq  V(G)\setminus T''$ of vertices of $G$ with $B'\cap A_{B'}\neq\emptyset$ and $\vol^+_G(A_{B'}) \leq \frac{1000\hat k\cdot m\cdot \log n}{N_{i}}$.  We then add $B'$ to the set $\bset_i^A$ of active level-$i$ batches, and we record its associated vertex set $A_{B'}$. 

This completes the description of the algorithm for Phase $i$. Note that, at the end of the algorithm, we obtain a partition $(\bset_i^A,\bset_i^I)$ of the set $\bset_i$ of level-$i$ batches into the set $\bset_i^A$ of active batches and the set $\bset_i^I$ of inactive batches. For every active level-$i$ batch $B'\in \bset^A_i$, we have computed a set  $A_{B'}\subseteq  V(G)\setminus T''$ of vertices of $G$ with $B'\cap A_{B'}\neq\emptyset$ and $\vol^+_G(A_{B'}) \leq \frac{1000\hat k\cdot m\cdot \log n}{N_{i}}$. 
Lastly,  assume that Event
$\event^*$ has happened and that $k\leq \hat k$, and consider some level-$i$ batch $B'\in \bset_i$ with $B'\cap L\neq \emptyset$.
Let $B\in \bset_i$ be the unique level-$(i-1)$ batch containing $B$. Then, from our assumption, $B\in \bset^A_{i-1}$ and $L\subseteq A_B$ hold. 
Therefore, our algorithm must have applied the algorithm from \Cref{claim: process one batch} to the batches $B$ and $B'$, and, moreover, the algorithm from \Cref{claim: process one batch} did not return ``FAIL''. Therefore, $B'$ was added to the set $\bset_i^A$ of active level-$i$ batches. Lastly, the algorithm from \Cref{claim: process one batch} guarantees that in this case, if $A_{B'}$ is the vertex set that it returns, then $L\subseteq A_{B'}$.

Recall that the running time of the algorithm from \Cref{claim: process one batch} is $O\left(\frac{\hat k\cdot m^{1+o(1)}}{N_{i-1}}\right )$, and we apply it to at most $|\bset_i|\leq N_i$ batches. Moreover, from Inequalities \ref{eq: bound 1 on Ni} and \ref{eq: compare set sizes}, $N_i\leq O(N_{i-1})$. Therefore, the running time of a single phase is:

\[O\left(\frac{\hat k\cdot m^{1+o(1)}}{N_{i-1}}\cdot N_i\right )\leq O\left(\hat k\cdot m^{1+o(1)}\right ).\]

\paragraph{Completing the Algorithm.}

Recall that every level-$z$ batch $B\in \bset_z$ of terminals contains exactly one terminal, that we denote by $s_B$. Recall also that, if Event $\event^*$ has happened, then $T\cap L\neq \emptyset$. If Event $\event^*$ has happened, then we let $s^*$ be any terminal in $T\cap L$, and otherwise we let $s^*$ be any terminal in $T$. We denote by $B^*$ the unique level-$z$ batch with $s_{B^*}=s^*$. Note that, if Event $\event^*$ has happened and $k\leq \hat k$, then our algorithm guarantees that $B^*\in \bset^A_z$ holds, and, moreover, $L\subseteq A_{B^*}$. 

If $\bset_z^A=\emptyset$, then we terminate the algorithm and return the arbitrary vertex-cut $(\tilde L,\tilde S,\tilde R)$ that we computed at the beginning of the algorithm; note that, from our discussion, this may only happen if Event $\event^*$ did not happen, or if $\hat k>k$.

Assume now that $\bset_z^A\neq \emptyset$. We process every active level-$z$ batch $B\in \bset_z^A$ one by one. When a batch $B\in \bset_z^A$ is processed, we compute the corresponding graph $G_B=G^{|A_B}$ that is derived from $G$ via $A_B$.  Recall that, from \Cref{claim: properties of defined graph}, graph $G_B$ can be computed in time $O\left(\vol^+_G(A_B)\right )\leq O\left(\frac{\hat k\cdot m\cdot \log n}{N_z}\right )$, and, moreover, $|E(G_B)|\leq O\left(\vol^+_G(A_B)\right )\leq O\left(\frac{\hat k\cdot m\cdot \log n}{N_z}\right )$. Recall also that $A_B\cap B\neq \emptyset$ must hold, so $s_B\in A_B$.
We then compute a minimum $s_B$-$t$ vertex-cut $(L_B,S_B,R_B)$ in $G_B$ in time  $O\left(|E(G_B)|^{1+o(1)}\right )\leq O\left(\frac{\hat k\cdot m^{1+o(1)}}{N_z}\right )$,
using the algorithm from \Cref{cor: min_vertex_cut}.
Finally, we select a batch $B\in \bset_z$ for which $|S_B|$ is minimized, and return the tripartition $(L',S',R')$ of $V(G)$ that is induced by the vertex set $L'=L_B$.
Recall that, from \Cref{claim: process one batch}, $A_B\subseteq V(G)\setminus T''$ must hold, so $t^*\not\in A_B\cup N_G^+(A_B)$, and, in particular, $A\cup N^+_G(A)\neq V(G)$. From \Cref{claim: transforming cuts from defined graph}, $(L',S',R')$ is a valid vertex-cut in $G$. We return this cut as the algorithm's output.

From the above discussion, the time required to process a single level-$z$ batch $B\in \bset_z^A$ is bounded by $O\left(\frac{\hat k\cdot m^{1+o(1)}}{N_z}\right )$, and, since the number of level-$z$ batches is bounded by $N_z$, the running time of this last step is bounded by $O\left(\hat k\cdot m^{1+o(1)}\right )$. Since the algorithm has $z=O(\log n)$ phases, and the running time of each phase is bounded by  $O\left(\hat k\cdot m^{1+o(1)}\right )$, and since the running time of the preprocessing step is bounded by $\tilde O(m)$, the total running time of the algorithm is  $O\left(\hat k\cdot m^{1+o(1)}\right )$.

As observed already, if Event $\event^*$ happened and $\hat k\leq k$, then $B^*\in \bset^A_z$ and $L\subseteq A_{B^*}$ hold.
From \Cref{claim: properties of defined graph}, $(L,S,V(G_{B^*})\setminus (L\cup S)))$ is a valid vertex-cut in $G^{B^*}$, and it is easy to see that it must be an $s^*$-$t$ vertex-cut in $G_{B^*}$ of value $k$.  In this case, we are guaranteed that the vertex-cut $(L',S',R')$ that our algorithm returns has value $|S'|=k$.

Recall that, if $\vol^+_G(L)\geq \frac{m}{4}$, then the vertex-cut that our algorithm returns is guaranteed to be optimal with probability at least $\frac{1}{32}$. If $\vol^+_G(L)< \frac{m}{4}$, and if $\hat k\leq k$ and Event $\event^*$ happens, then the vertex-cut that our algorithm returns is guaranteed to be optimal. Since, from 
\Cref{eq: bounding success prob}, if $k\leq \hat k$ and $\vol^+_G(L)<\frac{m}{4}$ then
$\prob{\event^*}\geq \frac{1}{2^{65}\cdot \log n}$, 
we get that, overall, if $\hat k\leq k$, the probability that the vertex-cut that our algorithm returns is optimal is $\Omega\left(\frac{1}{\log n}\right )$. By repeating the algorithm $O(\log^2n)$ times and returning the lowest-value cut among the resulting vertex-cuts, we guarantee that the probability that  the vertex-cut that our algorithm returns is optimal is at least $\left(1-\frac{1}{n}\right )$. The running time of the algorithm remains $O\left(\hat k\cdot m^{1+o(1)}\right )$.

In order to complete the proof of \Cref{thm: main: inner unweighted} it now remains to prove \Cref{claim: process one batch}, which we do next.


\subsection{Processing a Single Batch of Terminals: Proof of \Cref{claim: process one batch}}

We assume that we are given  an integer $1\leq i\leq z$, a level-$(i-1)$ batch $B\in \bset_{i-1}$ of terminals, a set $A_B\subseteq V(G)\setminus T''$ of vertices of $G$ with $B\cap A_B\neq \emptyset$ and $\vol^+_G(A_B) \leq \frac{1000\hat k\cdot m\cdot \log n}{N_{i-1}}$, and a level-$i$ batch $B'\subseteq B$ of terminals. 
We denote by $\hat G=G^{|A_B}$ the graph derived from $G$ via the vertex set $A_B$. From  \Cref{claim: properties of defined graph}, graph $\hat G$ can be computed in time $O\left(\vol^+_G(A_B)\right )\leq O\left(\frac{\hat k\cdot m\cdot \log n}{N_{i-1}}\right )$, and, moreover, $|E(\hat G)|\leq O\left(\vol^+_G(A_B)\right )\leq O\left(\frac{\hat k\cdot m\cdot \log n}{N_{i-1}}\right )$ must hold.

\paragraph{Condition (C).}
For brevity, we say that Condition (C) holds, if all of the following hold:

	\begin{itemize}
	\item 	Event $\event^*$ has happened;
	\item $k\leq \hat k$;
	\item $L\subseteq A_B$; and 
	
	\item $B'\cap L\neq\emptyset$,
\end{itemize} 

Recall that our goal is to compute a set $A_{B'}\subseteq  V(G)\setminus T''$ of vertices of $G$ with $B'\cap A_{B'}\neq\emptyset$ and $\vol^+_G(A_{B'}) \leq \frac{1000\hat k\cdot m\cdot \log n}{N_{i}}$, such that, if Condition (C) holds, then $L\subseteq A_{B'}$. Our algorithm may return ``FAIL'' if Condition (C) does not hold.  

%

The following observation follows immediately from the fact that $A_B\subseteq V(G)\setminus T''$, and from the definition of the graph $\hat G=G^{A_B}$.

\begin{observation}\label{obs: pseudoterminals connect to t}
	If $v\in T''$ is a vertex of $\hat G$, then edge $(v,t)$ lies in $\hat G$. In particular, if $(\hat L,\hat S,\hat R)$ is any vertex-cut in $\hat G$ with $t\in \hat R$, then $\hat L\cap T''=\emptyset$.
\end{observation}

If $B'\cap A_B=\emptyset$, then we terminate the algorithm and return ``FAIL''. Note that, in this case, Condition (C) does not hold. Indeed, if (C) holds, then, since  $B'\cap L\neq \emptyset$ and $B'\cap L\subseteq L\subseteq A_B$, we get that $B'\cap A_B\neq \emptyset$. Therefore, we assume from now on that $B'\cap A_B\neq \emptyset$.

\paragraph{High-level overview.}
At a high level, our algorithm will compute a graph $H$, that can be viewed as a slight modification of the graph $\hat G$, and it will then compute a vertex-cut in $H$. After minor modifications, the resulting cut will be transformed into a vertex-cut $(L',S',R')$ in $\hat G$, and we will then set $A_{B'}=L'$, and prove that the vertex set $A_{B'}$ has all the required properties. We now proceed with a formal description of our algorithm, starting with the construction of the graph $H$.

\paragraph{Graph $H$.} We denote by $\hat T=B'\cap A_B$; recall that, from our assumption,  $\hat T\neq \emptyset$. In order to construct the graph $H$, we start with $H=\hat G$, and then, for every terminal $v\in \hat T$, we add a new vertex $v'$ to $H$, and then connect $v'$ with an edge to every vertex $u\in N^+_G(v)$ (recall that each such vertex $u$ must lie in $\hat G=G^{|A_B}$ and edge $(v,u)$ must lie in $\hat G$, from the definition of a graph derived from $G$ via $A_B$). We refer to $v'$ as a \emph{copy of $v$}, and we denote by $\hat T'=\set{v'\mid v\in\hat T}$ the collection of all copies of the terminals in $\hat T$; we may sometimes also refer to the vertices in $\hat T'$ as \emph{pseudoterminals}. Lastly, we add a source vertex $s$ to $H$, and connect it with an edge to every vertex in $\hat T'$.  We set the weights of vertices $v\in V(H)$ as follows. If $v\in \hat T'$, then its weight is $\hat w(v)=2\hat k$, and otherwise its weight is $\hat w(v)=1$. This completes the description of the graph $H$. It is easy to verify that $H$ can be constructed in time $O(|E(\hat G)|)\leq O\left(\frac{\hat k\cdot m\cdot \log n}{N_{i-1}}\right )$, and that $|E(H)|\leq O(|E(\hat G)|)\leq O\left(\frac{\hat k\cdot m\cdot \log n}{N_{i-1}}\right )$ must hold.

\paragraph{A vertex-cut in $H$.}
As our next step, we compute a minimum $s$-$t$ vertex-cut $(\hat L,\hat S,\hat R)$ in $H$ using the algorithm from \Cref{cor: min_vertex_cut}, in time $O\left (|E(H)|^{1+o(1)}\cdot \log \hat k\right )\leq O\left(\frac{\hat k\cdot m^{1+o(1)}}{N_{i-1}}\right )$. In the following two simple observations we establish some basic properties of the cut $(\hat L,\hat S,\hat R)$.

\begin{observation}\label{obs: cut value}
$\hat w(\hat S)\leq 2\hat k\eta_i$.
\end{observation}
\begin{proof}
	It is enough to show that there is an $s$-$t$ vertex-cut $(L^*,S^*,R^*)$ in $H$ of value $\hat w(S^*)\leq 2\hat k\cdot \eta_i$. We define the vertex-cut 
	$(L^*,S^*,R^*)$ as follows: $L^*=\set{s}$; $S^*=\hat T'$; and $R^*=\set{t}\cup (R\cap V(H))$.
Clearly, $s\in L^*$ and $t\in R^*$ holds. Moreover, since all edges leaving $s$ in $H$ connect it to vertices of $\hat T'$, it is immediate to verify that $(L^*,S^*,R^*)$ is a valid $s$-$t$ cut in $H$. Finally:
	
	\[\hat w(\hat S)\leq \hat w(\hat T')\leq 2\hat k\cdot\eta_i.
	\]

\end{proof}
\begin{observation}\label{obs: terminal in L}
	Assume that Condition (C) holds, and let $v\in L\cap B'$ be any vertex. Then the copy $v'$ of $v$ must belong to $\hat T'$, and moreover, it must lie in $\hat L$.
\end{observation}
\begin{proof}
	Assume that Condition (C) holds and let $v\in L\cap B'$ be any vertex. Since $L\subseteq A_B$, $v\in A_B$ must hold, so $v\in \hat T$ and $v'\in \hat T'$ must hold. From the construction of the graph $H$, $v'\in V(H)$ and $(s,v')\in E(H)$. Since $(\hat L,\hat S,\hat R)$ is an $s$-$t$ cut in $H$, we get that $v'\in \hat L\cup \hat S$ must hold. Assume for contradiction that $v'\in \hat S$.
	
	Recall that $(L,S,R)$ is the distinguished minimum vertex-cut in $G$ with $|S|\leq k$, and it is a vertex-cut that is induced by $L$ in $G$. Additionally, since $L\subseteq A_B$ and $S=N^+_G(L)$, we get that $L\cup S\subseteq V(\hat G)\subseteq V(H)$.
	Consider the tripartition $(\hat L',\hat S',\hat R')$ of vertices of $H$, that is obtained from $(\hat L,\hat S,\hat R)$ by first moving $v'$ and all the vertices of $L$ to $\hat L$, and then moving all vertices of $S$ to $\hat S$.  Note that, in particular, vertex $v'$, whose weight in $H$ is $2\hat k\geq 2k$, lies in $\hat S\setminus \hat S'$, while the only vertices that lie in $\hat S'\setminus\hat S$ are the vertices of $S$, whose total weight in $H$ is bounded by $|S|\leq k$. Therefore, $\hat w(\hat S')<w(\hat S)$. We claim that $(\hat L',\hat S',\hat R')$ is a valid $s$-$t$ vertex-cut in $H$, contradicting the fact that $(\hat L,\hat S,\hat R)$ is a minimum $s$-$t$ vertex-cut in $H$.
	
	Indeed, it is immediate to verify that $s\in \hat L'$ and $t\in \hat R'$ hold. It is now enough to show that no edge may connect a vertex of $\hat L'$ to a vertex of $\hat R'$ in $H$. Assume for contradiction that there is an edge $(a,b)\in V(H)$ with $a\in \hat L'$ and $b\in \hat R'$, and notice that $b\in \hat R$ must hold, since $\hat R'\subseteq \hat R$. It is impossible that $a\in \hat L$ because $(\hat L,\hat S,\hat R)$ is a valid $s$-$t$ vertex-cut in $H$. Therefore, $a\in L$ or $a=v'$ must hold. In the former case, since $(L,S,R)$ is a valid vertex-cut in $G$, and since, from Condition (C), $a\in A_B$ must hold, $N^+_H(a)\subseteq N^+_G(a)\subseteq L\cup S$. 
	In the latter case, $N^+_H(v')\subseteq N^+_G(v)\subseteq L\cup S$.
	Therefore, $b\in L\cup S$ must hold. Since $L\subseteq \hat L'$ and $S\subseteq \hat S'$, this contradicts the fact that $b\in \hat R'$. We conclude that $(\hat L',\hat S',\hat R')$ is a valid $s$-$t$ vertex-cut in $H$, whose value is strictly lower than that of $(\hat L,\hat S,\hat R)$, a contradiction.
\end{proof}

Our next step is to slightly modify the $s$-$t$ vertex-cut $(\hat L,\hat S,\hat R)$ in $H$, to obtain a new $s$-$t$ vertex-cut $(\hat L'',\hat S'',\hat R'')$ in $H$, as follows. 
We start with $(\hat L'',\hat S'',\hat R'')=(\hat L,\hat S,\hat R)$, and then we consider every vertex $v\in \hat T$ one by one. If $v'\in \hat L''$ but $v\not\in \hat L''$, then we move $v$ to $\hat L''$. Notice that this move cannot increase the weight of the set $\hat S''$. Moreover, since $N^+_H(v)=N^+_H(v')$ by the construction of the graph $H$, we get that $N^+_H(v)\subseteq \hat L''\cup \hat S''$. Therefore, $(\hat L'',\hat S'',\hat R'')$ remains a valid $s$-$t$ vertex-cut in $H$, with $\hat w(\hat S'')\leq 3\hat k\eta_i$. Recall that, for every pseudoterminal $v'\in \hat T'$, edge $(s,v')$ lies in $H$, so $v'\in \hat L\cup \hat S$ must hold.
For convenience, abusing the notation, we denote $(\hat L'',\hat S'',\hat R'')$ by $(\hat L,\hat S,\hat R)$ in the remainder of the proof. 
We summarize the properties of the cut  $(\hat L,\hat S,\hat R)$ that we have established so far in the following observation.

\begin{observation}\label{obs: props of cut in H}
The tripartition $(\hat L,\hat S,\hat R)$ of $V(H)$ is a minimum $s$-$t$ vertex-cut in $H$ with $\hat w(\hat S)\leq 2\hat k\eta_i$. 
 Moreover, for every vertex $v\in \hat T$, $v'\in \hat L\cup \hat S$ holds, and, if $v'\in \hat L$, then $v\in \hat L$ as well.
 Lastly, if Condition (C) holds, then for every vertex $v\in L\cap B'$, $v\in \hat L$ must hold.
\end{observation}

(The last assertion follows immediately from \Cref{obs: terminal in L}).

\paragraph{A vertex-cut in $\hat G$.}
Next, we define a tripartition $(L'',S'',R'')$ of vertices of $\hat G$ as follows: $L''=\hat L\cap V(\hat G)$; $S''=\hat S\cap V(\hat G)$; and $R''=V(\hat G)\setminus\left(L''\cup S''\right )=\hat R\cap V(\hat G)$. Since $(\hat L,\hat S,\hat R)$ is a valid vertex-cut in $H$, it is easy to verify that no edge in $\hat G$ may connect a vertex of $L''$ to a vertex of $R''$. 
In the remainder of the proof, we denote by $\hat B=\hat T\cap L''$.
We will use the following central claim to establish the properties of the tripartition $(L'',S'',R'')$  of $V(\hat G)$.

\begin{claim}\label{claim: properties of tripartition in hat G}
	 Set $L''$ may not contain any vertices of $T''$. 
	Moreover if Condition (C) holds, then $(L'',S'',R'')$ is a valid vertex-cut in $\hat G$ with $|S''|\leq 3\hat k\eta_i$ and $B'\cap L\subseteq L''$. Additionally, in this case $\hat B\neq \emptyset$ must hold, and $(L'',S'',R'')$ is a minimum vertex-cut in $\hat G$ separating $\hat B$ from $t$.
\end{claim}
\begin{proof}
	Consider any vertex $v\in T''$ that lies in $\hat G$.
	From \Cref{obs: pseudoterminals connect to t}, edge $(v,t)$ lies in $\hat G$, and hence in $H$. Since $(\hat L,\hat S,\hat R)$ of $V(H)$ is a valid $s$-$t$ vertex-cut in $H$, vertex $v$ may not lie in $\hat L$, and so $v\not\in L''$. Therefore, $L''$ may not contain any vertices of $T''$.
	
	Next, we assume that Condition (C) holds. We have established already that no edge may connect a vertex of $L''$ to a vertex of $R''$ in $\hat G$. Since, from \Cref{obs: props of cut in H}, $(\hat L,\hat S,\hat R)$ is a valid $s$-$t$ vertex-cut in $H$, $t\in \hat R$ and hence $t\in R''$ holds. Since we assumed that Condition (C) holds, $B'\cap L\neq \emptyset$, and, from \Cref{obs: props of cut in H}, $B'\cap L\subseteq \hat L$ must hold. Therefore, $B'\cap L\subseteq L''$, and, since $B'\cap L=B'\cap A_B=\hat T$, we get that $B'\cap L\subseteq \hat B$, and $\hat B\neq \emptyset$. It is also immediate to verify that $|S''|\leq \hat w(\hat S)\leq 2\hat k\eta_i$, and, from our dicsussion so far, $(L'',S'',R'')$ is a valid vertex-cut in $\hat G$ separating $\hat B$ from $t$. It now only remains to show that $(L'',S'',R'')$ is a minimum vertex-cut in $\hat G$ separating $\hat B$ from $t$. The following observation will then complete the proof of the claim.
	
\begin{observation}\label{obs: minimality of the cut}
	Assume that Condition (C) holds and let $(\tilde L,\tilde S,\tilde R)$ be any vertex-cut in $\hat G$ separating $\hat B$ from $t$. Then $|\tilde S|\geq |S''|$.
\end{observation}
\begin{proof}
	Assume that  Condition (C) holds, and recall that, from \Cref{obs: props of cut in H}, for every vertex $v\in \hat T$, $v'\in \hat L\cup \hat S$ holds. We partition the set $\hat T'$ of pseudoterminals into two subsets, $\hat T'_1=\hat T'\cap \hat L$ and $\hat T'_2=\hat T'\cap \hat S$. Observe that $\hat w(\hat S)=|S''|+2\hat k\cdot |\hat T'_2|$, and that, from \Cref{obs: props of cut in H}, for every pseudoterminal $v'\in \hat T'_1$, the corresponding terminal $v\in \hat T=B'\cap A_B$ must lie in $\hat L$, and hence in $L''$; in particular, $v\in \hat B$ must hold.
	
	Consider now any vertex-cut $(\tilde L,\tilde S,\tilde R)$  separating $\hat B$ from $t$  in $\hat G$. Assume for contradiction that $|\tilde S|< |S''|$. We construct a tripartition $(\tilde L',\tilde S',\tilde R')$ of $V(H)$, by letting $\tilde L'=\tilde L\cup\set{s}\cup \hat T'_1$, $\tilde S'=\tilde S\cup \hat T'_2$ and $\tilde R'=\tilde R$. Note that $\hat w(\tilde S')=|\tilde S|+2\hat k\cdot |\hat T'_2|<|S''|+2\hat k\cdot |\hat T'_2|=\hat w(\hat S)$. Next, we show that $(\tilde L',\tilde S',\tilde R')$ is a valid $s$-$t$ vertex-cut in $H$,  contradicting the fact that, from \Cref{obs: props of cut in H}, $(\hat L,\hat S,\hat R)$ is a minimum $s$-$t$ vertex-cut in $H$.
	
	From our construction it is immediate to verify that $s\in \tilde L'$ and $t\in \tilde R'$. It is now enough to show that no edge in $H$ may connect a vertex of $\tilde L'$ to a vertex of $\tilde R'$. Indeed, assume otherwise, and let $e=(x,y)$ be an edge of $E(H)$ with $x\in \tilde L'$ and $y\in \tilde R'$, so, in particular, $y\in \tilde R$. Since $(\tilde L,\tilde S,\tilde R)$  is a valid vertex-cut in $\hat G$, no edge in $\hat G$ (and hence in $H$) may connect a vertex of $\tilde L$ to a vertex of $\tilde R$, so $x\not \in \tilde L$ must hold. Note that it is impossible that $x=s$, since $N^+_H(s)=\hat T'$ and $\hat T'\subseteq \tilde L'\cup \tilde S'$. Therefore, $x\in \hat T'$ must hold. We assume that $x$ is a copy of the terminal $v\in \hat T$, so $x=v'$. Since $v'\in \hat T'_1$, from \Cref{obs: props of cut in H}, $v\in \hat L$ must hold, so, in particular, $v\in L''$ and $v\in \hat B$. But then $v\in \tilde L$ must hold and $N^+_H(v)=N_{G}^+(v)\subseteq \tilde L\cup \tilde S$. Since $N^+_H(v')=N^+_H(v)$, we get that $N^+_H(v')\subseteq \tilde L\cup \tilde S$, contradicting the fact that an edge $(v',y)$ with $y\in \tilde R$ lies in $H$.
	We conclude that $(\tilde L',\tilde S',\tilde R')$ is a valid $s$-$t$ cut in $H$,  contradicting the fact that, from \Cref{obs: props of cut in H}, $(\hat L,\hat S,\hat R)$ of $V(H)$ is a minimum $s$-$t$ vertex-cut in $H$.
\end{proof}
\end{proof}

If $\hat B=\emptyset$, or if $(L'',S'',R'')$ is not a valid vertex-cut in $\hat G$ separating  $\hat B$ from $t$, or if $|S''|>2\hat k\eta_i$, then we terminate the algorithm and return ``FAIL''. Notice that, these conditions can be checked in time $O(|E(H)|)$, and, if any of them holds, then, from \Cref{claim: properties of tripartition in hat G}, Condition (C) does not hold. Therefore, we assume from now on that $\hat B\neq\emptyset$, and that $(L'',S'',R'')$ is a valid vertex-cut in $\hat G$ separating  $\hat B$ from $t$ with $|S''|\leq 3\hat k\eta_i$.

\paragraph{Vertex Set $A_{B'}$.}
We let $A_{B'}$ be the set of all vertices $v\in V(\hat G)$, such that there is a path in $\hat G$ connecting some vertex of $\hat B$ to $v$, that avoids the vertices of $S''$. Note that, in particular, $\hat B\subseteq A_{B'}$, and so $A_{B'}\neq \emptyset$. In the following claim, we establish some basic properties of the vertex set $A_{B'}$.

\begin{claim}\label{claim: properties of AB'}
Set	$A_{B'}$ may not contain a vertex of $T''$, and $B'\cap A_{B'}\neq\emptyset$. Moreover, if Condition (C) holds, then $\vol^+_G(A_{B'}) \leq \frac{1000\hat k\cdot m\cdot \log n}{N_{i}}$.
\end{claim}
\begin{proof}
	It is immediate to verify that $A_{B'}\subseteq L''$ must hold, since the set $S''$ of vertices separates $\hat B$ from $R''$ in $\hat G$. Since, from \Cref{claim: properties of tripartition in hat G}, set $L''$ may not contain any vertices of $T''$, the same holds for $A_{B'}$. Since we ensured that $\hat B\neq \emptyset$, and since $\hat B\subseteq B'$ and $\hat B\subseteq A_{B'}$, we get that $B'\cap A_{B'}\neq\emptyset$. 
	
	Assume now that $\vol^+_G(A_{B'}) > \frac{1000\hat k\cdot m\cdot \log n}{N_{i}}$. We claim that, in this case, $(A_{B'}, S'',V(G)\setminus(A_{B'}\cup S''))$ is a bad tripartition of $V(G)$ for level $i$ (see \Cref{def: bad tripartition}). Indeed, our algorithm ensures that $|S''|\leq 3\hat k\cdot \eta_i$, and we have already shown that $A_{B'}$ may not contain vertices of $T''\subseteq T'$, establishing properties \ref{prop: bad cuts: small separator} and \ref{prop: bad cuts: missed cut}. Property \ref{prop: bad cuts: large volume} follows from from our assumption that $\vol^+_G(A_{B'}) > \frac{1000\hat k\cdot m\cdot \log n}{N_{i}}$, and Property \ref{prop: bad cuts: terminals} follows from the definition of the vertex set $A_{B'}$. Therefore, if Condition (C) holds, then $\vol^+_G(A_{B'}) \leq \frac{1000\hat k\cdot m\cdot \log n}{N_{i}}$ must hold.
\end{proof}

If $\vol^+_G(A_{B'}) > \frac{1000\hat k\cdot m\cdot \log n}{N_{i}}$, then we return ``FAIL''. From \Cref{claim: properties of AB'}, Condition (C) does not hold in this case. Otherwise, we return the vertex set $A_{B'}$. 
In the next claim we show that, if Condition (C) holds, then $L\subseteq A_{B'}$ must hold.

\begin{claim}\label{claim: establising the cut property}
	Assume that Condition (C) holds. Then $L\subseteq A_{B'}$.
\end{claim}
\begin{proof}
	We assume that Condition (C) holds. From \Cref{claim: properties of tripartition in hat G}, $B'\cap L\subseteq L''$. Since we have denoted by $\hat B=\hat T\cap L''=(B'\cap A_B)\cap L''$, and since, from Condition (C), $L\subseteq A_B$, we get that $\hat B$ must contain some vertex of $B'\cap L$. We let $v$ be any such vertex. 

At a high level, we start by showing that $(L,S,V(\hat G)\setminus(L\cup S))$ is a minimum $v$-$t$ vertex-cut in $\hat G$, and that it minimizes $|L|$ among all such cuts. We also show that the tripartition $(L',S',R')$ of vertices of $\hat G$ that is induced by vertex set $L'=A_{B'}$ is a minimum $\hat B$-$t$ vertex-cut in $\hat G$. We then use the submodularity of vertex-cuts in order to argue that the tripartition of $V(\hat G)$ induced by $L\cap L'$ is a valid $v$-$t$ vertex-cut in $\hat G$ of value at most $|S|$. Since 
$(L,S,V(\hat G)\setminus(L\cup S))$ is a minimum $v$-$t$ vertex-cut in $\hat G$ that minimizes $|L|$, we then conclude that $L\subseteq L'$ must hold. We now proceed with the formal proof of the claim. In the next two simple observations we establish the properties of the tripartition $(L,S,V(\hat G)\setminus(L\cup S))$ of $\hat G$, and the tripartition $(L',S',R')$ of $V(\hat G)$ induced by the vertex set $L'=A_{B'}$, respectively.
Note that, since $L\subseteq A_B$, $L\cup S\subseteq V(\hat G)$ must hold, from the definition of the graph $\hat G=G^{|A_B}$ that is derived from $G$ via the vertex set $A_B$.

\begin{observation}\label{obs: first vertex cut in hat G}
	Assume that Condition (C) holds. 
	Then the tripartition $(L,S,V(\hat G)\setminus(L\cup S))$ of $V(\hat G)$ is a minimum $v$-$t$ vertex-cut in $\hat G$. Moreover, for every other $v$-$t$ vertex-cut $(X,Y,Z)$ in $\hat G$, either $|Y|>|S|$ or $|X|\geq |L|$ must hold.
\end{observation}
\begin{proof}
	From the definition of the distinguished cut $(L,S,R)$, it is induced by the set $L$ of vertices, and, from our assumption that Condition (C) holds, $L\subseteq A_B$. From \Cref{claim: properties of defined graph}, $(L,S,V(\hat G)\setminus(L\cup S))$ is a valid vertex-cut in $\hat G$. Clearly, $v\in L$ and $t\in V(\hat G)\setminus(L\cup S)$ must hold, so it is a valid $v$-$t$ vertex-cut in $\hat G$.
	
	Consider now any  other $v$-$t$ vertex-cut $(X,Y,Z)$ in $\hat G$. Note that, since $A_B\cap T''=\emptyset$, $V(G)\setminus (A_B\cup N^+_G(A_B))$ contains the vertex $t^*$, and so $A_B\cup N^+_G(A_B)\neq V(G)$. From \Cref{claim: transforming cuts from defined graph}, $(X,Y,V(G)\setminus(X\cup Y))$ is a valid vertex-cut in $G$. From the choice of the distinguished min-cut, either $|Y|>|S|$, or $|X|\geq |L|$ must hold.
\end{proof}

Next, we let $(L',S',R')$ be the tripartition of $V(\hat G)$ that is induced by the vertex set $L'=A_{B'}$. Since $A_{B'}\subseteq L''$, it is immediate to see that $t\in R'$. From the definition of the vertex set $A_{B'}$, we get that $\hat B\subseteq A_{B'}$. Therefore, $(L',S',R')$ is a $\hat B$-$t$ vertex-cut in $\hat G$. From \Cref{claim: properties of tripartition in hat G}, it is also a minimum $\hat B$-$t$ vertex-cut in $\hat G$. We summarize this in the following observation.

\begin{observation}\label{obs: cut 2 in hat G}
	If Condition (C) holds, then $(L',S',R')$ is  a minimum $\hat B$-$t$ vertex-cut in $\hat G$.
\end{observation}

The following standard observation, whose proof is deferred to Section \ref{subsec: proof of submodularity of cuts} of Appendix, establishes the submodularity of induced vertex-cuts in directed graphs.
	
	\begin{observation}[Submodularity of weighted vertex-cuts.]\label{obs:submodularity} Let $G$ be a directed graph with weights $w(v)\geq 0$ on its vertices $v\in V(G)$, and let $A,B\subseteq V(G)$ be a pair of vertex subsets. Then:
		
		$$w(N^+_G(A)) + w(N^+_G(B)) \geq w(N^+_G(A \cup B)) + w(N^+_G(A \cap B)).$$
	\end{observation}

Consider now the set $(L\cup L')$ of vertices of $\hat G$. Clearly, $N^+_{\hat G}(L\cup L')\subseteq N^+_{\hat G}(L)\cup  N^+_{\hat G}(L')$, so $t\not\in N^+_{\hat G}(L\cup L')$. Therefore, if we denote by$(L^*,S^*,R^*)$ the tripartition of $V(\hat G)$ induced by $L\cup L'$, then it is a valid vertex-cut in $\hat G$, and it is immediate to verify that it separates $\hat B$ from $t$. From \Cref{obs: cut 2 in hat G}, $w(S^*)\geq w(S')$ must hold. By applying \Cref{obs:submodularity}  to vertex sets $A=L$ and $B=L'$ in graph $\hat G$, and using the fact that $w(N^+_{\hat G}(L\cup L'))=w(N^+_{\hat G}(L^*))=w(S^*)\geq w(S')$, we get that $w(N^+_{\hat G}(L\cap L'))\leq w(N^+_{\hat G}(L))\leq w(S)$ must hold. It is immediate to verify that the tripartition of $V(\hat G)$ induced by $L\cap L'$ is a valid $v$-$t$ vertex-cut in $\hat G$, so, from \Cref{obs: first vertex cut in hat G}, $|L\cap L'|\geq |L|$ must hold. We conclude that $L\subseteq L'$ must hold.
\end{proof}

It now remains to bound the running time of the algorithm. Recall that graph $\hat G$ can be computed in time $O\left(\vol^+_G(A_B)\right )\leq O\left(\frac{\hat k\cdot m\cdot \log n}{N_{i-1}}\right )$, and, moreover, $|E(\hat G)|\leq O\left(\vol^+_G(A_B)\right )\leq O\left(\frac{\hat k\cdot m\cdot \log n}{N_{i-1}}\right )$ must hold. It is easy to verify that, given graph $\hat G$, graph $H$ can be constructed in time $O(|E(\hat G)|)\leq O\left(\frac{\hat k\cdot m\cdot \log n}{N_{i-1}}\right )$. The time required to compute 
a minimum $s$-$t$ vertex-cut $(\hat L,\hat S,\hat R)$ in $H$ using the algorithm from \Cref{cor: min_vertex_cut} is bounded by $ O\left(\frac{\hat k\cdot m^{1+o(1)}}{N_{i-1}}\right )$, as observed already. The time required to compute a modified vertex-cut in $H$ is bounded by $O(n)$, and this cut can be transformed into the tripartition $(L'',S'',R'')$ of $V(\hat G)$ in time $O(n)$. Lastly, the time required to compute the vertex set $A_{B'}$ is bounded by $O(|E(\hat G)|)\leq O\left(\frac{\hat k\cdot m\cdot \log n}{N_{i-1}}\right )$. The total running time of the algorithm is then bounded by  $ O\left(\frac{\hat k\cdot m^{1+o(1)}}{N_{i-1}}\right )$.

\section{Algorithm for Weighted Directed  Non-Dense Graphs}
\label{sec: weighted non-dense}

In this section we provide a randomized algorithm for the weighted directed global minimum vertex-cut problem with running time $O\left(mn^{11/12+o(1)}\cdot d^{1/12}\cdot \log^2 W\right )$, where $d$ is the average vertex degree in the input graph. The algorithm is most suitable for graphs that are not too dense, and is summarized in the following theorem.

\begin{theorem}\label{thm : sparse weighted algorithm}
	There is a randomized algorithm, whose input consists of a simple directed $n$-vertex and $m$-edge graph $G$ with integral weights $0 \leq w(v) \leq W$ on its vertices $v \in V(G)$, such that $G$ contains some vertex-cut. The algorithm computes a vertex-cut $(L',S',R')$ in $G$, so that, with probability at least $\frac{1}{2}$, $(L',S',R')$ is a global minimum vertex-cut.
	The running time of the algorithm is $O\left(m n^{11/12 + o(1)} \cdot d^{1/12} \cdot \log^2(W)\right)$, where $d$ is the average vertex degree in $G$.
\end{theorem}

In the remainder of this section we prove  \Cref{thm : sparse weighted algorithm}. We will use a parameter $0<\eps<1$, whose value will be set later. If there is a global minimum vertex-cut $(L,S,R)$ in the input graph $G$ for which either $|L|\geq n^{\eps}$ or $\vol_G(L)\ge n^{\eps}\cdot d$ hold, where $d$ is the average vertex degree in $G$, we will employ Algorithm $\algspec$ from \Cref{thm: alg for special} to compute a global minimum vertex-cut in $G$. Otherwise, we are guaranteed that every global minimum vertex-cut $(L,S,R)$ in $G$ has properties \ref{prop: small L}--\ref{prop: small volume of L}. 
It is therefore sufficient to design an algorithm for this latter case.

In the remainder of this section, we assume that we are given a directed $n$-vertex and $m$-edge graph $G=(V,E)$. Using the transformation from \Cref{subsec: positive weights}, we can assume that, for every vertex $v\in V(G)$, its weight is an integer $w(v)\geq 1$. We also assume that we are given a parameter $0<\eps<1$, whose value we will set later. For the sake of the analysis, we fix a global minimum vertex-cut $(L,S,R)$ in $G$, that we refer to as the \emph{distinguished min-cut}, as follows. If there exists a global minimum vertex-cut  in $G$ for which  Properties \ref{prop: small L}-\ref{prop: small volume of L} hold, then we let $(L,S,R)$ be any such cut, and we say that the distinguished min-cut is \emph{good}; otherwise, we let $(L,S,R)$ be an arbitrary global minimum vertex-cut and we say that the distinguished min-cut is \emph{bad}. As before, we use the convention that $w(L)\leq w(R)$ holds. Throughout, we denote by $d$ the average vertex degree in $G$. We also denote the parameters $\wmax'(G)$ and $\wmax(G)$ by $\wmax'$ and $\wmax$ respectively (see \Cref{subsec: split graph} for a definition).

\paragraph{Critical Threshold.}
One of the central notions that our algorithm uses is that of a \emph{critical 
threshold}, that we define next.

\begin{definition}[Critical Threshold]\label{def: cricical threshold}
Let $G$ be a directed graph with integral weights $w(v)\geq 1$ on its vertices $v\in V(G)$, let $0<\eps<1$ be a parameter, and let $(L',S',R')$ be a global minimum vertex-cut in $G$. An integer $\tau^*$ is the \emph{critical threshold} for $(L',S',R')$, if it is the smallest integral power of $2$ with $\tau^*\geq 1$, such that every vertex $u\in S'$ with $w(u)> \tau^*$ is a high-degree vertex (see \Cref{def: high-deg}). Notice that it is possible that $S'$ contains no vertex $u$ with $w(u)> \tau^*$. 
\end{definition}

Observe that $1\leq \tau^*\leq \wmax'$ must hold. Moreover, if $\tau^*=1$ then it is possible that $S$ only contains high-degree vertices. The critical thershold $\tau^*$ is not known to the algorithm. 
However, we will supply the algorithm with a \emph{guess} $\tau$ on the value of the critical threshold. We will then require that, if $\tau=\tau^*$, then the algorithm  returns a global minimum vertex-cut in $G$ with a sufficiently high probability.

\paragraph{Selecting parameters $\lambda$, $\tau$ and a set $\Gamma$ of vertex pairs.}
As our first step, we select values
 $1\le \lambda\leq n$ and $1\leq \tau\leq \wmax'$, that are both integral powers of $2$, together with a set $\Gamma$ of at most $\ceil{\frac{100n\log n}{\lambda}}$ pairs of vertices of $G$. Our goal is to ensure that, with a sufficiently high probability, $\tau=\tau^*$, $\frac{\lambda}{2}\leq |L|\leq\lambda$, and  there is a pair $(x,y)\in \Gamma$ with $x\in L$ and $y\in R$ (recall that $(L,S,R)$ is the distingiushed min-cut). 
The algorithm is summarized in the next claim, whose proof appears in Section \ref{appx: subsec: compute pairs 2} of Appendix.

\begin{claim}\label{claim: compute pairs2}
	There is a randomized algorithm, that we call $\algpairs$, whose input consists of a directed $n$-vertex and $m$-edge  graph $G$ with integral weights $w(v)\geq 1$ on its vertices $v\in V(G)$ and a parameter $0<\eps<1$. The algorithm returns 
 values $1\leq \tau\leq \wmax'$ and $1\le \lambda\leq 2n^{\eps}$ that are integral powers of $2$, together with a set $\Gamma$ of at most $\ceil{\frac{100n\log n}{\lambda}}$ pairs of vertices of $G$, such that, for every pair $(x,y)\in \Gamma$, $x\neq y$ and $(x,y)\not\in E(G)$. We say that the algorithm is \emph{successful} with respect to a fixed global minimum vertex-cut $(L,S,R)$, if  $\frac{\lambda}{2}\leq |L|\leq\lambda$, $\tau=\tau^*$, where $\tau^*$ is the critical threshold for $(L,S,R)$, and  there is a pair $(x,y)\in \Gamma$ with $x\in L$ and $y\in R$. If $(L,S,R)$ is a global minimum vertex-cut for which Properties \ref{prop: small L}--\ref{prop: small volume of L} hold, then the algorithm is successful with respect to this cut with probability at least  $\frac{1}{32\log n\cdot\log \wmax}$. The running time of the algorithm is $\tilde O\left (m\cdot n^{\eps}\right )$.
\end{claim}

In the remainder of the proof, we denote by $\event_1$ the good event that Algorithm $\algpairs$ was successful with respect to the distinguished min-cut $(L,S,R)$. From  \Cref{claim: compute pairs2}, if the distinguished cut was good, then $\prob{\event_1}\geq \frac{1}{32\log n\cdot\log(\wmax)}$.

For the sake of analysis, we designate a single pair $(x^*,y^*)\in \Gamma$ as the \emph{distinguished pair} as follows: if Algorithm $\algpairs$ was successful, then $(x^*,y^*)\in \Gamma$ is an arbitrary pair with $x^*\in L$ and $y^*\in R$. Otherwise, $(x^*,y^*)$ is an arbitrary pair in $\Gamma$.

\paragraph{Two Special Cases.}

Throughout, we denote by $\tau'=\frac{\tau}{64\lambda^2}$. 
Recall that we denoted by $d$ the average vertex degree in $G$, and that a vertex $v\in V(G)$ is \emph{high-degree} iff $\deg_G(v)\geq d\cdot n^{\eps}$. We denote by $\Vhd$ and $\Vld$ the sets of all high-degree and all low-degree vertices of $G$ respectively. We also denote by $V^{\geq \tau'}=\set{v\in V(G)\mid w(v)\geq \tau'}$.
Consider now the distinguished min-cut $(L,S,R)$. Throughout, we denote by $\hat S\subseteq S$ the set of vertices containing all low-degree vertices of $S$ whose weight is at least $\tau'$. In other words:

\[\hat S=S\cap \Vld\cap V^{\geq \tau'}.\]

Next we will consider two cases: Case 1 happens if $|\hat S|\leq 2^{12}\lambda^3$ and Case 2 happens otherwise.
We summarize our algorithms for each of the two cases in the following two lemmas. The two lemmas are essentially identical, except that the first one requires that  $|\hat S|\leq 2^{12}\lambda^3$, while the second requires that  $|\hat S|>2^{12}\lambda^3$, and they differ in the running times of their respective algorithms.

\begin{lemma}[Algorithm for Case 1]\label{lem: case 1 main}
		There is a randomized algorithm, that is given as input a directed $n$-vertex and $m$-edge graph $G=(V,E)$ with integral weights $w(v)\geq 1$ on its vertices $v\in V$, a set $\Gamma\subseteq V$ of at most $\ceil{\frac{100n\log n}{\lambda}}$ pairs of vertices, such that, for every pair $(x,y)\in \Gamma$, $x\neq y$ and $(x,y)\not\in E(G)$ hold, a parameter 	$0<\eps<1$, and two additional parameters $1\leq \lambda\leq 2n^{\eps}$ and $1\leq \tau\leq \wmax$, both integral powers of $2$. The algorithm returns a vertex-cut  $(L',S',R')$ in $G$. The algorithm guarantees that, if there exists a global minimum vertex-cut $(L,S,R)$ in $G$ for which all of the following hold:
		
	\begin{itemize}
		\item  Properties \ref{prop: small L}-\ref{prop: small volume of L} hold for $(L,S,R)$;
		\item $\frac{\lambda}{2}\leq |L|\leq\lambda$; 
		\item $\tau=\tau^*$, where $\tau^*$ is the critical threshold for $(L,S,R)$;  
		
		\item $|\hat S|\leq 2^{12}\lambda^3$, where $\hat S=S\cap \Vld\cap V^{\geq \tau'}$ and $\tau'=\frac{\tau}{64\lambda^2}$; and  
		\item there is a pair $(x^*,y^*)\in \Gamma$ with $x^*\in L$ and $y^*\in R$,
	\end{itemize}

 then, with probability at least $\frac{1}{4}$, the vertex cut $(L',S',R')$ is a global minimum vertex-cut. The running time of the algorithm is  $O\left ((m\cdot n^{1-\eps+o(1)}+md\cdot n^{11\eps+o(1)})\cdot \log (\wmax)\right )$, where $d=\frac{2m}{n}$ is the average vertex degree in $G$.
\end{lemma}

\begin{lemma}[Algorithm for Case 2]\label{lem: case 2 main}
	There is a randomized algorithm, that is given as input a directed $n$-vertex and $m$-edge graph $G=(V,E)$ with integral weights $w(v)\geq 1$ on its vertices $v\in V$, a set $\Gamma\subseteq V$ of at most $\ceil{\frac{100n\log n}{\lambda}}$ pairs of vertices, such that, for every pair $(x,y)\in \Gamma$, $x\neq y$ and $(x,y)\not\in E(G)$ hold, a parameter 	$0<\eps<1$, and two additional parameters $1\leq \lambda\leq 2n^{\eps}$ and $1\leq \tau\leq \wmax$, both integral powers of $2$. The algorithm returns a vertex-cut  $(L',S',R')$  in $G$.  The algorithm guarantees that, if there exists a global minimum vertex-cut $(L,S,R)$ in $G$ for which all of the following hold:
	
	\begin{itemize}
		\item  Properties \ref{prop: small L}-\ref{prop: small volume of L} hold for $(L,S,R)$;
		\item $\frac{\lambda}{2}\leq |L|\leq\lambda$; 
		\item $\tau=\tau^*$, where $\tau^*$ is the critical threshold for $(L,S,R)$;  
		
		\item $|\hat S|> 2^{12}\lambda^3$, where $\hat S=S\cap \Vld\cap V^{\geq \tau'}$ and $\tau'=\frac{\tau}{64\lambda^2}$; and  
		\item there is a pair $(x^*,y^*)\in \Gamma$ with $x^*\in L$ and $y^*\in R$,
	\end{itemize}
	
	then, with probability at least $\frac{1}{4}$, the vertex cut $(L',S',R')$ is a global minimum vertex-cut. The running time of the algorithm is  $O\left ((m\cdot n^{1-\eps+o(1)}+md\cdot n^{3\eps+o(1)})\cdot\log (\wmax)\right )$, where $d=\frac{2m}{n}$ is the average vertex degree in $G$.
\end{lemma}

We provide the proofs of Lemmas \ref{lem: case 1 main} and \ref{lem: case 2 main} in Sections \ref{sec: compute promising} and \ref{sec: case 2} respectively, after we complete the proof of \Cref{thm : sparse weighted algorithm} using them. 

We set $\eps=\frac{1}{12}-\frac{\log d}{12\log n}$, where $d$ is the average vertex degree in the input graph $G$.
We apply the algorithm \algspec from \Cref{thm: alg for special} to compute, in time $O(mn^{1-\eps+o(1)}\log W)$, a vertex-cut $(L_1,S_1,R_1)$ in $G$. Then we use Algorithm $\algpairs$ from \Cref{claim: compute pairs2} to compute, in time $\tilde O\left (m\cdot n^{\eps}\right )$, values $1\leq \tau\leq \wmax'$ and $1\le \lambda\leq 2n^{\eps}$ that are integral powers of $2$, together with a set $\Gamma$ of at most $\ceil{\frac{100n\log n}{\lambda}}$ pairs of vertices of $G$. We let $\event_1$ be the good event that the algorithm from \Cref{claim: compute pairs2} is successful with respect to the distinguished min-cut. Lastly, we apply  the algorithms from  Lemmas \ref{lem: case 1 main} and \ref{lem: case 2 main} to the input graph $G$, the set $\Gamma$ of vertex pairs, and the values $\lambda$ and $\tau$ computed by Algorithm $\algpairs$; we denote the cuts returned by these algorithms by $(L_2,S_2,R_2)$  and $(L_3,S_3,R_3)$, respectively. We output the smallest-value cut from among 
$(L_1,S_1,R_1)$,  $(L_2,S_2,R_2)$,  and $(L_3,S_3,R_3)$. Recall that, if either $|L|\geq n^{\eps}$ or $\vol_G(L)\geq n^{\eps}\cdot d$ hold, then the cut $(L_1,S_1,R_1)$ that Algorithm $\algspec$ returns is guaranteed to be optimal with probability at least $(1-1/n^2)$, and in this case, the cut that our algorithm returns is optimal as well. Otherwise, Properties \ref{prop: small L}-\ref{prop: small volume of L} hold for the distinguished min-cut. If Event $\event_1$ happens and  $|\hat S|\leq 2^{12}\lambda^3$, then, from \Cref{lem: case 1 main}, with probability at least $\frac{1}{4}$, cut $(L_2,S_2,R_2)$ is optimal, and in this case, the cut the our algorithm returns is optimal as well. If Event $\event_1$ happens and  $|\hat S|> 2^{12}\lambda^3$, then, from \Cref{lem: case 2 main}, with probability at least $\frac{1}{4}$, cut $(L_3,S_3,R_3)$ is optimal, and in this case, the cut the our algorithm returns is optimal as well.
Since, from \Cref{claim: compute pairs2}, $\prob{\event_1}\geq \frac{1}{32\log n\cdot\log (\wmax)}$, we get that the probability that the cut returned by our algorithm is optimal is at least $\Omega\left (\frac{1}{\log n\log (\wmax)}\right )$. By repeating the algorithm $O(\log n\log (\wmax))$ times and returning the smallest-value cut among the resulting vertex-cuts, we ensure that the cut that the algorithm returns is optimal with probability at least $\half$.

Note that the running time of the algorithm is bounded by:  

\[O\left ((m\cdot n^{1-\eps+o(1)}+md\cdot n^{11\eps+o(1)})\cdot \log^2 (\wmax)\right )\leq O\left(mn^{11/12+o(1)}\cdot d^{1/12}\cdot \log^2 W\right ),\] 

since  $\eps=\frac{1}{12}-\frac{\log d}{12\log n}$ and $\wmax\leq O(W\cdot n)$.
In order to complete the proof of \Cref{thm : sparse weighted algorithm}, it now remains to prove Lemmas \ref{lem: case 1 main} and \ref{lem: case 2 main}, which we do in the following two subsections.

\subsection{Algorithm for Case 1: Proof of \Cref{lem: case 1 main}}
\label{sec: compute promising}

Recall that we have denoted by $\tau'=\frac{\tau}{64\lambda^2}$  and by $\hat S=S\cap \Vld\cap V^{\geq \tau'}$, and that Case 1 happens if $|\hat S|\leq 2^{12}\lambda^3$.  In this section we provide an algorithm for the slightly more general case where $|\hat S|\leq \rho$  for some given parameter $\rho\geq \lambda$.
We note that the running time of the algorithm will depend on $\rho$, and that eventually we will set $\rho=2^{12}\lambda^3$.
 For brevity, we say that Condition (C1) holds for a vertex-cut $(L,S,R)$ in $G$, if all of the following hold:

\begin{itemize}
	\item  Properties \ref{prop: small L}-\ref{prop: small volume of L} hold for $(L,S,R)$;
	\item $\frac{\lambda}{2}\leq |L|\leq\lambda$; 
	\item $\tau=\tau^*$, where $\tau^*$ is the critical threshold for $(L,S,R)$;  
		\item we are given a value $\rho\geq \lambda$, such that $|\hat S|\leq \rho$; and  
	\item there is a pair $(x^*,y^*)\in \Gamma$ with $x^*\in L$ and $y^*\in R$,
\end{itemize}

For the sake of analysis, we designate a global minimum vertex-cut $(L,S,R)$ to be a \emph{distinguished min-cut} as follows. If there exists a global minimum vertex-cut for which Condition (C1) holds, then we let $(L,S,R)$ be this cut, and we say that the distinguished min-cut is \emph{good}; otherwise, we let $(L,S,R)$ be an arbitrary global minimum vertex-cut, and we say that it is \emph{bad}. In the following, whenever we say that Condition (C1) holds, we mean that it holds with respect to the distinguished min-cut $(L,S,R)$.
For the sake of the analysis, we also designate a pair $(x^*,y^*)\in \Gamma$ to be the \emph{distinguished pair} as follows: if there is a pair $(x,y)\in \Gamma$ with $x\in L$ and $y\in R$ then we let $(x^*,y^*)$ be any such pair; otherwise, we let $(x^*,y^*)$ be any pair in $\Gamma$.
We start by describing the preprocessing step of the algorithm.

\subsubsection{The Preprocessing Step}\label{subsec: Case 1 preprocessing}

Throughout, we use the parameter $\hat \tau=\frac{\tau}{8\lambda}$.
We start by defining some of the central notions that our algorithm uses: namely sets $\beta(v)$ and $\sigma(v)$ of vertices for $v\in V(G)$, as well as the notion of suspicious vertices.

\subsubsection*{$\beta$-Sets, $\sigma$-Sets, and Suspicious Vertices}

\paragraph{$\beta$-sets.}
For every vertex $v\in V(G)$, we construct a set $\beta(v)\subseteq V(G)$ of vertices of cardinality at most $2\lambda$, as follows. 
We perform a DFS search  in the graph $G$ starting at $v$, except that the DFS search is only allowed to explore vertices $u\in V(G)\setminus\set{v}$ with $u\in \Vld$ and $w(u)\geq \hat \tau$. Once the DFS search discovers $2\lambda$ vertices, we terminate it, and we let $\beta(v)$ be the vertex set containing all these discovered vertices, including $v$ itself. Alternatively, if the DFS search terminates before $2\lambda$ vertices are discovered (in which case there are fewer than $2\lambda$ vertices in $G$ that can be reached from $v$ via paths that only use low-degree vertices of weight at least $\hat \tau$ in addition to $v$), then we let $\beta(v)$ be the set of all vertices that the DFS search has discovered, including $v$. 
In order to compute the vertex sets $\beta(v)$ for all vertices $v\in V(G)$ efficiently, we first compute, 
for every vertex $u\in V(G)$, the set $\OUT(u)\subseteq V(G)$ of vertices, that contains all low-degree vertices $z\in N^+_G(u)$ with $w(z)\geq \hat \tau$; if $|\OUT(u)|>2\lambda$, then we discard vertices from $\OUT(u)$ arbitrarily, until $|\OUT(u)|=2\lambda+1$ holds. Clearly, all vertex sets $\OUT(u)$ for $u\in V(G)$ can be computed in time $O(m)$. Then when we process a vertex $v\in V(G)$ in order to compute the set $\beta(v)$, we perform the DFS search using the sets  $\OUT(u)$ of vertices that we have computed for all $u\in V(G)$, so that the time required to process $v$ is $O(\lambda^2)$, as we explore at most $2\lambda$ vertices $u$, and $|\OUT(u)|\leq 2\lambda+1$ holds for each such vertex. Overall, the total time required to compute the sets $\beta(v)$ of vertices for all $v\in V(G)$ is $O(m+n\lambda^2)$.

\paragraph{$\sigma$-sets and suspicious vertices.}
For every low-degree vertex $u\in V(G)$, we define the set $\sigma(u)$ of vertices of $G$, to contain every vertex $v\in V(G)$ with $u\in \beta(v)$, so:
$\sigma(u)=\set{v\in V(G)\mid u\in \beta(v)}$.
We say that a low-degree vertex $u\in \Vld$ is \emph{suspicious}, if $|\sigma(u)|>2n^{\eps}\cdot \lambda$. We denote by $U$ the set of all low-degree suspicious vertices.

Note that $\sum_{v\in V(G)}|\beta(v)|\leq 2n\cdot \lambda$. Therefore, at most $n^{1-\eps}$ vertices of $V(G)$ may be suspicious. Moreover, by inspecting the sets $\beta(v)$ of vertices for all $v\in V(G)$, we can compute the sets $\sigma(u)$ of vertices for all $u\in \Vld$, as well as the set $U\subseteq \Vld$ of suspicious vertices, in time $O(n\cdot \lambda)$.

\paragraph{Additional assumption.}

Assume that Condition (C1) holds and $L\cap U\neq \emptyset$. In this case, we show that we can efficiently compute a cut $(L',S',R')$, that, with probability at least $\half$, is a global minimum vertex-cut. In order to do so, for every vertex $v\in U$, we apply the algorithm from \Cref{lem: cut if vertex of L}, that, in time $O\left (m^{1+o(1)}\cdot \log (\wmax)\right )$, returns a vertex $y_v\in V(G)$, and the value $c_v$ of the minimum $v$-$y_v$ vertex-cut in $G$. Recall that the lemma guarantees that, if there is a global minimum vertex-cut $(L,S,R)$ in $G$ with $v\in L$, then with probability at least $1-1/n^4$, $c_v=w(S)$ holds. We then select a vertex $v'\in U$ for which the value $c_{v'}$ is the smallest, and use the algorithm from \Cref{cor: min_vertex_cut} to compute the minimum $v'$-$y_{v'}$ vertex-cut $(L',S',R')$ in $G$, whose value must be equal to $c_{v'}$, in time $O\left (m^{1+o(1)}\cdot \log (\wmax)\right )$, returning this cut as the algorithm's outcome. Since we have assumed that $L\cap U\neq \emptyset$, when a vertex $v\in L\cap U$ is considered,
the algorithm from \Cref{lem: cut if vertex of L} is guaranteed to output the value $c_v=\opt$ with probability at least $1-1/n^4$. If this happens, then the vertex-cut $(L',S',R')$ that our algorithm outputs is guaranteed to be optimal. Recall that we invoke the algorithm from \Cref{lem: cut if vertex of L}  $|U|\leq n^{1-\eps}$ times, and the algorithm from \Cref{cor: min_vertex_cut}  once. Therefore, the total running time of this part of the algorithm is $O\left (m\cdot n^{1-\eps+o(1)} \log (\wmax)\right )$. Overall, if Condition (C1) holds and $U\cap L\neq\emptyset$, then the running time of our algorithm is $O\left ((m\cdot n^{1-\eps+o(1)}+n\lambda^2) \log (\wmax)\right )\leq O\left ((m\cdot n^{1-\eps+o(1)}+n^{1+2\eps}) \log (\wmax)\right )$, and it is guaranteed to return the global minimum vertex-cut with probability at least $1/2$. 

Intuitively, it is now sufficient for us to consider the case where $L\cap U=\emptyset$. Formally, our algorithm chooses a random bit $b\in \set{0,1}$. 
If $b=0$, then it executes the algorithm described above, that is designed for the case where $L\cap U\neq \emptyset$. If $b=1$, then it executes the algorithm that we describe below, that is designed for the case where $L\cap U=\emptyset$. We say that a good event $\event'$ happens if either (i) $b=0$ and  $L\cap U\neq \emptyset$; or (ii) $b=1$ and $L\cap U= \emptyset$. Clearly, Event $\event'$ happens with probability $1/2$. It is now enough to design an algorithm for the case where $b=1$ and $L\cap U=\emptyset$. For shorthand, we say that Condition (C2) happens if Condition (C1) happens, and, additionally, $b=1$ and $L\cap U=\emptyset$.
We now focus on designing an algorithm that, if Condition (C2) holds, outputs a global minimum vertex-cut with a sufficiently high probability. Our algorithm will consider every pair $(x,y)\in \Gamma$ of vertices in turn. For each such pair, it will compute a relatively small subset $A_{x,y}\subseteq V(G)$ of vertices, such that, if Condition (C2) holds and $(x,y)$ is the distinguished pair, then $L\subseteq A_{x,y}$ with a sufficiently high probability. This, in turn, will allow us to sparsify the graph $G$, in order to efficiently compute the value of the minimum $x$-$y$ vertex-cut.
For every pair $(x,y)\in \Gamma$ of vertices, we will consider two graphs, $G'_{x,y}$ and $G''_{x,y}$, that we define next. We also need to perform additional preprocessing that will allow us to compute each such graph in turn efficiently, prior to processing the corresponding pair $(x,y)$ of vertices.

\subsubsection*{Graphs $G'_{x,y}$ and $G''_{x,y}$ for Pairs $(x,y)\in \Gamma$, and Valid Shortcut Operations}

\paragraph{Graph $G'_{x,y}$.}
Consider a pair $(x,y)\in \Gamma$ of vertices. Let $G'_{x,y}$ be the graph that is obtained from the split graph $G'$ for the graph $G$ (see the definition in \Cref{subsec: split graph}), by adding a destination vertex $t$ and connecting $y^{\inn}$ to $t$ with a regular edge of capacity $\wmax$; we also denote $x^{\out}$ by $s$, and view it as the source vertex of $G'_{x,y}$.
Note that the value of the  minimum $s$-$t$ edge-cut in $G'_{x,y}$ is equal to the value of the minimum $x$-$y$ vertex-cut in $G$. Moreover, if Condition (C2) holds and $(x,y)$ is the distinguished pair, then the set $E_S=\set{(v^{\inn},v^{\out})\mid v\in S}$ of edges of $G'_{x,y}$ is the minimum $s$-$t$ edge cut in $G'_{x,y}$, and $c(E_S)=\opt$, the value of the global minimum vertex-cut in $G$. 

Over the course of the algorithm, we may modify the graph $G'_{x,y}$ by performing \emph{shortcut operations}, that we define next; this definition is identical to that in \cite{CT24}.

\paragraph{Shortcut Operations.}
Consider a vertex pair $(x,y)\in \Gamma$ and the corresponding graph $G'_{x,y}$.
In a shortcut operation, we add an edge $(v,t)$ of capacity $\wmax$ to graph $G'_{x,y}$, for some vertex $v\in V(G'_{x,y})\setminus \set{s,t}$. All such newly added edges are considered regular edges. If Condition (C2) holds and $(x,y)\in \Gamma$ is the distinguished pair, then the shortcut operation in which an edge $(v,t)$  is inserted into $G'_{x,y}$ is a \emph{valid shortcut operation}, and edge $(v,t)$ is a \emph{valid shortcut edge}, if and only if either (i) $v$ is a copy of a vertex $u\in R$; or (ii) $v$ is an out-copy of a vertex $u'\in S$.  If Condition (C2) does not hold, or if $(x,y)$ is not the distinguished pair, then any such shortcut operation is a valid shortcut operation, and any edge $(v,t)$ is a valid shortcut edge.  We will use the following claim that summarizes the properties of the graph obtained from $G'_{x,y}$ via a sequence of shortcut operations. The claim and its proof for undirected graphs appeared in \cite{CT24}, but since we use it for directed graphs here, we include both the claim and its proof (which is identical to that in \cite{CT24}) here.

\begin{claim}\label{claim: shortcut}
	Let $G''_{x,y}$ be a graph that is obtained from $G'_{x,y}$ via a sequence of shortcut operations. Then the value of the minimum $s$-$t$ edge-cut in $G''_{x,y}$ is at least as high as the value of the minimum $s$-$t$ edge-cut in $G'_{x,y}$. Moreover, if Condition (C2) holds, $(x,y)$ is the distinguished pair, and all shortcut operations in the sequence are valid, then the value of the minimum $s$-$t$ edge-cut in $G''_{x,y}$ is exactly $\opt$.
\end{claim}
\begin{proof}
	From the Max-Flow/Min-Cut theorem, the value of the minimum $s$-$t$ edge-cut in $G'_{x,y}$ is equal to the value of the maximum $s$-$t$ flow in $G'_{x,y}$. The insertion of edges into $G'_{x,y}$ may not decrease the value of the maximum $s$-$t$ flow. Therefore, the value of the minimum $s$-$t$ edge-cut in $G''_{x,y}$ is at least as high as that in $G'_{x,y}$.

	Assume now that Condition (C2) holds and that $(x,y)$ is the distingiushed pair in $\Gamma$.  Let $E_S=\set{(v^{\inn},v^{\out})\mid  v\in S}$ be the set of special edges in graph $G'_{x,y}$ that represent the vertices of $S$.
	Let $X=L^*\cup S^{\inn}$ and let $Y=V(G'_x)\setminus X$. Note that $(X,Y)$ is an $s$-$t$ cut in $G'_{x,y}$. Clearly, $E_{G'_{x,y}}(X,Y)=E_S$, and so the edges of $E_S$ separate the vertices of $X$ from the vertices of $Y$ in $G'_{x,y}$. Moreover, since $(L,S,R)$ is a  minimum vertex-cut in $G$,  $\sum_{e\in E_S}c(e)=\sum_{v\in S}w(v)=\opt$. 
	Since $s\in X$ and $t\in Y$, and since the valid shortcut edges that we add all connect vertices of $Y$ to $t$, set $E_S$ of edges also separates $X$ from $Y$ in $G''_{x,y}$. Therefore, the value of the minimum $s$-$t$ edge-cut in $G''_{x,y}$ is $\sum_{e\in E_S}c(e)=\opt$.
\end{proof}

\paragraph{Initial Shortcuts and Graphs $G''_{x,y}$.}

For every pair $(x,y)\in \Gamma$, we define a new graph $G''_{x,y}$, that is obtained from graph $G'_{x,y}$ by adding,
for every high-degree vertex $v\in V(G)$, a shortcut edge $(v^{\out},t)$ of capacity $\wmax$ to it; these newly added edges are considered to be regular edges. Note that all these newly inserted edges are valid shortcut edges, from Property \ref{prop: small volume of L} and the definition of high-degree vertices (see \Cref{def: high-deg}).
Therefore, the value of the maximum $s$-$t$ flow in $G''_{x,y}$ is at least $\opt$, and, if Condition (C2) holds and $(x,y)$ is the distinguished pair, then this value is equal to $\opt$. 
We denote by $\opt_{x,y}$ the value of the maximum $s$-$t$ flow in $G''_{x,y}$. From the above discussion, $\opt_{x,y}\geq \opt$, and, if 
 Condition (C2) holds and $(x,y)$ is the distinguished pair, then $\opt_{x,y}=\opt$.

%

\paragraph{Efficient Construction of the graphs $G''_{x,y}$.}
Our algorithm will process every pair $(x,y)\in \Gamma$ one by one. When a pair $(x,y)$ is processed, we will consider the corresponding graph $G''_{x,y}$ and perform various computations in it, that require access to the modified adjacency-list representation of $G''_{x,y}$. However, since $|\Gamma|$ may be as large as $\Omega(n)$, we cannot afford to spend $\Omega(m)$ time to construct the modified adjacency-list representation of each such graph $G''_{x,y}$ from scratch. Instead, we will construct a ``universal'' graph $G^*$, from which one can compute the modified adjacency-list representation of each of the graphs $G''_{x,y}$ for $(x,y)\in \Gamma$ in turn efficiently.

Graph $G^*$ is defined to be the split-graph with shortcuts $G''(B)$, where the set $B$ contains the out-copy of every high-degree vertex of $G$. Notice that, given $G$, the modified adjacency-list representation of $G^*$ can be constructed in time $O(m+n)$.

Consider now any vertex pair $(x,y)\in \Gamma$. Notice that graph $G''_{x,y}$ can be equivalently defined as the graph that is obtained from $G^*$, by adding the edge $(y^{\inn},t)$ to it. The  modified adjacency-list representation of $G^*$ can be adapted in time $O(1)$ to obtain the modified adjacency-list representation of $G''_{x,y}$ from it.
To summarize, in the preprocessing step, we construct the modified adjacency-list representation of the graph $G^*$, in time $O(m+n)$. After that, we process every pair $(x,y)\in \Gamma$ one by one. When pair $(x,y)\in \Gamma$ is processed, we compute the modified adjacency-list representation of $G''_{x,y}$, from that of $G^*$, in time $O(1)$. Once we finish processing the pair $(x,y)$, we undo all the changes made to the adjacency-list representation of $G^*$ and proceed to process the next pair. The algorithm for processing every pair in $\Gamma$ is summarized in the following lemma.

\begin{lemma}\label{lem: Case 1}
	There is a randomized algorithm, that is given as input a directed $n$-vertex and $m$-edge graph $G$ with integral weights $w(v)\geq 1$ on its vertices $v\in V(G)$, whose average vertex degree is denoted by $d$, and  parameters $0<\eps<1$, $1\leq \tau\leq \wmax$, $1\leq \lambda\leq 2n^{\eps}$, and $\rho\geq \lambda$. Additionally, the algorithm is given a pair $(x,y)$ of distinct vertices of $G$ with $(x,y)\not\in E(G)$, and two copies of the corresponding graph $G''_{x,y}$ in the modified adjacency-list representation. Finally, the algorithm is given the set $U\subseteq \Vld$ of suspicious vertices, and, for every vertex $u\in\Vld$, the corresponding vertex set $\sigma(u)$. The algorithm returns two subsets $L',S'\subseteq V(G)$ of vertices of $G$, such that, if we denote by $R'=V(G)\setminus (L'\cup S')$, then $(L',S',R')$ is a vertex-cut in $G$.
	It ensures that, if all of the following conditions hold:
	\begin{itemize}
		\item there is a global minimum vertex-cut $(L,S,R)$ for which Properties \ref{prop: small L}-\ref{prop: small volume of L} hold;
		
		\item $\frac{\lambda}{2}\leq |L|\leq \lambda$;
		
		\item $\tau=\tau^*$, where $\tau^*$ is the critical threshold for $(L,S,R)$;
		\item $|\hat S|\leq \rho$, where $\hat S=S\cap \Vld\cap V^{\geq \tau'}$ and $\tau'=\frac{\tau}{64\lambda^2}$; 
		\item no vertex of $L$ is suspicious; and
		\item $x\in L$, and $y\in R$,
	\end{itemize}

then, with probability at least $\half$, $(L',S',R')$ is a global minimum vertex-cut in $G$. The running time of the algorithm is $ O(d^2\cdot n^{3\eps+o(1)}\cdot \lambda^3\cdot\rho^2\cdot \log(\wmax))$.
\end{lemma}

We prove \Cref{lem: Case 1} below, after we complete the proof of \Cref{lem: case 1 main} using it. Recall that the running time of the preprocessing step is $O(m+n\lambda^2)$. If the bit $b=0$, then the additional time that the algorithm takes is $O\left ((m\cdot n^{1-\eps+o(1)}+n^{1+2\eps}) \log (\wmax)\right )$. Otherwise, for every pair $(x,y)\in \Gamma$, we apply the algorithm from \Cref{lem: Case 1}, that returns the sets $L_{x,y},S_{x,y}\subseteq V(G)$ of vertices. We select a pair $(x,y)\in \Gamma$ for which $w(S_{x,y})$ is the smallest, and return the cut $(L_{x,y},S_{x,y},R_{x,y})$, where $R_{x,y}=V(G)\setminus (L_{x,y}\cup S_{x,y})$ as the algorithm's outcome. Since $|\Gamma|\leq \tilde O\left(\frac{n}{\lambda}\right)$, and since the running time of the algorithm from  \Cref{lem: Case 1} is $O(d^2\cdot n^{3\eps+o(1)}\cdot \lambda^3\cdot \rho^2\cdot\log(\wmax))$, we get that the total running time of the entire algorithm is:

\[
\begin{split}
&O\left ((m\cdot n^{1-\eps+o(1)}+n^{1+2\eps}+nd^2\cdot n^{3\eps+o(1)}\cdot \lambda^2\cdot\rho^2)\cdot \log (\wmax)\right )  \\
&\quad\quad\quad\quad\quad\quad\leq O\left ((m\cdot n^{1-\eps+o(1)}+md\cdot n^{5\eps+o(1)}\cdot \rho^2)\cdot\log (\wmax)\right )
\end{split}\]

Note that from the conditions of \Cref{lem: case 1 main}, we can set $\rho=2^{12}\lambda^3\leq O(n^{3\eps})$. The running time of the algorithm then becomes $O\left ((m\cdot n^{1-\eps+o(1)}+md\cdot n^{11\eps+o(1)})\cdot\log (\wmax)\right )$, as required.

Assume that Condition (C1) holds, and that Event $\event'$ happened. If $b=0$, then our algorithm is guaranteed to return a global minimum vertex-cut with probability at least $\half$. Otherwise, when the algorithm from  \Cref{lem: Case 1} is applied to the distinguished pair $(x^*,y^*)$, the corresponding cut $(L_{x^*,y^*},S_{x^*,y^*},R_{x^*,y^*})$ is guaranteed to be optimal with probability at least $\half$. Since $\prob{\event'}=\half$, we get that, overall, if Condition (C1) holds, then the cut that the algorithm returns is a global minimum vertex-cut with probability at least $\frac{1}{4}$.

In order to complete the proof of \Cref{lem: case 1 main}, it is now enough to prove \Cref{lem: Case 1}, which we do in the remainder of this subsection.

\subsubsection{Proof of \Cref{lem: Case 1}}
We assume that we are given a directed $n$-vertex and $m$-edge graph $G$ with integral weights $w(v)\geq 1$ on its vertices $v\in V(G)$, and  parameters $0<\eps<1$, $1\leq \tau\leq \wmax$, $1\leq \lambda\leq 2n^{\eps}$, and $\rho\geq \lambda$, together with a pair $(x,y)$ of distinct vertices of $G$ with $(x,y)\not\in E(G)$, and two copies of the corresponding graph $G''_{x,y}$ in the modified adjacency-list representation. We also assume that we are given the set $U\subseteq \Vld$ of suspicious vertices, and, for every vertex $u\in\Vld$, the corresponding vertex set $\sigma(u)$.

Given a vertex-cut $(L,S,R)$ in $G$, we say that Condition (C3) holds with respect to $(L,S,R)$, if all of the following hold:

\begin{itemize}
	\item  Properties \ref{prop: small L}-\ref{prop: small volume of L} hold for $(L,S,R)$;
	
	\item $\frac{\lambda}{2}\leq |L|\leq \lambda$;
	
	\item $\tau=\tau^*$, where $\tau^*$ is the critical threshold for $(L,S,R)$;
	\item $|\hat S|\leq \rho$, where $\hat S=S\cap \Vld\cap V^{\geq \tau'}$ and $\tau'=\frac{\tau}{64\lambda^2}$; 
	\item no vertex of $L$ is suspicious; and
	\item $x\in L$, and $y\in R$.
\end{itemize}

As before, for the sake of the analysis, we fix a distinguished global minimum vertex-cut $(L,S,R)$ in $G$, as follows.
If there is a global minimum vertex-cut in $G$ for which Condition (C3) holds, then we let any such cut be the distinguished min-cut $(L,S,R)$ and we say that it is good; otherwise, we let the distinguished min-cut be any global minimum vertex-cut in $G$, and we say that it is bad. In the remainder of this section, whenever we say that Condition (C3) holds, we mean that it holds with respect to the distinguished min-cut.

Note that, if $x$ is a high-degree vertex, then Condition (C3) may not hold. In this case, we let $L'\subseteq V(G)$ be a vertex set that contains a single arbitrary low-degree vertex $v$ of $G$ and $S'=N^+_G(v)$, and return $L',S'$, terminating the algorithm. Therefore, we assume from now on that $x$ is a low-degree vertex.

Since the pair $(x,y)\in \Gamma$ is fixed from now on, we denote the graph $G''_{x,y}$ by $G''$ in the remainder of the proof.
A central notion that our algorithm uses is that of a \emph{promising vertex set}, that we define next. Recall that we have defined a parameter $\hat \tau=\frac{\tau}{8\lambda}$, that was used in the definition of the vertex sets $\beta(v)$ for $v\in V(G)$.

\begin{definition}\label{def: promising vertex set}
	Let $Z\subseteq V(G)$ be a set of vertices. We say that $Z$ is a \emph{promising vertex set} if every vertex in $Z$ is a low-degree vertex, and moreover, every low-degree vertex $v\in S$ with $w(v)\geq 
	\hat \tau$ lies in $Z$.
\end{definition}

Our algorithm will compute a relatively small set $Z$ of vertices of $G$, such that, if Condition (C3) holds, then, with probability at least $\half$, $Z$ is a promising set of vertices. Before we provide the algorithm for computing such a set $Z$, we first provide an efficient algorithm for computing the vertex sets $L',S'\subseteq V(G)$ with the required properties, given $Z$.

\subsubsection*{Computing the Vertex Sets $L',S'$ given $Z$}

Our first step is to compute a relatively small set $A\subseteq V(G)$ of vertices, that has the following property: if Condition (C3) holds and $Z$ is a promising vertex set, then $L\subseteq A$.

\paragraph{Computing the vertex set $A$.}
Recall that our algorithm is given as input the set $U\subseteq \Vld$ of suspicious vertices, and, for every vertex $u\in \Vld$, the corresponding vertex set $\sigma(u)$. If $Z$ contains a high-degree vertex then it may not be a promising set by definition. In this case,  we terminate the algorithm and return a vertex set $L'\subseteq V(G)$ that contains a single arbirary low-degree vertex $v$ and $S'=N^+_G(v)$.  Therefore, we assume from now on that every vertex in $Z$ is a low-degree vertex. Recall that we also assumed that $x$ is  a low-degree vertex.
We start with $A=\set{x}$, and then process every vertex $v\in Z$ one by one. When a vertex $v\in Z$ is processed, we consider every vertex $u\in N^-_G(v)\cap \Vld$ one by one. For each such vertex $u$, we check whether $u$ is suspicious (that is, whether $u\in U$), and, if this is not the case, we add all low-degree vertices in $\sigma(u)$ to $A$. This completes the algorithm for computing the set $A$ of vertices.

Recall that every vertex $v\in Z$ is a low-degree vertex, so $|N^-_G(v)|\leq d\cdot n^{\eps}$. By definition, for every non-suspicious vertex $u\in\Vld\setminus U$, $|\sigma(u)|\leq 2n^{\eps}\cdot \lambda$. Therefore, $|A|\leq 2|Z|\cdot d\cdot n^{2\eps}\cdot \lambda$, and the time required to compute the set $A$ is bounded by $O\left (|Z|\cdot d\cdot n^{2\eps}\cdot \lambda\right )$.
The following claim summarizes the central property of the resulting set $A$.

\begin{claim}\label{claim: property of A}
Assume that Condition (C3) holds and that $Z$ is a promising vertex set. Assume further that $S$ contains at least one low-degree vertex $v$ of weight at least $\frac{\tau}{2}$. Then $L\subseteq A$.
\end{claim}
\begin{proof}
	Assume that Condition (C3) holds, and that $Z$ is a promising vertex set, so, in particular,  every vertex $v\in Z$ is a low-degree vertex. Let $Z'= \left(N^-_G(Z)\cap\Vld\right )\setminus U$ contain all low-degree in-neighbors of the vertices of $Z$ that are non-suspicious. Consider now some vertex $z\in L$ and recall that, from Property \ref{prop: small volume of L} it must be a low-degree vertex. Note that, if $\beta(z)$ contains some vertex $u\in Z'$, then $z\in \sigma(u)$, and so $z\in A$ must hold. Therefore, it is enough to prove that $\beta(z)\cap Z'\ne \emptyset$.

	Consider now the low-degree vertex $v\in S$ with $w(v)\geq \frac{\tau}{2}$. Since $\hat \tau=\frac{\tau}{8\lambda}$, from the definition of the promising vertex set, $v\in Z$ must hold.
	We use the following observation to complete the proof:
	
	\begin{observation}\label{obs: the path}
		There is a path $P$ connecting $z$ to $v$ in $G$, such that every inner vertex on $P$ lies in $L$ and has weight at least $\hat \tau$.
	\end{observation}
We provide the proof of \Cref{obs: the path} below, after we complete the proof of \Cref{claim: property of A} using it. Consider the set $\beta(z)$ of vertices, and assume first that $|\beta(z)|=2\lambda$. Since $|L|\leq \lambda$, $\beta(z)$ must contain some vertex $u\in S$.
From the definition of the set $\beta(z)$ of vertices, there is a path $P'$ in $G$ connecting $z$ to $u$, such that every vertex in $V(P')\setminus\set{z}$ is a low-degree vertex of weight at least $\hat \tau$ that lies in $\beta(z)$. We let $u'$ be the first vertex on $P'$ that lies in $S$, and we let $u''$ be the vertex preceding $u'$ on $P'$. Then $u''\in L$ must hold, and in particular it is not suspicious, from Condition (C3).
Since $u'\in \beta(z)$, and since every vertex in $\beta(z)$ is a low-degree vertex of weight at least $\hat \tau$, from the definition of a promising vertex set, we get that $u'\in Z$ must hold.  But then $u''\in Z'$ and $z\in \sigma(u'')$, so $z\in A$ must hold.

Finally, assume that $|\beta(z)|<2\lambda$. Then, by the construction, $\beta(z)$ must contain all vertices that lie on the path $P$ given by \Cref{obs: the path} (since, from Condition (C3), all vertices of $L$ are low-degree). If we denote by $v'$ the penultimate vertex on $P$, then $v'$ lies in $L$, so it must be a low-degree and a non-suspicious vertex. Therefore, $v'\in Z'$ and $z\in \sigma(v')$ must hold, so $z\in A$. In order to complete the proof of \Cref{claim: property of A}, it is now enough to prove \Cref{obs: the path}.

\begin{proofof}{\Cref{obs: the path}}
	Let $z'$ be any vertex of $R$. From the properties of the global minimum vertex-cut, there is a flow $f$ in $G$ from $z$ to $z'$ of value $\opt=w(S)$, that obeys all vertex capacities. In particular, for every vertex $u\in S$, the total amount of flow that $f$ sends through $u$ is equal to $w(u)$. 

Consider now a flow-path decomposition $\pset$ of $f$. It is easy to verify that every path $P\in\pset$ contains exactly one vertex of $S$, that we denote by $u(P)$. Moreover, if we consider the subpath $P'$ of $P$ from $z$ to $ u(P)$, then all inner vertices on $P'$ lie in $L$. Let $\pset'\subseteq \pset$ be the collection of all paths $P$ with $u(P)=v$, and let $\pset''=\set{P'\mid P\in \pset'}$. Then $\sum_{P\in \pset'}f(P)=w(v)$, and  every path $P'\in \pset''$ originates at $z$, terminates at $v$, and only contains vertices of $L$ as inner vertices.

We partition the set $\pset''$ of paths into two subsets: set $\pset''_1$ contains all paths $P'$ such that some inner vertex of $P'$ has weight below $\hat \tau$; and $\pset''_2$ contains all other paths of $\pset''$. Since $|L|\leq \lambda$, we get that $\sum_{\stackrel{P\in \pset':}{P'\in \pset''_1}}f(P)\leq \hat \tau\cdot \lambda$. Since $\sum_{P\in \pset'}f(P)=w(v)\geq \frac{\tau}{2}\geq 2\lambda\cdot \hat \tau$ (as $\hat \tau=\frac{\tau}{8\lambda}$), we get that $\pset''_2\neq\emptyset$. 

Let $P'$ be any path in $\pset''_2$. Then from the definition, $P'$ starts at $z$ and terminates at $v$, and moreover, every inner vertex $a$ on $P'$ lies in $L$ and has weight $w(a)\geq \hat \tau$.
\end{proofof}
\end{proof}

Recall that our algorithm ensures that $x\in A$. If $y\in A$ holds, then we delete $y$ from $A$; if Condition (C3) holds and $Z$ is a promising vertex set, then $L\subseteq A$ still continues to hold. Additionally, if some vertex $v\in N_G^-(y)$ lies in $A$, then we delete $v$ from $A$.  As before, if Condition (C3) holds and $Z$ is a promising vertex set, then $L\subseteq A$ still continues to hold (since $y\in R$, so it cannot have in-neighbors that lie in $L$).
Moreover, since $x\neq y$ and $(x,y)\not \in E(G)$, $x\in A$ continues to hold. Therefore, we assume from now on that $x\in A$, $y\not\in A$ and $y\not\in N_G^+(A)$.


Next, we construct the graph $G^{|A}$ that is derived from $G$ via the vertex set $A$ (see \Cref{def: defined graph}).
Since every vertex of $A$ is a low-degree vertex, it follows from \Cref{claim: properties of defined graph} that $|V(G^{|A})|+|E(G^{|A})|\leq O\left(|A|\cdot d\cdot n^{\eps}\right)\leq O\left(|Z|\cdot d^2\cdot n^{3\eps}\cdot \lambda\right)$, and that $G^{|A}$ can be constructed in time 
$O\left(|A|\cdot d\cdot n^{\eps}\right)\leq O\left(|Z|\cdot d^2\cdot n^{3\eps}\cdot \lambda\right)$.
Our next step is to compute a minimum $x$-$t$ vertex-cut $(\hat L,\hat S,\hat R)$ in $G^{|A}$ using the algorithm from \Cref{cor: min_vertex_cut}, in time $O\left(|E(G^{|A})|^{1+o(1)}\cdot \log(\wmax)\right )\leq O\left(|Z|\cdot d^2\cdot n^{3\eps+o(1)}\cdot \lambda\cdot \log(\wmax)\right)$. (Note that $(x,t) \notin E(G^{|A})$, since $x \in A$.)
We then output the vertex sets $L'=\hat L$ and $S'=\hat S$ as the algorithm's outcome.
Observe that, by \Cref{claim: transforming cuts from defined graph}, 
if we denote by $R' = V(G) \setminus (L' \cup S')$, then
$(L',S',R')$ is a valid vertex-cut in $G$.

So far we provided an algorithm that, given a set $Z\subseteq V(G)$ of vertices, computes the sets $L',S'\subseteq V(G)$ of vertices, such that, if we denote by $R'=V(G)\setminus (L'\cup S')$, then $(L',S',R')$ is a vertex-cut in $G$, in time $O\left(|Z|\cdot d^2\cdot n^{3\eps+o(1)}\cdot \lambda\cdot \log(\wmax)\right)$. Note that, if Condition (C3) holds, $S$ contains at least one low-degree vertex of weight at least $\frac{\tau}{2}$, and $Z$ is a promising vertex-set, then, from \Cref{claim: property of A}, $L\subseteq A$, and from 
\Cref{cl: finding global mincut using defined graph}, $(L',S',R')$ is guaranteed to be a global minimum vertex-cut in $G$.

From the definition of the critical threshold, if Condition (C3) holds and $\tau>1$, then $S$ must contain a low-degree vertex of weight at least $\frac{\tau}{2}$. If $\tau=1$, it is possible that this is not the case, and then every vertex in $S$ must be a high-degree vertex. Therefore, if $\tau=1$, then we use a slightly different strategy. 

We start by performing a DFS search in $G$ starting from vertex $x$, but only using low-degree vertices. If the DFS search terminates before $\lambda+1$ such vertices are discovered, then we let $A$ be the set of all vertices that the DFS search has discovered.
Assume that Condition (C3) holds, and let $\tilde L$ be the set of all vertices of $L$ that are reachable from $x$ via paths that only use vertices of $L$. Denote $\tilde R=R\cup (L\setminus\tilde L)$. Note that in this case, $\tilde L\subseteq A$, and $(\tilde L,S,\tilde R)$ is a global minimum vertex-cut in $G$. We then continue to construct the graph $G^{|A}$ and compute the vertex sets $L',S'$ exactly as before. Observe that, if Condition (C3) holds, then $x\in \tilde L$ and $y\in \tilde R$ must hold, and, using the same reasoning as before, the resulting cut $(L',S',R')$ is guaranteed to be a global minimum vertex-cut in $G$.
It is easy to verify that the running time of the algorithm in this case is bounded by $O(|A|\cdot d\cdot n^{\eps+o(1)}\cdot\log(\wmax))\leq O(dn^{\eps+o(1)}\lambda\cdot\log(\wmax))$. Otherwise, if the DFS search discovers $\lambda+1$ vertices, then we terminate it; in this case, we are guaranteed that either Condition (C3) does not hold, or $S$ contains a low-degree vertex of weight at least $1\geq\frac{\tau}{2}$. We then employ the algorithm that we described above to compute the vertex sets $L',S'$.

To summarize, we now provided an algorithm that, given a set $Z\subseteq V(G)$ of vertices, computes the sets $L',S'\subseteq V(G)$ of vertices, so that,  if we denote by $R'=V(G)\setminus (L'\cup S')$, then $(L',S',R')$ is a vertex-cut in $G$, in time $O((|Z|+1)\cdot d^2\cdot n^{3\eps+o(1)}\cdot \lambda\cdot \log(\wmax))$. Moreover, if Condition (C3) holds and $Z$ is a promising vertex-set, then $(L',S',R')$ is guaranteed to be a global minimum vertex-cut in $G$.

In order to complete the proof of \Cref{lem: Case 1}, it is now enough to provide an algorithm that computes a relatively small set $Z$ of vertices of $G$, such that, if Condition (C3) holds, then, with probability at least $\half$, $Z$ is a promising vertex set. We provide such an algorithm in the following lemma.

\begin{lemma}\label{lem: Case 1 promising}
	There is a randomized algorithm, that is given as input a directed $n$-vertex and $m$-edge graph $G$ with integral weights $w(v)\geq 1$ on its vertices $v\in V(G)$, and  parameters $0<\eps<1$, $1\leq \tau\leq \wmax$, $1\leq \lambda\leq 2n^{\eps}$, and $\rho\geq \lambda$. Additionally, the algorithm is given a pair $(x,y)$ of distinct vertices of $G$ with $(x,y)\not\in E(G)$, and two copies of the corresponding graph $G''_{x,y}$ in the modified adjacency-list representation. The algorithm computes a set $Z\subseteq V(G)$ of vertices with $|Z|\leq O(\lambda^3\rho)$.
	It ensures that, if Condition (C3) holds (that is, all of the following hold:
	\begin{itemize}
		\item there is a global minimum vertex-cut $(L,S,R)$ for which Properties \ref{prop: small L}-\ref{prop: small volume of L} hold;
		
		\item $\frac{\lambda}{2}\leq |L|\leq \lambda$;
		
		\item $\tau=\tau^*$, where $\tau^*$ is the critical threshold for $(L,S,R)$;
		\item $|\hat S|\leq \rho$, where $\hat S=S\cap \Vld\cap V^{\geq \tau'}$ and $\tau'=\frac{\tau}{64\lambda^2}$; 
		\item no vertex of $L$ is suspicious; and
		\item $x\in L$, and $y\in R$),
	\end{itemize}
	
	then, with probability at least $\half$, $Z$ is a promising vertex set for $(L,S,R)$. The running time of the algorithm is $O\left (d\cdot \lambda^5\cdot \rho^2\cdot n^{\eps+o(1)}\cdot \log(\wmax)\right )$, where $d$ is the average vertex degree in $G$.
\end{lemma}

We prove \Cref{lem: Case 1 promising} below, after we complete the proof of \Cref{lem: Case 1} using it. We use the algorithm from \Cref{lem: Case 1 promising} to compute a set $Z$ of $O(\lambda^3\rho)$ vertices, in time $O\left (d\cdot \lambda^5\cdot \rho^2\cdot n^{\eps+o(1)}\cdot \log(\wmax)\right )$. We then use the algorithm described above to compute the sets $L',S'\subseteq V(G)$ of vertices, in time 
$O((|Z|+1)\cdot d^2\cdot n^{3\eps+o(1)}\cdot \lambda\cdot \log(\wmax))\leq O(d^2\cdot n^{3\eps+o(1)}\cdot \lambda^4\cdot\rho\cdot  \log(\wmax))$, and output $L',S'$ as the algorithm's outcome. Since $\lambda\leq 2n^{\eps}$ and $\lambda\leq \rho$, the total running time of the algorithm is  $ O(d^2\cdot n^{3\eps+o(1)}\cdot \lambda^3\cdot \rho^2\cdot  \log(\wmax))$. Assume now that Condition (C3) holds. Then, from \Cref{lem: Case 1 promising}, with probability at least $\half$, the set $Z$ of vertices is promising, and therefore,  $(L',S',R')$ is a global minimum vertex-cut in $G$. 
In order to complete the proof of \Cref{lem: Case 1}, it is now enough to prove \Cref{lem: Case 1 promising}, which we do in the remainder of this subsection.

\subsection*{Proof of \Cref{lem: Case 1 promising}}
As before, since the vertices $x$ and $y$ are now fixed, we denote the graph $G''_{x,y}$ by $G''$ for brevity. We also denote the vertex $x^{\out}\in V(G'')$ by $s$. 
Our algorithm uses the graph $G^h\subseteq G''$ that we define next.

\paragraph{Graph $G^h$.}
We let $G^h$ be a graph that is obtained from $G''$ as follows. 
We start with $G^h=G''$. Then we delete, for every vertex $v\in V(G)\setminus\set{x,y}$ with $w(v)<\tau'$, both copies of  $v$, together with their incident edges, from the graph. 
For every vertex $u$ of the resulting graph, if edge $(u,t)$ lies in $G''$, then we delete all edges leaving $u$, except for the edge $(u,t)$. Finally, for every special edge $(u^{\inn},u^{\out})$ that remains in the graph, we decrease its capacity by at most $\tau'$ units, until it becomes an integral multiple of $\tau'$ (we note that, since both $\tau'$ and $\wmax$ are integral powers of $2$, the capacity of every regular edge is an integral multiple of $\tau'$). For every edge $e\in E(G^h)$, its capacity in $G^h$ is denoted by $\hat c(e)$; note that, by our construction, $\hat c(e)>0$ must hold.
This completes the definition of the graph $G^h$. While we cannot afford computing the graph $G^h$ explicitly, we can use the modified adjacency-list representation of $G''$ in order to simulate the adjacency-list representation of $G^h$.

We denote by $\opt^h$ the value of the maximum $s$-$t$ flow in $G^h$. 
Recall that, if Condition (C3) holds, then the edges in set $\set{(v^{\inn},v^{\out})\mid v\in S}$ disconnect vertices of $L^*\cup S^{\inn}$ from vertices of $R^*\cup S^{\out}$ in $G''$, and, in particular, they disconnect $s$ from $t$.  Let $S^{\geq \tau'}=\set{v\in S\mid  w(v)\geq \tau'}$.
It is then immediate to verify that the edges in set $\set{(v^{\inn},v^{\out})\mid v\in S^{\geq \tau'}}$ disconnect vertices of $L^*\cup S^{\inn}$ from vertices of $R^*\cup S^{\out}$ in $G^h$. We will use this fact later.

Our algorithm consists of three steps. In the first step, we compute an initial $s$-$t$ flow $f$ in $G^h$, whose value is quite close to $\opt^h$, so $\val(f)\geq \opt^h-2\rho\cdot \tau$. We also ensure that, for every vertex $v\in \hat S$, the flow on the corresponding special edge $(v^{\inn},v^{\out})$ is $0$. 
In the second step of the algorithm, we perform $O(\rho\cdot \lambda^2)$ iterations. In every iteration, we use the local flow augmentation paradigm in order to augment the current $s$-$t$ flow $f$ in $G^h$ by $\tau'$ units via an augmenting path of length at most $O(\lambda)$; all vertices of $G$ whose copies lie on the augmenting paths that the algorithm uses are added to the set $Z$. In every iteration of the local flow augmentation algorithm, we may perform a shortcut operation, but we will ensure that, with probability at least $\half$, all these shortcut operations are valid. Once the local flow augmentation algorithm terminates, we ensure that the total number of vertices reachable from $s$ in the resulting residual flow network $H=G^h_f$ is relatively small; in Step 3, we add the corresponding vertices of $G$ to the set $Z$ as well. We then show that, if Condition (C3) holds and all shortcut operations were valid, then $Z$ must contain all low-degree vertices of $S$ whose weight is at least $\hat \tau$. By discarding high-degree vertices from $Z$, we then obtain the desired set $Z$.
We now describe each of the steps in turn.

\subsubsection*{Step 1: Computing the Initial Flow}

The purpose of this step is to compute an initial $s$-$t$ flow $f$ in the graph $G^h$, whose value is close to $\opt^h$, so that, for every vertex $v\in \hat S$, $f(v^{\inn},v^{\out})=0$ holds. We will also ensure that the flow $f$ is $\tau'$-integral. In order to do so, we consider a new graph $\hat G\subseteq G''$, that is obtained from $G''$ as follows. We start with  $\hat G=G''$, and then delete from it the copies of all vertices in:

\[\set{v\in (V(G)\cap V^{\ld})\setminus\set{x,y}\mid w(v)\leq \tau} \cup \set{v\in (V(G)\cap V^{\hd})\setminus\set{x,y}\mid w(v)<\tau'}. \] 

In other words, we delete the copies of all vertices $v\not\in \set{x,y}$, such that either (i) $v$ is a low-degree vertex and its weight is at most $\tau$; or (ii) $v$ is a high-degree vertex and its weight is below $\tau'$.
Additionally, for every vertex $u$ that is connected to $t$ with an edge in the resulting graph, we delete all edges that leave $u$ except for the edge $(u,t)$. 
Finally, for every special edge $(v^{\inn},v^{\out})$ that remains in the graph, we decrease its capacity by at most $\tau'$ units, until it becomes an integral multiple of $\tau'$.
This completes the definition of the graph $\hat G$. It is immediate to verify that $\hat G\subseteq G^h$. Indeed, graph $\hat G$ can be obtained from $G^h$ by deleting from it the copies of every low-degree vertex $v\in V(G)\setminus\set{x,y}$, whose weight is at most $\tau$.
Moreover, for every edge $e\in E(\hat G)$, its capacity in $\hat G$ is the same as in $G^h$, and is denoted by $\hat c(e)$. From the above definitions, $\hat c(e)> 0$ must hold for all $e\in E(\hat G)$.

 Let $\hat J$ be the set of all vertices of $\hat G$ that are reachable from $s$, and let $J\subseteq V(G)$ be the set of all vertices $v$ such that a copy of $v$ lies in $\hat J$; note that $x\in J$. Lastly, we let $\Jld\subseteq J$ be the subset of all low-degree vertices of $J$. We use the following simple observation.

\begin{observation}\label{obs: large reachable set}
	If Condition (C3) holds, then $\Jld\subseteq L$ and $|\Jld|\leq \lambda$.
\end{observation}
\begin{proof}
	Assume that Condition (C3) holds. From the definition of the critical threshold, and since $\tau=\tau^*$, set $S$ may not contain a low-degree vertex $u$ with $w(u)> \tau$.
	We start by proving that $\Jld\subseteq L$. Indeed, assume otherwise, and let $v\in \Jld\setminus L$ be some vertex. Then there must be a path from $s$ to $v^{\out}$ in $\hat G$; let $P$ be any such path (if there is a path connecting $s$ to $v^{\inn}$ in $\hat G$ then the edge $(v^{\inn},v^{\out})$ must lie in $\hat G$ from the definition of $\hat G$, so there is a path connecting $s$ to $v^{\out}$ in $\hat G$ as well). 
	
	If $v\in S$, then the last edge on path $P$ must be $(v^{\inn},v^{\out})$. Since $v$ is a low-degree vertex by the definition of the set $\Jld$, from the definition of the critical threshold, it must be the case that $w(v)\leq \tau$, so the edge $(v^{\inn},v^{\out})$ and its endpoints may not lie in $\hat G$, a contradiction.
	
	Otherwise, $v\in R$, and there must be some vertex $u\in S$ such that the edge $(u^{\inn},u^{\out})$ lies on $P$. If $u$ is a low-degree vertex, then $w(u)\leq \tau$ must hold as before, and so the edge $(u^{\inn},u^{\out})$ may not lie in graph $\hat G$. Otherwise, the only edge leaving the vertex $u^{\out}$ in $\hat G$ is the edge $(u^{\out},t)$, so there is no way to reach $v^{\out}$ from $u^{\out}$ in $\hat G$. 
	
	We conclude that, if $v$ is a low-degree vertex in $S\cup R$, then there is no $s$-$v^{\out}$ path in $\hat G$, and so $v$ may not lie in $\Jld$.
	Therefore, $\Jld\subseteq L$ must hold. From the definition, if Condition (C3) holds, then $|L|\leq \lambda$ and so $|\Jld|\leq \lambda$ as well.
\end{proof}

\paragraph{Vertex set $N$.} Next, we define the set $N\subseteq V(G)$ of vertices of $G$, that will play an important role in the current step. Intuitively, $N$ is the set of all out-neighbors of the vertices in $\Jld$, whose weight is at least $\tau'$. Formally:

\[N=\set{v\in V(G)\setminus \Jld\mid w(v)\geq \tau' \mbox{ and } \exists u\in \Jld \mbox{ with } (u,v)\in E(G)}.\]
 
We partition $N$ into two subsets: set $\Nhd$ of high-degree vertices and set $\Nld$ of low-degree vertices. Notice that, from the construction of the graph $\hat G$, for every vertex $v\in \Nld$, $w(v)\leq \tau$ must hold. 
We also denote by $\hat G'\subseteq \hat G$ the subgraph of $\hat G$ containing all vertices and edges of $\hat G$ that are reachable from $s$ in $\hat G$. We also include the vertex $t$ in $\hat G'$, even if it is not reachable from $s$ in $\hat G$.
We  obtain the following corollary of \Cref{obs: large reachable set}.

\begin{corollary}\label{cor: candidate small graph}
Assume that Condition (C3) holds. Then all of the following hold:

\begin{itemize}
	\item $\Nld\subseteq \hat S\cup L$;
	\item $|N|\leq n^{\eps}\cdot d$; and
	\item  $|E(\hat G')|+|V(\hat G')|\leq 8n^{\eps}\cdot d$.
\end{itemize}
\end{corollary}
\begin{proof}
We assume that  Condition (C3) holds.  
	From \Cref{obs: large reachable set}, $\Jld\subseteq L$ and $|\Jld|\leq \lambda\leq n^{\eps}$. Moreover, from Property \ref{prop: small volume of L}, the total number of edges incident to the vertices of $L$, and hence of $\Jld$, in $G$ is bounded by $n^{\eps}\cdot d$. It follows that $|N|\leq n^{\eps}\cdot d$, and that the number of edges in $\hat G'$ that are incident to the copies of the vertices of $\Jld$ is bounded by $2n^{\eps}\cdot d$. It is easy to verify that every vertex of $\hat G'$ that is different from $t$ is a copy of a vertex in $\Jld\cup N^{\hd}$, so $|V(\hat G')|\leq 2n^{\eps}\cdot d$. Lastly, the only edges of $\hat G'$ that are not incident to the vertices of $\Jld$ lie in edge set $\set{(v^{\inn},v^{\out}),(v^{\out},t)\mid v\in \Nhd}$, 
			and their number is bounded by $2|\Nhd|\leq 2n^{\eps}\cdot d$. Altogether, we get that $|E(\hat G')|+|V(\hat G')|\leq 8n^{\eps}\cdot d$.
			
	Consider now some vertex $v\in \Nld$. Since $\Jld\subseteq L$ and since, from the definition of the set $N$, there is a vertex $u\in \Jld$ with edge $(u,v)$ in $G$, we get that $v\in S\cup L$ must hold. If $v\in S$, then, since it is a low-degree vertex of weight at least $\tau'$ (from the definition of $N$ and $\Nld$), it must lie in $\hat S=S\cap V^{\geq \tau'}\cap \Vld$.
\end{proof}

\paragraph{The algorithm for step 1.}
We are now ready to describe the algorithm for Step 1. The algorithm performs a search in the graph $\hat G$ starting from $s$, until either (i) more than $8n^{\eps}\cdot d$ edges and vertices are discovered; or (ii) all vertices and edges reachable from $s$ in $\hat G$ have been discovered, whichever happens first. Note that we can use the modified adjacency-list representation of the graph $G''$ in order to perform this search in time $O(n^{\eps}\cdot d)$. If the search terminates after more than  $8n^{\eps}\cdot d$ edges and vertices of $\hat G$ are discovered, then, from \Cref{cor: candidate small graph}, we are guaranteed that Condition (C3) does not hold. In this case, we terminate the algorithm and return $Z=\emptyset$. From now on we assume that the search algorithm terminated after discovering all vertices and edges of $\hat G$ that are reachable from $s$, so it has discovered all edges and vertices of graph $\hat G'$. We then compute the sets $J$, $\Jld$ and $N$ of vertices of $G$, in time $O(n^{\eps}\cdot d)$. If we discover that $|\Jld|>\lambda$, or  that $|N|>n^{\eps}\cdot d$, or that $|\Nld|>\lambda+\rho$, then, as before, we are guaranteed that Condition (C3) does not hold, and we terminate the algorithm, returning $Z=\emptyset$. Therefore, from now on we assume that our algorithm has correctly computed the graph $\hat G'$ and the sets $J$, $\Jld$ and $N$ of vertices of $G$, and moreover, that 
 $|\Jld|\leq \lambda$, $|N|\leq n^{\eps}\cdot d$, and $|\Nld|\leq \lambda+\rho\leq 2\rho$ hold. So far the running time of the algorithm for Step 1 is $O(n^{\eps}\cdot d)$. Our last step is to compute a maximum $s$-$t$ flow $f$ in graph $\hat G'$, in time $O\left(|E(\hat G')|^{1+o(1)}\right )\leq O\left (n^{\eps+o(1)}\cdot d\cdot\log(\wmax)\right )$, using the algorithm from \Cref{thm: maxflow}. 
 Since all edge capacities in $\hat G'$ are integral multiples of $\tau'$, we can ensure that the flow $f$ is $\tau'$-integral.

Recall that we denoted by $\opt^h$ be the value of the maximum $s$-$t$ flow in $G^h$. 
 In the following key claim we show that the value of $f$ is close to $\opt^h$.

\begin{claim}\label{claim: current deficit}
	Assume that  Condition (C3) holds. Then $\val(f) \geq \opt^h-2\rho\cdot \tau$.  
\end{claim}

\begin{proof}
	We assume that  Condition (C3) holds, and let $f'$ be the maximum $s$-$t$ flow in $G^h$, of value $\opt^h$. Recall that the edges of set $E'=\set{(v^{\inn},v^{\out})\mid v\in S^{\geq \tau'}}$ disconnect $s$ from $t$ in $G^h$, and so $\opt^h\leq \sum_{v\in S^{\geq\tau'}}\hat c(v^{\inn},v^{\out})$. Let $\pset$ be the flow-path decomposition of $f'$. For every path $P\in \pset$, let $e(P)=(u^{\inn},u^{\out})$ be the first edge of $E'$ that lies on $P$, and denote by $v(P)=u$ the corresponding vertex of $S^{\geq\tau'}$. Let $P'$ be the subpath of $P$ from $s$ to $u^{\out}$. It is immediate to verify that all vertices of $P'$, except for the two last ones, lie in $L^*$. 
	
	We partition the set $\pset$ of paths into three subsets. The first set, $\pset_1$ contains all paths $P$, such that there is a special edge $(v^{\inn},v^{\out})\in E(P')$ with $v\in L$ and $w(v)\leq \tau$. Since $|L|\leq \lambda$, $\sum_{P\in \pset_1}f'(P)\leq \lambda\cdot \tau\leq \rho\cdot \tau$.
	
	The second set, $\pset_2$, contains all paths $P$ that do not lie in $\pset_1$, for which $v(P)\in \Vld$ (note that, from the definition of the critical threshold $\tau^*$, and from our assumption that $\tau=\tau^*$, $w(v(P))\leq \tau$ must hold in this case).
	Consider any path $P\in \pset_2$, and let $P''$ be obtained from the corresponding path $P'$ by deleting the last edge $e(P)$, together with its endpoints, from it. Then every special edge on $P'$ corresponds to a vertex of $L$ whose weight is greater than $\tau$. Therefore, every vertex and edge of $P''$ lie in 
	$\hat G'$, and every vertex of $P''$ lies in $\hat J$. Moreover, there is a vertex $a\in \Jld$ with edge $(a,v(P))$ in $G$, and $w(v(P))\geq \tau'$ from the definition graph $G^h$. Therefore, $v(P)\in \Nld$ must hold, and, from \Cref{cor: candidate small graph}, $v(P)\in \hat S$.
	Therefore, $\sum_{P\in \pset_2}f'(P)\leq  |\hat S|\cdot \tau\leq \rho\cdot \tau$.

Lastly, set $\pset_3$ contains all remaining paths. 
Let $P\in \pset_3$ be any such path. Recall that every special edge in $E(P')\setminus \set{e(P)}$ must correspond to a vertex of $L$ whose weight is greater than $\tau$. As before, let $P''$ be the path obtained from $P'$ by deleting the edge $e(P)$, together with its both endpoints, from it.
Using the same argument as before, every edge and every vertex of $P''$ lie in $\hat G'$, and every vertex of $P''$ lies in $\hat J$. 
From the definition of the set $\pset_2$ of paths, vertex $u=v(P)$ must be a high-degree vertex, and, from the definition of graph $G^h$, $w(u)\geq \tau'$ must hold. It is then easy to verify that $u^{\inn}$ and $u^{\out}$ both lie in $\hat G$, and so $P'\subseteq \hat G'$. We let $\hat P$ be the path obtained by concatenating $P'$ with the edge $(u^{\out},t)$, that must lie in $G''$ since $u$ is a high degree vertex, so $\hat P\subseteq \hat G'$ must hold as well.

Let $f''$ be the flow in $\hat G'$ obtained by sending, for every path $P\in \pset_3$, $f'(P)$ flow units via the path $\hat P$. Then $f''$ is an $s$-$t$ flow in $\hat G'$, and its value is at least:

\[\sum_{P\in \pset_3}f'(P)\geq \sum_{P\in \pset}f'(P)-\sum_{P\in \pset_1\cup \pset_2}f'(P)\geq \opt^h-2\rho\cdot\tau.\]

Since $f$ is the maximum $s$-$t$ flow in $\hat G'$, we get that $\val(f)\geq \val(f'')\geq \opt^h-2\rho\cdot \tau$.	
\end{proof}

Recall that we have established already that $\hat G\subseteq G^h$. Therefore, $f$ is a valid $s$-$t$ flow in $G^h$, and, if Condition (C3) holds, then $\val(f)\geq \opt^h-2\rho\tau\geq \opt^h-128\tau'\lambda^2\rho$,
as $\tau'=\frac{\tau}{64\lambda^2}$ . Lastly, assume that Condition (C3) holds, and let $v$ be any vertex of $\hat S$. Then $v$ is a low-degree vertex, and, from the definition of the critical threshold, $w(v)\leq \tau$ must hold; in particular, the copies of $v$ do not lie in $\hat G$, and so $f(v^{\inn},v^{\out})=0$. We summarize the properties of the flow $f$ that we established in the following observation.

\begin{observation}\label{obs: initial flow}
	Flow $f$ is a valid $\tau'$-integral  $s$-$t$ flow in $G^h$. Moreover, if Condition (C3) holds, then $\val(f)\geq \opt^h-128\tau'\lambda^2\rho$, and, for every vertex $v\in \hat S$, $f(v^{\inn},v^{\out})=0$.
\end{observation}

\subsubsection*{Step 2: Local Flow Augmentations}
\label{subsubsec: local flow}

In this step we gradually augment the $s$-$t$ flow $f$ in graph $G^h$ that was computed in Step 1, while keeping it $\tau'$-integral.

Recall that the modified adjacency-list representation of graph $G''$ that was provided as part of input can be used in order to simulate the  adjacency-list representation of $G^h$. Over the course of the algorithm, we may perform some shortcut operations, by inserting edges $(v,t)$ into $G^h$, for some vertices $v\in V(G^h)$. We will use the first copy of the modified adjacency-list representation of $G''$ in order to maintain the resulting graph $G^h$ as it evolves.

Throughout, we denote by $H$ the residual flow network of $G^h$ with respect to the current flow $f$. We use the second  copy of the modified adjacency-list representation of $G''$ in order to simulate the adjacency-list representation of the graph $H$. For the initial flow $f$ computed in Step 1, we can modify the adjacency-list representation of $G''$ to obtain an adjacency-list representation of the initial residual flow network $H$ in time that is asymptotically bounded by the running time of Step 1. 
Note that, since all edge capacities in $G^h$ are integral multiples of $\tau'$, and since the flow $f$ that we maintain is $\tau'$-integral, the capacity of every edge in $H$ is at least $\tau'$ throughout the algorithm.

Our algorithm for Step 2 performs at most $z=256\lambda^2\rho$ iterations. In every iteration $i$, we select an $s$-$t$ path $P_i$ in the current graph $H$, whose length is at most $16\lambda$, and augment the current flow $f$ by sending $\tau'$ flow units along $P_i$. Every vertex $v\in V(G)$ whose copy lies on $P_i$ will be eventually added to $Z$. In order to compute the path $P_i$ efficiently, we may perform a shortcut operation.
We now describe every iteration more formally.

\paragraph{Iteration Description.}
In order to execute iteration $i$, we perform a BFS search in the current graph $H$, until one of the following conditions are met:

\begin{properties}{Q}
	\item Vertex $t$ is reached; \label{cond: found t}
	\item at least $16z\lambda$ vertices of $V^{\out}$ are discovered; \label{cond: many vertices} or
	\item Neither of the above conditions happened and the BFS has explored all vertices $v$ such that there is an $s$-$v$ path of length at most $16\lambda$ in $H$.\label{cond: BFS done}
\end{properties}

We now consider three cases, depending on which of the terminating conditions happened.

\paragraph{Case 1: Condition \ref{cond: found t} is met.} In this case, we let $P_i$ be the $s$-$t$ path of length at most $16\lambda$ through which $t$ was discovered. We augment $f$ by sending $\tau'$ flow units along this path, update the residual flow network $H$, and continue to the next iteration.

\paragraph{Case 2: Condition \ref{cond: many vertices} met.} In this case, we select a vertex $a^{\out}\in V^{\out}$ among the discovered vertices uniformly at random, and we connect $a^{\out}$ to $t$ in $G''$ and in $G^h$ via a shortcut operation. We also let $P_i'$ be the path of length at most $16\lambda$ via which the BFS discovered $a^{\out}$, and we denote by $P_i=P'_i\circ (a^{\out},t)$. We then augment $f$ by sending send $\tau'$ flow units on path $P_i$. As before, we update the flow $f$ and the residual network $H$, and continue to the next iteration. 

If Condition \ref{cond: BFS done} is met, then we terminate the algorithm for Step 2 and proceed to Step 3, where we will construct 
the desired set $Z$ of vertices.

\paragraph{Running time.} 
We now analyze the running time of a single iteration. Recall that the total number of vertices of $V^{\out}$ that may be discovered in a single iteration is $O(z\cdot \lambda)$. If, for any such vertex $v^{\out}$, the corresponding vertex $v\in V(G)$ is a high-degree vertex, then, since the only edge leaving $v^{\out}$ in $G^h$ is the edge $(v^{\out},t)$, the iteration terminates with outcome \ref{cond: found t} immediately upon discovering $v^{\out}$. Therefore, the total number of edges that may be discovered in a single iteration is bounded by $O\left (z\cdot \lambda\cdot d\cdot n^{\eps}\right )$. Since the 
number of iterations is bounded by $z=O(\lambda^2\rho)$, the total running time of Step 2 is:

\[O\left (z^2\cdot \lambda\cdot d\cdot n^{\eps}\right )\leq O\left (\lambda^5\cdot \rho^2\cdot d\cdot n^{\eps}\right ).\]

\paragraph{Valid shortcuts.} Let $\event''$ be the bad event that any of the shortcut operations that the algorithm performed were invalid. We use the following simple observation.

\begin{observation}\label{obs: valid shortcuts}
	If Condition (C3) holds, then $\prob{\event''}\leq \frac{1}{16}$.
\end{observation}
\begin{proof}
	Assume that Condition (C3) holds. Recall that we perform at most $z$ shortcut operations, and in each such operation, we insert an edge $(v^{\out},t)$ into $G''$, where $v^{\out}\in V^{\out}$. This shortcut operation may only be invalid if $v\in L$. In every iteration, we select a vertex of $V^{\out}$ to shortcut from among $16z\lambda$ such vertices uniformly at random. Therefore, the probability that a single shortcut operation is invalid is at most $\frac{|L|}{16z\lambda}\leq \frac{1}{16z}$. From the Union Bound, the probability that any of the at most $z$ shortcut operations that we perform is invalid is bounded by $\frac{1}{16}$. Therefore, 
	$\prob{\event''}\leq \frac{1}{16}$.
\end{proof}

Recall that, from \Cref{obs: initial flow}, if Condition (C3) holds, then $\val(f)\geq \opt^h-128\tau'\lambda^2\rho$, where $f$ is the initial flow computed in Step 1. Recall also that in every iteration, we augment $f$ by $\tau'$ flow units. Therefore, if, after $\ceil{128\lambda^2\rho}+1$ iterations the algorithm does not termiate with outcome \ref{cond: BFS done}, then we are guaranteed that either Condition (C3) does not hold, or Event $\event''$ happened. In any case, if, after $\ceil{128\lambda^2\rho}+1$ iterations the algorithm for Step 2 does not termiate with outcome \ref{cond: BFS done}, then we terminate the entire algorithm and return $Z=\emptyset$. If Condition (C3) holds, the probability for this happening is bounded by 
$\prob{\event''}\leq \frac{1}{16}$.

\subsubsection*{Step 3: Computing the Set $Z$}

We compute the initial set $Z\subseteq V(G)$ as follows. First, we include in $Z$ every vertex $v$, such that, in any iteration of Step 2, a copy of $v$ lied on the augmenting path that was used in that iteration. Since the number of iterations in Step 2 is $z=O(\lambda^2\rho)$, and the length of the augmenting path used in every iteration is $O(\lambda)$, $|Z|\leq O(\lambda^3\rho)$ now holds. Next, for every vertex $v\in V(G)\setminus Z$, if vertex $v^{\out}$ was discovered in the last iteration of the algorithm, we add $v$ to $Z$ as well. Since Condition \ref{cond: BFS done} must be met in the last iteration of the algorithm, the total number of vertices of $V^{\out}$ that are discovered in the last iteration is $O(z\lambda)$, and so $|Z|\leq O(\lambda^3\rho+z\lambda )\leq  O(\lambda^3\rho)$ must hold. Lastly, we delete from $Z$ all high-degree vertices.
This completes the algorithm for Step 3. Notice that its running time is subsumed by the running time of Step 2. Therefore, the total running time of the algorithm is:

\[O\left (n^{\eps+o(1)}\cdot d\cdot\log(\wmax)+\lambda^5\cdot \rho^2\cdot d\cdot n^{\eps}\right )\leq O\left (\lambda^5\cdot \rho^2\cdot d\cdot n^{\eps+o(1)}\cdot \log(\wmax)\right ).\]

We use the following claim to analyze Step 3.

\begin{claim}\label{claim: promising} 
	Assume that Condition (C3) holds and that Event $\event''$ did not happen. Then $Z$ is a promising set of vertices.
\end{claim}
\begin{proof}
		Assume that Condition (C3) holds and that Event $\event''$ did not happen.
	From our construction, every vertex in $Z$ is a low-degree vertex. 
	Consider now any low-degree vertex $v\in S$ with $w(v)\geq \hat \tau$. It is now enough to prove that $v\in Z$.
	Recall that $\hat \tau=\frac{\tau}{8\lambda}$, $\tau'=\frac{\tau}{64\lambda^2}<\hat \tau$, and  
	$\hat S=S\cap \Vld\cap V^{\geq \tau'}$. Clearly, $v\in \hat S$ must hold.
	
	Let $f'$ be the flow $f$ computed at the end of Step 1. Then, from \Cref{obs: initial flow}, $f'(v^{\inn},v^{\out})=0$. Let $f''$ be the flow $f$ computed at the end of Step 2. If $f''(v^{\inn},v^{\out})>0$, then edge $(v^{\inn},v^{\out})$ must lie on some augmenting path $P_i$ that the algorithm used, and so $v\in Z$ must hold. 
	
	Finally, assume that $f''(v^{\inn},v^{\out})=0$, and let $H'$ be the graph $H$ at the end of Step 2. It is enough to show that there exists a path connecting $s$ to $v^{\out}$ in $H'$ of length at most $16\lambda$, since in this case $v\in Z$ must hold.

	Let $\tilde H$ be the residual flow network of $G''$ with respect to the $s$-$t$ flow $f''$, and
	let $(X,Y)$ be a partition of $V(\tilde H)$, where $X=L^*\cup S^{\inn}$ and $Y=S^{\out}\cup R^*$. We use the following simple observation whose proof is deferred to Section \ref{sec: proof of min cut in residual} of Appendix.
	
	\begin{observation}\label{obs: min cut in residual}
	If Condition (C3) holds and Event $\event''$ did not happen, then	$(X,Y)$ is a minimum $s$-$t$ edge-cut in $\tilde H$.
	\end{observation}

Consider now the maximum $s$-$t$ flow $g$ in the graph $\tilde H$ and let $\pset$ be its flow-path decomposition. Then, from the Max-flow / Min-cut theorem, $\val(g)=\sum_{e\in E_{\tilde H}(X,Y)}\tilde c(e)$, where $\tilde c(e)$ is the capacity of the edge $e$ in $\tilde H$. Therefore, for every edge $e\in E_{\tilde H}(X,Y)$, $g(e)=\tilde c(e)$, and moreover, every flow-path $P\in \pset$ contains exactly one edge of  $E_{\tilde H}(X,Y)$; we denote this edge by $e(P)$. Note that all vertices preceding $e(P)$ on path $P$ must lie in $L^*\cup S^{\inn}$.

Consider now the edge $e^*=(v^{\inn},v^{\out})$. Since $f''(v^{\inn},v^{\out})=0$, its residual capacity is $\tilde c(e^*)=\hat c(e^*)\geq w(v)-\tau'\geq \hat \tau-\tau'\geq \frac{\tau}{16\lambda}$. Clearly, $e^*\in E_{\tilde H}(X,Y)$, and so, from the above discussion $g(e^*)\geq \tilde c(e^*)\geq \frac{\tau}{16\lambda}$. Let $\pset'\subseteq \pset$ be the collection of paths $P\in \pset$ with $e(P)=e^*$; notice that $\sum_{P\in \pset'}g(P)=g(e^*)\geq \frac{\tau}{16\lambda}$ must hold. For every path $P\in \pset'$, we denote by $P'$ the subpath of $P$ from $s$ to $v^{\out}$. Since all vertices on $P'\setminus\set{v^{\out}}$ lie in $L^*\cup S^{\inn}$, for every edge $e\in E(P'\setminus\set{v^{\out}})$, one endpoint of $e$ must lie in $L^{\out}$, and so $|E(P')|\leq 2|L|<16\lambda$.

We partition the set $\pset'$ into two subsets: set $\pset_1'$ containing all paths $P\in \pset'$, such that at least one edge $e'\in E(P')$ has residual capacity $\tilde c(e')<\tau'$, and set $\pset_2'$ contains all remaining paths of $\pset'$. 

We claim that, for every path $P\in \pset_2'$, the corresponding subpath $P'$ must lie in $H'$.
Indeed, consider any path $P\in \pset_2'$, any edge $e=(a,b)\in E(P')$, and let $e'$ be the corresponding edge of $G''$ (so if $e$ is a forward edge, $e'=e$, and otherwise $e'=(b,a)$). 
Since $\tilde c(e)\geq \tau'$, it must be the case that $\hat c(e')\geq \tau'$ as well. Moreover, since $a,b\in L^*\cup S^{\inn}$, and since we have assumed that Event $\event''$ did not happen, neither of the vertices $a,b$ may be connected to $t$ in $G''$. Therefore, $e'\in E(G^h)$ and $e\in E(H')$ must hold. Since the length of $P'$ is at most $16\lambda$, we conclude that, if $\pset'_2\neq\emptyset$, then $v^{\out}$ is reachable from $s$ in $H'$ via a path of length at most $16\lambda$, and so $v\in Z$ must hold.
Therefore, it is enough to show that $\pset'_2\neq \emptyset$, and, in order to do so, it is enough to show that $\sum_{P\in \pset'_1}g(P)<\frac{\tau}{16\lambda}$, since $\sum_{P\in \pset'}g(P)=g(e^*)\geq \frac{\tau}{16\lambda}$ holds.

Consider now any path $P\in \pset'_1$, and let $e'(P)\in E(P')$ be any edge that lies on the subpath $P'$ of $P$, for which $\tilde c(e'(P))<\tau'$. Since $V(P')\setminus\set{v^{\out}}\subseteq L^*\cup S^{\inn}$, at least one endpoint of $e'(P)$ must lie in $L^*$. Note that $e'(P)$ may not be a forward regular edge, since the initial capacity of every forward regular edge is $\wmax>\opt$. Moreover, $e'(P)$ may not be a backward  edge, since the flow $f''$ is $\tau'$-integral, so the residual capacity of every backward edge in $H'$ is at least $\tau'$. Therefore, $e'(P)$ must be a forward special edge, corresponding to some vertex $u\in L$. Since $|L|\leq \lambda$, we get that:

\[\sum_{P\in \pset_1'}g(P)< |L|\cdot \tau'\leq \lambda\tau'<\frac{\tau}{16\lambda},\]

since $\tau'=\frac{\tau}{64\lambda^2}$.
\end{proof}

\subsection{Algorithm for Case 2: Proof of \Cref{lem: case 2 main}}
\label{sec: case 2}

Recall that we have denoted by $\tau'=\frac{\tau}{64\lambda^2}$  and by $\hat S=S\cap \Vld\cap V^{\geq \tau'}$, and that Case 2 happens if $|\hat S|>2^{12}\lambda^3$. For brevity, we say that Condition (C'1) holds for a cut $(L,S,R)$, if all of the following hold:

\begin{itemize}
	\item  Properties \ref{prop: small L}-\ref{prop: small volume of L} hold for $(L,S,R)$;
	\item $\frac{\lambda}{2}\leq |L|\leq\lambda$; 
	\item $\tau=\tau^*$, where $\tau^*$ is the critical threshold for $(L,S,R)$;  
	\item $|\hat S|>2^{12}\lambda^3$; and  
	\item there is a pair $(x^*,y^*)\in \Gamma$ with $x^*\in L$ and $y^*\in R$,
\end{itemize}

As before, for the sake of analysis, we designate a global minimum vertex-cut $(L,S,R)$ as the \emph{distinguished min-cut} as follows. If there exists a global minimum vertex-cut for which Condition (C'1) holds, then we let $(L,S,R)$ be this cut, and we say that the distinguished min-cut is \emph{good}; otherwise, we let $(L,S,R)$ be an arbitrary global minimum vertex-cut, and we say that it is \emph{bad}. In the following, whenever we say that Condition (C'1) holds, we mean that it holds with respect to the distinguished min-cut $(L,S,R)$. We also define a distinguished pair $(x^*,y^*)\in \Gamma$ as follows: if there is a pair $(x,y)\in \Gamma$ with $x\in L$ and $y\in R$, then we let  $(x^*,y^*)$ be any such pair, and otherwise we let  $(x^*,y^*)$ be any pair in $\Gamma$.
We start by describing the preprocessing step of the algorithm.

\subsubsection{The Preprocessing Step}\label{subsec: Case 2 preprocessing}

The preprocessing step for this case is very similar to that for Case 1. For every vertex $v\in V(G)$, we construct a set $\beta'(v)\subseteq V(G)$ of vertices; this set is similar to the set $\beta(v)$ constructed for Case 1, but there are some differences, the chief among them is using the threshold of $\tau$ instead of $\hat \tau$. We then define vertex sets $\sigma'(u)$ for low-degree vertices $u\in V(G)$, that are similar to the sets $\sigma(u)$ constructed by the algorithm for Case 1, except that they are based on vertex sets $\beta'(v)$ for $v\in V(G)$ instead of $\beta(v)$. Finally, we define a set $U'\subseteq \Vld$ of suspicious vertices, based on the cardinalities of the sets $\sigma'(u)$ that we construct, and take care of a special case where $U'\cap L\neq\emptyset$ like in the algorithm for Case 1.
We now proceed to describe the preprocessing step formally.

\paragraph{$\beta'$-sets.}
For every vertex $v\in V(G)$, we construct a set $\beta'(v)\subseteq V(G)$ of vertices of cardinality at most $2\lambda$. The construction is identical to that used in the algorithm for Case 1, except that we include vertices of weight greater than $\tau$, instead of including vertices of weight at least $\hat \tau$.

Formally, we compute, for every vertex $v\in V(G)$, the set $\beta'(v)\subseteq V(G)$ of vertices as follows. We perform a DFS search  in the graph $G$ starting at $v$, except that the DFS search is only allowed to explore vertices $u\in V(G)\setminus\set{v}$ with $u\in \Vld$ and $w(u)>\tau$. Once the DFS search discovers $2\lambda$ vertices, we terminate it, and we let $\beta'(v)$ be the vertex set containing all these discovered vertices, including $v$ itself. Alternatively, if the DFS search terminates before $2\lambda$ vertices are discovered (in which case there are fewer than $2\lambda$ vertices in $G$ that can be reached from $v$ via paths that only use low-degree vertices of weight greater than $\tau$ in addition to $v$), then we let $\beta'(v)$ be the set of all vertices that the DFS search has discovered, including $v$. The algorithm for computing the sets $\beta'(v)$ of vertices for all $v\in V(G)$ is identical to that in Case 1 and we do not repeat it here; recall that its running time is  $O(m+n\lambda^2)$.

\paragraph{$\sigma'$-sets and suspicious vertices.}
For every low-degree vertex $u\in V(G)$, we define the vertex set $\sigma'(u)=\set{v\in V(G)\mid u\in \beta'(v)}$.
We say that a low-degree vertex $u\in \Vld$ is \emph{suspicious}, if $|\sigma'(u)|>2n^{\eps}\cdot \lambda$, and we denote by $U'$ the set of all low-degree suspicious vertices.

As in Case 1, $\sum_{v\in V(G)}|\beta'(v)|\leq 2n\cdot \lambda$, and so $|U'|\leq n^{1-\eps}$ must hold. By inspecting the sets $\beta'(v)$ of vertices for all $v\in V(G)$, we can compute the sets $\sigma'(u)$ of vertices for all $u\in \Vld$, as well as the set $U'\subseteq \Vld$ of suspicious vertices, in time $O(n\cdot \lambda)$.

\paragraph{Additional assumption.}

Assume that Condition (C'1) holds and that $L\cap U'\neq \emptyset$. In this case,  we show that we can efficiently compute a cut $(L',S',R')$, that, with probability at least $\half$, is a global minimum vertex-cut. The algorithm is essentially identical to that used for Case 1: For every vertex $v\in U'$, we apply the algorithm from \Cref{lem: cut if vertex of L}, that, in time $O\left (m^{1+o(1)}\cdot \log (\wmax)\right )$, returns a vertex $y_v\in V(G)$, and the value $c_v$ of the minimum $v$-$y_v$ vertex-cut in $G$. Recall that the lemma guarantees that, if there is a global minimum vertex-cut $(L,S,R)$ in $G$ with $v\in L$, then with probability at least $1-1/n^4$, $c_v=w(S)$ holds. We then select a vertex $v'$ for which the value $c_{v'}$ is the smallest, and use the algorithm from \Cref{cor: min_vertex_cut} to compute the minimum $v'$-$y_{v'}$ vertex-cut $(L',S',R')$ in $G$, whose value must be equal to $c_{v'}$, in time $O\left (m^{1+o(1)}\cdot \log (\wmax)\right )$, returning this cut as the algorithm's outcome. Since we have assumed that $L\cap U'\neq \emptyset$, when a vertex $v\in L\cap U'$ is considered,
the algorithm from \Cref{lem: cut if vertex of L} is guaranteed to output the value $c_v=\opt$ with probability at least $1-1/n^4$. If this happens, then the vertex-cut $(L',S',R')$ that our algorithm outputs is guaranteed to be optimal. Recall that we invoke the algorithm from \Cref{lem: cut if vertex of L} $O(n^{1-\eps})$ times, and the algorithm from \Cref{cor: min_vertex_cut}  once. Therefore, the total running time of this part of the algorithm is $O\left (m\cdot n^{1-\eps+o(1)} \log (\wmax)\right )$. Overall, if Condition (C'1) holds and $U'\cap L\neq\emptyset$, then the running time of our algorithm is $O\left ((m\cdot n^{1-\eps+o(1)}+n\lambda^2) \log (\wmax)\right )\leq O\left ((m\cdot n^{1-\eps+o(1)}+n^{1+2\eps}) \log (\wmax)\right )$, and it is guaranteed to return the global minimum vertex-cut with probability at least $1/2$. 

As in the algorithm for Case 1, we choose a random bit $b\in \set{0,1}$. 
If $b=0$, then we execute the algorithm described above, that is designed for the case where $L\cap U'\neq \emptyset$. If $b=1$, then we execute the algorithm that we describe below, that is designed for the case where $L\cap U'=\emptyset$. We say that a good event $\event'$ happens if either (i) $b=0$ and  $L\cap U'\neq \emptyset$; or (ii) $b=1$ and $L\cap U'= \emptyset$. Clearly, Event $\event'$ happens with probability $1/2$. It is now enough to design an algorithm for the case where $b=1$ and $L\cap U'=\emptyset$. For shorthand, we say that Condition (C'2) holds if Condition (C'1) holds, and, additionally, $b=1$ and $L\cap U'=\emptyset$.
We now focus on designing an algorithm that, if Condition (C'2) holds, outputs a global minimum vertex-cut with a sufficiently high probability.

Before we proceed, we define a set $B\subseteq V(G)$ of vertices to include all vertices $v\in V(G)$, for which either $|\beta'(v)|>\lambda$ or $\beta'(v)\cap U'\neq \emptyset$ hold:  $B=\set{v\in V(G)\mid |\beta'(v)|>\lambda\mbox{ or } \beta'(v)\cap U'\neq \emptyset}$. We refer to vertices in set $B$ as \emph{bad}, and to the remaining vertices of $G$ as \emph{good}. We use the following simple observation to argue that, if Condition (C'2) holds, then $B\cap L=\emptyset$.

\begin{observation}\label{obs: bad vertex}
	Assume that Condition (C'2) holds, and let $v\in L$ be a vertex. Then $v$ must be a good vertex, and, moreover, $\beta'(v)\subseteq L$.
\end{observation}
\begin{proof}
	Assume that Condition (C'2) holds, and consider a vertex $v\in L$. Let $u\in \beta'(v)$ be any vertex. We claim that $u\in L$ must hold. Indeed, assume otherwise. From the definition of vertex set $\beta'(v)$, there is a path $P$ connecting $v$ to $u$ in $G$, such that every vertex $a\in V(P)\setminus\set{v}$ is a low-degree vertex, whose weight is greater than $\tau$. From Condition (C'2), $\tau$ must be a critical threshold for $G$, so $S$ may not contain a low-degree vertex of weight greater than $\tau$. Therefore, $V(P)\cap S=\emptyset$ and $V(P)\subseteq L$ must hold, a contradiction. We conclude that $\beta'(v)\subseteq L$, and so, from Condition (C'2), $|\beta'(v)|\leq \lambda$ and $\beta'(v)\cap U'=\emptyset$. Therefore, $v$ must be a good vertex.
	\end{proof}

Notice that the time required to compute the set $B$ of vertices is subsumed by the running time of the preprocessing step so far.

Similarly to Case 1, our algorithm will consider every pair $(x,y)\in \Gamma$ of vertices in turn. For each such pair, it will compute a relatively small subset $A_{x,y}\subseteq V(G)$ of vertices, such that, if Condition (C'2) holds and $(x,y)$ is the distinguished pair, then $L\subseteq A_{x,y}$ will hold with a sufficiently high probability. This, in turn, will allow us to sparsify the graph $G$, in order to efficiently compute the value of the minimum $x$-$y$ cut. However, the specific algorithm for computing the vertex set $A_{x,y}$ for every pair $(x,y)\in \Gamma$ is significantly simpler. We use the following lemma, which is an analogue of \Cref{lem: Case 1} for Case 1.

\begin{lemma}\label{lem: Case 3}
	There is a randomized algorithm, that is given as input a directed $n$-vertex and $m$-edge graph $G$ with integral weights $w(v)\geq 1$ on its vertices $v\in V(G)$, and  parameters $0<\eps<1$, $1\leq \tau\leq \wmax$, and $1\leq \lambda\leq 2n^{\eps}$, together with a pair $(x,y)$ of distinct vertices of $G$ with $(x,y)\not\in E(G)$. The algorithm is also given the set $U'\subseteq \Vld$ of suspicious vertices, the set $B\subseteq V(G)$ of bad vertices, and, for every vertex $u\in\Vld$, the corresponding vertex set $\sigma'(u)$. The algorithm returns two subsets $L',S'\subseteq V(G)$ of vertices, such that, if we denote by $R'=V(G)\setminus(L'\cup S')$ then  $(L',S',R')$ is a vertex-cut in $G$.
	It also ensures that, if all of the following conditions hold:
	\begin{itemize}
		\item there is a global minimum vertex-cut $(L,S,R)$ for which Properties \ref{prop: small L}-\ref{prop: small volume of L} hold;
		
		\item $\frac{\lambda}{2}\leq |L|\leq \lambda$;
		
		\item $\tau=\tau^*$, where $\tau^*$ is the critical threshold for $(L,S,R)$;
		\item $|\hat S|>2^{12}\lambda^3$, where $\hat S=S\cap \Vld\cap V^{\geq \tau'}$ and $\tau'=\frac{\tau}{64\lambda^2}$; 
		\item no vertex of $L$ is suspicious; and
		\item $x\in L$, and $y\in R$,
	\end{itemize}
	
	then, with probability at least $\half$, $(L',S',R')$ is a global minimum vertex-cut in $G$. The running time of the algorithm is $O\left (d^2n^{3\eps+o(1)}\lambda\cdot \log (\wmax)\right )$, where $d$ is the average vertex degree in $G$.
\end{lemma}

We prove \Cref{lem: Case 3} below, after we complete the proof of \Cref{lem: case 2 main} using it. Recall that the running time of the preprocessing step is $O(m+n\lambda^2)$. If the bit $b=0$, then the additional time that the algorithm takes is $O\left ((m\cdot n^{1-\eps+o(1)}+n^{1+2\eps}) \log (\wmax)\right )$. Otherwise, for every pair $(x,y)\in \Gamma$, we apply the algorithm from \Cref{lem: Case 3}, that returns vertex sets $L_{x,y},S_{x,y}$. We select a pair $(x,y)$ for which the weight of set $S_{x,y}$ is minimized, and return
the corresponding vertex-cut $(L_{x,y},S_{x,y},R_{x,y})$, where $R_{x,y}=V(G)\setminus(S_{x,y}\cup L_{x,y})$. Since $|\Gamma|\leq \tilde O\left(\frac{n}{\lambda}\right)$, and since the running time of the algorithm from  \Cref{lem: Case 3} is $O\left (d^2n^{3\eps+o(1)}\lambda\cdot \log (\wmax)\right )$, we get that the total running time of the entire algorithm is:

\[
\begin{split}
&O\left ((m\cdot n^{1-\eps+o(1)}+n^{1+2\eps}+nd^2\cdot n^{3\eps+o(1)})\cdot \log (\wmax)\right )  \\
&\quad\quad\quad\quad\quad\quad\leq O\left ((m\cdot n^{1-\eps+o(1)}+md\cdot n^{3\eps+o(1)})\cdot\log (\wmax)\right )
\end{split}\]


Assume that Condition (C'1) holds, and that Event $\event'$ happened. If $b=0$, then our algorithm is guaranteed to return a global minimum vertex-cut with probability at least $\half$. Otherwise, when the algorithm from  \Cref{lem: Case 3} is applied to the distinguished pair $(x^*,y^*)$, the resulting cut $(L_{x^*,y^*},S_{x^*,y^*},R_{x^*,y^*})$ is guaranteed to be optimal with probability at least $\half$, and in that case, the cut that the algorithm outputs is a global minimum vertex-cut. Since $\prob{\event'}=\half$, we get that, overall, if Condition (C'1) holds, then the cut that the algorithm returns is a global minimum vertex-cut with probability at least $\frac{1}{4}$.

In order to complete the proof of \Cref{lem: case 2 main}, it is now enough to prove \Cref{lem: Case 3}, which we do in the remainder of this section.

\subsubsection{Proof of \Cref{lem: Case 3}}

We assume that we are given a directed $n$-vertex and $m$-edge graph $G$ with integral weights $w(v)\geq 1$ on its vertices $v\in V(G)$, and  parameters $0<\eps<1$, $1\leq \tau\leq \wmax$, and $1\leq \lambda\leq 2n^{\eps}$, together with a pair $(x,y)$ of distinct vertices of $G$ with $(x,y)\not\in E(G)$. We also assume that we are given the set $U'\subseteq \Vld$ of suspicious vertices, the set $B\subseteq V(G)$ of bad vertices, and, for every vertex $u\in\Vld$, the corresponding vertex set $\sigma'(u)$.

Given a vertex-cut $(L,S,R)$ in $G$, we say that Condition (C'3) holds with respect to $(L,S,R)$, if all of the following hold:

\begin{itemize}
	\item Properties \ref{prop: small L}-\ref{prop: small volume of L} hold for $(L,S,R)$;
	
	\item $\frac{\lambda}{2}\leq |L|\leq \lambda$;
	
	\item $\tau=\tau^*$, where $\tau^*$ is the critical threshold for $(L,S,R)$;
	\item $|\hat S|>2^{12}\lambda^3$, where $\hat S=S\cap \Vld\cap V^{\geq \tau'}$ and $\tau'=\frac{\tau}{64\lambda^2}$; 
	\item no vertex of $L$ is suspicious; and
	\item $x\in L$, and $y\in R$.
\end{itemize}

The following key lemma allows us to compute a relatively small set $A$ of vertices of $G$, such that, if Condition (C'3) holds for some global minimum vertex-cut $(L,S,R)$, then $L\subseteq A$.

\begin{lemma}\label{lem: Case 2 set A}
	There is a randomized algorithm, that is given as input a directed $n$-vertex and $m$-edge graph $G$ with integral weights $w(v)\geq 1$ on its vertices $v\in V(G)$, and  parameters $0<\eps<1$, $1\leq \tau\leq \wmax$, and $1\leq \lambda\leq 2n^{\eps}$, together with a pair $(x,y)$ of distinct vertices of $G$ with $(x,y)\not\in E(G)$. The algorithm is also given the set $U'\subseteq \Vld$ of suspicious vertices, the set $B\subseteq V(G)$ of bad vertices, and, for every vertex $u\in\Vld$, the corresponding vertex set $\sigma'(u)$. The algorithm returns a set $A\subseteq V(G)$ of at most $O(dn^{2\eps}\lambda\cdot \log n)$ vertices, where $d$ is the average vertex degree in $G$. It ensures that, if there exists a global minimum vertex-cut $(L,S,R)$ for which Condition (C'3) holds, then, with  probability at least $\half$, $L\subseteq A$. The running time of the algorithm is $O(dn^{2\eps}\lambda\cdot \log n)$.
\end{lemma}

We prove \Cref{lem: Case 2 set A} below, after we complete the proof of \Cref{lem: Case 3} using it. 
We use the algorithm from \Cref{lem: Case 2 set A} to compute a set $A$ of at most $O(dn^{2\eps}\lambda\cdot \log n)$ vertices of $G$.
Let $\hat \event$ be the bad event that $L\not\subseteq A$. From \Cref{lem: Case 2 set A}, if Condition (C'3) holds, then $\prob{\hat \event}\leq \half$.
We now provide an algorithm for computing the sets $L',S'\subseteq V(G)$ of vertices,  such that, if we denote by $R'=V(G)\setminus (L'\cup S')$, then $(L',S',R')$ is a vertex-cut in $G$. We will ensure that, if Condition (C'3) holds and Event $\hat \event$ did not happen, then $(L',S',R')$ is a global minimum vertex-cut in $G$. The remainder of the algorithm is essentially identical to the algorithm that was used in Case 1 in order to compute the vertex sets $L',S'$ from the set $A$ of vertices.

 Note that, if $x\not\in A$, then either  Condition (C'3) does not hold, or  Event $\hat \event$ has happened. Therefore, if  $x\not\in A$, then we return $L'=\set{v}$ and $S'=N^+_G(v)$, where $v$ is a lowest-degree vertex of $G$ and terminate the algorithm. We assume from now on that $x\in A$. If $y\in A$ holds, then we delete $y$ from $A$; if Condition (C3) holds and Event $\hat \event$ did not happen, then $L\subseteq A$ still continues to hold. Additionally, if some vertex $v\in N^-_G(y)$ lies in $A$, then we delete $v$ from $A$. As before, if Condition (C3) holds and Event $\hat \event$ did not happen, then $L\subseteq A$ still continues to hold (since $y\in R$, so it cannot have in-neighbors that lie in $L$). Since $(x,y)\not \in E(G)$, $x\in A$ continues to hold. Therefore, we assume from now on that $x\in A$, $y\not\in A$ and $y\not\in N_G^+(A)$. Lastly, we discard from $A$ all high-degree vertices, as these vertices may not lie in $L$ from Property \ref{prop: small volume of L}.


Next, as in Case 1, we construct the graph $G^{|A}$ that is derived from $G$ via $A$ (see \Cref{def: defined graph}).
Since every vertex of $A$ is a low-degree vertex, from \Cref{claim: properties of defined graph}, $|V(G^{|A})|+|E(G^{|A})|\leq O\left(|A|\cdot d\cdot n^{\eps}\right)\leq O\left(d^2n^{3\eps}\lambda\cdot \log n\right)$ must hold, and, moreover, graph $G^{|A}$ can be constructed in time 
$O\left(|A|\cdot d\cdot n^{\eps}\right)\leq O\left(d^2n^{3\eps}\lambda\cdot \log n\right)$.
Our next step is to compute a minimum $x$-$t$ vertex-cut $(\hat L,\hat S,\hat R)$ in $G^{|A}$ using the algorithm from \Cref{cor: min_vertex_cut}, in time $O\left(|E(G^{|A})|^{1+o(1)}\cdot \log(\wmax)\right )\leq O\left (d^2n^{3\eps+o(1)}\lambda\cdot \log (\wmax)\right )$. (Note that $(x,t) \notin E(G^{|A})$, since $x \in A$.)
We then output the vertex sets $L'=\hat L$ and $S'=\hat S$ as the algorithm's outcome.
Observe that, by \Cref{claim: transforming cuts from defined graph}, $(L',S',R')$ is a valid vertex-cut in $G$, where $R'$ denotes the set $R' = V(G) \setminus (L' \cup S')$.

We now bound the running time of the algorithm. The running time of the algorithm from \Cref{lem: Case 2 set A} is $O(dn^{2\eps}\lambda\cdot \log n)$, and the time required to construct the graph $\hat G$ and to compute the minimum $x$-$t$ cut in it is  $O\left (d^2n^{3\eps+o(1)}\lambda\cdot \log (\wmax)\right )$. Overall, the running time of the algorithm is $O\left (d^2n^{3\eps+o(1)}\lambda\cdot \log (\wmax)\right )$.

In order to complete the proof of \Cref{lem: Case 3}, it is now enough to prove \Cref{lem: Case 2 set A}, which we do next.

\subsubsection{Proof of \Cref{lem: Case 2 set A}}

We assume that we are given a directed $n$-vertex and $m$-edge graph $G$ with integral weights $w(v)\geq 1$ on its vertices $v\in V(G)$, and  parameters $0<\eps<1$, $1\leq \tau\leq \wmax$, and $1\leq \lambda\leq 2n^{\eps}$, together with a pair $(x,y)$ of distinct vertices of $G$ with $(x,y)\not\in E(G)$. We also assume that we are given the set $U'\subseteq \Vld$ of suspicious vertices, the set $B\subseteq V(G)$ of bad vertices, and, for every vertex $u\in\Vld$, the corresponding vertex set $\sigma'(u)$.

As before, for the sake of analysis, we designate a global minimum vertex-cut $(L,S,R)$ to be the \emph{distinguished min-cut} as follows. If there exists a global minimum vertex-cut for which Condition (C'3) holds, then we let $(L,S,R)$ be this cut, and we say that the distinguished min-cut is \emph{good}; otherwise, we let $(L,S,R)$ be an arbitrary global minimum vertex-cut, and we say that it is \emph{bad}. In the following, whenever we say that Condition (C'3) holds, we mean that it holds with respect to the distinguished min-cut $(L,S,R)$.
 Notice that, if $x\in B$, then, from \Cref{obs: bad vertex}, $x\not\in L$, and so Condition (C'3) may not hold. In this case, we simply return $A=\emptyset$. We assume from now on that $x\not\in B$.

The key to our algorithm for Case 2 is to define a property of the vertices of $L$ that is easy to compute, but such that only relatively few vertices of $G$ have it. In order to do so, we define, for every vertex $v\in V(G)\setminus B$, a subset $J(v)\subseteq V(G)$ of vertices. We then show that, if Condition (C'3) holds, then, for every vertex $v\in L$, the symmetric difference $J(x)\triangle J(v)$ must be small. This, in turn, will allow us to compute the desired set $A$ of vertices.
Formally, for every vertex $v\in V(G)\setminus B$, we define:

\[J(v)=\left(\bigcup_{u\in \beta'(v)}N^+_G(u)\right )\cap \Vld\cap V^{\geq \tau'}.\] 

In other words, $J(v)$ contains all out-neighbors of the vertices in $\beta'(v)$ that are low-degree vertices of weight at least $\tau'$.
The following observation is central to our algorithm and provides the motivation for the definition of the vertex sets $J(v)$.

\begin{observation}\label{obs: aux for case 2}
	Assume that Condition (C'3) holds. Then  for every vertex $v\in L$, $|\hat S\setminus J(v)|\leq  64\lambda^3$.
\end{observation}
\begin{proof}
	Assume that Condition (C'3) holds, and consider any vertex $v\in L$. Recall that, from \Cref{obs: bad vertex}, $v\not\in B$ must hold, so vertex set $J(v)$ is well defined. 
	
	From the properties of the global minimum vertex-cut, there is a flow $f$ of value $w(S)=\opt$ in $G$ from $v$ to the vertex $y$ that is given to the algorithm as input; since we assumed that Condition (C'3) holds, $y\in R$. Let $\pset$ be the flow-path decomposition of this flow. Then every path $P\in \pset$ contains exactly one vertex of $S$, that we denote by $u(P)$. For every vertex $u\in S$, let $\pset_u\subseteq \pset$ contain all paths $P\in \pset$ with $u(P)=u$. Then $\sum_{P\in \pset_u}f(P)=w(u)$, and moreover, for every path $P\in \pset_u$, all vertices preceding $u$ on $P$ lie in $L$.
	
	We denote by $\hat S'=\hat S\setminus J(v)$, and we let $\pset'=\bigcup_{u\in \hat S'}\pset(u)$. We claim that every path $P\in \pset'$ must contain a vertex of $L\setminus\set{v}$, whose weight is at most $\tau$. Indeed, consider any path $P\in \pset'$, and assume for contradiction that it does not contain any such vertex. Let $P'$ be the subpath of $P$ from $v$ to $u(P)$, and let $v'$ be the penultimate vertex on this path. Finally, let $P''$ be the subpath of $P'$ from $v$ to $v'$. Then all vertices of $P''$ lie in $L$, and all vertices of $V(P'')\setminus\set{v}$ have weight greater than $\tau$. 
	Since $v\not\in B$, $|\beta'(v)|\leq \lambda$, so $\beta'(v)$ must contain all vertices of $P''$. Since $u(P)\in \hat S$, we get that $u(P)\in \Vld\cap V^{\geq \tau'}$, and so $u(P)\in J(v)$ must hold; a contradiction.
	We conclude that every path $P\in \pset'$ must contain a vertex of $L\setminus\set{v}$, whose weight is at most $\tau$. Therefore, 
	$\sum_{P\in \pset'}f(P)\leq |L|\cdot \tau\leq 64\lambda^3\cdot \tau'$, since $\tau'=\frac{\tau}{64\lambda^2}$ and $|L|\leq \lambda$.

	Since $\sum_{P\in \pset'}f(P)=w(\hat S')$ must hold, we conclude that $w(\hat S')\leq 64\lambda^3\cdot \tau'$.
	Since every vertex in $\hat S'$ has weight at least $\tau'$, we get that $|\hat S'|\leq 64\lambda^3$.
\end{proof}

Consider now the following set of vertices:

\[A'=\set{v\in V(G)\setminus B \mid |J(v)\triangle J(x)|\leq 256\lambda^3}. \]

We show that, if Condition (C'3) holds, then $L\subseteq A'$, and, for every vertex $v\in A'$, $|J(v)\cap \hat S|\geq \frac{3|\hat S|}{4}$.
	
\begin{claim}\label{claim: central for case 2}
	Assume that Condition (C'3) holds. Then $L\subseteq A'$, and for every vertex $v\in A'$, $|J(v)\cap \hat S|\geq \frac{3|\hat S|}{4}$. 
\end{claim}

\begin{proof}
Assume that Condition (C'3) holds.	
We start by proving that $L\subseteq A'$. Specifically, we will prove that, for every pair $u,v\in L$ of distinct vertices, $|J(v)\setminus J(u)|\leq 128\lambda^3$. It will then follow that, for every vertex $a\in L$, $|J(x)\triangle J(a)|=|J(x)\setminus J(a)|+|J(a)\setminus J(x)|\leq 256\lambda^3$, and so $L\subseteq A'$.

We now fix any pair $(u,v)$ of distinct vertices of $L$ and prove that $|J(u)\setminus J(v)|\leq 128\lambda^3$. 
Recall that, from \Cref{obs: bad vertex}, $\beta'(u)\subseteq L$ holds. Therefore, if $a$ is a vertex of $J(u)$, then $a\in (L\cup S)\cap \Vld\cap V^{\geq \tau'}\subseteq L\cup \hat S$. 
We conclude that, if $a\in J(u)\setminus J(v)$, then either (i) $a\in L$; or (ii) $a\in \hat S\setminus J(v)$. However, from Condition (C'3), $|L|\leq \lambda$ holds, and from \Cref{obs: aux for case 2}, $|\hat S\setminus J(v)|\leq 64\lambda^3$. Altogether, we get that:

\[|J(u)\setminus J(v)|\leq |L|+|\hat S\setminus J(v)|\leq 128\lambda^3.\]

Consider now any vertex $v\in A'$.
Since $x\in L$, from \Cref{obs: aux for case 2}, 
$|\hat S\setminus J(x)|\leq 64\lambda^3$. Since $|J(v)\triangle J(x)|\leq 256\lambda^3$, while $|\hat S|>2^{12}\lambda^3$, we get that at least $\frac{3|\hat S|}4$ vertices of $\hat S$ lie in $J(v)$.
\end{proof}

\paragraph{The Algorithm.}
We are now ready to describe our algorithm. 
We compute the set $J(x)$ of vertices by inspecting the out-neighbors of every vertex $v\in \beta'(x)$. Since each such vertex  $v$ is a low-degree vertex, set $J(x)$ can be computed in time $O(d\cdot n^{\eps}\cdot |\beta'(x)|)\leq O(d\cdot n^{\eps}\cdot \lambda)$.
Next, we set $q=64\ceil{\log n}$, and select $q$ vertices $v_1,\ldots,v_q\in J(x)$ uniformly independently at random with repetitions. 

We set $A=\emptyset$, and then perform $q$ iterations. For all $1\leq i\leq q$, in order the execute the $i$th iteration, we start by adding to $A$ all vertices in $N^-_G(v_i)$. Then we consider every vertex $a\in N^-_G(v_i)\cap \Vld\cap V^{\geq \tau'}$ that is not suspicious. For each such vertex $a$, we add all vertices of $\sigma'(a)$ to the set $A$. 
We then return the set $A$ obtained at the end of iteration $q$.

Recall that every vertex in $J(x)$ is low-degree,
so $|N^-_G(v_i)|\leq n^{\eps}\cdot d$. Moreover, if a vertex $a\in N^-_G(v_i)\cap \Vld\cap V^{\geq \tau'}$ is not suspicious, then $|\sigma'(a)|\leq 2n^{\eps}\cdot \lambda$. Therefore, 
the running time of a single iteration is $O(dn^{2\eps}\lambda)$, and the number of vertices added to $A$ in a single iteration is also bounded by $O(dn^{2\eps}\lambda)$. Since $q\leq O(\log n)$,
The total running time of the algorithm is  bounded by $O(dn^{2\eps}\lambda\cdot \log n)$, and $|A|\leq O(dn^{2\eps}\lambda\cdot \log n)$ holds.

We use the following claim to finish the analysis of the algorithm.

\begin{claim}\label{claim: catch all2}
	Assume that Condition (C'3) holds. Then with probability at least $0.99$, $A'\subseteq A$.
\end{claim}
\begin{proof}
Assume that Condition (C'3) holds.
For every vertex $u\in A'$, we define a bad event $\hat \event_u$ that $u\not\in A$, and we let $\hat \event$ be the bad event that $A'\not\subseteq A$. Clearly, $\hat \event$ may only happen if $\hat \event_u$ happens for some $u\in A'$. We now fix a vertex $u\in A'$ and bound $\prob{\hat \event_u}$.

	Consider the vertices $v_1,\ldots,v_q$ that were randomly selected from $J(x)$ by the algorithm, and assume that, for some $1\leq i\leq q$, $v_i\in J(u)$. Then there is a vertex $u'\in \beta'(u)$ with edge $(u',v_i)\in E(G)$. 
	If $u'=u$, then $u$ was added to $A$, since all vertices in $N^-_G(v_i)$ are added to $A$.
	Otherwise, from the definition of set $A'$, $u\not\in B$ and so $u'$ may not be a suspicious vertex. Moreover, from the definition of the set $\beta'(u)$, $u'$ is a low-degree vertex whose weight is greater than $\tau$. In particular, $u'\in  N^-_G(v_i)\cap \Vld\cap V^{\geq \tau'}$ and it is not suspicious, so all vertices in $\sigma'(u')$ were added to $A$, including the vertex $u$. We conclude that Event $\hat \event_u$ may only happen if no vertices of $J(u)\cap J(x)$ was chosen in the random trials. We now prove that $|J(u)\cap J(x)|\geq |J(x)|/2$, which, in turn, will imply that Event $\hat \event_u$ has a low probability.
	
	Indeed, from the definition of the set $A'$: 
	
	\[|J(x)\setminus J(u)|\leq 256\lambda^3.\]

	On the other hand, since $|\hat S|\geq 2^{12}\lambda^3$, and since, from \Cref{claim: central for case 2}, $|J(x)\cap \hat S|\geq \frac{3|\hat S|}{4}$, we get that:
	
	\[|J(x)|\geq \frac{3|\hat S|}{4}\geq 1024 \lambda^3.\]

	By combining the two inequalities, we get that $|J(x)\setminus J(v)|<\frac{|J(x)|}2$. The probability that none of the vertices $v_1,\ldots,v_q$ lie in $J(x)\cap J(u)$ is then bounded by $\frac{1}{2^q}\leq \frac{1}{n^2}$. We conclude that $\prob{\hat \event_u}\leq \frac{1}{n^2}$. By using the Union bound over all $u\in A'$, the claim follows.
\end{proof}

Assume that Condition (C'3) holds. Then, from \Cref{claim: central for case 2}, $L\subseteq A'$, and, from \Cref{claim: catch all2}, with probability at least $\half$, $A'\subseteq A$. We conclude that, with probability at least $\half$, $L\subseteq A$ holds.

\section{Algorithm for Weighted Directed Dense Graphs}
\label{sec: weighted dense}

In this section we provide an algorithm for weighted directed global minimum vertex-cut that is summarized in the following theorem.
\begin{theorem}\label{thm: dense weighted main theorem}
	There is a randomized algorithm, whose input consists of a simple directed $n$-vertex and $m$-edge graph $G$ with integral weights $1 \leq w(v) \leq W$ on its vertices $v \in V(G)$, such that $G$ contains some vertex-cut. The algorithm computes a vertex-cut $(L',S',R')$ in $G$, and it guarantees that, if there exists a global minimum vertex-cut $(L,S,R)$ in $G$ with $|L| \leq n^{1/2}$, then, with probability at least $\frac{1}{48\log n}$, $(L',S',R')$ is a global minimum vertex-cut.
	The running time of the algorithm is $O\left(n^{2+ (5-\omega)/(11-3\omega) +o(1)} \cdot (\log W)^{O(1)}\right)\leq O\left(n^{2.677} \cdot (\log W)^{O(1)}\right)$,
	where $\omega$ is the matrix multiplication exponent.
\end{theorem}

We provide the proof of \Cref{thm: dense weighted main theorem} below, after we complete the proof of \Cref{thm: main: weighted} using it.

\subsection{Completing the Proof of  \Cref{thm: main: weighted}}
Recall that we are given as input a simple directed $n$-vertex and $m$-edge graph $G$ with integral weights $0 \leq w(v) \leq W$ on its vertices $v\in V(G)$, such that $G$ contains some vertex-cut. Let $d=\frac{2m}{n}$ denote the average vertex degree in $G$. Our goal is to compute a vertex-cut $(L',S',R')$ in $G$, so that, with probability at least $1/2$, $(L',S',R')$ is a global minimum vertex-cut in $G$. 
We design a randomized algorithm with these properties, whose running time is bounded by: 

$$O\left(\min\left\{mn^{11/12+o(1)}d^{1/12},n^{2+(5-\omega)/(11-3\omega)+o(1)}\right\}(\log W)^{O(1)}\right),$$ 

where $\omega$ is the matrix multiplication exponent.
Using the current bound $\omega \leq 2.371552$~\cite{WXXZ24}, we get that this running time is bounded by $O\left(\min\set{m n^{11/12+o(1)} \cdot d^{1/12}, n^{2.677}}\cdot (\log W)^{O(1)}\right )$, as required. 

Our algorithm distinguishes between two cases. The first case happens if  $mn^{11/12}\cdot d^{1/12} \leq n^{2+(5-\omega)/(11-3\omega)}$. In this case we apply the algorithm from \Cref{thm : sparse weighted algorithm} to the graph $G$, and return the vertex-cut $(L',S',R')$ that the algorithm outputs. From \Cref{thm : sparse weighted algorithm}, with probability at least $\half$, $(L',S',R')$ is a global minimum vertex-cut. Clearly, the running time of the algorithm in this case is bounded by:

$$O\left(m n^{11/12 + o(1)} \cdot d^{1/12} \cdot \log^2(W)\right)\leq O\left(\min\left\{mn^{11/12+o(1)}d^{1/12},n^{2+(5-\omega)/(11-3\omega)+o(1)}\right\}(\log W)^{O(1)}\right),$$

  as required.
From now on we assume that the second case happens, that is, $n^{2+(5-\omega)/(11-3\omega)}<mn^{11/12}\cdot d^{1/12}$. Notice that in this case, it is sufficient to design an algorithm whose running time is bounded by $O\left(n^{2+(5-\omega)/(11-3\omega)+o(1)} \cdot (\log W)^{O(1)}\right )$.

Using the transformation described in \Cref{subsec: positive weights}, it is sufficient to design an algorithm for the special case where all vertex weights are strictly positive (since this transformation inflates the vertex weights by $O(n^2)$ factor, this will only lead to an $O(\poly\log n)$ additional overhead in the running time). Therefore, we assume from now on that all vertex weights in the graph $G$ are strictly positive.

Consider the following algorithm: we apply the algorithm from  \Cref{thm: dense weighted main theorem} to the graph $G$, and denote by $(L',S',R')$ the vertex-cut that it returns.
We then apply the algorithm from \Cref{thm : alg for when L is large} to the graph $G$, with the parameter $\epsilon=1/2$, and denote by $(L'',S'',R'')$ the vertex-cut that it returns. We then return the smaller-value vertex-cut from among the cuts $(L',S',R')$ and $(L'',S'',R'')$.
Recall that the running time of the algorithm from \Cref{thm: dense weighted main theorem}  is $O\left(n^{2+ (5-\omega)/(11-3\omega) +o(1)} \cdot (\log W)^{O(1)}\right)$, while the running time of the algorithm from  \Cref{thm : alg for when L is large}
is  $O\left(m n^{1-\epsilon +o(1)} \log W\right) \leq O\left(m n^{1/2+o(1)} \log W\right) \leq O\left(n^{5/2+o(1)} \log W\right)$; notice that this expression is bounded by $O\left(n^{2+ (5-\omega)/(11-3\omega) +o(1)} \cdot (\log W)^{O(1)}\right)$ for any possible value of $2 \leq \omega \leq 3$.

For the sake of analysis, we fix an arbitrary global minimum vertex-cut $(L,S,R)$, that we refer to as the \emph{distinguished cut}. From \Cref{thm: dense weighted main theorem}, if  $|L| \leq n^{1/2}$, then the cut $(L',S',R')$ is guaranteed to be a global minimum vertex-cut with probability at least $\frac{1}{48 \log n}$.
Otherwise, from \Cref{thm : alg for when L is large}, the cut $(L'',S'',R'')$ is guaranteed to be a global minimum vertex-cut with probability at least  $\left(1-1/n^2\right)$. Therefore, in either case, with probability at least $\frac{1}{48 \log n}$, the cut that the algorithm returns is a global minimum vertex-cut. By repeating this algorithm $O(\log n)$ times and returning the smallest-value vertex-cut among the resulting cuts, we can guarantee that, with probability at least $\half$, the cut that the algorithm returns is a global minimum vertex-cut. The total running time of the algorithm for Case 2 remains bounded by:

\[O\left(n^{2+ (5-\omega)/(11-3\omega) +o(1)} \cdot (\log W)^{O(1)}\right),\]

and the total running time of the algorithm is bounded by:
	$$O\left(\min\set{m n^{11/12+o(1)}\cdot d^{1/12}, n^{2.677}}(\log W)^{O(1)}\right ) \leq O\left(m n^{0.976}(\log W)^{O(1)}\right).$$

In order to complete the proof of \Cref{thm: main: weighted}, it is now enough to prove \Cref{thm: dense weighted main theorem}, which we in the remainder of this section.

\subsection{Proof of \Cref{thm: dense weighted main theorem}} \label{sec: proof of dense weighted main}

Our proof follows the high-level approach from \cite{CT24}, but we extend their algorithm to directed graphs, while also significantly simplifying it and achieving a faster running time.

Recall that the algorithm receives as input a simple directed $n$-vertex and $m$-edge graph $G$ with integral weights $1 \leq w(v) \leq W$ on its vertices $v \in V(G)$, such that $G$ contains some vertex-cut.
We can assume that $n$ is greater than a large enough constant, since otherwise we can use, for example, the algorithm from \Cref{thm : sparse weighted algorithm}, whose running time is bounded by $O(n^3\log^2W)\leq O(\log^2W)$. We can also assume without loss of generality that $W=\max_{v\in V(G)}\set{w(v)}$, since otherwise we can reduce the value of $W$ to be equal to $\max_{v\in V(G)}\set{w(v)}$, and it is sufficient to prove \Cref{thm: dense weighted main theorem} with this new value of the parameter $W$.
For the sake of the analysis, we fix a global minimum vertex-cut $(L,S,R)$, that we refer to as the \emph{distinguished cut}, as follows. If there exists a global minimum vertex-cut $(L^*,S^*,R^*)$ with $|L^*| \leq n^{1/2}$, then  we let $(L,S,R)$ be any such cut, and we say that the distinguished cut is \emph{good}.
Otherwise, we let $(L,S,R)$ be any global minimum vertex-cut, and we say that the distinguished cut is \emph{bad}.

Recall that our goal is to compute a vertex-cut $(L',S',R')$ in $G$, such that, if the distinguished cut is good, then,  with probability at least $\frac{1}{48\log n}$, $(L',S',R')$ is a global minimum vertex-cut.

At a very high level, our algorithm follows a rather standard approach.
We start by employing a simple randomized algorithm to select a value $1\leq\lambda \leq n$ and a collection $\Gamma$ of pairs of vertices of $G$. Intuitively, we would like to ensure that $\lambda$ is close to $|L|$ (where $(L,S,R)$ is the distinguished cut), and that there is some pair $(x,y)\in \Gamma$ with $x\in L$ and $y\in R$. We would also like to ensure that $|\Gamma|$ is not too large.
We will then compute, for every pair $(x,y)\in\Gamma$, an $x$-$y$ vertex-cut $(L_{x,y},S_{x,y},R_{x,y})$ in $G$, so that, if $|L|\leq \lambda$, $x\in L$ and $y\in R$ hold, then, with a sufficiently high probability, the corresponding vertex-cut $(L_{x,y},S_{x,y},R_{x,y})$ is a global minimum vertex-cut. We note that a similar high-level approach has been used in numerous prior works. The main difficulty with this approach is in designing an  algorithm that computes, for every pair $(x,y)\in \Gamma$ of vertices, the vertex-cut $(L_{x,y},S_{x,y},R_{x,y})$ with the desired properties efficiently; we note that we cannot afford to spend $\Theta(m)$ time per pair on such an algorithm.

We start with the following simple algorithm for computing the parameter $\lambda$ and the set $\Gamma$ of pairs of vertices of $G$.
The claim is similar to \Cref{claim: compute pairs2}, and its proof appears in Section \ref{subsec: proof of select pairs 2} of  Appendix.

\begin{claim}\label{cl: selecting pairs}
	There is a randomized algorithm, that we call $\algpairs'$, whose input consists of an $n$-vertex directed graph $G$ with integral weights $w(v) \geq 1$ on its vertices $v \in V(G)$. The algorithm returns 
	a value $1\le \lambda\leq n$ that is an integral power of $2$, together with a set $\Gamma$ of at most $\ceil{\frac{100n\log n}{\lambda}}$ pairs of vertices of $G$. We say that the algorithm is \emph{successful} with respect to a fixed global minimum vertex-cut $(L,S,R)$, if  $\frac{\lambda}{2}\leq |L|\leq\lambda$, and  there is a pair $(x,y)\in \Gamma$ with $x\in L$ and $y\in R$.  Let $(L,S,R)$ be any global minimum vertex-cut in $G$.
	The running time of the algorithm is $\tilde O\left(n^2\right)$, and it is successful with respect to $(L,S,R)$ with probability at least $\frac{1}{12\log n}$. 
\end{claim}

We apply Algorithm  $\algpairs'$ from \Cref{cl: selecting pairs} to graph $G$, and we let $\event$ be the good event that the algorithm is successful with respect to the distinguished cut $(L,S,R)$. From \Cref{cl: selecting pairs}, $\prob{\event}\geq \frac{1}{12\log n}$.

 If $\lambda >2n^{1/2}$, then  we terminate the algorithm and output an arbitrary vertex-cut. In this case, we say that the algorithm \emph{terminates early}. Note that, if the distinguished cut is good, then we may only terminate the algorithm early in this step if Event $\event$ did not happen. From now on we assume that the algorithm did not terminate early. Let $\Gamma'\subseteq\Gamma$ be a collection of pairs of vertices obtained from $\Gamma$ by discarding all pairs $(x,y)$ where either $x=y$ or $(x,y)\in E(G)$ hold. Note that, if Event $\event$ happened, then there is a pair $(x^*,y^*)\in \Gamma$ of vertices  with $x^*\in L$ and $y^*\in R$; in particular, $x^*\neq y^*$ and $(x^*,y^*)\not \in E(G)$ must hold, so $(x^*,y^*)\in \Gamma'$ holds as well. 
Therefore, if $\Gamma'=\emptyset$, then it must be the case that Event $\event$ did not happen. If $\Gamma'=\emptyset$, we also terminate the algorithm and return an arbitrary vertex-cut; we say that the algorithm terminates early in this case as well.
In the remainder of the algorithm, for brevity, we denote $\Gamma'$ by $\Gamma$. 
For the sake of the analysis, we designate a single pair $(x^*,y^*)\in \Gamma$ of vertices as the \emph{distinguished pair}, as follows. If Event $\event$ happened, then we let $(x^*,y^*)$ be any pair of vertices with $x^*\in L$ and $y^*\in R$ (such a pair must exist from our discussion); in this case, we say that the distinguished pair is \emph{good}. Otherwise, we let $(x^*,y^*)$ be any pair of vertices in $\Gamma$, and we say that the distinguished pair is \emph{bad}.

The main technical ingredient of our algorithm is a procedure that is summarized in the following theorem. Given a pair $(x,y)$ of vertices in the graph $G$, the procedure computes a vertex-cut $(L_{x,y},S_{x,y},R_{x,y})$, so that, if Event $\event$ happened and $(x,y)$ is the distinguished pair in $\Gamma$, then $(L_{x,y},S_{x,y},R_{x,y})$ is guaranteed to be a global minimum vertex-cut with a sufficiently high probability.
The subroutine is given a query access to the $(\delta,\tau)$-subgraph oracle on the input graph $G$ for appropriate values of $\delta$ and $\tau$ (see \Cref{def: subgraph oracle}).

\begin{theorem}\label{thm : main-subroutine}
	 There is a randomized algorithm, whose input consists of a simple directed $n$-vertex graph $G$ with integral weights $1 \leq w(v) \leq W$ on its vertices $v \in V(G)$ in the adjacency-list representation, together with the adjacency-list representation of the split-graph $G'$ of $G$. Additionally, the algorithm is given as input
	  integers $\lambda,\gamma \geq 1$ with $\lambda\cdot\gamma \leq \frac{n}{100}$, and a pair $x,y \in V(G)$ of distinct vertices, such that $(x,y) \notin E(G)$. Lastly, the algorithm is given access to the $(\delta,20\wmax(G))$-subgraph oracle for $G$, for $\delta=1-\frac{\log(30000\lambda\gamma \cdot \log n)}{\log n}$, so that $n^{1-\delta}=30000\lambda\gamma \cdot \log n$. The algorithm returns an $x$-$y$ vertex-cut $(L',S',R')$ in $G$ with the following guarantee: if there exists a minimum $x$-$y$ vertex-cut $(L,S,R)$ in $G$ with $|L| \leq \lambda$, and if the subgraph oracle never erred in its responses to queries, then $(L',S',R')$ is a minimum $x$-$y$ vertex-cut in $G$. The expected running time of the algorithm is $O\left(\left (\frac{n^{2+o(1)}\cdot \lambda}{\gamma} + n^{1+o(1)}\cdot \gamma\cdot \lambda\right ) \cdot \log^5W\right)$, and it accesses the $(\delta,20\wmax(G))$-subgraph oracle at most $\log (80n(W+1))$ times.
\end{theorem}

We note that the algorithm from \Cref{thm : main-subroutine} may modify the adjacency-list representations of the graphs $G$ and $G'$ that it is given as input.
We prove \Cref{thm : main-subroutine} in \Cref{sec : main-subroutine}.
At a high level, \Cref{thm : main-subroutine} can be viewed as an adaptation of Theorem 3.16 from \cite{CT24} to directed graphs, that additionally obtains a better tradeoff between the running time and the parameter $\delta$ of the subgraph oracle.
Additionally, the algorithm of \cite{CT24} only works in the setting where there exists a minimum $x$-$y$ vertex-cut $(L,S,R)$ with $|L| \leq \frac{|S|}{n^{44/45}} \leq n^{1/45}$, while our algorithm works for much higher cardinalities of the set $L$, albeit with running time that depends linearly on $|L|$.
While it is generally not difficult to extend the algorithm  of \cite{CT24} to the directed setting, we both significantly simplify the algorithm and its analysis, and obtain a faster running time, while also extending the algorithm to the setting where $|L|$ may be large. The latter is crucial in order to obtain the running time guarantees of \Cref{thm: main: weighted}.
In the remainder of this subsection, we complete the proof of \Cref{thm: dense weighted main theorem} using \Cref{thm : main-subroutine}. This part of our algorithm follows a standard approach that is similar to that used by \cite{CT24}.

\paragraph{Completing the proof of \Cref{thm: dense weighted main theorem}.}
After computing the parameter $\lambda$ and the set $\Gamma$ of vertex pairs as described above, we compute the adjacency-list representations of the input graph $G$ and its split-graph $G'$. 
Clearly, all these can be done in time $\tilde O(n^2)$. In the remainder of the description of the algorithm we assume that it did not terminate early.
The remainder of the algorithm consists of at most $z=\floor{\log (80n(W+1))}+1$ phases.

We now fix an integer $1 \leq i \leq z$ and describe the execution of Phase $i$. 
At a very high level, our algorithm will apply the algorithm from \Cref{thm : main-subroutine}  to every pair $(x,y)\in \Gamma$ of vertices one by one, while recording all random choices that the algorithm makes, until it performs the $i$th call to the subgraph oracle; we denote the corresponding oracle query by $Z_{i,(x,y)}$. Once it collects the $i$th query $Z_{i,(x,y)}$ for the subgraph oracle for all pairs $(x,y)\in \Gamma$, it calls to the algorithm from \Cref{lem: oracle in bulk} in order to compute responses to all these queries at bulk. For every pair $(x,y)\in \Gamma$, when we execute the algorithm from \Cref{thm : main-subroutine} in the subsequent phase, we repeat all random choices and reuse all responses to oracle queries computed in previous phases (it is easy to verify that the first $i$ oracle queries $Z_{1,(x,y)},\ldots,Z_{i,(x,y)}$ that the algorithm asks will then remain unchanged). Therefore, when applying the algorithm from  \Cref{thm : main-subroutine}  to pair $(x,y)\in \Gamma$ in Phase $i$, we can assume that we have already computed the responses to the first $i-1$ queries to the subgraph oracle that the algorithm makes, and can therefore execute the algorithm until it makes its $i$th call to the oracle. We now describe the execution of the $i$th phase more formally.

We process every pair $(x,y)\in \Gamma$ of vertices one by one. 
When pair $(x,y)$ is processed, we apply the algorithm from \Cref{thm : main-subroutine} to graph $G$, the the pair $(x,y)$ of its vertices, the parameter $\lambda$ computed above, and the parameter $\gamma$ that is set to $\gamma=\ceil{n^{(6-2\omega)/(11-3\omega)}}$, where $\omega$ is the matrix multiplication exponent. Note that, using the current bounds $2\leq \omega\leq 2.371552$~\cite{WXXZ24}, and by the assumption that $n$ is greater than a large enough constant, $n^{1/4} \leq \gamma \leq n^{2/5}$ must hold.
Moreover, since we assumed that the algorithm did not terminate early, $\lambda \leq 2n^{1/2}$, and, since we assumed that $n$ is greater than a large enough constant, $\lambda\cdot \gamma\leq 2n^{9/10}\leq \frac{n}{100}$ must hold.

We execute the algorithm from \Cref{thm : main-subroutine}  until its $i$th call to the subgraph oracle, while recording all random choices that it makes. Once the algorithm makes the $i$th call to the subgraph oracle, we terminate it, and we denote by $Z_{i,(x,y)}$ its $i$th query to the oracle. (If the algorithm terminates before making $i$ queries to the subgraph oracle, then we instead let $Z_{i,(x,y)}=\emptyset$.) When the algorithm from \Cref{thm : main-subroutine} is subsequently executed with the pair $(x,y)$ of vertices in Phase $(i+1)$, we repeat all the recorded random choices that we used when processing this pair in Phase $i$, so that the first $i$ calls to the subgraph oracle are identical to the execution of the algorithm in Phase $i$. Lastly, once the algorithm from \Cref{thm : main-subroutine} terminates, we undo all the modifications that it made to the adjacency-list representations of $G$ and $G'$. We then continue to process the next pair of vertices of $\Gamma$.

Once all pairs of vertices in $\Gamma$ are processed in this manner, we execute the algorithm from \Cref{lem: oracle in bulk} with parameters $\delta=1-\frac{\log(30000\lambda\gamma \cdot \log n)}{\log n}$ and $\tau=20\wmax(G)$ in order to compute the responses to all queries $\set{Z_{i,(x,y)} \mid (x,y) \in \Gamma}$. 
Observe that, since we assumed that the algorithm did not terminate early, $\lambda \leq 2n^{1/2}$ must hold. 
Since, as observed already, $n^{1/4} \leq \gamma \leq n^{2/5}$, and since we assumed that $n$ is greater than a large enough constant, we get that $\lambda\cdot \gamma\leq 2n^{9/10}< \frac{n}{30000 \cdot 2^{\sqrt{\log n}} \cdot \log n}$ and $\log(30000\lambda\gamma \cdot \log n)<\log n-\sqrt{\log n}$ must hold. Therefore, 
$\delta=1-\frac{\log(30000\lambda\gamma \cdot \log n)}{\log n} > \frac{1}{\sqrt{\log n}}$. Additionally, since, as observed already, $\gamma\geq n^{1/4}$, we get that $\delta\leq \frac{3}{4}$ must hold as well. We conclude that we obtain a valid input to the algorithm from \Cref{lem: oracle in bulk}.
Once the algorithm from \Cref{lem: oracle in bulk} computes the responses to all queries $\set{Z_{i,(x,y)} \mid (x,y) \in \Gamma}$, the current phase terminates and we continue to Phase $(i+1)$.
When processing the vertex pair $(x,y)\in\Gamma$ in Phase $(i+1)$, we repeat all random choices made in Phase $i$. Note that at this time we have already computed responses to the first $i$ queries of this algorithm to the subgraph oracle, so we can now execute this algorithm until its $(i+1)$th query to the oracle. 

Finally, for every pair $(x,y)\in \Gamma$ of vertices, we denote by $(L_{x,y},S_{x,y},R_{x,y})$ the vertex-cut that the algorithm from 
 \Cref{thm : main-subroutine} returned when applied to the pair $(x,y)$ in the last phase. We let $(x',y')\in \Gamma$ be the pair for which $w(S_{x',y'})$ is minimized, and we return the corresponding vertex-cut $(L_{x',y'},S_{x',y'},R_{x',y'})$.
This concludes the description of the algorithm. We now proceed to analyze it. 
 
Let $\event^{\oracle}$ be the bad event that the algorithm from \Cref{lem: oracle in bulk} errs in its response to any query in any of the phases.
Recall that the number of phases in the algorithm is bounded by:

\[z=\floor{\log (80n(W+1))}+1\leq 4\log (80nW).\]

Recall also that we apply the subgraph oracle with the parameter $\tau=20\wmax(G) \geq 20n \wmax'(G)\geq 8nW$ (from the definitions of $\wmax(G)$ and $\wmax'(G)$ in \Cref{subsec: split graph}, and from our assumption that $W=\max_{v\in V(G)}\set{w(v)}$).

By combining \Cref{lem: oracle in bulk} with the Union Bound, we get that: 

\[\prob{\event^{\oracle}} \leq \frac{z}{n^9 \cdot \log^4 \tau}\leq  \frac{4\log(80nW)}{n^9\cdot \log^4(8nW)}\leq \frac{1}{64n}.\]

Let $\event^*$ be the good event that Event $\event$ happened and 
 Event $\event^{\oracle}$ did not happen. Recall that, from \Cref{cl: selecting pairs}, $\prob{\event}\geq \frac{1}{12\log n}$, and from our discussion above, $\prob{\event^{\oracle}} \leq  \frac{1}{64n}$. Altogether, we get that:
 
 \[\prob{\neg\event^*}\leq \prob{\neg \event}+\prob{\event^{\oracle}}\leq 1-\frac{1}{12\log n}+\frac{1}{64n}\leq 1-\frac{1}{24\log n},\]

and so $\prob{\event^*}\geq \frac{1}{24\log n}$. 
 We use the following observation to complete the correctness analysis of the algorithm.

\begin{observation}\label{obs: good outcome}
	Assume that the distinguished cut is good and that Event $\event^*$ has happened. Then the cut that the algorithm outputs is a global minimum vertex-cut.
\end{observation}
\begin{proof}
	Recall that if the distinguished cut is good and Event $\event$ has happened, then the algorithm may not terminate early, and moreover, the distinguished pair $(x^*,y^*)\in \Gamma$ of vertices is good; in other words, $x^*\in L$ and $y^*\in R$ holds. Therefore, the distinguished cut $(L,S,R)$ is also a minimum $x^*$-$y^*$ vertex-cut in $G$, and, since the distinguished cut is good, $|L|\leq \lambda$ holds. Since we assumed that Event $\event^{\oracle}$ did not happen, when the algorithm from \Cref{thm : main-subroutine} is applied to the pair  $(x^*,y^*)$ in the last phase, the cut $(L_{x^*,y^*},S_{x^*,y^*},R_{x^*,y^*})$ that it returns must be a minimum $x^*$-$y^*$ vertex-cut in $G$, so its value is $w(S_{x^*,y^*})=w(S)$. Since our algorithm returns a smallest-value vertex-cut from among the cuts in $\set{(L_{x,y},S_{x,y},R_{x,y})\mid (x,y)\in \Gamma}$, the value of the cut that it returns must be $w(S)$, and so it must be a global minimum vertex-cut.
\end{proof}

We conclude that, with probability at least $\frac{1}{24\log n}$, the vertex-cut that the algorithm returns is a global minimum vertex-cut. 

It now remains to bound the running time of the algorithm. As noted already, the time required to execute Algorithm $\algpairs'$, and to compute the adjacency-list representations of the graphs $G$ and $G'$ is bounded by $\tilde O(n^2)$. The number of phases in the remainder of the algorithm is $z\leq O(\log n\log W)$. In every phase, we perform  $|\Gamma|\leq \tilde O\left (\frac{n}{\lambda}\right )$  iterations, each of which executes the  algorithm from 
\Cref{thm : main-subroutine}, whose expected running time is  $O\left(\left (\frac{n^{2+o(1)}\cdot \lambda}{\gamma} + n^{1+o(1)}\cdot \gamma\cdot \lambda\right ) \cdot \log^5W\right)$. 
Therefore, the total expected running time of all applications of the algorithm from \Cref{thm : main-subroutine} is bounded by:

\[\begin{split}
&O\left(\left (\frac{n^{2+o(1)}\cdot \lambda}{\gamma} + n^{1+o(1)}\cdot \gamma\cdot \lambda\right )\cdot z\cdot |\Gamma| \cdot (\log W)^{O(1)}\right)\\
&\quad\quad\quad\quad \leq O\left(\left (\frac{n^{3}}{\gamma} + n^{2}\cdot \gamma\right )\cdot n^{o(1)} \cdot (\log W)^{O(1)}\right)\\
&\quad\quad\quad\quad \leq O\left(\left (n^{2+(5-\omega)/(11-3\omega)}+n^{2+(6-2\omega)/(11-3\omega)}\right )\cdot n^{o(1)} \cdot (\log W)^{O(1)}\right)\\
&\quad\quad\quad\quad \leq O\left(\left (n^{2+(5-\omega)/(11-3\omega)+o(1)}\right )\cdot (\log W)^{O(1)}\right).
\end{split}\]

(We have used the fact that  $\gamma=\Theta\left(n^{(6-2\omega)/(11-3\omega)}\right)=\Theta\left(n^{1-(5-\omega)/(11-3\omega)}\right)$).

Lastly, in every phase we also execute the algorithm from \Cref{lem: oracle in bulk}, whose running time is bounded in the following simple but technical claim, whose proof is deferred to Section \ref{subsec: runtime bound} of Appendix.

\begin{claim}\label{claim: bound time of oracle}
	The total running time that the algorithm spends, over all phases, on the algorithm from \Cref{lem: oracle in bulk} is bounded by  $O\left(n^{2+(5-\omega)/(11-3\omega)+o(1)}\cdot (\log W)^{O(1)}\right )$.
\end{claim}

Overall, the expected running time of the entire algorithm is bounded by
\begin{equation}\label{eq : expected running time of dense alg}
O\left(n^{2+(5-\omega)/(11-3\omega)+o(1)} \cdot (\log W)^{O(1)}\right).
\end{equation}
In order to achieve the worst-case running time bound stated in \Cref{thm: dense weighted main theorem}, we terminate the algorithm whenever its running time exceeds the expected running time by a factor of $48 \log n$, in which case we output an arbitrary vertex-cut in $G$. 
By Markov's inequality, the probability that the algorithm is terminated for this reason is at most $\frac{1}{48 \log n}$.
Overall, if the distinguished cut is good, then the probability that algorithm outputs a global minimum vertex-cut remains at least $\frac{1}{24\log n}-\frac{1}{48 \log n} \geq \frac{1}{48 \log n}$.
In order to complete the proof of \Cref{thm: dense weighted main theorem}, it is now enough to prove \Cref{thm : main-subroutine}, which we do next. 

\subsection{The Main Subroutine: Proof of \Cref{thm : main-subroutine}} \label{sec : main-subroutine}

In this section we provide the algorithm for our main subroutine, proving \Cref{thm : main-subroutine}.
 A key notion used in our proof is that of well-behaving cuts, that we define next.

\begin{definition}[A well-behaving cut]\label{def : well-behaving cut}
	Let $G$ be a directed $n$-vertex graph, and let $x,y\in V(G)$ be a pair of its distinct vertices. Let 
	$\lambda,\gamma\geq 1$ be integral parameters, and let $M'>0$ be another parameter. Given a pair $(L,S,R)$ and $(L',S',R')$ of $x$-$y$ vertex-cuts in $G$, we say that $(L',S',R')$ is \emph{$(\lambda,\gamma)$-well-behaving for scale $M'$ with respect to  $(L,S,R)$}, if all of the following hold:

	\begin{properties}{W}
		\item\label{prop : P1} $\vol^+_G(L') \leq \frac{100n^2}{\gamma}$;
		\item\label{prop : P2} $w(S') \leq w(S) + 10nM'$; and
		\item\label{prop : P3} set $(L \cup S) \setminus (L' \cup S')$ contains at most $20\lambda\gamma$ vertices of weight at least $M'$.
	\end{properties}
\end{definition}

Note that in the above definition it is possible that $M'<1$.
Note also that the notion of a well-behaving cut is somewhat similar to the notion of vertex-\AFMC (see \Cref{def : almost-minimum-furthest-vertex-cut}). The main difference is the specific setting of the parameters; that  Property \ref{prop : P3} only extends to vertices of weight at least $M'$ and is only required to hold for a specific vertex-cut $(L,S,R)$; and the additional Requirement 
\ref{prop : P1} that $L'$ has a small out-volume.


\paragraph{The Distinguished Cut.}
For the sake of the analysis, we fix a minimum $x$-$y$ vertex-cut $(L,S,R)$, that we refer to as the \emph{distinguished cut}, as follows: if there exists a minimum $x$-$y$ vertex-cut $(\hat L,\hat S,\hat R)$ with $|\hat L| \leq \lambda$, then let $(L,S,R)$ be any such vertex-cut, and we say that the distinguished cut is \emph{good}.
Otherwise, we let $(L,S,R)$ be any minimum $x$-$y$ vertex-cut, and we say that the distinguished cut is \emph{bad}.


The main idea of our algorithm is to iteratively construct $x$-$y$ vertex-cuts that are well-behaving with respect to the distinguished cut for increasingly smaller scales.

For an integer $i\geq 1$, let $M_i = \frac{\wmax'(G)}{2^i}$;  observe that $M_i$ is an integral power of $2$.
Our algorithm consists of $z=\ceil{\log(20n\wmax'(G))}$ phases.
 Intuitively, the $i$th phase will focus on obtaining an $x$-$y$ vertex-cut that is well-behaving with respect to the distinguished cut for the scale $M_i$ 

Specifically, for all $1\leq i\leq z$,
the input to Phase $i$ is an $x$-$y$ vertex-cut $(L_{i-1},S_{i-1},R_{i-1})$ in $G$ that has the following property: if the distinguished cut $(L,S,R)$ is good, then the cut $(L_{i-1},S_{i-1},R_{i-1})$ is $(\lambda,\gamma)$-well-behaving  for scale $M_{i-1}$ with respect to $(L,S,R)$. The output of Phase $i$ is an $x$-$y$ vertex-cut $(L_i,S_i,R_i)$, that can serve as a valid input to Phase $(i+1)$. In other words, it must have the following property: if the distinguished cut $(L,S,R)$ is good, then $(L_i,S_i,R_i)$ is $(\lambda,\gamma)$-well-behaving  for scale $M_{i-1}$ with respect to $(L,S,R)$.
The algorithm for a single phase is summarized in the following theorem.

\begin{theorem}\label{thm : phase-i-of-main-subroutine}
	There is a randomized algorithm, whose input consists of a simple directed $n$-vertex graph $G$ with integral weights $1 \leq w(v) \leq W$ on its vertices $v \in V(G)$ in the adjacency-list representation, together with the adjacency-list representation of the split-graph $G'$ of $G$. Additionally, the algorithm is given as input integral parameters $\lambda,\gamma \geq 1$ with $\lambda\cdot\gamma \leq \frac{n}{100}$, a pair $x,y \in V(G)$ of distinct vertices such that $(x,y) \notin E(G)$, a scale parameter $\frac{1}{100n} \leq M' \leq \frac{\wmax'(G)}{2}$ that is an integral power of $2$,
	and an $x$-$y$ vertex-cut $(L',S',R')$ in $G$.
	Lastly, the algorithm is allowed to make a single query to the $(\delta,20\wmax(G))$-subgraph oracle for $G$, for $\delta=1-\frac{\log(30000\lambda\gamma \cdot \log n)}{\log n}$, so that $n^{1-\delta}=30000\lambda\gamma \cdot \log n$.
	The algorithm returns an $x$-$y$ vertex-cut $(L'',S'',R'')$  with the following guarantee: if there exists a minimum $x$-$y$ vertex-cut $(L,S,R)$ with $|L| \leq \lambda$ such that the cut $(L',S',R')$ is $(\lambda,\gamma)$-well-behaving  for scale $(2M')$ with respect to $(L,S,R)$, and if the subgraph oracle did not err in its response to the query,  then cut $(L'',S'',R'')$ is $(\lambda,\gamma)$-well-behaving for scale $M'$  with respect to $(L,S,R)$.
	The expected running time of the algorithm is $O\left (\left (\frac{n^{2+o(1)}\cdot \lambda}{\gamma} + n^{1+o(1)}\cdot \gamma\cdot \lambda\right ) \cdot \log^4W\right )$.
\end{theorem}

We note that the algorithm from  \Cref{thm : phase-i-of-main-subroutine}  may modify the adjacency-list representations of the graphs $G$ and $G'$ that it is given as input.
We provide the proof of \Cref{thm : phase-i-of-main-subroutine} in \Cref{sec : phase-i-of-main-subroutine}, after we complete the proof of \Cref{thm : main-subroutine} using it.
Our algorithm starts by computing an initial $x$-$y$ vertex-cut $(L_0,S_0,R_0)$ with $L_0=\{x\}$, $S_0=N^+_G(x)$, and $R_0 = V(G) \setminus (L_0 \cup S_0)$; we later show that, if the distinguished cut is good, then $(L_0,S_0,R_0)$ is $(\lambda,\gamma)$-well-behaving for scale $M_{0}$ with respect to it.
We then perform $z=\ceil{\log(20n\wmax'(G))}$ phases. For all $1\leq i\leq z$, we use the algorithm from \Cref{thm : phase-i-of-main-subroutine} with scale parameter $M'=M_i$ and cut $(L',S',R')=(L_{i-1},S_{i-1},R_{i-1})$. We denote the cut that the algorithm from \Cref{thm : phase-i-of-main-subroutine} outputs by $(L_i,S_i,R_i)$. 
Recall that the algorithm from  \Cref{thm : phase-i-of-main-subroutine}  may modify the adjacency-list representations of the graphs $G$ and $G'$ that it is given as input; at the end of every phase we undo all these changes, so that the algorithm for the subsequent phase is given access to a valid adjacency-list representations of the graphs $G$ and $G'$.
After the completion of the last phase, the algorithm returns the cut $(L_z,S_z,R_z)$ computed in that phase. This completes the description of the algorithm. We now proceed to analyze its correctness, starting with the following simple claim. 

\begin{claim}\label{cl : input-to-phase-1-is-well-behaving}
Cut	$(L_0,S_0,R_0)$ is a valid $x$-$y$ vertex-cut in $G$, and, moreover, it is $(\lambda,\gamma)$-well-behaving  for scale $M_0$ with respect to the distinguished cut $(L,S,R)$.
\end{claim}
\begin{proof}
	It is immediate to verify that $(L_{0},S_{0},R_{0})$ is a valid $x$-$y$ vertex-cut from its definition, and from the fact that $(x,y) \notin E(G)$.

	It now remains to prove that $(L_0,S_0,R_0)$ is $(\lambda,\gamma)$-well-behaving for scale $M_0$ with respect to the distinguished cut $(L,S,R)$. First, since $|L_0| = 1$, and $\gamma \leq \lambda\cdot\gamma \leq n$, we get that $\vol^+_G(L_0) \leq n\leq \frac{100n^2}{\gamma}$, establishing Property \ref{prop : P1}.
Next, recall that $M_0 = \wmax'(G) > \max_{v \in G} \set{w(v)}$, so $w(S_0)\leq w(V(G))\leq n\cdot \wmax'(G)\leq 10nM_0$ holds, establishing Property \ref{prop : P2}.
Lastly, since the weight of every vertex in $G$ is less than $M_0$, we get that the set $(L \cup S) \setminus (L_0 \cup S_0)$ contains no vertices with weight at least $M_0$, establishing Property \ref{prop : P3}.
\end{proof}

We then obtain the following immediate corollary.

\begin{corollary}\label{cor : the-vertex-cut-produced-by-each-phase-is-well-behaving}
	For all $0\leq i\leq z$, $(L_i,S_i,R_i)$ is a valid $x$-$y$ vertex-cut in $G$. Moreover, if the distinguished cut $(L,S,R)$ is good, and if the subgraph oracle never erred in its responses to its queries, then the vertex-cut $(L_z,S_z,R_z)$ is $(\lambda,\gamma)$-well-behaving for scale $M_z$ with respect to $(L,S,R)$.
\end{corollary}

The proof of the corollary follows immediately by induction on $i$, where the base of the induction, with $i=0$, follows from  \Cref{cl : input-to-phase-1-is-well-behaving}, and the step follows from the statement of \Cref{thm : phase-i-of-main-subroutine}.

It remains to show that, if the distinguished cut is good, and if the subgraph oracle never erred in its responses to the queries, then the $x$-$y$ vertex-cut $(L_z,S_z,R_z)$ returned by our algorithm is a minimum $x$-$y$ vertex-cut. Indeed, assume that the distinguished cut $(L,S,R)$ is good, and  that the subgraph oracle never erred in its responses to the queries. 
Then, from  \Cref{cor : the-vertex-cut-produced-by-each-phase-is-well-behaving},  the vertex-cut $(L_z,S_z,R_z)$ is $(\lambda,\gamma)$-well-behaving for scale $M_z$ with respect to $(L,S,R)$.
By property \ref{prop : P2} of well-behaving cuts, $w(S_z) \leq w(S) + 10n M_z<w(S)+1$ holds, since, from the definition of the parameter $z$, $M_z\leq \frac{1}{20n}$.
Since all vertex weights are integral, it must be the case that $w(S_z) \leq w(S)$, and so $(L_z,S_z,R_z)$ is a minimum $x$-$y$ vertex-cut.

\paragraph{Running Time and Oracle Access.}
The time required to compute the initial vertex-cut $(L_0,S_0,R_0)$ is $O(n)$.
Since the number of phases is $z\leq \ceil{\log(20n\wmax'(G))} \leq O(\log (nW)) \leq O(\log n \cdot \log W)$, and the expected running time of every phase, from  \Cref{thm : phase-i-of-main-subroutine}, is $O\left (\left (\frac{n^{2+o(1)}\cdot \lambda}{\gamma} + n^{1+o(1)}\cdot\gamma\cdot\lambda\right ) \cdot \log^4(W)\right )$, we get that the total expected running time of the algorithm is:

\[O\left(n+\left (\frac{n^{2+o(1)}\cdot \lambda}{\gamma} + n^{1+o(1)}\gamma\lambda\right ) \cdot \log^5(W)\right)\leq O\left(\left (\frac{n^{2+o(1)}\cdot \lambda}{\gamma} + n^{1+o(1)}\cdot \gamma\cdot \lambda\right ) \cdot \log^5(W)\right).\]


Lastly, observe that the algorithm accesses the $(\delta,20\wmax(G))$-subgraph oracle at most once per phase, so the total number of times it accesses the oracle is bounded by $\ceil{\log(20n\wmax'(G))} \leq \log(40n\wmax'(G)) \leq \log(80n (W+1))$.
In order to complete the proof of \Cref{thm : main-subroutine} it is now enough to prove \Cref{thm : phase-i-of-main-subroutine}, which we do next.

\subsubsection{Algorithm for a Single Phase: Proof of \Cref{thm : phase-i-of-main-subroutine}}\label{sec : phase-i-of-main-subroutine}

In this subsection we prove \Cref{thm : phase-i-of-main-subroutine}. We start by defining the notion of a \emph{promising vertex set}, that is central to the proof of \Cref{thm : phase-i-of-main-subroutine}.


\begin{definition}[Promising vertex set]\label{def : promising-set}
	Let $G$ be a directed $n$-vertex graph, let $x,y\in V(G)$ be a pair of its distinct vertices, and let $(L,S,R)$ be a minimum $x$-$y$ vertex-cut in $G$. Assume further that we are given integral parameters $\lambda,\gamma\geq 1$ and a scale parameter $M'>0$.
	We say that a set $A \subseteq V(G)$ of vertices is a \emph{$(\lambda,\gamma)$-promising vertex set for scale $M'$ with respect to $(G,x,y)$ and the cut $(L,S,R)$}, if all of the following hold:
	
	\begin{properties}{R}
		\item\label{prop : P'1} $\vol^+_G(A) \leq \frac{100n^2}{\gamma}$; and
	\end{properties}
	
	there exists a set $\Tilde{L} \subseteq A$ of vertices, such that, if we denote by $(\Tilde{L},\Tilde{S},\Tilde{R})$ the tripartition of $V(G)$ induced by $\Tilde{L}$, then:
	
	\begin{properties}[1]{R}
		\item\label{prop : P'2} $w(\Tilde{S}) \leq w(S) + 8\lambda\gamma M'$; 
		\item\label{prop : P'3} set $(L\cup S)\setminus (\Tilde{L} \cup \Tilde{S})$ contains at most $12\lambda\gamma$ vertices of weight at least $M'$; and
		\item \label{prop : P'4}   $(\Tilde{L},\Tilde{S},\Tilde{R})$ is a valid $x$-$y$ vertex-cut in $G$.
	\end{properties}
\end{definition}

Note that Properties \ref{prop : P'2} and \ref{prop : P'3}  of the vertex-cut $(\Tilde{L},\Tilde{S},\Tilde{R})$ in the above definition can be viewed as a strengthened version of Properties \ref{prop : P2} and \ref{prop : P3} from the definition of a well-behaving vertex-cut with respect to a cut $(L,S,R)$ (see \Cref{def : well-behaving cut}), if $\lambda\cdot \gamma\leq n$. Therefore, intuitively, a vertex set $A$ is promising if there exists an $x$-$y$ vertex-cut $(\tilde L,\tilde S,\tilde R)$ in $G$ that is induced by $\tilde L$, with $\tilde L\subseteq A$, such that $(\tilde L,\tilde S,\tilde R)$  is well-behaving with respect to $(L,S,R)$, even when using slightly strengthened versions of properties  \ref{prop : P2} and \ref{prop : P3}. 

The main idea of our algorithm is to first compute a $(\lambda,\gamma)$-promising vertex set $A$, and then to use this set in order to compute the cut $(L'',S'',R'')$ with the required properties, that will serve as the algorithm's output. The following theorem, whose proof is deferred to \Cref{sec : step-1}, provides an efficient randomized algorithm for computing such a vertex set $A$.
The input to the algorithm is identical to that from \Cref{thm : phase-i-of-main-subroutine}, but its output is a subset $A$ of vertices of $G$ (instead of a vertex-cut), that is guaranteed to be a promising vertex set, if the input vertex-cut $(L',S',R')$ is well-behaving for scale $(2M')$ with respect to some minimum $x$-$y$ vertex-cut, and if the subgraph oracle does not err.

\begin{theorem}
	\label{thm : algorithm-for-step-1}
	There is a randomized algorithm, whose input consists of a simple directed $n$-vertex graph $G$ with integral weights $1 \leq w(v) \leq W$ on its vertices $v \in V(G)$ given in the adjacency-list representation, together with the adjacency-list representation of the split-graph $G'$ of $G$. Additionally, the algorithm is given as input integral parameters $\lambda,\gamma \geq 1$ with $\lambda\cdot\gamma \leq \frac{n}{100}$, a pair $x,y \in V(G)$ of distinct vertices such that $(x,y) \notin E(G)$, a scale parameter $\frac{1}{100n} \leq M' \leq \frac{\wmax'(G)}{2}$ that is an integral power of $2$,
	and an $x$-$y$ vertex-cut $(L',S',R')$ in $G$.
	Lastly, the algorithm is allowed to make a single query to the $(\delta,20\wmax(G))$-subgraph oracle for $G$, for $\delta=1-\frac{\log(30000\lambda\gamma \cdot \log n)}{\log n}$, so that $n^{1-\delta}=30000\lambda\gamma \cdot \log n$.
	The algorithm returns a subset $A \subseteq V(G)$ of vertices with $x \in A$ and $y \notin A \cup N^+_G(A)$ with the following guarantee: if there exists a minimum $x$-$y$ vertex-cut $(L,S,R)$ with $|L| \leq \lambda$ such that the input vertex-cut $(L',S',R')$ is $(\lambda,\gamma)$-well-behaving for scale $(2M')$ with respect to $(L,S,R)$, and if the subgraph oracle does not err in its response to the query, then vertex set $A$ is $(\lambda,\gamma)$-promising for scale $M'$ with respect to $(G,x,y)$ and the cut $(L,S,R)$.
	The expected running time of the algorithm is $O\left (\left (\frac{n^{2+o(1)}\cdot \lambda}{\gamma} + n^{1+o(1)}\cdot \gamma\cdot \lambda\right ) \cdot \log^4 W\right )$.
\end{theorem}

We prove \Cref{thm : algorithm-for-step-1} in \Cref{sec : step-1}, after we complete the proof of  \Cref{thm : phase-i-of-main-subroutine} using it.
Recall that we are given as input a simple directed $n$-vertex graph $G$ with integral weights $1 \leq w(v) \leq W$ on its vertices $v \in V(G)$ in the adjacency-list representation, together with the adjacency-list representation of its split-graph $G'$, integers $\lambda,\gamma \geq 1$ with $\lambda\cdot\gamma \leq \frac{n}{100}$, and a pair $x,y \in V(G)$ of distinct vertices such that $(x,y) \notin E(G)$. Additionally, we are given a scale parameter $\frac{1}{100n} \leq M' \leq \frac{\wmax'(G)}2$ that is an integral power of $2$,
and an $x$-$y$ vertex-cut $(L',S',R')$.
Lastly, the algorithm is allowed to make a single query to the $(\delta,20\wmax(G))$-subgraph oracle for $G$, for $\delta=1-\frac{\log(30000\lambda\gamma \cdot \log n)}{\log n}$, so that $n^{1-\delta}=30000\lambda\gamma \cdot \log n$.

For the sake of analysis, we define a \emph{distinguished cut} $(L,S,R)$ as follows.
If there exists a minimum $x$-$y$ vertex-cut $(L^*,S^*,R^*)$ with $|L^*| \leq \lambda$ such that the input vertex-cut $(L',S',R')$ is $(\lambda,\gamma)$-well-behaving  for scale $(2M')$ with respect to $(L^*,S^*,R^*)$, then we let $(L,S,R)=(L^*,S^*,R^*)$, and we say that the distinguished cut is \emph{good}. Otherwise, we let $(L,S,R)$ be any minimum $x$-$y$ vertex-cut, and we say that the distinguished cut is \emph{bad}.
Since, in the remainder of the proof, the graph $G$, the vertices $x,y$, the parameters $\lambda,\gamma$, and the distinguished cut are fixed, we will say that a vertex set $A\subseteq V(G)$ is \emph{promising} if it is a $(\lambda,\gamma)$-promising vertex set for scale $M'$ with respect to $(G,x,y)$ and the distinguished cut $(L,S,R)$. Similarly, we will say that a vertex-cut $(L'',S'',R'')$ is \emph{well-behaving} for some scale $M$, if it is $(\lambda,\gamma)$-well-behaving for scale $M$ with respect to the cut $(L,S,R)$. 

Recall that our goal is to compute a vertex-cut  $(L'',S'',R'')$ in $G$, such that, if the distinguished cut is good, and if the subgraph oracle does not err in its response to the query, then $(L'',S'',R'')$  is well-behaving for scale $M'$.
As our first step, we apply the algorithm from \Cref{thm : algorithm-for-step-1} to the input graph $G$, the pair $x,y$ of vertices, the parameters $\lambda,\gamma$ and $M'$, and the $x$-$y$ vertex-cut $(L',S',R')$. Recall that the expected running time of the algorithm from  \Cref{thm : algorithm-for-step-1} is $O\left (\left (\frac{n^{2+o(1)}\cdot \lambda}{\gamma} + n^{1+o(1)}\cdot \gamma\cdot \lambda\right ) \cdot \log^4 W\right )$. Let $A\subseteq V(G)$ be the output of the algorithm, and  recall that $x\in A$ and $y\not\in A\cup N_G^+(A)$ must hold. 
Recall that  \Cref{thm : algorithm-for-step-1} guarantees that, if the distinguished cut $(L,S,R)$ is good, and if the oracle does not err in its response to the query, then  $A$ is a promising vertex set for scale $M'$.
We check, in time $O\left(\frac{n^2}{\gamma}\right )$, whether $\vol^+_G(A)>\frac{100n^2}{\gamma}$; if this is indeed the case, then we terminate the algorithm and return the $x$-$y$ vertex-cut $(L'',S'',R'')$, where $L''=\{x\}$, $S''=N^+_G(x)$, and $R'' = V(G) \setminus (L'' \cup S'')$. We say that the algorithm \emph{terminated early} in this case. Notice that the algorithm may only terminate early if the distinguished cut is bad, or if the subgraph oracle erred. We assume from now on that the algorithm did not terminate early, and so $\vol^+_G(A)\leq \frac{100n^2}{\gamma}$ holds.

We now provide intuition for the remainder of the algorithm. Observe that it is now enough to compute a vertex-cut $(L'',S'',R'')$ in $G$, so that, if the distinguished cut $(L,S,R)$ is good and the subgraph oracle did not err, then $(L'',S'',R'')$ is well-behaving for scale $M'$ with respect to $(L,S,R)$. Therefore, for the sake of intuition, assume that the distinguished cut $(L,S,R)$ is good, and that the subgraph oracle did not err, so that the vertex set $A$ is promising. Consider now the vertex-cut $(\tilde L,\tilde S,\tilde R)$ with $\tilde L\subseteq A$ that is guaranteed to exist by the definition of the promising vertex set (\Cref{def : promising-set}). 
Recall that, from the definition of the promising vertex set, cut $(\tilde L,\tilde S,\tilde R)$ has all required properties of a well-behaving cut, and in fact it even satisfies a strengthened version of these properties. Unfortunately, even if we are given a promising vertex set $A$, it is not immediately clear how to recover the corresponding vertex-cut $(\tilde L,\tilde S,\tilde R)$ from it. Fortunately, we can overcome this difficulty by 
using an algorithm for the vertex-\AFMC problem from \Cref{thm : vertex weighted approximately minimum containing cut} with  parameter $\alpha=M'n$, in order to recover a vertex-cut $(L'',S'',R'')$ in $G$, that, while not guaranteed to be exactly equal to the cut $(\tilde L,\tilde S,\tilde R)$, is sufficiently ``close'' to it, so that it is guaranteed to be well-behaving with respect to $(L,S,R)$. One remaining obstacle is that the running time of the algorithm from 
\Cref{thm : vertex weighted approximately minimum containing cut} is at least linear in the size of the input graph, which may be higher than the required running time bound if the graph $G$ is sufficiently dense. To overcome this difficulty, we apply the algorithm for the  vertex-\AFMC problem to the graph $\Tilde{G}=G^{|A}$ that is derived from $G$ via $A$ instead (see \Cref{def: defined graph}); since $\vol^+_G(A)\leq O\left (\frac{n^2}{\gamma}\right )$, this graph is guaranteed to be sufficiently small, and, since $\tilde L\subseteq A$ holds, a solution to the vertex-\AFMC problem on this graph naturally defines an $x$-$y$ vertex-cut in the original graph with the required properties. 

We are now ready to complete the description of the algorithm.
We start by constructing the graph $\Tilde{G}=G^{|A}$ that is derived from $G$ via $A$ (see \Cref{def: defined graph}). From \Cref{claim: properties of defined graph}, $|E(\tilde G)|\leq O\left(\vol^+_G(A)\right)\leq O\left(\frac{n^2}{\gamma}\right )$, and graph $\tilde G$ can be constructed in time $O\left(\vol^+_G(A)\right)\leq O\left(\frac{n^2}{\gamma}\right )$.

Next, we define an instance of the vertex-\AFMC problem (see \Cref{def : almost-minimum-furthest-vertex-cut}) with the graph $\Tilde{G}$, the pair $x,t \in V(\Tilde{G})$ of its vertices and the parameter $\alpha=n M'$.\footnote{Observe that, in the case where $M' < 1$, it is possible that $\alpha$ is not an integer. In this case, we cannot set up the instance of the vertex-\AFMC problem exactly as described, since this problem is only defined when $\alpha$ is an integer. However, we can get around this issue by rescaling both $\alpha$ and all vertex-weights by a multiplicative factor of $\frac{1}{M'}$, which is guaranteed to make them all integral (because, in this case, $M'$ is an integral power of $2$ satisfying $M' <1$.)} Note that $(x,t)\not\in E(\tilde G)$ since $x\in A$, and, moreover, since $M'\leq\frac{\wmax'(G)}{2}$, $\alpha<\frac{\wmax(G)}{2}$ must hold. Therefore, we obtain a valid instance of vertex-\AFMC.
We apply the algorithm from \Cref{thm : vertex weighted approximately minimum containing cut} to this instance of vertex-\AFMC, to obtain a solution $\left(\tilde L'',\tilde S'',\tilde R''\right)$ to the problem, so, in particular,
$\left(\tilde L'',\tilde S'',\tilde R''\right)$  is an $x$-$t$ vertex-cut in $\tilde G$.
We then compute a tripartition $(L'',S'',R'')$ of vertices of $G$, with $L''=\tilde L''$, $S''=\tilde S''$, and $R''=V(G)\setminus (L''\cup S'')$. 
Recall that $x\in \tilde L''=L''$ must hold, and, since, from \Cref{thm : algorithm-for-step-1}
$y \notin A \cup N^+_G(A)$, we get that $y$ is not a vertex of $\tilde G$, so $y\in R''$ must hold.
From \Cref{claim: transforming cuts from defined graph},  $(L'',S'',R'')$ is a valid vertex-cut in $G$, and $L''\subseteq A$. Moreover, from the above discussion, $(L'',S'',R'')$ is a valid $x$-$y$ cut in $G$.
The following claim will complete the correctness analysis of the algorithm.

\begin{claim}\label{claim: correctness of well behaving}
	Assume that the distinguished cut is good and that the subgraph oracle did not err. Then the algorithm does not terminate early, and moreover, the $x$-$y$ cut $(L'',S'',R'')$ that it returns is a well-behaving cut for scale $M'$ with respect to $(L,S,R)$.
\end{claim}
\begin{proof}
	We assume that  the distinguished cut is good and that the subgraph oracle did not err, so $A$ is a promising vertex set.
	In particular from Property \ref{prop : P'1} of a promising set, $\vol^+_G(A) \leq \frac{100n^2}{\gamma}$ holds, so the algorithm does not terminate early. Consider now the vertex-cut $(L'',S'',R'')$ that the algorithm returns. It is now enough to show that $(L'',S'',R'')$ is a well-behaving cut for scale $M'$ with respect to $(L,S,R)$, or, equivalently, that it satisfies Properties \ref{prop : P1}--\ref{prop : P3}.
	
	As observed already, $L''\subseteq A$ must hold. Since, from Property \ref{prop : P'1} of a promising vertex set, $\vol^+_G(A) \leq \frac{100n^2}{\gamma}$, we conclude that $\vol^+_G(L'')\leq  \frac{100n^2}{\gamma}$, establishing Property \ref{prop : P1}. 
	
	Consider now the set  $\tilde L\subseteq A$ of vertices that is guaranteed to exist by the definition of a promising vertex set, and the corresponding $x$-$y$ vertex-cut $(\tilde L,\tilde S,\tilde R)$ in $G$ (see Property \ref{prop : P'4} of a promising vertex set). It is immediate to verify that $\tilde L,\tilde S\subseteq V(\tilde G)$, and, if we denote by $\tilde R'=V(\tilde G)\setminus(\tilde L\cup\tilde S)$, then, from \Cref{claim: properties of defined graph}, $(\tilde L,\tilde S,\tilde R')$ is a valid vertex-cut, and hence a valid $x$-$t$ vertex-cut in $\tilde G$.
	Let $\opt^V_{\Tilde{G}}$ denote the value of the minimum $x$-$t$ vertex-cut in $\Tilde{G}$. 
	From the above discussion, $\opt^V_{\Tilde{G}}\leq w(\tilde S)$, and, from Property 
	\ref{prop : P'2} of a promising vertex set, $w(\tilde S)\leq  w(S) + 8\lambda\gamma M'$ must hold. Altogether,  
since $\lambda\gamma\leq \frac{n}{100}$, we get that:

\begin{equation}\label{eq: bound on OPT}	
\opt^V_{\Tilde{G}}\leq w(\tilde S)\leq  w(S) + 8\lambda\gamma M'\le w(S)+8n\cdot M'.
\end{equation}

Since $(\tilde L'',\tilde S'',\tilde R'')$ is a valid solution to the instance of the vertex-\AFMC problem on graph $\tilde G$, with the pair $x,t$ of vertices and the parameter $\alpha=n\cdot M'$, we get that:
	
	\[w(S'')=w(\tilde S'')\leq  \opt^V_{\Tilde{G}} + 2\alpha = \opt^V_{\Tilde{G}} + 2n M'\leq w(S)+10n\cdot M',\]
	
establishing Property \ref{prop : P2}.

It remains to show that $(L'',S'',R'')$ satisfies property \ref{prop : P3} of a  well-behaving vertex-cut for scale $M'$. 
In other words, it remains to show that the vertex set $(L \cup S)\setminus (L'' \cup S'')$ contains at most $20\lambda\gamma$ vertices of weight at least $M'$. Consider the following two subsets of vertices of $G$:

\[Z_1=\set{v\in (L \cup S)\setminus (\tilde L\cup \tilde S)\mid w(v)\geq M'},\]

and

\[Z_2=\set{v\in (\tilde L\cup \tilde S)\setminus (L'' \cup S'')\mid w(v)\geq M'}.\]

Notice that, if $v\in  (L \cup S)\setminus (L'' \cup S'')$ is a vertex of weight at least $M'$, then $v\in Z_1$ or $v\in Z_2$ must hold. Therefore, it is enough to prove that $|Z_1|+|Z_2|\leq 20\lambda\gamma$. Observe that, from Property \ref{prop : P'3} of a promising set, $|Z_1|\leq 12\lambda\gamma$ must hold. 
The following observation will then finish the proof of the claim.

\begin{observation}\label{obs: bound Z2}
	$|Z_2|\leq 8\lambda\gamma$.
\end{observation}
\begin{proof}
Recall that we denoted by $\tilde R'=V(\tilde G)\setminus(\tilde L\cup\tilde S)$, and we have established that $(\tilde L,\tilde S,\tilde R')$ is a valid $x$-$t$ vertex-cut in $\tilde G$. From the definition of the vertex-\AFMC problem (see \Cref{def : almost-minimum-furthest-vertex-cut}):

\[
	|(\tilde L \cup \tilde S) \setminus (\tilde L'' \cup \tilde S'')| \leq \left(w(\tilde S) - \opt^{V}_{\tilde G}\right )\cdot \frac{n}{\alpha} \leq \frac{w(\tilde S) - \opt^{V}_{\tilde G}}{M'},
\]

since $\alpha=nM'$.
From Property \ref{prop : P'2} of a promising vertex set, $w(\Tilde{S}) \leq w(S) + 8\lambda\gamma M'$.
Moreover, from \Cref{cor: minimum cut in defined graph}, $\opt_{\Tilde{G}}^V$ is at least as large as the value of the minimum $x$-$y$ vertex-cut in $G$, which, in turn, is equal to $w(S)$, so  $\opt^V_{\Tilde{G}} \geq w(S)$ must hold. Altogether, we get that:

\[	|(\tilde L \cup \tilde S) \setminus (\tilde L'' \cup \tilde S'')| \leq \frac{w(\tilde S) - \opt^{V}_{\tilde G}}{M'}\leq 8\lambda\gamma.\]
\end{proof}
\end{proof}

We conclude that, if the distinguished cut is good, and if the subgraph oracle does not err, then the $x$-$y$ vertex-cut that our algorithm returns is well-behaving for scale $M'$ with respect to $(L,S,R)$. 

Recall that the expected running time of the algorithm from  \Cref{thm : algorithm-for-step-1} is bounded by: 

\[O\left (\left (\frac{n^{2+o(1)}\cdot \lambda}{\gamma} + n^{1+o(1)}\cdot \gamma\cdot \lambda\right ) \cdot \log^4 W\right ).\]

As observed already, we can check whether $\vol^+_G(A)\leq \frac{100n^2}{\gamma}$, and if so, to compute the corresponding graph $\tilde G$, in time $O\left(\frac{n^2}{\gamma}\right )$.
Once the algorithm for the vertex-\AFMC problem computes the vertex-cut $(\tilde L'',\tilde S'',\tilde R'')$ in $\tilde G$, we can transform it into the vertex-cut $(L'',S'',R'')$ in $G$ in time $O(n)$. Lastly, the running time of the algorithm 
from \Cref{thm : vertex weighted approximately minimum containing cut} for the vertex-\AFMC problem is:

\[O\left ((|V(\tilde G)+|E(\tilde G)|)^{1+o(1)}\log W\cdot \log(2+\alpha)\right )\leq O\left(\frac{n^{2+o(1)}\cdot \log^2W}{\gamma}\right ).\]

since $\alpha=M'n\leq (W+1)n$ and $|V(\tilde G)|,|E(\tilde G)|\leq O\left(\frac{n^2}{\gamma}\right )$.
Altogether, the expected running time of the algorithm is $O\left (\left (\frac{n^{2+o(1)}\cdot \lambda}{\gamma} + n^{1+o(1)}\cdot \gamma\cdot \lambda\right ) \cdot \log^4W\right )$. 

\subsubsection{Computing a Promising Set: Proof of \Cref{thm : algorithm-for-step-1}} \label{sec : step-1}

In this subsection we prove \Cref{thm : algorithm-for-step-1}.
Recall that we are given as input a simple directed $n$-vertex graph $G$ with integral weights $1 \leq w(v) \leq W$ on its vertices $v \in V(G)$ in the adjacency-list representation, together with the adjacency-list representation of the split-graph $G'$ of $G$, integers $\lambda,\gamma \geq 1$ with $\lambda\cdot\gamma \leq \frac{n}{100}$, and a pair $x,y \in V(G)$ of distinct vertices such that $(x,y) \notin E(G)$.
Additionally, we are given as input a scale parameter $\frac{1}{100n} \leq M' \leq \frac{\wmax'(G)} 2$ that is an integral power of $2$,
and an $x$-$y$ vertex-cut $(L',S',R')$ in $G$.

\paragraph{Distinguished cut.}
For the sake of analysis, we fix  a minimum $x$-$y$ vertex-cut $(L,S,R)$ in $G$, that we refer to as the \emph{distinguished cut}, as follows.
If there exists a minimum $x$-$y$ vertex-cut $(L^*,S^*,R^*)$ with $|L^*| \leq \lambda$ such that the input vertex-cut $(L',S',R')$ is $(\lambda,\gamma)$-well-behaving for scale $(2M')$ with respect to $(L^*,S^*,R^*)$ (see \Cref{def : well-behaving cut}), then we let $(L,S,R)=(L^*,S^*,R^*)$, and we say that the distinguished cut is \emph{good}. Otherwise, we let $(L,S,R)$ be any 
minimum $x$-$y$ vertex-cut in $G$, and we say that the distinguished cut is \emph{bad}.
Since the vertices $x,y$, the parameters $\lambda,\gamma$, and the distinguished cut $(L,S,R)$ are now fixed, for simplicity, we will refer to a vertex set $A\subseteq V(G)$ that is $(\lambda,\gamma)$-promising for scale $M'$ with respect to $(G,x,y)$ and the vertex-cut $(L,S,R)$, as a \emph{promising vertex set for graph $G$ and scale $M'$}.

Recall that our algorithm is allowed to make a single query to the $(\delta,20\wmax(G))$-subgraph oracle for $G$, for $\delta=1-\frac{\log(30000\lambda\gamma \cdot \log n)}{\log n}$; observe that $n^{1-\delta}=30000\lambda\gamma \cdot \log n$.
Our goal is to compute a subset $A \subseteq V(G)$ of vertices with $x \in A$ and $y \notin A \cup N^+_G(A)$, so that, if the distinguished cut is good, and the subgraph oracle does not err in its response to the query, then $A$ is a promising vertex set for graph $G$ and scale $M'$. 

It will be more convenient for us to work with a graph $\hat G$ that is identical to $G$, but with slightly modified vertex weights. Specifically, for every vertex $v\in V(G)$, we define a new integral weight $\hat w(v)\geq 0$ as follows: if $w(v)<2M'$ then we set $\hat w(v)=0$, and otherwise we set $\hat w(v)=\floor{\frac{w(v)-M'}{M'}}$. 
Clearly, for every vertex $v\in V(G)$:

\begin{equation}\label{eq: bound on new weight}
\frac{w(v)}{M'}-2\leq  \hat w(v)\leq \frac{w(v)}{M'}.
\end{equation}

We let $\hat G$ be the graph that is identical to $G$, except that, for every vertex $v\in V(G)$, its weight in $\hat G$ is $\hat w(v)$.
We note that we do not compute the graph $\hat G$ explicitly. Instead, we can use the adjacency-list representation of the graph $G$ that is given as part of input in a natural way in order to simulate access to the adjacency-list representation of $\hat G$.
Let $\opt_{\hat{G}}$ denote the value of the minimum $x$-$y$ vertex-cut in $\hat{G}$. In the following simple observation we bound $\hat w(S)$ in terms of $\opt_{\hat{G}}$.

\begin{observation}\label{obs : value-of-distinguish-cut-in-hat-G-is-not-far-from-optimal}
$\hat{w}(S) \leq \opt_{\hat{G}} + 2n$. 
\end{observation}
\begin{proof}
	Let $(\hat L,\hat S,\hat R)$ be a minimum $x$-$y$ vertex-cut in $\hat{G}$, so that $\hat{w}(\hat S)=\opt_{\hat{G}}$.
	Since $(L,S,R)$ is a minimum $x$-$y$ vertex-cut in $G$, $w(S) \leq w(\hat S)$ must hold. 
From 	\Cref{eq: bound on new weight} we then get that:

\[\hat{w}(S)\leq \frac{w(S)}{M'}\leq \frac{w(\hat S)}{M'}\leq \hat w(\hat S)+2|\hat S|\le \opt_{\hat{G}} + 2n.\]

\end{proof}

In the following observation we show that, if the distinguished cut $(L,S,R)$ is good, then the input vertex-cut $(L',S',R')$ is $(\lambda,\gamma)$-well-behaving for scale $4$ with respect to the distinguished cut $(L,S,R)$ in the graph $\hat{G}$.

\begin{observation}\label{obs : cut-is-well-behaving-in-graph-with-reduced-weights}
If the distinguished cut $(L,S,R)$ is good, then the vertex-cut $(L',S',R')$ is $(\lambda,\gamma)$-well-behaving for scale $4$ with respect to the distinguished cut $(L,S,R)$ in $\hat{G}$.
\end{observation}
\begin{proof}
	Recall that, if the distinguished cut $(L,S,R)$ is good, then 
	vertex-cut $(L',S',R')$ is $(\lambda,\gamma)$-well-behaving for scale $(2M')$ with respect to it in $G$. In particular, from Property \ref{prop : P1} of well-behaving cuts, $\vol^+_{\hat G}(L')=\vol^+_G(L') \leq \frac{100n^2}{\gamma}$ must hold, establishing Property \ref{prop : P1} for cut $(L',S',R')$ in $\hat G$.

	From Property \ref{prop : P2} of the cut $(L',S',R')$ in $G$, $w(S') \leq w(S) + 20nM'$ must hold. 
	From Inequality \ref{eq: bound on new weight}:

\[\hat w(S')\leq \frac{w(S')}{M'}\leq \frac{w(S)}{M'}+20n\leq \hat w(S)+2|S|+20n\leq \hat w(S)+40n,\]
		
 establishing Property \ref{prop : P2} for cut $(L',S',R')$ in $\hat G$ for scale $4$.
	
Lastly, from Property \ref{prop : P3} of the cut $(L',S',R')$ in $G$, set $(L \cup S) \setminus (L' \cup S')$ may contain at most  $20\lambda\gamma$ vertices $v\in V(G)$ with $w(v)\geq 2M'$. Note that if, for some vertex $v\in V(G)$, $\hat w(v)\geq 4$ holds, then, from the definition of the function  $\hat w$, $w(v)\geq 2M'$ must hold. Therefore, the total number of vertices $v\in V(G)$ with $\hat w(v)\geq 4$ in set  $(L \cup S) \setminus (L' \cup S')$ is bounded by  $20\lambda\gamma$ as well, establishing Property \ref{prop : P3} for cut $(L',S',R')$ in $\hat G$ for scale $4$.
\end{proof}

Next, we define the notion of a \emph{quasi-promising vertex set} in the graph $\hat G$, and show that, if $A$ is a quasi-promising vertex set in $\hat G$, and  if the distinguished cut $(L,S,R)$ is good, then $A$ is also a promising vertex set for graph $G$ and scale $M'$.

\begin{definition}[Quasi-promising vertex set]\label{def : quasi-promising}
A set $A\subseteq V(\hat G)$ of vertices is \emph{quasi-promising for graph $\hat G$}, if all of the following hold:
	
	\begin{properties}{Q}
		\item\label{prop : Q1} $\vol^+_{\hat G}(A) \leq \frac{100n^2}{\gamma}$;
	\end{properties}
	and there exists a set $\Tilde{L} \subseteq A$ of vertices, such that, if we denote by $(\Tilde{L},\Tilde{S},\Tilde{R})$ the tripartition of vertices of $\hat G$ induced by $\Tilde{L}$, then:
	\begin{properties}[1]{Q}
		\item\label{prop : Q2} $\hat w(\Tilde{S}) \leq \hat w(S) + 6\lambda\cdot \gamma$;
		\item\label{prop : Q3} $\Tilde{L} \subseteq L$; and
		\item \label{prop : Q4} $(\Tilde{L},\Tilde{S},\Tilde{R})$ is an $x$-$y$ vertex-cut in $\hat G$.	
	\end{properties}
\end{definition}

The proof of the following claim follows directly from the definitions but is somewhat technical. We defer it to Section \ref{subsec: appx: quasi to promising} of Appendix.
\begin{claim}\label{claim: from quasi to promising}
Assume that the distinguished cut $(L,S,R)$ is good, and	let $A\subseteq V(\hat G)$ be a vertex set that is quasi-promising for $\hat G$. Then $A$ is a promising vertex set for graph $G$ and scale $M'$.
\end{claim}

It is now sufficient to provide an algorithm for computing a set $A\subseteq V(\hat G)$ of vertices with $x\in A$ and $y\not\in A\cup N^+_G(A)$, such that, if the distinguished cut is good and the subgraph oracle does not err, then $A$ is a quasi-promising vertex set for graph $\hat G$. We now provide  a high-level overview and intuition for our algorithm.
Consider a graph $\hat H$ obtained from $\hat{G}$ by multiplying all vertex-weights by roughly $\frac{n}{\gamma}$.
Denote by $\Delta=\frac{50n^2}{\gamma}$, and assume first that the value of the maximum $x$-$y$ flow in graph $\hat H$ (where the capacity of every vertex $v$ is its weight in $\hat H$) is significantly smaller than $\Delta$. Let $\hat H'$ be the split-graph of the graph $\hat H$, and denote $s=x^{\out}$ and $t=y^{\inn}$.
Let $\hat H''$ be the graph obtained from $\hat H'$ by adding, for every vertex $v\in V(\hat H')$ with $\deg^+_{\hat H'}(v)>0$, the edge $(v,t)$ of capacity $\deg^+_{\hat H'}(v)$ to the graph. As we show later, in this case, the value of the maximum $s$-$t$ flow in $\hat H''$ is bounded by $\Delta$.
Consider now applying Algorithm \alglocal from \Cref{thm : unit-flow} for  computing a local $s$-$t$ flow $f$ in $\hat H''$ with the parameter $h=2\lambda+4$. Since the value of the maximum $s$-$t$ flow in $\hat H''$ is less than $\Delta$, the algorithm must output a vertex set $A'\subseteq V(\hat H')$ with $s\in A'$ and $\vol^+_{\hat H'}(A')\leq \Delta$, such that the distance from $s$ to every vertex in $V(\hat H')\setminus A'$ in the residual flow network of the graph $\hat H''$ with respect to the flow $f$ is at least $h$. Let $A\subseteq V(G)$ contain every vertex $v\in V(G)\setminus\left(\set{y}\cup N^-_G(y)\right )$ with $v^{\out}\in A'$. We then show that vertex set $A$ must be quasi-promising for graph $\hat G$. Notice that the running time of the algorithm from \Cref{thm : unit-flow} is 
$\Tilde{O}(h \cdot \Delta)\leq \tilde O\left(\frac{n^2\lambda}{\gamma}\right )$. The main difficulty with using this approach directly is that the value of the maximum $x^{\out}$-$y^{\inn}$ flow in graph $\hat H$ may be large, and it may not be bounded by $O(n^2/\gamma)$. If we increase the value of the parameter $\Delta$ appropriately in order to overcome this difficulty, then the running time of the algorithm from \Cref{thm : unit-flow} will also grow, and, moreover, we may no longer be able to guarantee that the bound on $\vol^+_{\hat G}(A)$ required by Property \ref{prop : Q1} of a quasi-promising vertex set holds.

Intuitively, in order to overcome this difficulty, we would like to use the above approach with a ``warm start''. In other words, we will first compute an initial $x$-$y$ flow $f_0$ in graph $\hat G$, whose value is close to the value of the maximum $x$-$y$ flow in $\hat G$. In order to do so efficiently, we first sparsify the graph $\hat G$ while approximately preserving the value of the maximum $x$-$y$ flow.
The sparsification procedure constructs a \emph{sparsifier} for the graph $\hat G$, that consists of a sparse subgraph $\tilde G\subseteq \hat G$, and an additional collection $\tilde E$ of edges, that we refer to as \emph{shortcut edges}. Each shortcut edge $e\in \tilde E$ connects some vertex of $\hat G$ to the destination vertex $y$, and we ensure that, if the distinguished cut is good, and if the subgraph oracle does not err in its response to the query that the sparsification procedure performs, then, for each such edge $e=(v,y)\in \tilde E$, $v\not\in L$ must hold. This latter property ensures that $(L,S,R)$ remains a valid $x$-$y$ vertex-cut in both the sparsifier $\tilde G\cup \tilde E$ and in the graph $J=\hat G\cup \tilde E$, that is obtained from $\hat G$ by inserting the shortcut edges into it. Our sparsification procedure also ensures that the value of the maximum $x$-$y$ flow in the sparsifier $\tilde G\cup \tilde E$ is close to that in $\hat G$. We then use the algorithm from \Cref{thm: maxflow} in order to compute a maximum $x$-$y$ flow $f_0$ in the sparsified graph $\tilde G\cup \tilde E$; since this graph is relatively sparse, this flow can be computed efficiently. Note that $f_0$ also defines a valid $x$-$y$ flow in the graph  $J=\hat G\cup \tilde E$, and, moreover, its value is close to the value of the maximum $x$-$y$ flow in $J$.

Intuitively, this initial flow $f_0$ provides us with a ``warm start'', that allows us to carry out the local-flow based approach described above efficiently. Specifically, we let $J'$ denote the split-graph of the graph $J$, and we let $f'_0$ be the $x^{\out}$--$y^{\inn}$ flow in $J'$ of value $\val(f_0)$ that corresponds to flow $f_0$ in $J$; as we later show, the value of the flow $f'_0$ is close to the value of the maximum $x^{\out}$--$y^{\inn}$ flow in $J'$. We then let $Z$ be the residual flow network of the graph $J'$ with respect to the flow $f'_0$, and we define a graph $\hat{Z}$ that is obtained from $Z$ by multiplying all capacities by roughly $\frac{n}{\gamma}$; as we show later, the value of the maximum $x^{\out}$--$y^{\inn}$ flow in $\hat{Z}$ is bounded by $\Delta=\frac{50n^2}{\gamma}$. We then apply  the local-flow based approach described above to the graph $\hat{Z}$ instead of the split-graph $\hat H'$ of $\hat H$, eventually obtaining a set $A\subseteq V(G)$ of vertices, with the guarantee that, if the distingiushed cut is good and the subgraph oracle did not err in its response to the query, then $A$ is a quasi-promising vertex set for graph $\hat G$.

The remainder of our algorithm then consists of three steps. In the first step, we compute the sparsifier $(\tilde G,\tilde E)$ for graph $\hat G$. In the second step, we compute a maximum $x$-$y$ flow in graph $J=\hat G\cup \tilde E$, define the additional graphs $J'$ and $\hat{Z}$ that our algorithm uses, and analyse their properties. Then in the third step we use the local-flow based approach with the graph $\hat{Z}$ in order to construct the desired vertex set $A$. We now describe each of the steps in turn.

\subsubsection{Step 1: Sparsification}

The goal of this step is to compute a \emph{sparsifier} of the graph $\hat G$. Intuitively, a sparsifier is the union of a sparse subgraph  $\tilde G\subseteq \hat G$ with a collection $\tilde E$ of edges that may not lie in $\hat G$,  and that we refer to as \emph{shortcut edges}. Each shortcut edge connects some vertex of $\hat G$ to the destination vertex $y$. We would like to enure that, on the one hand, the value of the minimum $x$-$y$ vertex-cut in $\tilde G\cup \tilde E$ is close to that in $\hat G$, while, on the other hand, no vertex of $L$ serves as an endpoint of an edge in $\tilde E$. Recall that we denoted by $\opt_{\hat G}$ the value of the minimum $x$-$y$ vertex-cut in $\hat G$, and that $(L,S,R)$ is the distinguished minimum $x$-$y$ vertex-cut in $G$. We now formally define the notion of a sparsifier for $\hat G$.

\begin{definition}[A sparsifier for $\hat G$]
	A \emph{sparsifier} for the graph $\hat G$ consists of a subgraph 
$\tilde G\subseteq \hat G$ with $V(\tilde G)=V(\hat G)$ and $|E(\tilde G)|\leq \Tilde{O}\left(\left(\frac{n^2}{\gamma}+n\lambda\gamma\right) \cdot \log^2(\wmax(G))\right)$, and a collection $\tilde E$ of edges that may not lie in $\hat G$, where each edge connects a distinct vertex of $V(\hat G)$ to $y$; we refer to the edges of $\tilde E$ as \emph{shortcut edges}.
We say that $(\tilde G,\tilde E)$ is a \emph{perfect sparsifier}, if, additionally, the following properties hold:

\begin{itemize}
	\item no vertex of $L$ serves as an endpoint of an edge in $\tilde E$; and
	\item the value of the minimum $x$-$y$ vertex-cut in the graph $\tilde G\cup \tilde E$ is at least $\hat w(S)-48n$.
\end{itemize}
\end{definition}

In this subsection we provide a randomized algorithm for computing a sparsifier $(\tilde G,\tilde E)$ for $\hat G$, in time $\Tilde{O}\left(\left(\frac{n^2}{\gamma} + n\gamma\lambda\right) \cdot \log^2(\wmax(G))\right)$. The algorithm may perform a single query to the $(\delta,20\wmax(G))$-subgraph oracle for $G$. We will ensure that, if the distinguished cut $(L,S,R)$ is good and the subgraph oracle does not err, then $(\tilde G,\tilde E)$ is a perfect sparsifier for $\hat G$.

Our algorithm for constructing the sparsifier for $\hat G$ is based on  a similar procedure from \cite{CT24} (see section 7 of \cite{CT24}.)
However, our setting is somewhat simpler, allowing for a simpler algorithm and a more concise correctness proof. The main difference is that we work directly with the vertex-weighted graph $\hat G$, whereas the algorithm of \cite{CT24} was applied to an edge-capacitated graph  corresponding to the residual flow network of the split-graph of the input graph $G$, with respect to a flow that their algorithm maintains.

We start by defining a partition of the edges of $\hat G$ into four subsets $E_1,\ldots,E_4$, that will play a central role in the construction of the sparsifier.

\paragraph{Set $E_1$.} The first set of edges $E_1$ contains all edges $(u,v)\in E(\hat G)$ with $u \in L'$. Recall, from \Cref{obs : cut-is-well-behaving-in-graph-with-reduced-weights}, if the distinguished cut is good, then cut $(L',S',R')$ is $(\lambda,\gamma)$-well-behaving for scale $4$ with respect to the cut $(L,S,R)$  in $\hat G$. In particular, from Property \ref{prop : P1} of  well-behaving cuts, $|E_1|=\vol^+_{\hat G}(L')\leq  \frac{100n^2}{\gamma}$ must hold in this case.

\paragraph{Set $E_2$.} The second set of edges $E_2$ contains all edges $(u,v) \in E(\hat G)$ with $v \in R'\setminus\set{y}$ and $\hat w(v) \leq 4$. Intuitively, in any maximum $x$-$y$ flow in $\hat G$, the edges of $E_2$ may only carry at most $4n$ flow units, so deleting them from $\hat G$ may only decrease the value of the minimum $x$-$y$ vertex-cut in $\hat G$ by at most $4n$. Eventually we will not include the edges of $E_2$ in the graph $\tilde G$.

\paragraph{Set $E_3$.} The third set of edges $E_3$ contains all edges $(u,v) \in E(\hat G)$ with $v \in R'$ that do not lie in $E_2$; in other words, either $v=y$, or $\hat w(v)>4$ must hold for all such edges $(u,v)$.
The following observation will be helpful in order to sparsify the edges of $E_3$.

\begin{observation}\label{obs: few edges of E3}
	Assume that the  distinguished cut $(L,S,R)$ is good, and let $u\in V(\hat G)$ be any vertex with $|\delta^+_{\hat G}(u)\cap E_3|>20\lambda\gamma+1$. Then $u\not\in L$.
\end{observation}
\begin{proof}
	Assume that the distinguished cut $(L,S,R)$ is good. Then, from 
	\Cref{obs : cut-is-well-behaving-in-graph-with-reduced-weights},
	the vertex-cut $(L',S',R')$ is $(\lambda,\gamma)$-well-behaving for scale $4$ with respect to the distinguished cut $(L,S,R)$ in $\hat{G}$, and, from Property   \ref{prop : P3}, set $(L \cup S) \setminus (L' \cup S')$ contains at most $20\lambda\gamma$ vertices $a$ with $\hat w(a)\geq 4$.
	
	Consider now some vertex $u\in V(\hat G)$ with $|\delta^+_{\hat G}(u)\cap E_3|>20\lambda\gamma+1$. Denote by: 
	
	$$U'=\set{v\neq y\mid (u,v)\in \delta^+_{\hat G}(u)\cap E_3}.$$
	
	 Recall that, by the definition of the set $E_3$ of edges, $U'\subseteq R'$ must hold, and every vertex $v\in U'$ has $\hat w(v)\geq 4$. Since all but at most $20\lambda\gamma$ vertices of $\set{a\in L\cup S\mid \hat w(a)\geq 4}$ lie in $L'\cup S'$, we get that $U'$ may contain at most $20\lambda\gamma$ vertices of $L\cup S$. Therefore, some vertex $v\in U'$ must lie in $R$. Since the edge $(u,v)$ lies in $\hat G$ and hence in $G$, we get that $u\not\in L$ must hold.
\end{proof}

\paragraph{Set $E_4$.} The fourth and the last set $E_4$ contains all remaining edges of $\hat G$, so $E_4=E(\hat G)\setminus\left(E_1\cup E_2\cup E_3\right )$. It is immediate to verify that, if $e=(u,v)\in E_4$, then $u\in S'\cup R'$ and $v\in L'\cup S'$ must hold. We will not include the edges of $E_4$ in the sparsifier $\tilde G$. Intuitively, if we consider any maximum $x$-$y$ flow in $\hat G$ and its flow-path decomposition $\pset$, then every flow-path $P\in \pset$ that contains an edge of $E_4$ must also contain at least two vertices of $S'$. We will use standard arguments to show that the deletion of such edges from $\hat G$ may not decrease the value of the maximum $x$-$y$ flow, and hence of the minimum $x$-$y$ vertex-cut, by too much.

We are now ready to describe the algorithm for constructing a sparsifier for the graph $\hat G$. We start by checking, in time $O\left(\frac{n^2}{\gamma}\right )$, whether $\vol^+_{\hat G}(L') \leq \frac{100n^2}{\gamma}$ holds. If this is not the case, then we terminate the algorithm and return the sparsifier $(\tilde G,\tilde E)$ for $\hat G$ with $V(\tilde G)=V(\hat G)$ and $E(\tilde G)=\tilde E=\emptyset$. In this case, we say that the algorithm \emph{terminated early}. Recall that, if the distinguished cut is good, then, from \Cref{obs : cut-is-well-behaving-in-graph-with-reduced-weights}, cut $(L',S',R')$ is $(\lambda,\gamma)$-well-behaving for scale $4$ with respect to the cut $(L,S,R)$  in $\hat G$, and, from Property \ref{prop : P1} of  well-behaving cuts, $\vol^+_{\hat G}(L')\leq  \frac{100n^2}{\gamma}$ must hold. Therefore, the algorithm may only terminate early at this step if the distinguished cut is bad.  We assume from now on that $\vol^+_{\hat G}(L')\leq  \frac{100n^2}{\gamma}$, and we construct the set $E_1=\delta^+_{\hat G}(L')$ of edges, in time $O\left(\frac{n^2}{\gamma}\right )$. Notice that $|E_1|\leq \frac{100n^2}{\gamma}$ must hold.

Next, we denote by $U=\set{v\in V(\hat G)\mid \hat w(v)\geq 4}$. 
 Recall that our algorithm is allowed to make one query to the
$(\delta,20\wmax(G))$-subgraph oracle for $G$ (see Definition \ref{def: subgraph oracle}), for $\delta=1-\frac{\log(30000\lambda\gamma \cdot \log n)}{\log n}$; observe that $n^{1-\delta}=30000\lambda\gamma \cdot \log n$. Recall also that graph $\hat G$ is identical to the graph $G$, except for the weights of its vertices, so querying the $(\delta,20\wmax(G))$-subgraph oracle for $G$ is equivalent to querying the $(\delta,20\wmax(G))$-subgraph oracle for $\hat G$. We query the $(\delta,20\wmax(G))$-subgraph oracle for $G$ with the set $Z=U\cap R'$ of vertices, and denote by 
 $\left(Y^h,Y^\ell\right)$ and $E' \subseteq E(\hat G)$ the response of the oracle to the query.
Recall the subgraph oracle guarantees that $|E'| \leq n^{2-\delta} \cdot \log^2(20\wmax(G)) \leq  \Tilde{O}\left(\lambda\gamma n \cdot \log^2(\wmax(G))\right)$.
Moreover, if the oracle does not err, then $E' = E_{\hat G}\left(Y^\ell,Z\right) = E_{\hat G}\left(Y^\ell,R' \cap U\right)$ must hold, and, for every vertex $v \in Y^h$:

\[ \left |\delta^+_{\hat G}(v)\cap E_3\right |=\left|N^+_{\hat G}(v) \cap \left(R' \cap U\right)\right|=|N^+_{\hat G}(v)\cap Z| \geq \frac{n^{1-\delta}}{1000\log n} \geq 30\lambda\gamma.\]

The following corollary then immediately follows from \Cref{obs: few edges of E3}.

\begin{corollary}\label{cor: valid shortcuts}
	Assume that the distinguished cut is good and that the subgraph oracle did not err in its response to the query. Then, for every vertex $v\in Y^h$, $v\not\in L$ must hold.
\end{corollary}

We construct the sparsifier $(\tilde G,\tilde E)$ of $\hat G$ as follows. We let $\tilde E=\set{(v,y)\mid v\in Y^h}$ and $V(\tilde G)=V(\hat G)$. We then let the set $E(\tilde G)$ of edges of $\tilde G$ contain all edges $e=(u,v)\in E'\cup E_1$ for which either $\hat w(v)\ge 4$ or $v=y$ holds.
From the above, discussion, $|E(\tilde G)|\leq |E_1|+|E'|\leq \Tilde{O}\left(\left(\frac{n^2}{\gamma}+n\lambda\gamma\right) \cdot \log^2(\wmax(G))\right)$,
and it is immediate to verify that the time required to construct the sparsifier, excluding the time that the subgraph oracle takes to compute its response to the query, is bounded by $\Tilde{O}\left(\left(\frac{n^2}{\gamma}+n\lambda\gamma\right) \cdot \log^2(\wmax(G))\right)$. It is immediate to verify that $(\tilde G,\tilde E)$ is indeed a valid sparsifier for $\hat G$.

It now remains to show that, if the distinguished cut is good, and if the subgraph oracle did not err in its response of the query, then $(\tilde G,\tilde E)$ is a perfect sparsifier for $\hat G$. We establish this in the following claim.

\begin{claim}\label{claim: perfect sparsifier}
	Assume that the distinguished cut $(L,S,R)$ is good and that the 
subgraph oracle did not err in its response of the query. Then 
$(\tilde G,\tilde E)$ is a perfect sparsifier for $\hat G$.
\end{claim}
\begin{proof}
	Assume that the distinguished cut $(L,S,R)$ is good and that the 
	subgraph oracle did not err in its response of the query. We have already shown that $(\tilde G,\tilde E)$ is a valid sparsifier for $\hat G$. Moreover, from \Cref{cor: valid shortcuts}, no vertex of $L$ may lie in $Y^h$, and so no vertex of $L$ may serve as an endpoint of an edge in $\tilde E$. 
	For convenience, we denote $J=\tilde G\cup \tilde E$, and we denote by $\opt_{J}$ the value of the minimum $x$-$y$ vertex-cut in  graph $J$. It now remains to show that $\opt_{J}\geq \hat w(S)-48n$. Recall that we denoted by $\opt_{\hat{G}}$ the value of the minimum $x$-$y$ vertex-cut in $\hat G$.
	We use the following claim in order to complete the proof.
	
	\begin{claim}\label{claim: perfect sparsifier inner}
		Assume that the distinguished cut $(L,S,R)$ is good and that the 
		subgraph oracle did not err in its response of the query. Then  $\opt_{J}\geq 2\opt_{\hat G}-\hat w(S')-4n$.
	\end{claim}
	
	We provide the proof of \Cref{claim: perfect sparsifier inner} below, after we complete the proof of \Cref{claim: perfect sparsifier} using it.
	 Recall that, since we have assumed that the distinguished cut $(L,S,R)$ is good, from \Cref{obs : cut-is-well-behaving-in-graph-with-reduced-weights}, the vertex-cut $(L',S',R')$ is $(\lambda,\gamma)$-well-behaving for scale $4$ with respect to the distinguished cut $(L,S,R)$ in $\hat{G}$. From Property \ref{prop : P2} of well-behaving cuts, $\hat w(S') \leq \hat w(S) + 40n$ must then hold. Moreover, from \Cref{obs : value-of-distinguish-cut-in-hat-G-is-not-far-from-optimal}
	 $\hat{w}(S) \leq \opt_{\hat{G}} + 2n$. Altogether, we get that:
	
	\[\begin{split}
	\opt_{J}&\geq 2\opt_{\hat G}-\hat w(S')-4n\\
	&\geq 2\opt_{\hat G}-\hat w(S) - 44n\\
	&\geq \hat w(S)-48n.
	\end{split}
	\]

In order to complete the proof of \Cref{claim: perfect sparsifier} it is now enough to prove \Cref{claim: perfect sparsifier inner}, which we do next.

\begin{proofof}{\Cref{claim: perfect sparsifier inner}}
Let $\edel_1$ denote the set of all edges $e=(u,v)\in E(\hat G)$, such that $v\neq y$ and $\hat w(v)\leq 4$ holds.
Let $\edel_2$ denote the set of all edges 
$e\in E_4\setminus \edel_1$; note that, equivalently, set $\edel_2$ contains all edges $e=(u,v)\in E(\hat G)\setminus \edel_1$ with $u\in S'\cup R'$ and $v\in L'\cup S'$. Lastly, let $\edel_3$ denote the set of all edges $e=(u,v)\in E(\hat G)\setminus \left(\edel_1\cup \edel_2\right )$ with $u\in Y^h$ and $v\neq y$. It is easy to verify that $J$ is precisely the graph that is obtained from $\hat G$ by deleting all edges of $\edel_1\cup\edel_2\cup \edel_3$ from it, and inserting the edges of $\tilde E$ into it.

Consider now a maximum $x$-$y$ flow $f$ in $\hat G$ that obeys the vertex capacities defined by their weights $\hat w(v)$; from the max-flow / min-cut theorem, $\val(f)=\opt_{\hat G}$. Let $\pset$ be the flow-path decomposition of  $f$. We partition the set $\pset$ into four subsets: set $\pset_1$ containing all paths $P\in \pset$ with $E(P)\cap \edel_1\neq \emptyset$; set $\pset_2$ containing all paths $P\in \pset\setminus \pset_1$ with $E(P)\cap \edel_2\neq \emptyset$;  set $\pset_3$ containing all paths $P\in \pset\setminus \left(\pset_1\cup \pset_2\right )$ with $E(P)\cap \edel_3\neq \emptyset$; and set $\pset_4$ containing all remaining paths of $\pset$.

It is immediate to verify that every path $P\in \pset_4$ is contained in $J$. Consider now some path $P\in \pset_3$, and let $u(P)$ be the first vertex of $P$ that lies in $Y^h$; such a vertex must exist by the definition of the path set $\pset_3$. Recall that edge $(u(P),y)$ lies in $\tilde E$ and hence in $J$, and, from the definition of the set $\pset_3$ of paths, the subpath of $P$ from $x$ to $u(P)$ is contained in $J$. We denote by $P'$ the path obtained by concantenating the subpath of $P$ from $x$ to $u(P)$ with the edge $(u(P),y)$. From our discussion, $P'\subseteq J$.

Consider now the flow $f'$ in $J$ that is defined as follows: for every path $P\in \pset_4$, we send $f(P)$ flow units on $P$, and, for every path $P\in \pset_3$, we send $f(P)$ flow units on $P'$. It is immediate to verify that $f'$ is a valid $x$-$y$ flow in $J$, and, moreover:

\[\val(f')=\sum_{P\in \pset_3\cup \pset_4}f(P)=\val(f)-\sum_{P\in \pset_1\cup \pset_2}f(P)=\opt_{\hat G}-\sum_{P\in \pset_1\cup \pset_2}f(P).\]

Note that every path $P\in \pset_1$ contains some vertex $v\in V(\hat G)\setminus\set{x,y}$ with $\hat w(v)\leq 4$ as an inner vertex. Therefore, $\sum_{P\in \pset_1}f(P)\leq 4n$ must hold.

Lastly, we show that $\sum_{P\in \pset_2}f(P)\leq \hat w(S')-\opt_{\hat G}$. Consider any path $P\in \pset_2$. Recall that $P$ connects $x$ to $y$ in $\hat G$, and that $(L',S',R')$ is an $x$-$y$ vertex-cut in $\hat G$. Moreover, from the definition of the path set $\pset_2$, there must be an edge $e=(u,v)\in \edel_2$ on path $P$, so that $u\in S'\cup R'$ and $v\in L'\cup S'$ holds. We claim that $P$ must contain at least two vertices of $S'$. Indeed, consider the subpath $P_1$ of $P$ from $x$ to $u$ and the subpath $P_2$ of $P$ from $v$ to $y$. Since $u\in  S'\cup R'$, and since $(L',S',R')$ is a valid vertex-cut in $\hat G$ with $x\in L'$, path $P_1$ must contain a vertex of $S'$. Since $v\in S'\cup L'$ and $y\in R'$, it must be the case that some vertex of $S'$ lies on $P_2$. Since the paths $P_1,P_2$ are disjoint in their vertices, $P$ must contain at least two vertices of $S'$. We conclude that every path $P\in \pset_2$ must contain at least two vertices of $S'$. 
Since vertex set $S'$ separates $x$ from $y$ in $\hat G$, we get that:

\[\hat w(S')\geq \sum_{P\in \pset}f(P)+\sum_{P\in \pset_2}f(P)=\opt_{\hat G}+\sum_{P\in \pset_2}f(P).\]

We conclude that $\sum_{P\in \pset_2}f(P)\leq \hat w(S')-\opt_{\hat G}$, and, altogether:

\[\val(f')=\opt_{\hat G}-\sum_{P\in \pset_1}f(P)-\sum_{P\in \pset_2}f(P)\geq 2\opt_{\hat G}-\hat w(S')-4n.\]

From the max-flow/min-cut theorem, the value of the minimum $x$-$y$ vertex-cut in $J$ is at least $2\opt_{\hat G}-\hat w(S')-4n$.
\end{proofof}
\end{proof}

\paragraph{Condition (C).}
In the remainder of the algorithm, for brevity, we say that Condition (C) holds, if the distinguished cut is good and the subgraph oracle did not err in its response to the query.
Recall that, if the distinguished cut $(L,S,R)$ is good, then $|L| \leq \lambda$, and the input vertex-cut $(L',S',R')$ is $(\lambda,\gamma)$-well-behaving for scale $(2M')$ with respect to $(L,S,R)$ in $G$, and, from \Cref{obs : cut-is-well-behaving-in-graph-with-reduced-weights}, it is also  $(\lambda,\gamma)$-well-behaving for scale $4$ with respect to   $(L,S,R)$ in $\hat{G}$.
From \Cref{claim: perfect sparsifier}, if Condition (C) holds, then
$(\tilde G,\tilde E)$ is a perfect sparsifier for $\hat G$.

\subsubsection{Step 2: Initial Flow and Auxiliary Graphs}

In this step, we compute an initial flow $f_0$ in graph $\hat G\cup \tilde E$, and define a series of auxiliary graphs, that will play a key role in computing the set $A$ of vertices in the third and final step of the algorithm.

We start by applying the algorithm from \Cref{thm : computing-flow-with-few-edges} to the graph $\tilde G\cup \tilde E$, with the pair $x,y$ of its vertices.
Recall that the expected running time of the algorithm is: 

\[O\left(\left(|E(\tilde G)|+|\tilde E|\right )^{1+o(1)}\cdot\log(\wmax(G))\right )
\leq O\left(\left(\frac{n^{2+o(1)}}{\gamma}+n^{1+o(1)}\cdot \lambda\gamma\right) \cdot \log^{3+o(1)}(\wmax(G))\right)
.\] 

The algorithm computes  a maximum $x$-$y$ flow $\tilde G\cup \tilde E$, that we denote by $f_0$, with the following properties: the flow $f_0$ obeys the vertex capacities (defined by their weights in $\tilde G$); it is integral, and, moreover, the total number of edges $e\in E(\tilde G)\cup \tilde E$ with $f_0(e)>0$ is at most $4n$. We assume without loss of generality that the algorithm also produces a list $E^f\subseteq E(\tilde G)\cup \tilde E$  of all edges $e$ with $f_0(e)>0$.

For convenience, we denote by $J=\hat G\cup \tilde E$. Clearly, $\tilde G\cup \tilde E\subseteq J$, so we can also view $f_0$ as a valid $x$-$y$ flow in $J$, which we do from now on. We also denote by $\opt_J$ the value of the minimum $x$-$y$ vertex-cut in $J$. We will use the following simple observation.

\begin{observation}\label{obs: flow in $J$}
	Assume that Condition (C) holds. Then $(L,S,R)$ is a valid $x$-$y$ vertex-cut in graph $J$ and, moreover, $\val(f_0)\geq \hat w(S)-48n$.
\end{observation}
\begin{proof}
	Assume that Condition (C) holds. Then,
	from \Cref{claim: perfect sparsifier},  $(\tilde G,\tilde E)$ is a perfect sparsifier for $\tilde G$. In particular, 
	no vertex of $L$ serves as an endpoint of an edge in $\tilde E$, and
	 the value of the minimum $x$-$y$ vertex-cut in the graph $\tilde G\cup \tilde E$ is at least $\hat w(S)-48n$. Recall that $(L,S,R)$ is a valid $x$-$y$ vertex-cut in $G$, and hence in $\hat G$. Since $J=\hat G\cup \tilde E$, and since no vertex of $L$ serves as an endpoint of an edge in $\tilde E$, we get that $(L,S,R)$ is a valid $x$-$y$ vertex-cut in $J$.
	 
	 Since the value of the minimum $x$-$y$ vertex-cut in the graph $\tilde G\cup \tilde E$ is at least $\hat w(S)-48n$,
	 from the max-flow/min-cut theorem, the value of the maximum $x$-$y$ flow in 
	 $\tilde G\cup \tilde E$ is at least $\hat w(S)-48n$, and so $\val(f_0)\geq \hat w(S)-48n$.
\end{proof}

\paragraph{Graph $J'$ and flow $f'_0$.}
We let $J'$ be the split-graph corresponding to the graph $J$. Recall that we denoted by $G'$ the split-graph of the input graph $G$. Note that graph $J'$ can be equivalently obtained from the graph $G'$ as follows: first, for every vertex $v\in V(G)$, we change the capacity of the corresponding special edge $(v^{\inn},v^{\out})$ from $w(v)$ to $\hat w(v)$. Next, for every regular edge $(v^{\out},u^{\inn})\in E(G')$, we change the capacity of the edge from $\wmax(G)$ to $\wmax(J)=\wmax(\hat G)$. Lastly, for every edge $e=(a,y)\in \tilde E$, we insert the edge $(a^{\out},y^{\inn})$ of capacity $\wmax(\hat G)$ into the graph.

Let $V^{\inn}=\set{v^{\inn}\mid v\in V(G)}$ and $V^{\out}=\set{v^{\out}\mid v\in V(G)}$. Notice that $V(J')=V^{\inn}\cup V^{\out}$, and that every vertex $v^{\inn}\in V^{\inn}$ has exactly one outgoing edge in $J'$ -- the special edge $(v^{\inn},v^{\out})$. For every vertex $v^{\out}\in V^{\out}$, every edge leaving $v^{\out}$ in $J'$ connects it to some vertex $u^{\inn}\in V^{\inn}$, and, unless $u=y$, this edge must also lie in $G'$. In particular, every edge of $J'$ connects a vertex of $V^{\inn}$ to a vertex of $V^{\out}$ or the other way around.

Recall that our algorithm is given as input the adjacency-list representation of the graph $G'$, that we denote by $\dset$. We modify $\dset$ in order to simulate the adjacency-list representation of the graph $J'$, as follows. For every vertex $v\in V(G)$, we change the capacity of the corresponding edge $e_v=(v^{\inn},v^{\out})$ from $w(v)$ to $\hat w(v)$ in the lists $\OUT(v^{\inn})$ and $\IN(v^{\out})$, and we mark the edge $e_v$ in both these lists to indicate that its capacity is stored explicitly. For every edge $e=(a,y)\in \tilde E$, we insert the edge $(a^{\out},y^{\inn})$ of capacity $\wmax(\hat G)$ into the lists $\OUT(a^{\out})$ and $\IN(y^{\inn})$, and mark the edge in both these lists to indicate that its capacity is stored explicitly. We also compute the value $\wmax(J)$. For regular edges $(v^{\out},u^{\inn})\in E(G')$, we do not change their capacities in the lists $\OUT(v^{\out}),\IN(u^{\inn})$. Since we did not mark these edges to indicate that their capacities are stored explicitly, whenever an algorithm tries to access such an edge $e$, we report that its capacity is $\wmax(J)$ to the algorithm. Note that all the above modifications to $\dset$ can be performed in time $O(n)$, and we can now use $\dset$ in order to simulate the adjacency-list representation of the graph $J'$.

Recall that we have computed an $x$-$y$ flow $f_0$ in the graph $J$. This flow naturally defines an $x^{\out}$-$y^{\inn}$ flow $f'_0$ in the graph $J'$, that obeys the edge capacities, with $\val(f'_0)=\val(f_0)$, as follows. Consider any edge $e=(u,v)\in E(J)$ with $f_0(e)>0$ (recall that all such edges are stored in the list $E^f$). We set $f'_0(u^{\out},v^{\inn})=f_0(u,v)$. Additionally, we store the edge 
$(u^{\out},v^{\inn})$ together with the vertex $u\in V(G)$. Once we process all edges of $E^f$ in this manner, we consider every vertex $v\in V(G)\setminus\set{x,y}$ one by one. For each such vertex $v$, we set $f'_0(v^{\inn},v^{\out})=\sum_{e\in \delta^+_{J'}(v^{\out})}f'_0(e)$; recall that all edges $e\in \delta^+_{J'}(v^{\out})$ with $f'_0(e)>0$ are explicitly stored with $v$. We denote by $E^{f'}$ the set of all edges $e\in E(J')$ with $f'_0(e)>0$. It is immediate to verify that $f'_0$ is indeed a valid $x^{\out}$-$y^{\inn}$ flow in $J'$, that obeys the edge capacities, with $\val(f'_0)=\val(f_0)$, and that $|E^{f'}|\leq 2|E^f|\leq 8n$. Moreover, given the flow $f_0$ and the list $E^f$, we can compute the set $E^{f'}$ of edges of $J'$ and the flow value $f'_0(e)$ for every edge $e\in E^{f'}$, in time $O(n)$.
For every edge $e\in E(J')$, we denote its capacity in $J'$ by $c_{J'}(e)$.

Next, we define a partition $(X,Y)$ of the vertices of $G'$ (and hence of $J'$) as follows: $X=L^*\cup S^{\inn}$ and $Y=S^{\out}\cup R^*$. We will use the following simple corollary of \Cref{obs: flow in $J$}.

\begin{corollary}\label{cor: flow in $J'$}
	Partition $(X,Y)$ of $V(J')$ is a valid $x^{\out}$-$y^{\inn}$ edge-cut in graph $J'$. Moreover, if Condition (C) holds, then $\val(f'_0)\geq \sum_{e\in E_{J'}(X,Y)}c_{J'}(e)-48n$.
\end{corollary}
\begin{proof}
	Since $(L,S,R)$ is an $x$-$y$ cut in $G$, it is immediate to verify that $x^{\out}\in X$ and $y^{\inn}\in Y$ must hold, so 
	$(X,Y)$ is a valid $x^{\out}$-$y^{\inn}$ edge-cut in graph $J'$.
	
	Assume now that Condition (C) holds.
	Then, from \Cref{obs: flow in $J$}, 
	$(L,S,R)$ is a valid $x$-$y$ vertex-cut in graph $J$.
	It is then easy to verify that $E_{J'}(X,Y)=\set{(v^{\inn},v^{\out}\mid v\in S}$, and that, for every vertex $v\in S$, $c_{J'}(v^{\inn},v^{\out})=\hat w(v)$.   Since, from \Cref{obs: flow in $J$},
 $\val(f_0)\geq \hat w(S)-48n$, and since $\val(f'_0)=\val(f_0)$, we get that $\val(f'_0)\geq \hat w(S)-48n= \sum_{e\in E_{J'}(X,Y)}c_{J'}(e)-48n$.
\end{proof}

\paragraph{Graph $Z$.}
We let $Z$ be the residual flow network of the graph $J'$ with respect to the flow $f'_0$. We modify data structure $\dset$ so that it can be used in order to provide the algorithm with access to the adjacency-list representation of the graph $Z$, as follows. For every edge $e=(u,v)\in E^{f'}$, we set the capacity of the edge $e$ in $\dset$ to $c_{J'}(e)-f'_0(e)$, and we insert the edge $(v,u)$ of capacity $f'_0(e)$ into the lists $\OUT(u)$ and $\IN(v)$. We mark both these edges to indicate that their capacities are stored in $\dset$ explicitly. It is now immediate to verify that $\dset$ can be used in order to simulate access to the adjacency-list representation of the graph $Z$. 
For every edge $e\in E(Z)$, we denote its capacity in $Z$ by $c_{Z}(e)$.
We will use the following simple corollary of \Cref{cor: flow in $J'$}.

\begin{corollary}\label{cor: flow in $Z$}
	The partition $(X,Y)$ of $V(J')=V(Z)$ is a valid $x^{\out}$-$y^{\inn}$ edge-cut in graph $Z$.
	Moreover, if Condition (C) holds, then $\sum_{e\in E_{Z}(X,Y)}c_{Z}(e)\leq 48n$.
\end{corollary}
\begin{proof}
	From \Cref{cor: flow in $J'$}, $(X,Y)$ is a valid $x^{\out}$-$y^{\inn}$ edge-cut in graph $J'$. It is then immediate to verify that it is also a valid  $x^{\out}$-$y^{\inn}$ edge-cut in graph $Z$.
	
	Assume now that Condition (C) holds. Then,  from \Cref{cor: flow in $J'$}, $\sum_{e\in E_{J'}(X,Y)}c_{J'}(e)\leq \val(f'_0)+48n$, and from \Cref{fact : residual-capacity-of-edge-cut}, $c_Z(X,Y)=c_{J'}(X,Y)-\val(f'_0)\leq 48n$.
\end{proof}

We will also use the following simple obervation.

\begin{observation}\label{obs: bounding volume of in-vertices}
	$\sum_{v^{\inn}\in V^{\inn}}\deg^+_Z(v^{\inn})\leq 5n$.
\end{observation}
\begin{proof}
	Consider any vertex $v^{\inn}\in V^{\inn}$. Recall that, in graph $G'$, there is exactly one edge leaving $v^{\inn}$ -- the special edge $(v^{\inn},v^{\out})$. If $e=(v^{\inn},a^{\out})$ is an edge of $Z$ that leaves $v^{\inn}$ and is different from $(v^{\inn},v^{\out})$, then it must be the case that $f'_0(a^{\out},v^{\inn})>0$, and so $f_0(a,v)\in E^f$. Therefore, if $n_v$ is the total number of edges in $E^f$ that enter $v$, then $\deg^+_Z(v)\leq 1+n_v$. Altogether, we get that 
	$\sum_{v^{\inn}\in V^{\inn}}\deg^+_Z(v^{\inn})\leq n+|E^f|\leq 5n$.
\end{proof}

\paragraph{Graph $\hat Z$.}
We let $\hat Z$ be a graph that is identical to the graph $Z$, except that, every edge $e\in E(Z)$, its capacity in graph $\hat Z$ is set to $c_{\hat Z}(e)=\ceil{\frac{n}{\gamma}}\cdot c_Z(e)$.  

Note that data structure $\dset$ can be used to simulate access to the adjacency-list representation of the graph $\hat Z$ in a straightforward manner: we do not update the capacities of the edges of $E(Z)$; instead, when such an edge $e$ is accessed, and its capacity $c_Z(e)$ is recovered, we report its new scaled capacity  $c_{\hat Z}(e)=\ceil{\frac{n}{\gamma}}\cdot c_Z(e)$. 

\paragraph{Graph $\hat Z'$.}
The last graph that we define is graph $\hat Z'$. This graph is obtained from graph $\hat Z$ by adding,  for every vertex $v\in V(\hat Z)\setminus\set{y^{\inn}}$ with $\deg^+_{\hat Z}(v)>0$, an edge $(v,y^{\inn})$, of capacity $\deg^+_{\hat Z}(v)$.
For every edge $e\in E(\hat Z')$, we denote its capacity in $\hat Z'$ by $c_{\hat Z'}(e)$.

%
Let $\Delta=\frac{50n^2}{\gamma}$. We will use the following easy corollary of \Cref{cor: flow in $Z$}.

\begin{corollary}\label{cor: flow in hat Z}
	The partition $(X,Y)$ of $V(\hat Z')$ is a valid $x^{\out}$-$y^{\inn}$ edge-cut in $\hat Z'$.
	Moreover, if Condition (C) holds, then $\sum_{e\in E_{\hat Z'}(X,Y)}c_{\hat Z'}(e)< \Delta$.
\end{corollary}
\begin{proof}
	From \Cref{cor: flow in $Z$}, $(X,Y)$ is a valid $x^{\out}$-$y^{\inn}$ edge-cut in graph $Z$, so it is immediate to verify that it is also a valid $x^{\out}$-$y^{\inn}$ edge-cut in $\hat Z'$.
	
	Assume now that Condition (C) holds. Then, from \Cref{cor: flow in $Z$}, $\sum_{e\in E_{Z}(X,Y)}c_{Z}(e)\leq 48n$. 
		We partition the edges of  $E_{\hat Z'}(X,Y)$ into two subsets: set $E_1=E_{\hat Z'}(X,Y)\cap E_Z(X,Y)$ and set $E_2=E_{\hat Z'}(X,Y)\setminus E_1$. Since, for every edge $e\in E(Z)$, its capacity in graph $\hat Z'$ is $c_{\hat Z'}(e)=\ceil{\frac{n}{\gamma}}\cdot c_Z(e)$, we get that:

\[ \sum_{e\in E_1}c_{\hat Z'}(e)=\sum_{e\in E_Z(X,Y)} \ceil{\frac{n}{\gamma}}\cdot c_Z(e)\leq \ceil{\frac{n}{\gamma}}\cdot 48n. \]

Recall that, if Condition (C) holds, then the distinguished cut $(L,S,R)$ is good, and, in particular, $|L|\le \lambda$ must hold. Since $X=L^*\cup S$, we get that $|X\cap V^{\out}|\leq \lambda$ in this case. Clearly, for every vertex $v^{\out}\in V^{\out}$, $\deg^+_Z(v^{\out})\leq |V^{\inn}|\leq n$. From \Cref{obs: bounding volume of in-vertices}, $\sum_{v^{\inn}\in X\cap V^{\inn}}\deg^+_Z(v^{\inn})\leq \sum_{v^{\inn}\in  V^{\inn}}\deg^+_Z(v^{\inn})\leq 5n$. Altogether, we get that:

\[\sum_{e\in E_2}c_{\hat Z'}(e)\leq \sum_{v\in X}\deg^+_Z(v)\leq \lambda\cdot n+5n.\]

Finally, recall that $\lambda\cdot\gamma\leq \frac{n}{100}$. Therefore, altogether:

\[\begin{split}
\sum_{e\in E_{\hat Z'}(X,Y)}c_{\hat Z'}(e)&\leq \sum_{e\in E_1}c_{\hat Z'}(e)+\sum_{e\in E_2}c_{\hat Z'}(e)\\
&\leq  \ceil{\frac{n}{\gamma}}\cdot 48n+\lambda\cdot n+5n\\
&\leq \frac{48n^2}{\gamma}+53n+\frac{n^2}{10\gamma}
\\& < \frac{50n^2}{\gamma}\\
&=\Delta.
\end{split}\]
\end{proof}

Note that the running time of the algorithm for Step 2 is dominated by the time required to compute the initial flow $f_0$ in graph $\tilde G\cup \tilde E$, which is bounded by  $O\left(\left(\frac{n^{2+o(1)}}{\gamma}+n^{1+o(1)}\cdot \lambda\gamma\right) \cdot \log^{3+o(1)}(\wmax(G))\right)$.

\subsubsection{Step 3: Computing the Vertex Set $A$}
\label{subsec: step 3: quasi promising}

In this step we compute a set $A\subseteq V(\hat G)$ of vertices with $x\in A$ and $y\not\in A\cup N^+_G(A)$. We will ensure that, if Condition (C) holds, then vertex set $A$ is quasi-promising for $\hat{G}$ (see \Cref{def : quasi-promising}). From \Cref{claim: from quasi to promising}, it will then follow that, if Condition (C) holds, then $A$ is also a promising vertex set for graph $G$ and scale $M'$.

As our first step, we apply the algorithm for local flow from \Cref{thm : unit-flow} to graph $\hat Z$, with the pair $s=x^{\out}$ and $t=y^{\inn}$ of its vertices, parameter $\Delta$ that remains unchanged, and parameter $h=2\lambda+4$. We use data structure $\dset$ in order to simulate access to the adjacency-list representation of graph $\hat Z$ for the algorithm. 
Since $\Delta=\frac{50n^2}{\gamma}>4n$ (as $\gamma\leq \frac{n}{100}$ must hold), it is immediate to verify that $\deg^+_{\hat Z}(x^{\out})\leq \Delta$, as required.
Recall that the running time of the algorithm is $\tilde O(h\cdot \Delta)\leq \tilde O\left(\frac{\lambda\cdot n^2}{\gamma}\right )$, since  $\Delta=\frac{50n^2}{\gamma}$.
Recall also that the algorithm returns an $x^{\out}$-$y^{\inn}$ flow $f$ in $\hat Z'$, by explicitly listing the flow values $f(e)$ for all edges $e\in E(\hat Z')$ with $f(e)>0$, and a subset $A'\subseteq V(\hat Z')\setminus\set{y^{\inn}}$ of vertices with $x^{\out}\in A'$ and $\vol^+_{\hat Z}(A')\leq \Delta$. If $\val(f)\geq \Delta$, then we return the vertex set $A=\set{x}$. Notice that, from \Cref{cor: flow in hat Z}, if Condition (C) holds, then there is an $x^{\out}$-$y^{\inn}$ edge-cut in $\hat Z'$, whose value is less than $\Delta$. Therefore, if $\val(f)\ge \Delta$, then Condition (C) does not hold, so either the distinguished cut is bad, or the subgraph oracle erred in its response to the query. Since the edge $(x,y)$ does not lie in $G$, we get that $y\not\in A\cup N^+_G(A)$, so vertex set $A$ is a valid outcome for the algorithm from \Cref{thm : algorithm-for-step-1} in this case.

We assume from now on that $\val(f)<\Delta$. Let $H=\hat Z'_{f}$ be the residual flow network of the graph $\hat Z'$ with respect to the flow $f$. Then the algorithm from \Cref{thm : unit-flow} guarantees that every path in $H$ from $x^{\out}$ to $V(H)\setminus A'$ has length at least $h$.

We compute the set $A\subseteq V(G)$ of vertices as follows: we include in $A$ every vertex $v\in V(G)\setminus \left(\set{y}\cup N^-_G(y)\right )$ with $v^{\out}\in A'$. We then return the vertex set $A$ as the outcome of the algorithm.

Note that the total running time of Step 3 is $\tilde O\left(\frac{\lambda\cdot n^2}{\gamma}+n\right )\leq \tilde O\left(\frac{\lambda\cdot n^2}{\gamma}\right )$, while the running times of Steps 1 and 2 of the algorithm are bounded by: 

$$\Tilde{O}\left(\left(\frac{n^2}{\gamma}+n\lambda\gamma\right) \cdot \log^2(\wmax(G))\right)$$ 

and  

$$O\left(\left(\frac{n^{2+o(1)}}{\gamma}+n^{1+o(1)}\cdot \lambda\gamma\right) \cdot \log^{3+o(1)}(\wmax(G))\right),$$ 

respectively. Therefore, the total running time of the algorithm is bounded by: 

$$O\left(\left(\frac{n^{2+o(1)}\cdot\lambda}{\gamma}+n^{1+o(1)}\cdot \lambda\gamma\right) \cdot \log^4(\wmax(G))\right).$$

From the construction of the vertex set $A$, and since $x^{\out}\in A'$ and $(x,y)\not\in E(G)$, it is immediate to verify that  $x \in A$ and $y \notin A \cup N^+_G(A)$. 
Next, we prove that, if Condition (C) holds, then vertex set $A$ is quasi-promising in $\hat G$ (see \Cref{def : quasi-promising}).
We start with the following simple observation bounding  $\vol^+_{\hat G}(A)$.

\begin{observation}\label{obs: volume of A}
	 $\vol^+_{\hat G}(A) \leq \frac{100n^2}{\gamma}$.
\end{observation}
\begin{proof}
Recall that the algorithm from \Cref{thm : unit-flow} guarantees that $\vol^+_{\hat Z}(A')\leq \Delta=\frac{50n^2}{\gamma}$.
Consider now any vertex $v\in A$, and recall that $v^{\out}\in A'$ must hold. Let $e=(v,u)\in \delta^+_G(v)$ be any edge. Then edge $e'=(v^{\out},u^{\inn})$ must lie in the split graph $G'$, and hence in $J'$. Since graph $\hat Z$ is obtained from the residual flow network $Z$ of graph $J'$ with respect to the flow $f'_0$ by only making modifications to the capacities, either $e'\in E(\hat Z)$, or $f'_0(e')>0$ must hold. We conclude that, if $e=(v,u)$ is an edge in $\delta^+_G(v)$, and $(v^{\out},u^{\inn})$ is not an edge in $E(\hat Z)$, then $f'_0(v^{\out},u^{\inn})>0$ must hold. Therefore:

\[\vol^+_{\hat G}(A)\leq \vol^+_{\hat Z}(A')+|E^f|\leq \frac{50n^2}{\gamma}+4n\leq \frac{100n^2}{\gamma},\]

since $\gamma\leq\frac{n}{100}$, so $4n\leq \frac{4n^2}{\gamma}$.
\end{proof}

Next, we prove that, if Condition (C) holds, then there is a set $\Tilde{L} \subseteq A$ of vertices, such that, if we denote by $(\Tilde{L},\Tilde{S},\Tilde{R})$ the tripartition of vertices of $\hat G$ induced by $\Tilde{L}$, then  $\hat w(\Tilde{S}) \leq \hat w(S) + 6\lambda\gamma$; $\Tilde{L} \subseteq L$; and
 $(\Tilde{L},\Tilde{S},\Tilde{R})$ is an $x$-$y$ vertex-cut in $\hat G$.	
 Recall that we defined an $x^{\out}$-$y^{\inn}$ edge-cut $(X,Y)$ in $\hat Z'$, with $X=L^*\cup S^{\inn}$ and $Y=S^{\out}\cup R^*$.
 The following claim is key to completing the analysis of our algorithm.

 \begin{claim}\label{claim: good cut in $H$}
 Assume that Condition (C) holds. Then there is an $x^{\out}$-$y^{\inn}$ edge-cut $(\Tilde{X},\Tilde{Y})$ in $Z$ with  $\Tilde{X} \subseteq X \cap A'$ and
 		 $c_{Z}(\Tilde{X},\Tilde{Y}) \leq c_{Z}(X,Y) + 6\lambda\gamma$.
 \end{claim}

We prove Claim \ref{claim: good cut in $H$} below, after we complete the proof of \Cref{thm : algorithm-for-step-1} using it.
We obtain the following corollary of \Cref{claim: good cut in $H$}.

\begin{corollary}\label{cor: quasi-promising}
	Assume that Condition (C) holds. Then there is a set $\Tilde{L} \subseteq A$ of vertices, such that, if we denote by $(\Tilde{L},\Tilde{S},\Tilde{R})$ the tripartition of vertices of $\hat G$ induced by $\Tilde{L}$, then  $\hat w(\Tilde{S}) \leq \hat w(S) + 6\lambda\gamma$; $\Tilde{L} \subseteq L$; and
	$(\Tilde{L},\Tilde{S},\Tilde{R})$ is an $x$-$y$ vertex-cut in $\hat G$.
\end{corollary}
\begin{proof}
	We assume that Condition (C) holds. Let $(\tilde X,\tilde Y)$ be the $x^{\out}$-$y^{\inn}$ edge-cut in $Z$ that is guaranteed to exist from \Cref{claim: good cut in $H$}. Recall that 
		$\Tilde{X} \subseteq X \cap A'$ and
		$c_{Z}(\Tilde{X},\Tilde{Y}) \leq c_{Z}(X,Y) + 6\lambda\gamma$ must hold. Note that $(\Tilde{X},\Tilde{Y})$ is also a valid $x^{\out}$-$y^{\inn}$ edge-cut in graph $J'$. In the following observation, we bound the value of this cut in $J'$ in terms of $\hat w(S)$.
		
\begin{observation}\label{cl : properties-of-existential-cut-in-split-graph}
Assume that Condition (C) holds. Then $c_{J'}(\Tilde{X},\Tilde{Y}) \leq \hat w(S) +6\lambda\gamma$, and in particular, $c_{J'}(\Tilde{X},\Tilde{Y})<\wmax(\hat G)$.
\end{observation}

\begin{proof}
	Recall that the graph $Z$ is the residual flow network of the graph $J'$ with respect to the flow $f'_0$. 
	From \Cref{fact : residual-capacity-of-edge-cut}, $c_Z(\tilde X,\tilde Y)=c_{J'}(\tilde X,\tilde Y)-\val(f'_0)$, and similarly, $c_Z(X,Y)=c_{J'}(X,Y)-\val(f'_0)$. Recall that, if Condition (C) holds, then, from \Cref{obs: flow in $J$}, 
	$(L,S,R)$ is a valid $x$-$y$ vertex-cut in graph $J$. From the definition of the edge-cut $(X,Y)$ in $J'$, it is immediate to verify that $c_{J'}(X,Y)=\hat w(S)$. Altogether, we get that:
	
	\[
	\begin{split}
	c_{J'}(\Tilde{X},\Tilde{Y})&=	c_Z(\tilde X,\tilde Y)+\val(f'_0)\\
	&\leq c_Z(X, Y)+6\lambda\gamma+\val(f'_0)\\
	&=c_{J'}(X,Y)+6\lambda\gamma\\
	&=\hat w(S)+6\lambda\gamma.
	\end{split}
	\]
	
	Moreover, since $\lambda\gamma\leq\frac{n}{100}$, and since $\wmax(\hat G)\geq 2n\cdot \max_{v\in V(\hat G)}\set{\hat w(v)}$, we get that $c_{J'}(\Tilde{X},\Tilde{Y})\leq \hat w(S)+6\lambda\gamma<\wmax(\hat G)$.
\end{proof}

We let $\tilde L=\set{v\in V(G)\setminus \set{y}\mid v^{\out}\in \tilde X}$. Consider the tripartition $(\tilde L,\tilde S,\tilde R)$  of $V(\hat G)$ induced by $\tilde L$. Recall that $J'$ is the split-graph of the graph $J=\hat G\cup \tilde E$.
By \Cref{obs: existential from split to regular} $(\tilde L,\tilde S,\tilde R)$ is a valid $x$-$y$ vertex-cut in $J$ (and hence in $\hat G$), and, moreover: 

\[\hat w(\tilde S)\leq c_{J'}(\tilde X,\tilde Y)\leq \hat w(S)+6\lambda \gamma.\]

It now remains to show that $\tilde L\subseteq L$ and that $\tilde L\subseteq A$.
Indeed, consider a vertex $v\in \tilde L$ and recall that $v^{\out}\in \tilde X$ must hold, by the definition of the set $\tilde L$.
Recall that, from \Cref{claim: good cut in $H$},  $\Tilde{X} \subseteq X \cap A'$, so $v^{\out}\in X\cap A'$ must hold. Since $X=L^*\cap S^{\inn}$, we get that $X\cap V^{\out}\subseteq L^{\out}$, so $v\in L$. Moreover, vertex set $A$ contains every vertex $u\in V(G)\setminus\left(\set{y}\cup N_G^-(y)\right )$ with $u^{\out}\in A'$. Since we have established that $(\tilde L,\tilde S,\tilde R)$ is a valid $x$-$y$ cut in $\hat G$, and since $v\in \tilde L$, vertex $v$ may not lie in $\set{y}\cup N_G^-(y)$. Since $v^{\out}\in A'$, we conclude that $v\in A$ must hold.
\end{proof}

To summarize, we have shown that, if Condition (C) holds,  then vertex set $A$ is quasi-promising in $\hat G$. From \Cref{claim: from quasi to promising} it then follows that, if Condition (C) holds, then $A$  is a promising vertex set for graph $G$ and scale $M'$. Therefore, if the set $A$ that our algorithm returned is not 
a promising vertex set for graph $G$ and scale $M'$, then Condition (C) does not hold, so either the distinguished cut is bad, or the subgraph oracle erred in its response to the query.
In order to complete the proof of \Cref{thm : algorithm-for-step-1} it now remains to prove \Cref{claim: good cut in $H$}, which we do next.

\begin{proofof}{\Cref{claim: good cut in $H$}}
	Assume that Condition (C) holds.
We let $\tilde X$ be the collection of all vertices $v\in V(\hat Z')$, such that there is a path $P_v$ connecting $x^{\out}$ to $v$ in $H$ that avoids the vertices of $Y$, and we let $\tilde Y=V(\hat Z')\setminus \tilde X$. Since $x^{\out}\in X$ and $y^{\inn}\in Y$, it is immediate to verify that $x^{\out}\in \tilde X$ and $y^{\inn}\in \tilde Y$ must hold, so $(\tilde X,\tilde Y)$ is a valid $x^{\out}$-$y^{\inn}$ cut in $\hat Z'$, and hence in $Z$.
Since $(X,Y)$ is a partition of $V(Z)$, and since, from the definition of the vertex set $\tilde X$ it may not contain vertices of $Y$, it is immediate to verify that $\tilde X\subseteq X$ must hold. It now remains to show that $\tilde X\subseteq A'$ and that $c_{Z}(\Tilde{X},\Tilde{Y}) \leq c_{Z}(X,Y) + 6\lambda\gamma$ hold.
We establish this in the following two observations.

\begin{observation}\label{obs: in A'}
If Condition (C) holds, then	$\tilde X\subseteq A'$.
\end{observation}
\begin{proof}
	Recall that, if Condition (C) holds, then the algorithm from 
	\Cref{thm : unit-flow} guarantees that every path in $H$ that connects  $x^{\out}$ to a vertex of $V(H)\setminus A'$ has length at least $h=2\lambda+4$. Therefore, if, for some vertex $v\in V(\hat Z')$, there is a path $P'_v$ of length at most $2\lambda+3$ connecting $x^{\out}$ to $v$  in $H$, then $v\in A'$ must hold. It is now enough to prove that, for every vertex $v\in \tilde X$, such a path $P'_v$ exists in $H$.
	
	Consider now any vertex $v\in \tilde X$. Recall that, from the definition of the set $\tilde X$ of vertices, there is a path $P_v$ connecting $x^{\out}$ to $v$ in $H$, that avoids the vertices of $Y$. It is now enough to prove that the length of the path $P_v$ is bounded by $2\lambda+3$. 
	
	Recall that we have established that every edge of graph $J'$ connects a vertex of $V^{\out}$ to a vertex of $V^{\inn}$ or the other way around. Since graph $Z$ is the residual flow network of the graph $J'$ with respect to the flow $f'_0$, it is easy to verify that it has the same property: namely, every edge of $Z$ connects a vertex of $V^{\out}$ to a vertex of $V^{\inn}$ or the other way around. Recall also that graph $\hat Z'$ is obtained from $Z$ by scaling the capacities of its edges, and adding edges that are incident to the vertex $y^{\inn}$. Lastly, graph $H$ is a residual flow network of the graph $\hat Z'$ with respect to the flow $f$. It is then immediate to verify that, if $e$ is an edge of $H$ that is not incident to $y^{\inn}$, then it connects a vertex  of $V^{\out}$ to a vertex of $V^{\inn}$ or the other way around.
	
	Recall that $X=L^*\cup S^{\inn}$, and that, if Condition (C) holds, then $|L|\leq \lambda$, so $|X\cap V^{\out}|\leq \lambda$ holds. Therefore, path $P_v$ may contain at most $\lambda$ vertices of $V^{\out}$, and, since it may not contain the vertex $y^{\inn}\in Y$, the vertices of $V^{\inn}$ and $V^{\out}$ must alternate on it. We conclude that the length of the path is bounded by $2\lambda+1$, and so $v\in A'$ must hold.
\end{proof}
 	
\begin{observation}\label{obs: bounding cut value}
	If Condition (C) holds, then $c_{Z}(\Tilde{X},\Tilde{Y}) \leq c_{Z}(X,Y) + 6\lambda\gamma$.
\end{observation}
\begin{proof}
	Assume that  Condition (C) holds.
	Recall that we denoted by $H$ the residual flow network of the graph $\hat Z'$ with respect to the flow $f$ that the algorithm from \Cref{thm : unit-flow} computed.
	For every edge $e\in E(H)$, we denote its capacity in graph $H$ by $c_H(e)$. Notice that, if $e=(u,v)$ is an edge of graph $H$ with $u\in \tilde X$ and $v\in \tilde Y$, then $v\in Y$ must hold. Indeed, assume otherwise. Then, from the definition of the vertex set $\tilde X$, there is a path $P_u$ connecting $x^{\out}$ to $u$ in graph $H$, that avoids the vertices of $Y$. If $v\not\in Y$, then this path can be extended by appending the edge $(u,v)$ at the end of it, to obtain a path $P_v$ connecting
	$x^{\out}$ to $u$ in graph $H$, that avoids the vertices of $Y$, contradicting the fact that $u\not\in \tilde X$. Since, as we have shown already, $\tilde X\subseteq X$ must hold, we conclude that, for every edge $e=(u,v)\in E(H)$ with $u\in \tilde X$ and $v\in \tilde Y$, $(u,v)\in E_H(X,Y)$ must hold. Therefore, $c_H(\tilde X,\tilde Y)\leq c_H(X,Y)$ must hold.

Since, from \Cref{fact : residual-capacity-of-edge-cut},  $c_{H}(X,Y) = c_{\hat Z'}(X,Y) - \val(f)$, and similarly, $c_{H}(\tilde X,\tilde Y) = c_{\hat Z'}(\tilde X,\tilde Y) - \val(f)$, we conclude that:

\begin{equation}\label{eq: compare 2 cuts}
c_{\hat Z'}(\tilde X,\tilde Y)=c_{H}(\tilde X,\tilde Y)+\val(f)\leq c_H(X,Y)+\val(f)=c_{\hat Z'}(X,Y).
\end{equation}

Recall that graph $\hat Z'$ is obtained from graph $Z$ by first setting, for every edge $e\in E(Z)$, its capacity to $c_{\hat Z'}(e)=\ceil{\frac{n}{\gamma}}\cdot c_Z(e)$, and then adding, for every vertex $v\in V(Z)\setminus\set{y^{\inn}}$ with $\deg^+_Z(v)>0$, an edge $(v,y^{\inn})$ of capacity $c_{\hat Z'}(v,y^{\inn})=\deg^+_Z(v)$.

Therefore, on the one hand, $c_{\hat Z'}(\tilde X,\tilde Y)\geq \ceil{\frac{n}{\gamma}}\cdot c_Z(\tilde X,\tilde Y)$, while, on the other hand, 
$c_{\hat Z'}(X,Y)=\ceil{\frac{n}{\gamma}}\cdot c_Z(X,Y)+\vol^+_Z(X)$.

Recall that, from \Cref{obs: bounding volume of in-vertices},
$\sum_{v^{\inn}\in V^{\inn}}\deg^+_Z(v^{\inn})\leq 5n$. Moreover, if Condition (C) holds, then $|L|\leq \lambda$, and so $X=L^*\cup S^{\inn}$ contains at most $\lambda$ vertices of $V^{\out}$. Therefore, $\vol^+_Z(X)\leq 5n+\lambda n\leq 6\lambda n$ must hold.
We conclude that $c_{\hat Z'}(X,Y)\leq \ceil{\frac{n}{\gamma}}\cdot c_Z(X,Y)+6\lambda n$.

Altogether, we get that:

\[
\begin{split}
c_Z(\tilde X,\tilde Y)&\leq \frac{c_{\hat Z'}(\tilde X,\tilde Y)}{\ceil{n/\gamma}}\\
&\leq \frac{c_{\hat Z'}( X, Y)}{\ceil{n/\gamma}}\\
&\leq c_Z(X,Y)+\frac{6\lambda n}{\ceil{n/\gamma}}\\
&\leq c_Z(X,Y)+6\lambda\gamma.
\end{split}
\]
\end{proof}
\end{proofof}

\newpage

\appendix
\section{Proofs Omitted from Section \ref{sec: prelims}}
\label{sec: appx-proofs from prelims}

\subsection{Proof of \Cref{obs: existential from regular to split}}\label{subsec : proof of obs existential from regualr to split}
It is immediate to verify that $s^{\out} \in X$ and $t^{\inn} \in Y$, so, in particular, $X,Y\neq \emptyset$, and $(X,Y)$ is an $s^{\out}$-$t^{\inn}$ edge-cut in $G'$. It now remains to show that $c(X,Y)=w(S)$. 

Consider any edge $e\in E_{G'}(X,Y)$. We claim that $e$ must be a special edge representing a vertex of $S$. Indeed, if $e$ is a regular edge, then it must connect a vertex $v^{\out}\in X\cap V^{\out}$ to a vertex $u^{\inn}\in Y\cap V^{\inn}$. But then, by the definition of the sets $X$ and $Y$, $v\in L$ and $u\in R$ must hold, and, by the definition of the split graph, edge $(v,u)$ must lie in $G$, a contradiction. We conclude that $e$ must be a special edge. Since, for every vertex $v\in L$, both copies of $v$ lie in $X$, while for every vertex $u\in R$, both copies of $u$ lie in $Y$, it must be the case that $e$ represents a vertex of $S$. We conclude that all edges in $E_{G'}(X,Y)$ are special edges representing  vertices of $S$. Moreover, from the construction of the cut $(X,Y)$, it is immediate to verify that, for every vertex $v\in S$, the corresponding special edge $(v^{\inn},v^{\out})\in E_{G'}(X,Y)$. Therefore, $c(E_{G'}(X,Y))=w(S)$.

\subsection{Proof of \Cref{obs: existential from split to regular}}
\label{subsec: proof of obs existential from split to regular}

It is immediate to verify that every vertex of $G$ belongs to exactly one of the sets $L,S$ or $R$, and, moreover, since $s^{\out}\in X$ and $t^{\inn}\in Y$, we get that $s\in L$ and $t\in R$ must hold. In order to prove that $(L,S,R)$ is a valid $s$-$t$ vertex-cut in $G$, it remains to prove that no edge connects a vertex of $L$ to a vertex of $R$. Assume otherwise, and let $e=(u,v)\in E(G)$ be an edge with $u\in L$ and $v\in R$. Then, from the definition of the set $L$, $u^{\out}\in X$, and, from the definition of the set $R$, $v^{\inn}\in Y$ must hold. Therefore, the regular edge $(u^{\out},v^{\inn})$, whose capacity is $\wmax(G)$ lies in $E_{G'}(X,Y)$, contradicting the fact that  
$c(X,Y)<\wmax(G)$. We conclude that $(L,S,R)$ is a valid $s$-$t$ vertex-cut in $G$.
Note that, for every vertex $v\in S$, the corresponding special edge $(v^{\inn},v^{\out})$, whose capacity is $w(v)$, lies in $E_{G'}(X,Y)$, and so $w(S) \leq c(X,Y)$ must hold.

Consider now the tripartition $(L,S',R')$ of $V(G)$ induced by $L$, and recall that $S'=N_G^+(L)$ and $R'=V(G)\setminus (L\cup S')$. Since $(L,S,R)$ is a valid vertex-cut, it is immediate to verify that $S'\subseteq S$, and therefore $R\subseteq R'$ must hold. In particular, $t\in R'$. Therefore, $(L,S',R')$ is a vertex-cut in $G$ that is induced by the set $L$ of vertices, and it is a valid $s$-$t$ vertex-cut. As observed already, $S'\subseteq S$ must hold, and therefore $w(S')\le w(S) \leq c(X,Y)$.

\subsection{Proof of \Cref{cor : minimum cut in split graph}}\label{subsec: proof of cor minimum cut in split graph}
Let $\opt_{G}$ denote the value of the minimum $s$-$t$ vertex cut in $G$, and let $\opt_{G'}$ denote the value of the minimum $s^{\out}$-$t^{\inn}$ edge cut in $G'$.
We start by proving that $\opt_{G} < \frac{\wmax(G)}{2}$.
Indeed, since $(s,t) \notin E(G)$, the tripartition  $(L,S,R)$ of vertices of $G$, where $L=\{s\}$, $R=\{t\}$ and $S = V(G) \setminus \{s,t\}$ defines a valid $s$-$t$ vertex-cut in $G$.
The value of this cut is: 

$$w(S)\leq (n-2) \cdot \max_{v \in V(G)}\{w(v)\}\leq \frac{\wmax}{2},$$

from the definition of the parameter $\wmax(G)$. Therefore,  $\opt_{G} \leq w(S) < \frac{\wmax(G)}{2}$ must hold.

It is easy to verify that $\opt_{G'} \leq \opt_{G}$ must hold. Indeed, consider any minimum $s$-$t$ vertex cut $(L^*,S^*,R^*)$ in $G$. From \Cref{obs: existential from regular to split}, there exists an $s^{\out}$-$t^{\inn}$ edge-cut in $G'$ whose value is at most $w(S^*)=\opt_{G}$, so $\opt_{G'} \leq \opt_{G}$ must hold.

Lastly, we prove that $\opt_{G} \leq \opt_{G'}$. Let $(X,Y)$ be a minimum $s^{\out}$-$t^{\inn}$ edge-cut in $G'$. As we have shown, $\opt_{G'} \leq \opt_{G} < \frac{\wmax(G)}{2}$, so in particular $c(X,Y)<\wmax(G)$.
From \Cref{obs: existential from split to regular}, there exists an $s$-$t$ vertex-cut in $G$ whose value is at most $c(X,Y)=\opt_{G'}$. We conclude that $\opt_{G} \leq \opt_{G'}$ and $\opt_G=\opt_{G'}$ must hold.

\subsection{Proof of \Cref{thm : unit-flow}}\label{sec: proof-of-unit-flow}

Our proof closely follows the analysis from Section 4 of \cite{OZ14}, with only minor modifications. Like the algorithm of \cite{OZ14}, we use a local version of Dinic's Blocking Flow algorithm \cite{dinic1970algorithm}.
We note that another commonly used approach in local-flow algorithms is to rely instead on the Push-Relabel framework of Goldberg and Tarjan \cite{GT88}, see e.g. \cite{HRW17,NSWN17,SW19}. One can equivalently obtain the proof of \Cref{thm : unit-flow} via this latter approach, but we choose to follow the approach of \cite{OZ14}, since we find it conceptually simpler.

We start with some notation and definitions that our algorithm uses. Suppose we are given a directed (possibly non-simple) graph $G$ with capacities $c(e)\geq 0$ on its edges $e \in E(G)$, and a flow $f$ in $G$. For a pair $x,y\in V(G)$ of vertices, we denote by $c(x,y)$ the total capacity of all edges $(x,y)$ in $G$, and by $f(x,y)$ the total amount of flow on all such edges; if no such edges lie in $G$ then $c(x,y)=f(x,y)=0$. Next, we define the notion of distancing flows.

\begin{definition}[A Distancing Flow]\label{def: distnacing flow}
	Consider a directed edge-capacitated graph $G$, a pair $s,t \in V(G)$ of its vertices, an $s$-$t$ flow $f$ in $G$, and an integer $h$. Let $G_f$ be the residual network of $G$ with respect to $f$.
	We say that $f$ is an \emph{$h$-distancing $s$-$t$ flow in $G$}, if every $s$-$t$ path in $G_f$ contains at least $h$ edges.
\end{definition}

A central tool in the proof of  \Cref{thm : unit-flow} is an algorithm that, given a simple directed edge-capacitated graph $G$, a pair $s,t \in V(G)$ of its vertices, and a parameter $h$, computes an $h$-distancing $s$-$t$ flow in $G$. 
For technical reasons, we also require that $(s,t) \notin E(G)$, and that the flow that the algorithm outputs does not use the edges of $\delta^+_G(t)$. The algorithm is summarized in the following theorem.

\begin{theorem}\label{thm: simplified local flow}
	There is a deterministic algorithm, whose input consists of a simple directed graph $G$ with integral capacities $c(e) \geq 0$ on its edges $e \in E(G)$ given in the adjacency-list representation, a pair $s,t \in V(G)$ of vertices with $(s,t) \notin E(G)$, and an integer $h> 0$.
	The algorithm returns an $h$-distancing $s$-$t$ flow $f$ in $G$, by providing the list $E^f$ of all edges $e\in E(G)$ with $f(e)>0$ and the flow value $f(e)$ for each such edge, and it ensures that, for every edge $e\in \delta^+_G(t)$, $f(e)=0$.
Let $G_f$ be the residual flow network of $G$ with respect to the flow $f$, and let $\hat{A} = V(G) \setminus \left(\{t\} \cup N^-_{G_f}(t)\right)$ be the collection of all vertices of $G$, excluding $t$ and its in-neighbors in $G_f$.
	The running time of the algorithm is $\Tilde{O}\left(h \cdot \left(\vol^+_G(\hat{A}) +\val(f)\right)\right)$, and  $|E^f|\leq O\left(h \cdot \val(f)\right)$.
\end{theorem}

We prove \Cref{thm: simplified local flow} in Section \ref{subsec: proof of simplified local flow}, after we complete the proof of  \Cref{thm : unit-flow} using it.
Recall that our algorithm receives as input a simple directed graph $G$ with integral capacities $c(e) \geq 0$ on its edges $e \in E(G)$ in the adjacency-list representation, integral
 parameters $\Delta,h \geq 1$, and a pair $s,t \in V(G)$ of distinct vertices with $\deg^+_{G}(s) \leq \Delta$. We let $\hat G$ be the graph that is obtained from $G$ by adding, for every vertex $v\in V(G)\setminus\set{t}$ with $\deg_G^+(v)>0$, the edge $(v,t)$ of capacity $\deg^+_G(v)$; if such an edge already lies in $G$, then we increase its capacity by $\deg^+_G(v)$ units. 
 
 Let $G^*$ be a directed graph that is obtained from $\hat{G}$ by adding a new source vertex $s'$, and an edge $(s',s)$ of capacity $\Delta$. Note that the input adjacency-list representation of the graph $G$ can be used to simulate access to the adjacency-list representation of $G^*$ in a straightforward way.  We then apply the algorithm from \Cref{thm: simplified local flow} to graph $G^*$, the pair $s',t$ of its vertices, and the parameter $h' = h+2$. 
Let $f$ be the $(h+2)$-distancing $s'$-$t$ flow in $G^*$ that the algorithm from  \Cref{thm: simplified local flow} returns and let $E^f$ be the set of all edges $e\in E(G^*)$ with $f(e)>0$.
Note that, since the only edge leaving the vertex $s'$ in $G^*$ is the edge  $(s,s')$ of capacity $\Delta$, $\val(f)\leq \Delta$ must hold.

Let $G^*_f$ be the residual flow network of the graph $G^*$ with respect to the flow $f$, and let $\hat{A} = V(G^*) \setminus \left(\{t\} \cup N^-_{G^*_f}(t)\right)$. In order to bound the running time of the algorithm from \Cref{thm: simplified local flow}, we use the following claim.

\begin{claim}\label{cl: volume of hat A not larger than Delta}
	$\vol^+_{G}(\hat{A} \setminus \set{s'}) \leq \Delta$ and $\vol^+_{G^*}(\hat{A}) \leq 2\Delta+1$.
\end{claim}
\begin{proof}
Consider any vertex 
$v \in \hat{A} \setminus \set{s'}$ with $\deg^+_{G}(v) \neq 0$, and recall that an edge $(v,t)$ of capacity at least $\deg^+_{ G}(v)$ lies in $\hat G$. Since $v\in \hat A$, edge $(v,t)$ does not lie in $G^*_f$, and so it must be saturated. We conclude that, for every vertex $v \in \hat{A} \setminus \set{s'}$, $f(v,t)=c(v,t) \geq \deg^+_{G}(v)$ must hold. Since the algorithm from \Cref{thm: simplified local flow} guarantees that, for every edge $e\in \delta^+_G(t)$, $f(e)=0$, and since, as we have established, $\val(f)\leq \Delta$, we get that:

\[ \vol^+_{G}(\hat{A} \setminus \set{s'})=\sum_{v\in\hat{A} \setminus \set{s'}}\deg^+_G(v)\leq \sum_{e\in \delta^-_G(t)}f(e)= \val(f)\leq \Delta. \]

Lastly, it is immediate to verify that, for every vertex $v\in V(G^*)\setminus\set{s'}$, $\deg^+_{G^*}(v)\leq 2\deg^+_G(v)$, and that $\deg^+_{G^*}(s')=1$. Therefore, $\vol^+_{G^*}(\hat{A}) \leq 
2\vol^+_{G}(\hat{A} \setminus \set{s'})+1\leq
 2\Delta+1$.
\end{proof}

We conclude that the running time of the algorithm from \Cref{thm: simplified local flow} is bounded by:

\[\Tilde{O}\left(h' \cdot \left(\vol^+_{G^*}(\hat{A}) +\val(f)\right)\right) \leq \Tilde{O}\left(h \cdot \Delta\right),\]

and that $|E^f|\leq O\left(h' \cdot \val(f)\right)\leq O\left(h \cdot \Delta\right)$.
Note that the $s'$-$t$ flow $f$ in graph $G^*$ naturally defines an $s$-$t$ flow $f'$ in graph $\hat G$ of the same value as follows: for every edge $e\in E(\hat G)$, we set $f'(e)=f(e)$.
If $\val(f')\geq \Delta$, then we output the flow $f'$ and the set $A=\set{s}$ of vertices. In the remainder of the algorithm, we assume that $\val(f')<\Delta$. Let $H$ be the residual flow network of the graph $\hat G$ with respect to the flow $f'$. 
Next, we show that $f'$ must be an $(h+1)$-distancing $s$-$t$ flow in $\hat{G}$.

\begin{claim}\label{cl : f' is distancing flow in hat G}
	If $\val(f') < \Delta$, then $f'$ is an $(h+1)$-distancing $s$-$t$ flow in $\hat{G}$.
\end{claim}
\begin{proof}
	Assume that $\val(f')<\Delta$, and assume for contradiction that $f'$ is not an $(h+1)$-distancing flow in $\hat{G}$; in other words, there is an $s$-$t$ path $P$ in $H$ that contains at most $h$ edges.
	
	By the construction of the graph $G^*$, and by the construction of the flow $f'$ from the flow $f$, it is easy to see that the residual flow network $G^*_f$ must contain all edges of $H$, and, additionally, the edge $(s',s)$ (since $\val(f')<\Delta=c(s',s)$). Therefore, if we let $P'$ be the path obtained by appending the edge $(s',s)$ at the beginning of the path $P$, then $P'$ is an $s'$-$t$ path in $G^*_f$, and it contains fewer than $h+2$ edges, contradicting the fact that $f$ is an $(h+2)$-distancing $s'$-$t$ flow in $G$.
\end{proof}

It now remains to describe the algorithm for computing the set $ A\subseteq V(\hat G)$ of vertices for the case where $\val(f') < \Delta$. We perform a BFS in $H$ from the vertex $s$ up to distance $h-1$, and we then let $A$ be the set of all vertices discovered in this BFS search; in other words, we let $A=\set{v\in V(\hat G)\mid \dist_H(s,v)\leq h-1}$; we discuss the efficient implementation of the BFS search below. We then output the flow $f'$ by listing all edges $e\in E^f\cap E(\hat G)$, together with the corresponding flow value $f'(e)$, and the set $A$ of vertices. 
It is immediate to verify that, for every vertex $v\in V(\hat G)\setminus A$, every path connecting $s$ to $v$ in $H$ must contain at least $h$ edges.
It is also easy to see that $s \in A$.
We use the following claim in order to complete the correctness proof of the algorithm.

\begin{claim}\label{cl : A contains A hat in proof of local flow}
	If $\val(f') < \Delta$, then $A \subseteq \hat{A} \setminus \set{s'}$.
\end{claim}
\begin{proof}
	Consider any vertex $v\in A$. Since $A\subseteq V(\hat G)$ and $s'\not\in V(\hat G)$, $v\neq s'$ must hold. Therefore, it is enough to prove that $v \in \hat A$. Assume for contradiction that $v\not\in \hat A$. Then either $v= t$, or $(v,t)\in E(G^*_f)$ must hold. It is easy to see that, in the latter case, $(v,t)\in E(H)$.
	
	Recall that, from the definition of $A$, there is a path $P$ in graph $H$ connecting $s$ to $v$, that contains at most $h-1$ edges. Therefore, if $v\not\in \hat A$, then there is a path connecting $s$ to $t$ in $H$, that contains at most $h$ edges (either the path $P$ itself if $v=t$, or the path obtained from $P$ by concatenating the edge $(v,t)$ to the end of $P$ if $(v,t)\in E(H)$). But this contradicts the fact that, from \Cref{cl : f' is distancing flow in hat G},  $f'$ is an $(h+1)$-distancing $s$-$t$ flow in $\hat{G}$.
\end{proof}

From \Cref{cl : A contains A hat in proof of local flow}, it is immediate to verify that $t\not\in A$, and, by combining \Cref{cl : A contains A hat in proof of local flow} with \Cref{cl: volume of hat A not larger than Delta}, we get that $\vol_{G}^+(A)\leq \vol^+_{G}(\hat{A} \setminus \set{s'}) \leq \Delta$.

Lastly, we describe an efficient implementation of the BFS search and bound the total running time of the algorithm. Recall that our algorithm is given as input the adjacency-list representation of the graph $G$, which, in turn, can be used in order to simulate the adjacency-list representation of $\hat G$ and of $G^*$. Given the list $E^f$ of the edges of $E(G^*)$ with non-zero flow values $f(e)>0$, and the flow value $f(e)$ for each such edge, we can modify this data structure so that it can be used to simulate the adjacency-list representation of $H$, in a straightforward manner: for every edge $e=(x,y)\in E^f\cap E(\hat G)$, if $f'(e)<c(e)$, then we modify the capacity of the edge $e$ in lists $\OUT(x)$ and $\IN(y)$ and mark the edge to signify that its capacity is stored explicitly. We also insert the backward edge $(y,x)$ of capacity $f'(e)$ into $\OUT(y)$ and $\IN(x)$, and mark it as well. It is easy to verify that the resulting data structure can be used to simulate access to the adjacency-list representation of the graph $H$, and that this modification can be performed in time $O(|E^f|)\leq  O(h\cdot \Delta)$. We can then use this data structure to perform the  BFS search in graph $H$ up to depth $h-1$ in time $O\left (\vol^+_{\hat G}(A)\right )\leq  O(\Delta)$. Overall, the total running time of the algorithm is $\tilde O(h\cdot \Delta)$. In order to complete the proof of  \Cref{thm : unit-flow} it now remains to prove \Cref{thm: simplified local flow}, which we do next.

\subsubsection{Proof of \Cref{thm: simplified local flow}}\label{subsec: proof of simplified local flow}

Recall that our algorithm is given as input a simple directed graph $G$ with integral capacities $c(e) \geq 0$ on its edges $e \in E(G)$, in the adjacency-list representation. We are also given a pair $s,t \in V(G)$ of vertices with $(s,t) \notin E(G)$, and an integer $h> 0$.

At a high level, we follow the standard strategy of \cite{edmonds1972theoretical, hopcroft1973n} and its efficient implementation via blocking flows by \cite{dinic1970algorithm}.
Specifically, our algorithm maintains a valid $s$-$t$ flow $f$ in the graph $G$, starting with $f=0$, and then gradually augmenting it. In order to do so efficiently, we only explicitly maintain a list $E^f$ of all edges $e\in E(G)$ with $f(e)>0$, together with the flow value $f(e)$ for each edge $e\in E^f$. At all times, we denote by $H=G_f$ the residual flow network of the input graph $G$ with respect to the current flow $f$. The algorithm consists of a number of iterations, where in the $i$th iteration we compute a shortest $s$-$t$ path $P_i$ in $H$, and augment the flow $f$ by $1$ unit along this path. It is well known that, following this procedure, the length of the shortest $s$-$t$ path in $H$ may not decrease from iteration to iteration (and in fact this also follows from some of the arguments that we provide below). The algorithm terminates once every $s$-$t$ path in $H$ contains more than $h$ edges. 
Since the length of the shortest $s$-$t$ path in $H$ may not decrease from iteration to iteration, we can partition this algorithm into $h$ phases, where, for all $1\leq j\leq h$, the $j$th phase encompasses all iterations in which the length of the augmenting path is exactly $j$ (note that some phases may contain no iterations). The key to the algorithm of Dinic \cite{dinic1970algorithm} is an efficient implementation of each phase. We follow his high level approach with some modifications.
In addition to the flow $f$, our algorithm maintains the adjacency-list representation of the graph $H=G_f$. At the beginning of the algorithm, $E^f=\emptyset$, $f=0$, and $H=G$, so we can use the adjacency-list representation of the graph $G$ that the algorithm receives as input, that we denote by $\dset$, in order to access the adjacency-list representation of $H$. We will gradually modify $\dset$, so that it remains a valid adjacency-list representation of $H$ throughout the algorithm.

The following standard observation shows that, for every edge $e=(x,y)$ that is inserted into the residual flow network $H$ following an augmentation of the flow $f$ via an augmenting path $P\subseteq H$, the antiparallel edge $(y,x)$ must lie on $P$. The proof immediately follows from the definition of the residual flow network and is omitted here.

\begin{observation}\label{obs: residual network change}
	Let $G$ be a simple directed graph with capacities $c(e)\geq 0$ on its edges $e\in E(G)$ and two special vertices $s$ and $t$, and let $f'$ be an $s$-$t$ flow in $G$. Let $H'$ be the residual flow network of $G$ with respect to $f'$, and let $P$ be a shortest $s$-$t$ path in $H'$. Finally, let $f''$ be the $s$-$t$ flow in $G$ that is obtained by augmenting the flow $f'$ via the augmenting path $P$ by $1$ unit, and let $H''$ be the residual flow network of $G$ with respect to the flow $f''$. Then for every edge $e=(x,y)\in E(H'')\setminus E(H')$, edge $(y,x)$ must lie on the path $P$.
\end{observation}

Throughout the algorithm, we let $\Lambda\subseteq V(G)$ be a vertex set that contains all vertices $v\in V(G)$, such that either (i) $(v,t)\not\in E(G)$; or (ii) $(v,t)\in E(G)$ and $f(v,t)=c(v,t)$. We do not maintain the vertex set $\Lambda$ explicily, and only use it for analysis. Recall that, in each iteration $i$, we augment the current flow $f$ via an augmenting path $P_i$, which is a shortest $s$-$t$ path in $H$; in particular, $P_i$ may not contain edges that leave the vertex $t$. It then follows that, throughout the algorithm, for every edge $e\in \delta^+_G(t)$, $f(e)=0$, and, moreover, the vertex set $\Lambda$ is incremental, so, as the algorithm progresses, new vertices may join $\Lambda$, but no vertex may ever leave $\Lambda$. Note also that vertex set $\Lambda$ may be equivalently defined as $\Lambda= V(G) \setminus \left(\{t\} \cup N^-_{H}(t)\right)$. For all $1\le j\leq h$, we denote by $\Lambda_j$ the vertex set $\Lambda$ at the beginning of the $j$th phase, and we denote by $\hat A$ the set $\Lambda$ at the end of the algorithm. We will ensure that, throughout the algorithm, the following invariant holds:

\begin{properties}{I}
	\item \label{inv: flow} for every edge $e=(x,y)\in E(G)$ with $f(e)>0$, either $x\in \Lambda$ or $y=t$ holds.
\end{properties}

 We now fix an integer $1\leq j\leq h$, and describe the algorithm for the $j$th phase. We assume that, at the beginning of the phase, $\dset$ is a valid adjacency-list representation of the current residual flow network $H$, and that Invariant \ref{inv: flow} holds for $f$. Additionally, we assume that  the length of the shortest $s$-$t$ path in $H$ is at least $j$.
The algorithm for the $j$th phase consists of three procedures: initialization, augmentation and cleanup, that we now describe in turn.

\paragraph{Initialization.}
At the beginning of the $j$th phase, we compute a subgraph $Z_j\subseteq H$, that is defined as follows. For all $1\leq a<j$, let $L_a=\set{v\in V(H)\mid \dist_H(s,v)=a}$, and let $L_0=\set{s}$ and $L_j=\set{t}$. We let $V(Z_j)=L_0\cup\cdots\cup L_j$, and we let $E(Z_j)$ contain all edges $e=(x,y)\in E(H)$ where, for some  index $0\leq a<j$, $x\in L_a$ and $y\in L_{a+1}$.
Since, at the beginning of the $j$th phase, $\dist_H(s,t)\geq j$, it is immediate to verify that $L_0\cup\ldots\cup L_{j-2}\subseteq \Lambda_j$ and $L_{j-1}\subseteq N^+_H(\Lambda_j)$. Since every vertex of $L_{j-1}$ has out-degree at most $1$ in $Z_j$, we get that $|E(Z_j)|\leq O\left(\vol^+_H(\Lambda_j)\right )$. 
Moreover, by performing a BFS in $H$ using the adjacency-list representation $\dset$ of the graph $H$, graph $Z_j$ can be constructed in time $O\left(\vol^+_H(\Lambda_j)\right )$. 

Consider now any edge $e=(x,y)\in E(H)$ where $x\in \Lambda_j$. From Invariant \ref{inv: flow}, if $e\not\in E(G)$, then $y\in \Lambda_j$ must hold. It is then easy to verify that $\vol^+_H(\Lambda_j)\leq O(\vol^+_G(\Lambda_j))$. We conclude that 
$|E(Z_j)|\leq O\left(\vol^+_G(\Lambda_j)\right )$, and that the running time of the initialization procedure is bounded by $ O\left(\vol^+_G(\Lambda_j)\right )$.

\paragraph{Cleanup.}
Over the course of the phase, graph $Z_j$ may be updated as follows. Whenever, for any vertex $v\in V(Z_j)\setminus\set{t}$,  $\delta^+_{Z_j}(v)=\emptyset$ holds, we delete $v$  from $Z_j$. Additionally, if some edge $e\in E(Z_j)$ is deleted from $H$ (due to the changes in the flow $f$), then this edge is deleted from graph $Z_j$ as well. We never insert any edges or vertices into $Z_j$. For convenience, we use $Z_j$ to denote the graph that is computed at the beginning of the phase, and by $Z'_j$ the corresponding decremental graph that is updated with edge- and vertex-deletions over the course of the phase. We use the cleanup procedure in order to update the graph $Z'_j$. The cleanup procedure is performed once at the beginning of the phase after the initialization, and then after each iteration in the phase. Over the course of the cleanup procedure we iteratively delete from $Z'_j$ every vertiex $v$ whose out-degree in the current graph $Z'_j$ is $0$. Immediately after the graph $Z_j$ is initialized, we compute the out-degree of every vertex in $Z_j$ and initialize a list $\tau$ of all vertices whose out-degree in $Z_j$ is $0$. Then as long as $\tau\neq \emptyset$, we select an arbitrary vertex $v\in \tau$, delete it from $Z'_j$ and from $\tau$, and then inspect every vertex $u\in V(Z'_j)$ that was an in-neighbor of $v$ in $Z'_j$. For each such vertex $u$, we decrease its out-degree in $Z'_j$ by $1$, and, if it reaches $0$, we add $u$ to $\tau$. Once $\tau=\emptyset$, the cleanup procedure terminates. Similarly, after each iteration of the phase, for every edge $e=(u,v)\in E(Z'_j)$ that was deleted from $H$ (and hence from $Z'_j$), we decrease the out-degree of $u$ in $Z'_j$ by $1$, and, if it reaches $0$, we add $u$ to $\tau$. We then process the list $\tau$ exactly as before. It is immediate to verify that the total running time of the cleanup procedure, over the course of the entire phase, is bounded by $O(|E(Z_j)|)\leq O\left(\vol^+_G(\Lambda_j)\right )$.

\paragraph{Augmentations.}
As long as $s\in V(Z'_j)$, we perform iterative augmentations of the flow $f$. Recall that we ensure that, for every vertex $v\in V(Z'_j)\setminus\set{t}$, $\delta^+_{Z'_j}(v)\neq\emptyset$. Moreover, graph $Z'_j$ is a DAG with $j+1$ layers $L_0,\ldots,L_j$, where $s\in L_0$ and $t\in L_j$. In order to perform a single iteration, we simply perform a DFS in the graph $Z'_j$ starting from the vertex $s$, until we reach vertex $t$. From the structure of the graph $Z'_j$, it is immediate to verify that, after exploring $j$ edges, the DFS will compute an $s$-$t$ path $P$ in $Z'_j$, whose length is $j$. Since we ensure that $Z'_j\subseteq H$, $P\subseteq H$ holds, and, moreover, from our discussion, $P$ must be a shortest $s$-$t$ path in $H$. We augment the flow $f$ by sending one flow unit along the path $P$, update the residual flow network $H$, and, for every edge $e\in E(Z'_j)$ that is deleted from $H$ as the result of this update, we delete $e$ from $Z'_j$ as well. This completes the algorithm for performing a single augmentation. Note that the running time of the augmentation procedure is $O(j)\leq O(h)$. Note also that, for every edge $e=(u,v)\in E(P)$, if $e$ is a forward edge, then either $u\in \Lambda_j$ or $v=t$ holds, so Invariant \ref{inv: flow} continues to hold after the augmentation.

\paragraph{Phase termination.}
Phase $j$ terminates once $s$ no longer lies in the graph $Z'_j$. We claim that, at this time, $H$ contains no $s$-$t$ path of length at most $j$. Indeed, assume for contradiction that $H$ contains an $s$-$t$ path $P$, whose length is at most $j$. We partition the edges of $P$ into two types: an edge $e\in E(P)$ is of \emph{type 1} if $e$ lied in graph $H$ at the beginning of the current phase and it was never deleted from $H$ over the course of the phase; otherwise, $e$ is of \emph{type 2}. Note that, if $e=(x,y)$ is of type 2, then it was inserted into the graph $H$ over the course of the current phase. Consider any type-2 edge $e=(x,y)$. Since $e$ must have been inserted into $H$ at some time during the current phase, from  \Cref{obs: residual network change}, some augmenting path $P'$ that we used in the current phase contained the anti-parallel edge $(y,x)$, so $(y,x)\in E(Z_j)$ must hold. Therefore, there is some index $0\leq a<j$ with $y\in L_a$ and $x\in L_{a+1}$. To conclude, a type-2 edge $e=(x,y)$ must connect a vertex from a higher-indexed layer of $Z_j$ to a vertex from a lower-indexed layer. Since $s\in L_0$ and $t\in L_j$, and since every type-2 edge connects a pair of vertices of $Z_j$, while every type-1 edge lied in $H$ at the beginning of the current phase, the only way that the path $P$ has length at most $j$ is if $P$ has length exactly $j$ and it visits the vertices from layers $L_0,L_1,\ldots,L_j$ in this order. In particular, every edge $e\in E(P)$ must be of type 1 and it must lie in graph $H$ throughout the entire phase. But then all vertices of $P$ must lie in $Z'_j$ throughout the entire phase. This contradicts the fact that $s$ no longer lies in $Z'_j$ at the end of the phase. We conclude that, when the phase terminates, the length of the shortest $s$-$t$ path in $H$ is at least $j+1$.

\paragraph{Running time analysis.}
As discussed already, the total time spent in the $j$th phase on initialization and cleanup is bounded by $O\left(\vol^+_G(\Lambda_j)\right )$. The time spent on each augmentation is $O(j)\leq O(h)$, and each augmentation increases the value of the flow $f$ by $1$. Since the total number of phases in the algorithm is $h$, and since, for all $1\le j\leq h$, $\Lambda_j\subseteq \hat A$, we get that the running time of the entire algorithm is bounded by:

\[\sum_{j=1}^hO\left(\vol^+_G(\Lambda_j)\right )+O(h\cdot\val(f))\leq O\left(h \cdot \left(\vol^+_G(\hat{A}) +\val(f)\right)\right);\] 

here, $f$ refers to the flow $f$ at the end of the algorithm.
Since the length of every augmenting path is $O(h)$, it is immediate to verify that, at the end of the algorithm, $|E^f|\leq O(h\cdot \val(f))$ holds.

\subsection{Proof of \Cref{lem: oracle in bulk}}
\label{subsec: implementing subgraph oracle}

We first consider the case in which $\epsilon+\delta\geq 1$. In this case, we use a simple deterministic procedure to compute the response to each query in time $O(n^2)$, resulting in a total running time of $O(n^{3-\epsilon})\leq O(n^{2+\delta})$ for responding to all $q=n^{1-\epsilon}$ queries.
This procedure works as follows: to respond to a query $Z_i$, we scan the set of out-neighbors of each vertex $v \in V(G)$; we add to $Y^\ell_i$ all vertices $v$ with $|Z_i \cap N^+_G(v)|<n^{1-\delta}$, and we add to $E'_i$ all edges connecting these vertices to the vertices of $Z_i$; lastly, we set $Y^h_i=V(G) \setminus Y^\ell_i$. It is immediate to verify that this procedure correctly computes the response to each query in time $O(|E(G)|) \leq O(n^2)$ per query.
In the remainder of the proof we assume that $\epsilon+\delta<1$.

We use the following algorithm, that is a slightly simplified version  of the algorithm from Claim 2.13 of \cite{CT24}.

\begin{claim}[Restatement of Claim 2.13 in \cite{CT24}]
	\label{claim: heavy weight est}
	There is a randomized algorithm, whose input is a 
	directed $n$-vertex graph $G=(V,E)$ in the adjacency-list representation, a  subset $Z\subseteq V$ of its vertices, and two parameters $\gamma\geq 1$ and $\tau$ that is greater than a large enough constant. The algorithm computes a set $A\subseteq V$ of vertices, such that, with probability at least $\frac{1}{n^{10}\log^4\tau}$, both of the following hold:
	\begin{itemize}
		\item for every vertex $u\in A$, $|N^+_G(u)\cap Z|\geq \gamma$; and
		\item for every vertex $u'\in V\setminus A$, $|N^+_G(u')\cap Z|\leq 1000\gamma\log n\cdot\log \tau$.
	\end{itemize} 
	The running time of the algorithm is $O\left(\frac{n\cdot |Z|\cdot \log n\cdot \log \tau}{\gamma}\right )$.
\end{claim}

\Cref{claim: heavy weight est} can be obtained from Claim 2.13 in \cite{CT24} by setting all edge capacities and the parameter $c^*$ to $0$, using the vertex set $V$ instead of the set $Z$ in the statement of Claim 2.13 and the set $Z$ instead of $Z'$. 

We are now ready to complete the proof of \Cref{lem: oracle in bulk}.
We start by computing, for all $1\leq i\leq q$, the partition $(Y^h_i,Y^{\ell}_i)$ of $V(G)$, as follows. Set $\gamma=\frac{n^{1-\delta}}{1000\log n}$. For every index $1\leq i\leq q$, we apply the algorithm from \Cref{claim: heavy weight est}  to graph $G$, the set $Z=Z_i$ of its vertices, and parameters $\tau$ and $\gamma$. Let $A_i\subseteq V(G)$ be the subset of vertices that the algorithm returns. We then set $Y_i^h=A_i$ and $Y_i^{\ell}=V(G)\setminus A_i$. Let $\event_i$ be the bad event that the algorithm from \Cref{claim: heavy weight est} erred; in other words, either there is a vertex $u\in A_i$ with $|N^+_G(u)\cap Z_i|< \gamma$; or there is a vertex $u'\in V(G)\setminus A_i$, $|N^+_G(u')\cap Z_i|> 1000\gamma\log n\cdot\log \tau$. From \Cref{claim: heavy weight est}, $\prob{\event_i}\le \frac{1}{n^{10}\log^4\tau}$. Note that, if the event $\event_i$ did not happen, then we are guaranteed that (i) for every vertex $v\in Y_i^{h}$, $|N^+_G(v)\cap Z_i|\geq \gamma=\frac{n^{1-\delta}}{1000\log n}$; and (ii) for every vertex $v'\in Y_i^{\ell}$, $|N^+_G(v')\cap Z_i|\leq 1000\gamma\log n\log \tau=n^{1-\delta}\cdot \log \tau$.
Recall that the running time of the algorithm from \Cref{claim: heavy weight est} is:

\[O\left(\frac{n\cdot |Z_i|\cdot \log n\cdot\log\tau}{\gamma}\right )\leq O\left(n^{1+\delta}\cdot \log^2n \cdot\log\tau \right ),
\]

so, the time required to execute this step for all $1\leq i\leq q$ is $O\left(n^{2+\delta}\cdot \log^2n\cdot\log\tau\right )$.
We let $\event$ be the bad event that $\event_i$ happened for any $1\leq i\leq q$. By using the Union Bound, and since $q\leq n$, we get that $\prob{\event}\leq \frac{1}{n^9\log^4\tau}$.

Next, we construct a tripartite undirected graph $H=(A\cup B\cup C,E')$, as follows. We let $A=\set{a_v\mid v\in V(G)}$, $B=\set{b_v\mid v\in V(G)}$, and $C=\set{c_1,\ldots,c_{q}}$. For every (directed) edge $e=(u,v)\in E(G)$, we add the (undirected) edge $(a_u,b_v)$ to $H$. Additionally, for every index $1\leq i\leq q$, for every vertex $v\in Z_i$, we add the edge $(b_v,c_i)$, and for every vertex $u\in Y^{\ell}_i$, we add the edge $(c_i,a_u)$ to $H$.

For a triple $\Pi=(a_v,b_u,c_i)$ of vertices of $H$ with $a_v\in A$, $b_u\in B$, and $c_i\in C$, we say that $\Pi$ is a \emph{good triple}, if $(v,u)\in E_G(Y^{\ell}_i,Z_i)$. Notice that, if Event $\event$ did not happen, then $|E_G(Z_i,Y^{\ell}_i)|\leq n^{2-\delta}\cdot\log\tau$ must hold. Therefore, if Event $\event$ did not happen, then the total number of good triples in $H$ is bounded by $n^{2-\delta}\cdot q\cdot\log\tau= n^{3-(\delta+\epsilon)}\cdot\log\tau$. Moreover, every good triple defines a triangle in $H$ and vice versa. Therefore, in order to compute, for every index $1\leq i\leq q$, the set $E_G(Y^{\ell}_i,Z_i)$ of edges, it is sufficient to compute all triangles in $H$. Like in \cite{CT24}, we use the following algorithm of 
\cite{bjorklund2014listing} in order to do so.

\begin{theorem}[Theorem 1 in \cite{bjorklund2014listing}]\label{thm: triangle counting}
	There is a deterministic algorithm that lists all triangles in an $n$-vertex graph in time:  
	
	\[O\left(n^{\omega+o(1)}+n^{3(\omega-1)/(5-\omega)+o(1)}\cdot t^{2(3-\omega)/(5-\omega)}\right ),\] 
	
	where $\omega$ is the matrix multiplication exponent, and $t$ is the total number of triangles in $G$.
\end{theorem}

By substituting $t\leq n^{3-(\delta+\epsilon)}\cdot\log\tau$, we get that, if Event $\event$ did not happen, then the running time of the algorithm from \Cref{thm: triangle counting} on graph $H$ is bounded by:

\[
\begin{split}
&O\left(\left(n^{\omega}+n^{3(\omega-1)/(5-\omega)}\cdot n^{2(3-\omega)\cdot (3-(\delta+\epsilon))/(5-\omega)}\right )\cdot n^{o(1)}\cdot \poly\log\tau\right )\\
&\quad\quad\quad\quad\quad\quad\quad\quad\quad\leq 
O\left (\left (n^{\omega}+n^{3-(\delta+\epsilon)\cdot 2(3-\omega)/(5-\omega)}\right )\cdot n^{o(1)}\cdot \poly\log\tau\right ).
\end{split}
\] 

Since we assumed that $\epsilon+\delta<1$, we get that: 

\[n^{3-(\delta+\epsilon)\cdot 2(3-\omega)/(5-\omega)} \geq n^{3-2(3-\omega)/(5-\omega)} = n^{(9-\omega)/(5-\omega)} = n^{\omega+(3-\omega)^2/(5-\omega)} \geq n^{\omega}.\]

Therefore, if Event $\event$ does not happen, then the running time of the algorithm from \Cref{thm: triangle counting} on graph $H$ is bounded by:

\[
 O\left (\left (n^{3-(\delta+\epsilon)\cdot 2(3-\omega)/(5-\omega)}\right )\cdot n^{o(1)}\cdot \poly\log\tau\right ).
\]

If the algorithm from \Cref{thm: triangle counting} does not terminate within this running time, then we terminate it and return, for all $1\le i\leq q$, an arbitrary partition $(Y^{\ell}_i,Y^h_i)$ of $V(G)$, together with the edge set $E'_i=\emptyset$; notice that in this case, event $\event$ must have happened.

Finally, recall that the running time of the first step, that performed $q\leq n$ applications of the algorithm from \Cref{claim: heavy weight est} is $O(n^{2+\delta}\cdot\log^2n\cdot\log\tau)$
, and that the time required to compute the graph $H$, given the graph $G$ and the sets $\set{Z_i,Y_i^{\ell}}_{1\leq i\leq q}$ of vertices is bounded by $O(n^2)$. Therefore, the total running time of the algorithm is 
$O\left ( \left(n^{3-(\delta+\epsilon)\cdot 2(3-\omega)/(5-\omega)} + n^{2+\delta}\right)\cdot n^{o(1)}\cdot \poly\log\tau\right )$, or $O\left (\left(n^{3-0.478(\delta+\epsilon)+o(1)} + n^{2+\delta+o(1)}\right)\cdot \poly\log\tau\right )$ using the current bound $\omega\leq 2.371552$~\cite{WXXZ24}.
\section{Proofs Omitted from \Cref{sec: alg: unweigthed}}
\label{appx: proofs unweighted}

\subsection{Proof of \Cref{obs: preprocessing successful}}
\label{subsec: proof of good terminal selection}
Recall that we denoted by $\event_1$ the good event that  $\frac{\lambda} 2\leq \vol^+_G(L)\leq \lambda$ holds, and we showed that $\prob{\event_1}\geq\frac{1}{4\log n}$. We also denoted by $\event_2$ the good event that $t^*\in R$, and showed that, if $k\leq \hat k$, then  $\prob{\event_2}\geq \frac{1}{4}$. Lastly, we denoted by $\event_3$ the good event that either (i) $b=0$ and $\vol^+(L)\geq \frac{m}{4}$; or (ii) $b=1$ and $\vol^+(L)< \frac{m}{4}$, and we  showed that $\prob{\event_3}=\half$. Since we assumed that  $\vol^+_G(L)<\frac{m}{4}$, if the good event $\event_3$ happens, then $b=1$ must hold.

Consider now a single trial, where we select a vertex of $V(G)$ at random, and the probability for selecting a vertex $v$ is $\frac{\deg^+_G(v)}{m}$. Note that the probability that the selected vertex $v$ lies in $L$ is precisely $\frac{\vol^+_G(L)}{m}$.

Let $\event_4$ be the good event that $T_0\cap L\neq \emptyset$. From the above discussion:

\[\begin{split}
\prob{\neg \event_4\mid \event_1\band \event_3}&=\left(1-\frac{\vol^+_G(L)}{m}\right )^r\\
&=\left(1-\frac{\vol^+_G(L)}{m}\right )^{\ceil{m/\lambda}}\\
&\leq \left(1-\frac{\vol^+_G(L)}{m}\right )^{m/(2\vol^+_G(L))}\\
&\le \frac{1}{\sqrt {e}},
\end{split}\]

since, if Event $\event_1$ happened, $\lambda\leq 2\vol^+_G(L)$.
Lastly, we let $\event_5$ be the good event that $T'\cap L=\emptyset$. As before:

\[\begin{split}
\prob{\event_5\mid \event_1\band\event_3}&=\left(1-\frac{\vol^+_G(L)}{m}\right )^r\\
&=\left(1-\frac{\vol^+_G(L)}{m}\right )^{\ceil{m/\lambda}}\\
&\geq \left(1-\frac{\vol^+_G(L)}{m}\right )^{2m/\vol^+_G(L)}\\
&\geq \frac{1}{e^{8}}.
\end{split}\]

We have used the fact that, if Event $\event_1$ happened, then $\lambda\geq \vol^+_G(L)$. We also used the assumption that $\vol^+_G(L)< \frac{m}{4}$, and that, for all $x\geq 2$, $\left(1-\frac{1}{x}\right )^{x-1}\geq \frac{1}{e^2}$ holds (see e.g. Claim 2.9 in \cite{CT24}).

Overall, we get that, if $\vol^+_G(L)< \frac{m}{4}$ and $k\leq \hat k$, then:

\[\begin{split}
\prob{\event}&=\prob{\event_1\band\event_2\band\event_3\band\event_4\band\event_5}\\
&\geq \prob{\event_1}\cdot \prob{\event_2}\cdot\prob{\event_3} \prob{\event_4\mid \event_1\band\event_3}\cdot \prob{\event_5\mid \event_1\band\event_3}\\
&\geq \frac{1}{4\log n}\cdot \frac{1}{4}\cdot \frac{1}{2}\cdot \left(1-\frac 1 {\sqrt{e}}\right )\cdot \frac{1}{e^{8}}\\
&\geq \frac{1}{2^{64}\cdot \log n}.
\end{split}\]

\subsection{Proof of \Cref{obs:submodularity}}
\label{subsec: proof of submodularity of cuts}

The proof follows standard arguments and is almost identical to the proof of Observation A.4 in \cite{CT24} for undirected graphs.
We consider the contribution of every vertex $v\in V(G)$ to the expressions: 

\begin{equation}\label{eq 1}
w(N^+_G(A))+w(N^+_G(B)) 
\end{equation}

and

\begin{equation}\label{eq: 2}
w(N^+_G(A \cup B)) + w(N^+_G(A \cap B)).
\end{equation}

Notice that a vertex $v \notin N^+_G(A) \cup N^+_G(B)$ contributes $0$
to expression \ref{eq 1}. We now show that it also contributes $0$ to expression \ref{eq: 2}. Indeed, if no in-neighbor of $v$ lies in $A\cup B$, then $v$ may not lie in either of the sets $N^+_G(A \cup B)$ or $N^+_G(A \cap B)$, so its contribution to expression \ref{eq: 2} is $0$. Assume now that some in-neighbor $u\in N^-_G(v)$ of $v$ lies in $A$, but no in-neighbor of $v$ lies in $B$. Since $v\not\in N^+_G(A)$, $v\in A$ must hold. But then $v\in  A\cup B$, and no in-neighbor of $v$ may lie in $A\cap B$. Therefore, $v\not\in N^+_G(A \cup B)$ and $v\not \in N^+_G(A \cap B)$ must hold. The case where some in-neighbor of $v$ lies in $B$ but no in-neighbor of $v$ lies in $A$ is symmetric. Finally, assume that some in-neighbor of $v$ lies in $A$, and the same holds for $B$. Since $v\not \in  N^+_G(A) \cup N^+_G(B)$, it must be the case that $v\in A$ and $v\in B$. But then $v\in A\cap B$, so it may not lie in either of the sets $N^+_G(A \cup B)$ or in $N^+_G(A \cap B)$. To conclude, if $v \notin N^+_G(A) \cup N^+_G(B)$, its contribution to both expressions is $0$.

Assume now that $v\in N^+_G(A)\cup N^+_G(B)$. 
Observe that, if $v \in N^+_G(A) \cap N^+_G(B)$, then it contributes $2w(v)$ to the first expression, and its contribution to the second expression  is at most $2w(v)$. 

Assume next that $v\in N^+_G(A) \setminus N^+_G(B)$. 
Then $v$ contributes $w(v)$ to expression \ref{eq 1}. We show that in this case, either $v\not \in N^+_G(A\cap B)$, or $v\not \in N^+_G(A\cup B)$ must hold, so $v$ contributes at most $w(v)$ to expression \ref{eq: 2}. In order to do so, we consider two cases. The first case happens if  $v\not\in B$. In this case, no vertex of $N^-_G(v)$ may lie in $B$, so $v\not \in N^+_G(A\cap B)$. Assume now that the second case happens, so $v\in B$. In this case, $v\not\in N^+_G(A\cup B)$ must hold, since, from the definition, $N^+_G(A\cup B)$ does not contain vertices of $A\cup B$.

The only remaining case is where  $v \in N^+_G(B) \setminus N^+_G(A)$, and it follows a symmetric argument.

\section{Proofs Omitted from \Cref{sec: weighted non-dense}}
\label{sec: appx weighted non-dense}

\subsection{Proof of \Cref{claim: compute pairs2}}
\label{appx: subsec: compute pairs 2}

	The algorithm starts by selecting two integers $i\in \set{1,\ldots,\ceil{\log(n^{\eps})}}$ and $j\in \set{1,\ldots,\log(\wmax')}$  independently uniformly at random. It then sets $\lambda=2^i$ and $\tau=2^j$. Additionally, it performs $\ceil{\frac{100n\log n}{\lambda}}$ independent trials, as follows. 
	
	In every trial, a low-degree vertex $x\in \Vld$ is selected uniformly at random. 
	Let $A_x=V(G)\setminus \left (\set{x}\cup N^+_G(x)\right )$. Then a vertex $y_x\in A_x$
	is selected at random, where the probability of a vertex $v\in A_x$ to be selected is $\frac{w(v)}{w(A_x)}$. The resulting pair $(x,y_x)$ is then added to $\Gamma$; notice that $x\neq y_x$ and $(x,y_x)\not\in E(G)$ must hold. At the end, the algorithm returns $\lambda$, $\tau$ and $\Gamma$. 
	This completes the description of the algorithm. We now discuss its efficient implementation, and then analyze its success probability.

	\paragraph{Efficient implementation.}
	In order to execute the algorithm  efficiently, the main challenge is the efficient implementation of each trial: given a vertex $x\in V(G)$, we need to compute a random vertex $y_x\in A_x$, where the probability of a vertex $v\in A_x$ to be selected is $\frac{w(v)}{w(A_x)}$. In order to execute each such trial efficiently, we start with the following preprocessing step.
	
	We compute the value $\hat W=\sum_{v\in V(G)}w(v)$, and, for convenience, we denote $V(G)=\set{v_1,\ldots,v_n}$.
	We refer to a contiguous collection of integers from $\set{1,\ldots,n}$ as an \emph{interval}. We compute a hierarchical partition of the set $\set{1,\ldots,n}$ into collections of intervals. Let $k=\ceil{2\log n}$. For $0\leq k'\leq k$, we define a collection $\iset_{k'}$ of disjoint intervals, called \emph{level-$k'$ intervals}, that partition $\set{1,\ldots,n}$, as follows. We let the collection $\iset^0$ of level-$0$ intervals contain a single interval $\set{1,\ldots,n}$. Assume now that, for some $0\leq k'<k$, we computed the collection $\iset^{k'}$ of level-$k'$ intervals. For each level-$k'$ interval $I\in \iset^{k'}$, we create two level-$(k'+1)$ intervals $I',I''$, where $I'$ contains the first $\floor{\frac{|I|}{2}}$ elements of $I$ and $I''$ contains the remaining elements. If $|I|\leq 1$, then we let $I'=\emptyset$ and $I''=I$. It is immediate to verify that every level-$k$ interval contains at most one integer. For every level $0\leq k'\leq k$, for every level-$k'$ interval $I$, we compute the value $W(I)=\sum_{i\in I}w(v_i)$. Note that for all $0< k'\leq k$, we can compute the collection $\iset^{k'}$ of level-$k'$ intervals, given the collection $\iset^{k'-1}$ of level-$(k'-1)$ intervals, and the values $W(I)$ for every interval $I\in \iset^{k'}$, in time $O(n)$. Therefore, the preprocessing step can be implemented in time $O(n\log n)$. 
	
	Consider now a single trial of the algorithm, where a low-degree vertex $x\in V(G)$ is selected. Since $x$ is a low-degree vertex, $|N^+_G(x)|\leq d\cdot n^{\eps}$.
	We compute the value $W_x=\hat W-\sum_{v\in N^+_G(x)\cup\set{x}}w(v)$. For every level $0\leq k'\leq k$ and level-$k'$ interval $I\in \iset_{k'}$, if $I\cap \left ( N^+_G(x)\cup\set{x}\right )\neq\emptyset$, then we update the value $W(I)$ to ensure that it is equal to $\sum_{\stackrel{i\in I:}{v_i\in A_x}}w(v_i)$. In order to do so, we consider every vertex $v_j\in N^+_G(x)\cup\set{x}$ one by one. For each such vertex we identify, for every level $0\leq k'\leq k$, the unique level-$k'$ interval $I\in \iset_{k'}$ that contains the integer $j$, in time $O(\log n)$, by traversing the hierarchy of intervals from level $0$ to higher levels. For each such interval $I$, we decrease $W(I)$ by $w(v_j)$. In order to select a vertex $y_x$ from $A_x$, we select a value $0<\rho\leq W_x$ uniformly at random. We then let $i$ be the largest index for which $\sum_{\stackrel{v_j\in A_x:}{j<i}}w(v_j)<\rho$, and set $y_x=v_i$. It is easy to verify that, for every vertex $v_j\in A_x$, the probability to be selected as $y_x$ is exactly $\frac{w(v_j)}{W_x}=\frac{w(v_j)}{W(A_x)}$. Once the value $\rho$ is selected, the 
	corresponding vertex $y_x$ can be computed in time $O(\log n)$ in a straightforward manner, by traversing the hierarchy of intervals from lower to higher levels and exploiting the values $W(I)$ for all such intervals $I$ that we have computed. Recall that the preprocessing step takes time $\tilde O(n)$, and every trial can be implemented in time $\tilde O(dn^{\eps})$, since the selected vertex $x$ must be a low-degree vertex. Since the total number of trials is bounded by $\tilde O(n)$, the total running time of the algorithm is $\tilde O(n^{1+\eps}d)\leq \tilde O(mn^{\eps})$.

	\paragraph{Probability of success.}
	Next, we assume that there is a global minimum vertex-cut $(L,S,R)$ for which Properties \ref{prop: small L}--\ref{prop: small volume of L} hold. Let $\event_1$ be the good event that $\frac{\lambda}{2}\leq |L|\leq \lambda$ and $\tau=\tau^*$. It is easy to verify that the probability that the algorithm selects the value $\lambda$ with 
	$\frac{\lambda}{2}\leq |L|\leq \lambda$ is at least $\frac{1}{\ceil{2\log n}}$, and the probability that $\tau=\tau^*$ is selected is at least $\frac{1}{\log(\wmax')}$. Therefore, $\prob{\event_1}\geq \frac{1}{4\log n\cdot \log(\wmax')}$. 
	
	Next, we let $\event_2$ be the good event that in some trial, a vertex $x\in L$ was selected. Note that, since all vertices in $L$ are low-degree, if Event $\event_1$ happened, then the probability that in a single trial a vertex of $L$ is selected is at least $\frac{|L|}{n}\geq \frac{\lambda}{2n}$. Therefore:

	\[\prob{\neg \event_2\mid \event_1}\leq \left(1-\frac{\lambda}{2n}\right )^{\ceil{(100n\log n)/\lambda}}\leq e^{-25\log n},\]
	
	and $\prob{\event_2\mid \event_1}\geq \half$.
	
	Lastly, assume that both events $\event_1$ and $\event_2$ happened, and let $x$ be a vertex that was selected by any of the trials that lies in $L$. We define the good event $\event_3$ that the corresponding vertex $y_x$ lies in $R$. 
	
	Recall that vertex $y_x$ is selected from among the vertices of $A_x=V(G)\setminus \left(\set{x}\cup N^+_G(x)\right )$, where 
	the probability of a vertex $v\in A_x$ to be selected is $\frac{w(v)}{\sum_{u\in A_x}w(u)}$. 
	We partition the set $A_x$ into three subsets: set $A^{L}_x=A_x\cap L$; set $A^S_x=A_x\cap S$; and set $A^R_x=A_x\cap R$. Note that, since $x\in L$, $A^R_x=R\setminus\left(\set{x}\cup N^+_G(x)\right )=R$.
	
	Clearly, $w(A^L_x)\leq w(L)$.
	Moreover, since $A^S_x=S\setminus N^+_G(x)$, from \Cref{claim: neighbors in S}, $w(A^S_x)\leq w(L)$, and, from our assumption, $w(R)\geq w(L)$. Therefore, $w(A^R_x)=w(R)\geq w(L)\geq \frac{w(A_x)}{3}$, and the probability that $y_x\in R$ is at least $\frac{1}{3}$. We conclude that $\prob{\event_3\mid\event_1\band \event_2}\geq \frac{1}{3}$. 
	
	Overall, if all of the events $\event_1,\event_2$ and $\event_3$ happen, then the algorithm is successful. Therefore, the probability that the algorithm is successful with respect to a global minimum vertex-cut $(L,S,R)$ for which Properties \ref{prop: small L}--\ref{prop: small volume of L} hold is:
	
	\[\prob{\event_1\band\event_2\band\event_3}\geq \prob{\event_1}\cdot \prob{\event_2\mid \event_1}\cdot \prob{\event_3\mid \event\band\event_2}\geq \frac{1}{32\log n\cdot \log(\wmax)}. \]	

\subsection{Proof of \Cref{obs: min cut in residual}}
\label{sec: proof of min cut in residual}



Let $E'=\set{(v^{\inn},v^{\out})\mid v\in S}$. Since we assumed that Event $\event''$ did not happen, our algorithm never inserted an edge connecting a vertex of $X$ to $t$ into $G''$, and so $E_{G''}(X,Y)=E'$ holds at the end of the algorithm. Moreover $(X,Y)$ is the minimum $s$-$t$ edge-cut in $G''$.
Then, it follows from \Cref{fact : residual-capacity-of-edge-cut} that $(X,Y)$ is the minimum $s$-$t$ edge-cut in the residual network $\tilde H$.

\section{Proofs Omitted from \Cref{sec: weighted dense}}

\subsection{Proof of \Cref{cl: selecting pairs}}
\label{subsec: proof of select pairs 2}
We use Algorithm \algterm from \Cref{claim: compute terminals} to compute, in time $O(n)$, a value $\lambda$ and a collection  $T$ of at most $\ceil{\frac{100n\log n}{\lambda}}$ vertices of $G$. 

Next, we process every vertex $x\in T$ one by one. For each such vertex $x$, we compute the set of vertices $Z_x = V(G) \setminus \left(\set{x} \cup N_G^+(x)\right)$, and we select a vertex $y_x\in Z_x$ at random, where the probability to select each vertex $v \in Z_x$ is $\frac{w(v)}{w(Z_x)}$. The resulting pair $(x,y_x)$ is then added to $\Gamma$.
This completes the description of the algorithm.
It is immediate to verify that $|\Gamma|\leq \ceil{\frac{100n\log n}{\lambda}}$ as required. Note that in order to construct the set $\Gamma$, we perform $\tilde O(n/\lambda)\leq \tilde O(n)$ iterations, and each iteration can be executed in time $\tilde{O}(n)$. The total running time of the algorithm is then bounded by 
$\tilde O(n^2)$.

Next, we consider a fixed global minimum vertex-cut $(L,S,R)$.
Let $\event_1$ be the good event that Algorithm \algterm is successful with respect to $(L,S,R)$, and recall that, from \Cref{claim: compute terminals},
$\prob{\event_1}\geq \frac{1}{4\log n}$. Assume now that Event $\event_1$ happened
and let $x$ be an arbitrary vertex in $T\cap L$. Let $\event_2$ be the good event that the corresponding vertex $y_x$ lies in $R$.

Let $S'=S\setminus N^+_G(x)$. Then from \Cref{claim: neighbors in S},  $w(S')\leq w(L)$ must hold.
We conclude that $w(Z_x)\leq w(L)+w(R)+w(S')\leq 2w(L)+w(R)\leq 3w(R)$, since we always assume by default that, in a global minimum vertex-cut $(L,S,R)$, $w(L)\leq w(R)$ holds. Therefore, the probability that the algorithm chooses $y_x\in R$ is $\frac{w(R)}{w(Z)}\geq \frac{1}{3}$. We conclude that $\prob{\event_2\mid\event_1}\geq \frac{1}{3}$. 

Overall, if both of the events $\event_1$ and $\event_2$ happen, then Algorithm $\algpairs'$ is successful. Therefore, the probability that $\algpairs'$ is successful is at least:

\[\prob{\event_1\band\event_2}\geq \prob{\event_1}\cdot \prob{\event_2\mid \event_1}\geq \frac{1}{12\log n}. \]

\subsection{Proof of \Cref{claim: bound time of oracle}}
\label{subsec: runtime bound}
Recall that  $|\Gamma| \leq \ceil{\frac{100n\log n}{\lambda}} \leq \frac{200n\log n}{\lambda}$. Let $\epsilon = 1-\frac{\log \left(\frac{200n\log n}{\lambda}\right)}{\log n}$ so that $n^{1-\epsilon} = \frac{200n\log n}{\lambda}$. Notice that in every phase, the number of queries in the input to \Cref{lem: oracle in bulk} is bounded by $|\Gamma|\leq n^{1-\eps}$. Therefore, the running time of the algorithm from \Cref{lem: oracle in bulk} in a single phase is bounded by:

\[
\begin{split}
&O\left (\left(n^{3-(\delta+\epsilon)\cdot 2(3-\omega)/(5-\omega)+o(1)} + n^{2+\delta+o(1)}\right)\cdot \poly\log \tau\right )\\
&\quad\quad\quad\quad\leq  O\left (\left(n^{3-(\delta+\epsilon)\cdot 2(3-\omega)/(5-\omega)} + n^{2+\delta}\right)\cdot n^{o(1)}\cdot \poly\log W\right ).
\end{split}
\]

Recall that the parameters $\delta$ and $\epsilon$ were selected so that $n^{1-\delta} = 30000\lambda\gamma \cdot \log n$ and $n^{1-\epsilon} = \frac{200n\log n}{\lambda}$. Therefore:

\[n^{2-\delta - \epsilon} =cn\gamma\log^2n
\]

for some constant $c$, and so:

\[n^{\delta+\eps}=\frac{n}{c\gamma\log^2n}.\]

Therefore, we can bound:

\begin{equation}\label{eq: run time 1}
n^{2+\delta}\leq n^{2+\delta+\eps}\leq \tilde O\left (\frac{n^3}{\gamma}\right ).
\end{equation}

Additionally,
since $\gamma=\Theta\left(n^{(6-2\omega)/(11-3\omega)}\right)$,
it is easy to verify that $\left(\frac{n}{\gamma}\right)^{2(3-\omega)/(5-\omega)} = \Theta(\gamma)$ must hold. 
We can now lower-bound:

\[n^{(\delta+\epsilon)\cdot 2(3-\omega)/(5-\omega)}
= \left(\frac{n}{c\gamma\log^2n}\right )^{ 2(3-\omega)/(5-\omega)}\geq \frac{\gamma}{\poly\log n}.
\]

The total running time of the algorithm from \Cref{lem: oracle in bulk} over all $z\leq O(\log n\log W)$ phases can now be bounded by:

\[
\begin{split}
&O\left (\left(n^{3-(\delta+\epsilon)\cdot 2(3-\omega)/(5-\omega)} + n^{2+\delta}\right)\cdot z\cdot n^{o(1)}\cdot (\log W)^{O(1)}\right )\\
&\quad\quad\quad \leq O\left (\frac{n^{3+o(1)}\cdot (\log W)^{O(1)}}{\gamma}\right )\\
&\quad\quad\quad \leq O\left(n^{2+(5-\omega)/(11-3\omega)+o(1)}\cdot (\log W)^{O(1)}\right ), 
\end{split}
\]

since
 $\gamma = \Theta\left(n^{(6-2\omega)/(11-3\omega)}\right) = \Theta\left(n^{1-(5-\omega)/(11-3\omega)}\right)$.

\subsection{Proof of \Cref{claim: from quasi to promising}}
\label{subsec: appx: quasi to promising}
	Assume that the distinguished cut is good, and	let $A\subseteq V(\hat G)$ be a vertex set that is quasi-promising for $\hat G$.
	Recall that graph $\hat G$ is identical to graph $G$, except for the weights of its vertices. Therefore, $A$ is also a subset of vertices of $G$. From Property \ref{prop : Q1} of a quasi-promising set, $\vol^+_G(A)=\vol^+_{\hat G}(A) \leq \frac{100n^2}{\gamma}$, establishing Property \ref{prop : P'1} of a promising vertex set for $A$ in $G$. 
	
	Consider now the vertex set  $\Tilde{L} \subseteq A$  that is guaranteed to exist by the definition of a quasi-promising vertex set and the tripartition $(\Tilde{L},\Tilde{S},\Tilde{R})$ of $V(\hat G)$ induced by $\tilde L$. Clearly, $\tilde L\subseteq V(G)$, and $(\Tilde{L},\Tilde{S},\Tilde{R})$  is also a tripartition of $V(G)$ induced by $\tilde L$ in $G$. Moreover, since, from Property \ref{prop : Q4} of a quasi-promising vertex set, $(\Tilde{L},\Tilde{S},\Tilde{R})$ is an $x$-$y$ vertex-cut in $\hat G$, it must also be a valid $x$-$y$ vertex-cut in $G$, establishing Property \ref{prop : P'4} of a promising vertex set for $A$ in $G$.
	
	It now remains to establish Properties \ref{prop : P'2} and \ref{prop : P'3} from the definition of a promising vertex set for $A$ in $G$. 
	The key idea in showing both these properties is to observe that $w(\tilde S)-w(S)=w(\tilde S\setminus S)-w(S\setminus \tilde S)$, and, similarly, 
	$\hat w(\tilde S)-\hat w(S)=\hat w(\tilde S\setminus S)-\hat w(S\setminus \tilde S)$. Property \ref{prop : Q2} of a quasi-promising vertex set allows us to bound the latter expression by $6\lambda\cdot \gamma$. We will then use the following observation, that 
	bounds $w(\Tilde{S} \setminus S)$ in terms of $\hat{w}(\Tilde{S} \setminus S)$, together with an appropriate lower bound on $w(S\setminus \tilde S)$ in terms of $\hat{w}(S \setminus \tilde S)$, in order to bound $w(\tilde S)-w(S)$ in a way that allows us to establish each of these properties. We start with the following simple observation that bounds $w(\Tilde{S} \setminus S)$ in terms of $\hat{w}(\Tilde{S} \setminus S)$.

\begin{observation}\label{obs : tilde-S-contains-few-vertices-outside-of-S}
	$w(\Tilde{S} \setminus S)\leq M'\cdot \hat{w}(\Tilde{S} \setminus S)+2M'\cdot \lambda\gamma $.
\end{observation}
\begin{proof}
	We start by showing that $|\Tilde{S} \setminus S| \leq \lambda\cdot \gamma$ must hold. Indeed, 
	from the definition of a tri-partition induced by $\tilde L$, we get that $\Tilde{S} = N^+_{\hat{G}}(\Tilde{L}) = N^+_{G}(\Tilde{L})$. 
	Since, from Property \ref{prop : Q3} of a quasi-promising vertex set, $\Tilde{L} \subseteq L$, we get that $\tilde S=N^+_{G}(\Tilde{L})\subseteq L\cup N^+_G(L)= L\cup S$. But then $\tilde S\setminus S\subseteq L$, and, since we assumed that the distinguished cut is good, $|L|\leq \lambda$ must hold. We conclude that $|\Tilde{S} \setminus S| \leq \lambda\leq \lambda\cdot \gamma$ since $\gamma\geq 1$.
	
	Notice that, from Inequality \ref{eq: bound on new weight}, for every vertex $v\in V(G)$, $w(v)\leq M'\cdot \hat w(v)+ 2M'$. We conclude that:
	\[w(\Tilde{S} \setminus S)\leq M'\cdot \hat{w}(\Tilde{S} \setminus S)+ |\tilde S\setminus S|\cdot 2M'\leq  M'\cdot \hat{w}(\Tilde{S} \setminus S)+2\lambda\gamma M'.\]
\end{proof}

We are now ready to establish Property \ref{prop : P'2} of 	a promising vertex set for $A$ in $G$, namely that $w(\Tilde{S}) \leq w(S) + 8\lambda\gamma M'$. Equivalently, it is enough to prove that $w(\tilde S)-w(S)\leq 8\lambda\gamma M'$. Note that we can equivalently write 
$w(\tilde S)-w(S)=w(\tilde S\setminus S)-w(S\setminus \tilde S)$, and that, from \Cref{obs : tilde-S-contains-few-vertices-outside-of-S}, 
$w(\Tilde{S} \setminus S)\leq M'\cdot \hat{w}(\Tilde{S} \setminus S)+2M'\cdot \lambda\gamma $.
Moreover, from Inequality \ref{eq: bound on new weight}, for every vertex $v\in V(G)$, $w(v)\geq M'\cdot \hat w(v)$, and so $w(S\setminus \tilde S)\geq M'\cdot \hat w(S\setminus \tilde S)$. Therefore, altogether, we get that:

\[
\begin{split}
w(\tilde S)-w(S)&=w(\tilde S\setminus S)-w(S\setminus \tilde S)	\\
&\leq 2M'\cdot \lambda\gamma + M'\cdot \hat{w}(\Tilde{S} \setminus S)-M'\cdot \hat w(S\setminus \tilde S)\\
&=2M'\cdot \lambda\gamma+M'\cdot \left(\hat w(\tilde S)-\hat w( S)\right )\\
&\leq 8M'\cdot \lambda \gamma,
\end{split}
\]

where the last inequality follows from Property \ref{prop : Q2} of a quasi-promising vertex set. This establishes Property \ref{prop : P'2} of 	the promising vertex set for $A$ in $G$.

Lastly, we establish Property \ref{prop : P'3} of a promising vertex set for $A$ in $G$, namely, that set $(L\cup S)\setminus (\Tilde{L} \cup \Tilde{S})$ contains at most $12\lambda\gamma$ vertices $v$ with $w(v)\geq M'$.
Assume for contradiction that this is not the case, so 
$(L\cup S)\setminus (\Tilde{L} \cup \Tilde{S})$ contains more than $12\lambda \gamma$ vertices $v$ with $w(v)\geq M'$. Since we have assumed that the distinguished cut is good, $|L|\leq \lambda$ must hold, and so there must be more than $12\lambda \gamma-\lambda\geq 11\lambda \gamma$ vertices $v\in S\setminus \tilde S$ with $w(v)\geq M'$.

Recall that, for every vertex $v\in V(G)$, if $w(v)<2M'$ then  $\hat w(v)=0$, and otherwise, $\hat w(v)=\floor{\frac{w(v)-M'}{M'}}$. Therefore, in particular, if  $w(v)\geq M'$, then $M'\cdot \hat w(v)\leq w(v)-M'$. Since there are more than $11\lambda \gamma$ vertices $v\in S\setminus \tilde S$ with $w(v)\geq M'$, we get that: 

\[w(S\setminus \tilde S)\geq M'\cdot \hat w(S\setminus \tilde S)+ 11\lambda\gamma M'.\]

On the other hand, from \Cref{obs : tilde-S-contains-few-vertices-outside-of-S}:

\[w(\Tilde{S} \setminus S)\leq M'\cdot \hat{w}(\Tilde{S} \setminus S)+2M'\cdot \lambda\gamma.\]

By combining these inequalities, we get that:

\[\begin{split}
w(\tilde S)-w(S)&=w(\Tilde{S} \setminus S)-w(S\setminus \tilde S)
\\&\leq M'\cdot \left(\hat w(\tilde S\setminus S)-\hat w( S\setminus \tilde S)\right )-9M'\cdot \lambda \gamma\\
&=M'\cdot \left(\hat w(\tilde S)-\hat w( S)\right )-9M'\cdot \lambda \gamma\\
&<0,
\end{split} \]

where the last inequality follows from Property \ref{prop : Q2} of a quasi-promising vertex set. Altogether, we get that $w(\tilde S)<w(S)$ must hold. However, since $(\tilde L,\tilde S,\tilde R)$ is a valid $x$-$y$ vertex-cut in $G$, this contradicts the fact that $(L,S,R)$ is a minimum $x$-$y$ vertex-cut in $G$. We conclude that Property \ref{prop : P'3} of a promising vertex holds for $A$ in graph $G$.

\newpage
\bibliographystyle{alpha}
\bibliography{global-min-cut-p3}

\end{document}